\documentclass[thmsa,12pt]{article}
\normalfont\sffamily
\usepackage{amsfonts}
\usepackage{amssymb}
\usepackage[active]{srcltx}
\usepackage{color}
\usepackage[x11names]{xcolor}
\usepackage{amsmath}
\usepackage{amsthm}
\usepackage{bm} 
\usepackage{graphicx}
\usepackage{textcomp}
\usepackage{a4wide}
\usepackage{multirow}
\usepackage{diagbox}
\usepackage{color}
\usepackage[format=hang]{caption}
\usepackage{subcaption}
\usepackage{url}
\usepackage{xfrac}

\usepackage[normalem]{ulem}
\usepackage{comment}
\usepackage{pxfonts}
\usepackage[margin=2cm]{geometry} 
\usepackage{titlesec}

\usepackage{setspace}
\usepackage{footmisc}
\usepackage{accents}
\usepackage{appendix}

\usepackage{array,booktabs,ragged2e}
\newcolumntype{R}[1]{>{\RaggedRight}p{#1}}

\usepackage{cellspace}
\usepackage{hyperref}
\usepackage{fancyhdr}

\newtheorem{theorem}{Theorem}
\newtheorem{lemma}{Lemma}
\newtheorem{conjecture}{Conjecture}
\newtheorem{proposition}{Proposition}
\newtheorem{corollary}{Corollary}
\newtheorem{definition}{Definition}
\newtheorem{remark}{Remark}
\newtheorem{example}{Example}

\def \be {\begin{equation}}
\def \ee {\end{equation}}

\newcommand{\F}{\mathbb{F}}
\newcommand{\N}{\mathbb{N}}
\newcommand{\Q}{\mathbb{Q}}
\newcommand{\Z}{\mathbb{Z}}
\newcommand{\cA}{\mathcal{A}}
\newcommand{\cB}{\mathcal{B}}
\newcommand{\cM}{\mathcal{M}}
\newcommand{\cK}{\mathcal{K}}
\newcommand{\cL}{\mathcal{L}}
\newcommand{\cP}{\mathcal{P}}
\newcommand{\cS}{\mathcal{S}}
\newcommand{\cV}{\mathcal{V}}
\newcommand{\cX}{\mathcal{X}}
\newcommand{\PG}{\operatorname{PG}}

\newcommand{\wt}{\operatorname{wt}}
\newcommand{\supp}{\operatorname{supp}}

\newcommand{\qbin}[3]{\genfrac{[}{]}{0pt}{}{#1}{#2}_{#3}}

\begin{document}

\def\title #1{\begin{center}
{\Large {\sc #1}}
\end{center}}
\def\author #1{\begin{center} {#1}
\end{center}}

\setstretch{1.1}

\begin{titlepage}

\renewcommand{\thefootnote}{\fnsymbol{footnote}}\addtocounter{footnote}{1}
\title{\sc 
Additive codes attaining the Griesmer bound \\ \medskip \  }

\author{Sascha Kurz 
	\\ {\small Dept.\ of Mathematics, University of Bayreuth, 95440 Bayreuth, Germany\\email: sascha.kurz@uni-bayreuth.de}}

\vspace{0.3cm}

\begin{center} {\tt 
\hspace{-1.9em} May 08, 2026 
} 
\end{center}

\vspace{0.3cm}

\begin{center} {\bf {\sc Abstract}} \end{center}
{\small 
Additive codes may have better parameters than linear codes. However, still very few cases are known and the explicit construction of such codes is 
a challenging problem. Here we show that a Griesmer type bound for the length of additive codes can always be attained with equality if the minimum distance 
is sufficiently large. This solves the problem for the optimal parameters of additive codes when the minimum distance is large and yields many infinite series  
of additive codes that outperform linear codes.  
} 

\vspace{0.2cm}

\begin{description}
{\small
\item[Keywords:] additive codes $\cdot$ linear codes $\cdot$ Griesmer bound $\cdot$ Galois geometry\\[-7mm]
\item[2020 Mathematics Subject Classification:] Primary 94B05; Secondary 51E21, 94B27.
}
\end{description}


\vspace{2cm}

\vfill
\noindent {\footnotesize I am grateful to Simeon Ball for feedback on earlier drafts and generously sharing his insights on additive codes. Further thanks go to Ferdinand Ihringer. I also benefited from
feedback on presentations at the workshops {\lq\lq}Finite Geometries 2025{\rq\rq} in Irsee and {\lq\lq}New Mathematical Directions in Coding Theory{\rq\rq} in Oberwolfach.}

\end{titlepage}

\addtocounter{footnote}{-1}
\pagenumbering{roman}

\setstretch{1.26} 

\tableofcontents

\pagebreak

\pagenumbering{arabic}
\pagestyle{fancy}
\fancyhf{}
\fancyheadoffset{0cm}
\renewcommand{\headrulewidth}{0pt}
\renewcommand{\footrulewidth}{0pt}
\fancyhead[R]{
  \color{lightgray}\leftmark
  }
\fancyfoot[C]
{
\thepage
}
\fancypagestyle{plain}{%
  \fancyhf{}%
  \fancyhead[R]{\thesection \thepage}%
  }


\section{Introduction}\label{sec:Introduction}

For a finite set $\cA$, called \emph{alphabet}, a \emph{code} $C$ of \emph{length} $n$ and \emph{minimum distance} $d$ is a subset of $\cA^n$ such that any two elements,
called \emph{codewords}, differ in at least $d$ positions. Given parameters $n$, $d$, and the alphabet size $\#\cA$, the aim is to maximize the code size $\# C$. For some
prime power $q$ consider the finite field $\F_q$ as alphabet. An $[n,k,d]_q$ code $C$ is a $k$-dimensional subspace
of the vector space $\F_q^n$ with minimum distance $d$. The number of codewords is given by $\#C=q^k$. We also say that $C$ is \emph{linear} over $\F_q$, since $C$ is linearly closed, i.e., for 
every $u,v\in C$ and every $\alpha,\beta\in\F_q$ we have $\alpha u+\beta v\in C$. The parameters of an $[n,k,d]_q$ code $C$ are related by the so-called
\emph{Griesmer bound} \cite{griesmer1960bound,solomon1965algebraically}
\begin{equation}
  \label{eq_griesmer_bound}
  n\ge \sum_{i=0}^{k-1} \left\lceil\frac{d}{q^i}\right\rceil=:g_q(k,d).
\end{equation}   
Interestingly enough, this bound can always be attained with equality if the minimum distance $d$ is sufficiently large and a nice geometric construction was given by Solomon and Stiffler \cite{solomon1965algebraically}.

If a code $C$ is additively closed for alphabet $\cA=\F_q$, i.e., $u+v\in C$ for all $u,v\in C$, we say that $C$ is \emph{additive} over $\F_q$, see e.g.~\cite{bierbrauer2000quantum,calderbank1998quantum,delsarte1973algebraic,kumar2001f,huffman2013theory}.
There indeed exist parameters $n$, $d$, $q$ such that each linear code $C_1$ with length $n$, minimum distance $d$, and alphabet $\F_q$ satisfies
$\# C_1<\# C_2$ for a suitable additive code $C_2$ with length $n$, minimum distance $d$, and alphabet $\F_q$. In this situation some authors say that additive codes outperform linear codes.\footnote{For examples of general block codes outperforming linear codes see Remark~\ref{remark_outperform_blockcode}.}
Typically $k:=\log_q\# C_2$ is fractional, so that no $[n,k,d]_q$ code can exist. Indeed, very few cases where additive codes outperform linear codes and $k$ is integral are known.
In \cite{clerck2001perp} an additive code with length $n=21$, minimum distance $d=18$, alphabet $\cA=\F_9$, and size $9^3$ is given. However, the largest linear code with length $n=21$,
minimum distance $d=18$, and alphabet $\cA=\F_9$ has size $9^2$. More precisely, the largest $n$ such that an $[n,3,n-3]_9$ code exists is $n=17$ \cite{barnabei1978small}. So,
no $[21,3,18]_9$ code exists. Instead of maximizing the size of the code we can also minimize its length while fixing the other parameters. In \cite{guan2023some} several additive codes
over $\F_4$ that outperform the best known linear codes were constructed using cyclic codes. E.g.\ for length $n=63$ and minimum distance $d=45$ size $4^5$ can be achieved while the 
existence of an $[63,5,45]_4$ code is unknown. Recently, in \cite{kurz2024optimal} four infinite series and five sporadic examples of additive codes, with size $4^4$ and
alphabet $\F_4$, that outperform linear codes with the same length and minimum distance were constructed.

There is some renewed interest in additive codes due to applications in the construction of quantum codes, see e.g.\ \cite{dastbasteh2024new,polynomialRepresentation,grassl2021algebraic,li2024ternary}. In the context of MDS codes additive codes where studied e.g. in \cite{blaum1999lowest} as subclasses of group codes \cite{forney1992hamming}, see also \cite{blaum1996mds,cardell2012constructions,li2025some,louidor2006lowest}.
In \cite{guruswami2008explicit} \emph{Folded Reed--Solomon codes} were introduced with the property that they  admit much better list-decoding
algorithms than the original Guruswami--Sudan algorithm \cite{guruswami1999improved}. Another subclass of additive codes
allowing improved decoding algorithms, see e.g.~\cite{tamo2024tighter}, and also obtained by some folding (or grouping) construction is given by so-called \emph{multiplicity codes}
\cite{kopparty2014high,rosenbloom1997codes}. Also the notion of the \emph{folded Hamming distance} appears in the literature, see \cite{martinez2025linear}.

The aim of this paper is to show that a Griesmer type bound for additive codes, see \cite[Theorem 12]{ball2024additive} or Lemma~\ref{lemma_indirect_upper_bounds} for details, can always be attained with equality if the minimum distance $d$ is sufficiently large. This explains the mentioned four infinite series and gives many more such examples. The underlying construction generalizes the
Solomon--Stiffler construction and will be formulated in geometric terms. For relatively small minimum distances the problem of the determination of the optimal parameters of additive codes is widely open and a challenging research direction, as it is for linear codes. Here we restrict our considerations to linear and additive codes over finite fields. However, similar questions also arise for codes over chain rings, see e.g.\ \cite{jose2024eisenstein,shiromoto2001griesmer}. Several non-linear binary codes can be constructed via $\Z_4$-linear codes, see e.g.\ \cite{hammons1994z} and Remark~\ref{remark_z4}.

For other constructions attaining the Griesmer bound for linear codes we refer e.g.\ to \cite{belov1974construction,chen2024griesmer,helleseth1981characterization,helleseth1983new,hyun2024griesmer}. For constructions of additive codes with
good parameters we refer e.g.\ to \cite{arrieta2021go,guan2023some,10693309,panario2024additive,sharma2024some}. Many constructions are based on cyclic codes or generalizations thereof, see e.g.\ \cite{guneri2018additive,shi2024additive}. For bounds and constructions of codes over general alphabets we refer e.g.\ to \cite{ashikhmin1998delsarte,bogdanova2001error,bogdanova2001bounds}. Additive codes have e.g.\ applications in quantum information \cite{calderbank1998quantum,ketkar2006nonbinary,shor1995scheme,steane1996multiple}, computer memory systems \cite{chen1992symbol,chen1984error}, deep space communication \cite{hattori1998subspace}, storage systems \cite{blaum1996mds,blaum1998array}, secret sharing \cite{kim2017secret}, and distance-regular graphs  \cite{shi2019new}.

Other subclasses of more structured block codes, besides linear and additive codes, are e.g.\ almost affine codes \cite{simonis1998almost}. In Section~\ref{sec_almost_affine_codes} we give a brief definition and mention relations to larger classes of block codes.

The remaining part of this paper is structured as follows. In Section~\ref{sec_preliminaries} we introduce the necessary preliminaries. In Section~\ref{sec_solomon_stiffler} we generalize the Solomon--Stiffler construction for linear codes to additive codes. Our main result is
Theorem~\ref{thm_attained_asymptotically}, or the more explicit version in Corollary ~\ref{cor_attained_asymptotically}, stating that a Griesmer upper bound for additive codes can always be attained with equality if the minimum distance
$d$ is sufficiently large. We list some parameterized series of improvements for additive codes over linear codes in Table~\ref{table_improvements}. More extensive data is moved to appendices ~\ref{sec_parameterized_outperform} and \ref{generic_results}. Relations between our problem and linear equation systems over $\Z$ are outlined in Section~\ref{sec_SNF}. Results on optimal additive codes for small parameters are summarized in Section~\ref{sec_small_parameters}. Our generalization of the Solomon--Stiffler construction involves a rather strong technical assumption which is relaxed in Appendix~\ref{sec_generalization_type}. A relation of some bounds for additive codes to divisible codes is briefly outlined in Appendix~\ref{sec_divisible_multisets}. Additive two-weight codes are considered in Appendix~\ref{sec_two_weight}. Examples of additive codes that have been found by computer searches are stated explicitly in Section~\ref{sec_searches}.

\pagebreak

\section{Preliminaries}
\label{sec_preliminaries}

In this section we collect the necessary preliminaries. I.e., we introduce the
coding theoretic notation in Subsection~\ref{subsec_coding_notation} and the
geometric notation in Subsection~\ref{subsec_geometric_notation}. Basic constructions and bounds are summarized in Subsection~\ref{subsec_basic}. None of this is essentially new. However, since different notions have been used in the literature and in order to keep the paper self-contained, we provide short proofs or explanations. In Subsection~\ref{subsec_asymptotic} we formalize the idea of asymptotic results which hold when the minimum distance is sufficiently large.

\subsection{Coding theoretic notation}
\label{subsec_coding_notation}
Let $\F_q$ denote the finite field with $q$ elements, where $q=p^l$ is a prime power.
We call the prime $p$ the characteristic of $\F_q$. An additive code $C$ of length $n$
over the alphabet $\cA=\F_{q'}$ is a subset of $\F_{q'}^n$ such that $u+v\in C$
for all $u,v\in C$. It turns out that each code $C$ that is additive over $\F_{q'}$ is
linear over some subfield $\F_{q}\le\F_{q'}$, i.e., $\alpha u+\beta v\in C$ for all
$u,v\in C$ and all $\alpha,\beta\in\F_{q}$ \cite{ball2023additive,ball2024additive}.
So, we use the notation $[n,r/h,d]_q^h$ for an additive code $C$ that is linear
over $\F_q$ and has length $n$, minimum distance $d$,  alphabet
$\cA=\F_{q^h}$, and size $q^r$, where $r\in \N$. We also call
$k=r/h\in\Q$ the \emph{dimension} of $C$, so that $\#C=\#\cA^k$
and an $[n,k,d]_q^1$ additive code is an $[n,k,d]_q$ linear code. Note
that $k$ can be fractional.

An $[n,k,d]_q$ linear code $C$ can be defined as the rowspace of a $k\times n$ matrix
with entries in $\F_q$, called a \emph{generator matrix} for $C$. Similarly,  an
$[n,r/h,d]_q^h$ additive code $C$ can be defined as the $\F_q$-space spanned by the
rows of an $r\times n$ matrix $G$ with entries in $\F_{q^h}$,  again called a
\emph{generator matrix} for $C$. Let $\cB$ be a basis for $\F_{q^h}$ over
$\F_q$ and write out the elements of $G$ over the basis $\cB$ to obtain
an $r\times nh$ matrix $\widetilde{G}$ with entries from $\F_q$. Here we assume that the
columns of $\widetilde{G}$ are grouped together into $n$ groups of $h$ columns and say
that $\widetilde{G}$ is a \emph{subfield generator matrix} for $C$.

\begin{example}
  \label{ex_generator_matrix}
  Write $\F_4\simeq \F_2[\omega]/\left(\omega^2+\omega+1\right)$ and consider the linear
  code $C$ with generator matrix
  $$
    \begin{pmatrix}
      0  & 1 & 1 & 1      & 1        \\
      1  & 0 & 1 & \omega & \omega^2
    \end{pmatrix}.
  $$
  It can be easily checked that $C$ is a $[5,2,4]_4$ code. If we interprete $C$ as an $[5,4/2,4]_4^2$ additive code a generator matrix is e.g.\ given by
  $$
    G=\begin{pmatrix}
      0  & 1 & 1 & 1      & 1        \\
      0  & \omega & \omega & \omega      & \omega        \\
      1  & 0 & 1 & \omega & \omega^2 \\
      \omega  & 0 & \omega & \omega^2 & 1
    \end{pmatrix}.
  $$
  Here we have
  $$
    \widetilde{G}=\begin{pmatrix}
      00  & 10 & 10 & 10 & 10 \\
      00  & 01 & 01 & 01 & 01 \\
      10  & 00 & 10 & 01 & 11 \\
      01  & 00 & 01 & 11 & 10
    \end{pmatrix}
  $$
  choosing the basis $\cB=(1,\omega)$ and using $\omega^2=1+\omega$.
\end{example}

Given a subfield generator matrix $\widetilde{G}$ for an additive code $C$, the column space of each block of $h$ columns defines an $\F_q$-subspace of dimension at most $h$.
By $\cX_G(C)$ we define the multiset of the $n$ subspaces spanned by the
$n$ blocks of $h$ columns of $\widetilde{G}$ in this way.
We say that $C$ is \emph{faithful} if all elements of $\cX_G(C)$ have dimension
$h$ and \emph{unfaithful} otherwise.
This is indeed a property of the code $C$ and does not
depend on the choice of a generator matrix $G$ or the choice of a basis
$\cB$, see e.g.\ \cite{ball2024additive}, so that we also write $\cX(C)$. We remark that a linear code $C$, i.e.\ an additive code with
$h=1$ is unfaithful iff an arbitrary generator matrix for $C$ contains a column
consisting entirely of zeroes. An unfaithful additive code for $h=2$ is given in
Example~\ref{ex_projection}.

Given a linear code $C$ the \emph{weight} of a codeword $c\in C$ is the number of
non-zero entries in $c$, i.e.\ there is exactly one codeword of weight zero
and besides that the minimum possible weight equals the minimum distance of
the code.\footnote{The same is true for additive codes, see
\cite[Lemma 3]{ball2023additive}.} We call $C$ \emph{$\Delta$-divisible} if the
weight of each codeword is divisible by $\Delta\in\N$. The \emph{maximum weight}
of $C$ is the maximum of the weights of the codewords of $C$. We say that
$C$ is a \emph{$t$-weight} code if only $t$ different non-zero weights occur.
one-weight codes are repetitions of simplex codes, i.e.\ exactly those
codes that attain the upper bound from Lemma~\ref{lemma_one_weight_bound},
see Theorem~\ref{thm_partition}. A linear code $C$ is \emph{projective} if the
minimum distance of its dual code is at least $3$, i.e., if each pair of columns of a generator matrix of $C$ is linearly independent. In geometric terms the
latter property means that each point in the corresponding multiset of points
has multiplicity at most one, see the subsequent subsection. Projective two-weight
codes have received at lot of attention, see e.g.\
\cite{brouwer2021two,calderbank1986geometry}.
\begin{lemma}(\cite[Corollary 2]{delsarte1972weights})
  \label{lemma_two_weight}
  Let $0<w_1<w_2$ be the two non-zero weights of a projective linear code over
  $\F_q$. Then, there exist positive integers $u$, $t$ such that $w_1=up^t$ and
  $w_2=(u+1)p^t$, where $p$ is the characteristic of $\F_q$.
\end{lemma}
A similar statement also holds for nonprojective two-weight codes, see
\cite[Theorem 3]{ward1999introduction} and \cite{kurz2024non}. General two-weight codes were e.g.\ studied in \cite{boyvalenkov2021two}. If $C$ is an
$[n,k]_q$ code where all non-zero weights are contained in $\left\{w_1,\dots,w_t\right\}$, then we also speak of an $\left[n,k,\left\{w_1,\dots,w_t\right\}\right]_q$ code. There is a vast literature on linear codes with few weights, but rather little seems to be known for additive codes with few weights, see e.g.~\cite{panario2024additive}. We remark that a faithful
projective $h-(n,r,s,\mu)_q$ system with type $\sigma[r]-\sum_{i=1}^{r-1}\varepsilon_i[i]$, see Definition~\ref{def_system}  and Definition~\ref{def_partitionable}, where $t$ of the
$\varepsilon_i$ are non-zero, corresponds to an additive $(t+1)$-weight code. In Section~\ref{sec_solomon_stiffler} we give constructions for such codes. A few more observations on the special case of additive two-weight codes are given in Section~\ref{sec_two_weight} in the appendix.

\subsection{Geometric notation}
\label{subsec_geometric_notation}
The set of all subspaces of $\F_q^r$ , ordered by the incidence relation
$\subseteq$, is called \emph{$(r-1)$-dimensional projective geometry over
$\F_q$} and denoted by $\PG(r-1,q)$. Employing this algebraic notion
of dimension instead of the geometric one, we will use the term $i$-space
to denote an $i$-dimensional subspace of $\F_q^r$. To highlight the
important geometric interpretation of subspaces we will call $1$-, $2$-,
and $(r-1)$-spaces points, lines, and hyperplanes, respectively. For two
subspaces $S$ and $S'$ we write $S\le S'$ if $S$ is contained in $S'$. Moreover,
we say that $S$ and $S'$ are \emph{incident} iff $S\le S'$ or $S\ge S'$.
Let $[i]_q:=\tfrac{q^i-1}{q-1}$ denote the number of points
of an arbitrary $i$-space in $\PG(r-1,q)$ where $r\ge i$. By convention we
set $[0]_q:=0$. We have the
following obvious but useful observations.
\begin{lemma}
  For $b\ge a\ge 1$ we have
  \begin{equation}
    \label{cross_difference}
    [a]_q[b-1]_q-[a-1]_q[b]_q=q^{a-1}\cdot [b-a]_q,
  \end{equation}
  using the convention $[0]_q=0$, and
  \begin{equation}
    \label{eq_gcd}
    \gcd\!\left([a]_q,[b]_q\right)=\left[\gcd(a,b)\right]_q.
  \end{equation}
\end{lemma}

A \emph{premultiset of points} $\cM$ of $\PG(r-1,q)$ is a mapping from the
set of points of $\PG(r-1,q)$ to $\Z$. The value $\cM(P)$ is called the
\emph{multiplicity} of $\cM$. We extend this notation additively to each subspace
$S$ via $\cM(S)=\sum_{P\le S} \cM(P)$. So, the cardinality of $\cM$ equals
$\cM(V)$, where $V$ denotes the ambient space. For a hyperplane $H$
the number of points of $\cM$ in $H$ is given by $\cM(H)$. If $\cM(P)\in \N$
for each point $P$, we say that $\cM$ is a \emph{multiset of points}. For each
subspace $S$ we denote its \emph{characteristic function} by $\chi_S$,
i.e.\ $\chi_S(P)=1$ if $P\le S$ and $\chi_S(P)=0$ otherwise.

Multisets of points can be generalized as follows,
cf.~\cite[Definition 4]{ball2024additive}.
\begin{definition}
  \label{def_system}
  A projective $h-(n,r,s)_q$ system is a multiset $\cS$ of $n$ subspaces
  of $\PG(r-1, q)$ of dimension at most $h$ such that each hyperplane
  contains at most $s$ elements of $\cS$ and some hyperplane contains
  exactly $s$ elements of $\cS$. We say that $\cS$ is faithful if all
  elements have dimension $h$. A projective $h-(n,r,s)_q$ system $\cS$
  is a projective $h-(n,r,s,\mu)_q$ system if each point is contained
  in at most $\mu$ elements from $\cS$ and there is some point that is
  contained in exactly  $\mu$ elements from $\cS$.
\end{definition}

Note that the elements of $\cS$ span the entire ambient space
$\PG(r-1,q)$ iff $s<n$. It is well known that a projective
$1-(n,r,s)_q$ system $\cS$ with $s<n$ is in one-to-one correspondence to
a linear $[n,r,n-s]_q$ code $C$ which has full length iff $\cS$ is
faithful, see e.g.\ \cite[{\S}1.1.2]{tsfasman1991agcodes} or \cite{bose1966characterization,burton1964application,dodunekov1998codes,macwilliams1962combinatorial,slepian1956class}.\footnote{Every $ \mathbb{F}_q $-linear code is
also equivalent to a multi-twisted code \cite{multitwisted}. There is also a
relation to subsets of vector spaces with pairwise different linear
combinations, see \cite[Theorem 3.1]{pantoja2024s_h} for details and
\cite[Definition 3.2]{pantoja2024s_h} for the concept of an $S_h$-linear set.}
Each hyperplane
$H$ that contains $i$ elements of $\cS$ corresponds to $q-1$ codewords of
weight $n-i$.
In general, projective $h-(n,r,s)_q$ systems (with $s<n$) are in
one-to-one correspondence
to additive codes:

\begin{theorem}(\cite[Theorem 5]{ball2024additive})
   \label{thm_connection}
   If $C$ is an additive $[n,r/h,d]_q^h$ code with generator matrix $G$, then
   $\cX_G(C)$ is a projective $h-(n, r, n-d)_q$ system $\cS$, and conversely,  each projective $h-(n,r, s)_q$ system $\cS$ defines an additive $[n,r/h,n-s]_q^h$ code $C$.
\end{theorem}
As mentioned before, $C$ is faithful iff $\cS$ is faithful. Whenever the
specific choice of a generator matrix $G$ or a subfield generator matrix
$\widetilde{G}$ is irrelevant, we write $\cS=\cX(C)$ or $C=\cX^{-1}(\cS)$.
We also use the notion $C=\cX^{-1}(\cM)$ for a multiset of points $\cM$
interpreting $\cM$ as a faithful projective $1-(n,r,s)_q$ system. In the case
where $n=s>0$ let $R$ be the $r'$-space spanned by the points of positive
multiplicity in $\cM$ and consider the multiset of points $\cM'$ in $\PG(r'-1,q)$
as the restriction of $\cM$ to $K$ instead. We say that
a (pre-)multiset of points $\cM$ in $\PG(r-1,q)$ is
\emph{$\Delta$-divisible} for some positive integer $\Delta$ iff
\begin{equation}
  \#\cM\equiv \cM(H) \pmod\Delta
\end{equation}
for every hyperplane $H$ in $\PG(r-1,q)$ if $r\ge 2$ and for $r=1$ iff
$\#\cM\equiv 0\pmod \Delta$. Note that $\cM$ is $\Delta$-divisible
iff the linear code $\cX^{-1}(\cM)$ is $\Delta$-divisible. We collect a
few basic properties in the following lemma and refer e.g.\ to the surveys
\cite{kurz2021divisible,ward2001divisible} for more details.
\begin{lemma}
  \label{lemma_divisible_properties}
  Let $\cM,\cM'$ be two $\Delta$-divisible (pre-)multisets of points in
  $\PG(r-1,q)$ for $r\ge 2$. Then, $\cM+\cM'$, $\cM-\cM'$ are also
  $\Delta$-divisible and $\lambda\cdot\cM$ is $\Delta\cdot\lambda$-divisible
  for every positive integer $\lambda$. Moreover, the characteristic
  function  $\chi_S$ of every $i$-space $S$ is $q^{i-1}$-divisible.
\end{lemma}

Note that each projective $h-(n,r,s)_q$ system $\cS$ can be modified to a
faithful projective $h-(n,r,\le s)_q$ system $\cS'$ by replacing each
element $S\in\cS$ with dimension less than $h$ by an arbitrary
$h$-space containing $S$ if $r\ge h$. We remark that it is also possible
to obtain a faithful projective $h-(n,r,s)_q$ system $\cS'$ if we
choose the replacing $h$-spaces carefully. On the other side, the knowledge that
a projective system is unfaithful sometimes allows to deduce tighter bounds on
its parameters, see e.g.\ the proof of Lemma~\ref{coding_bound_improved_e1}.

To each faithful projective $h-(n,r,s)_q$ system we can
also associate a multiset of points and a linear
code over $\F_q$ with certain properties.
\begin{definition}
  For a faithful projective  $h-(n,r,s,\mu)_q$ system $\cS$ let
  $\cP(\cS)$ denote the multiset of points that we obtain
  by replacing each element of $\cS$ by its contained $[h]_q$ points.
\end{definition}
\begin{lemma}
  \label{lemma_linear_code}
  Let $\cS$ be a faithful projective  $h-(n,r,s,\mu)_q$ system.
  Then $\cP(\cS)$ is a faithful projective $1-(n',r,s',\mu)_q$ system,
  where $n'=n[h]_q$ and $s'=n\cdot[h-1]_q+s\cdot  q^{h-1}$. 
  Moreover,
  $C:=\cX^{-1}(\cP(\cS))$ is a  
  $q^{h-1}$-divisible linear $[n',r,d']_q$ code $C$ with maximum weight at most $n\cdot q^{h-1}$, where $d'=q^{h-1}\cdot (n-s)$.
\end{lemma}
\begin{proof}
  As an abbreviation we set $\cS':=\cP(\cS)$. By construction we
  have $n'=\#\cS'=\#\cS\cdot [h]_q=n[h]_q$. If a hyperplane $H$ contains
  $0\le i\le s$ elements from $\cS$, then $H$ contains
  $i\cdot[h]_q+(n-i)\cdot[h-1]_q=n\cdot[h-1]_q+i\cdot q^{h-1}$ elements
  from $\cS'$, so that  $s'=n\cdot[h-1]_q+s\cdot  q^{h-1}$. Moreover,
  each hyperplane $H$ with $\cS'(H)=n\cdot[h-1]_q+i\cdot q^{h-1}$
  corresponds to $q-1$ codewords of weight $q^{h-1}\cdot(n-i)$, so that
  $C$ is $q^{h-1}$-divisible, has minimum weight $d'$, and
  a maximum weight of at most $n\cdot q^{h-1}$.
\end{proof}

\begin{example}
  \label{ex_derived_linear_code}
 Consider the $[5,4/2,4]_4^2$ additive code given by the generator matrix $G$ from
 Example~\ref{ex_generator_matrix} and the corresponding projective $2-(5,4,1)_2$
 system $\cS$. The linear $[15,4,8]_2$ code $C$ constructed from $\cS$ in
 Lemma~\ref{lemma_linear_code} has e.g.\
 $$
   \begin{pmatrix}
      000  & 101 & 101 & 101 & 101 \\
      000  & 011 & 011 & 011 & 011 \\
      101  & 000 & 101 & 011 & 110 \\
      011  & 000 & 011 & 110 & 101
    \end{pmatrix}
 $$
 as a generator matrix, where we group the columns into blocks of size $[h]_q=3$
 corresponding to the points of the five lines. Here the maximum weight equals the
 minimum distance, i.e.\ every hyperplane contains exactly one element from $\cS$.
\end{example}

Sometimes the stated restrictions on the weights of the linear code $C$ in
Lemma~\ref{lemma_linear_code} can be used to prove tailored upper bounds
like e.g.\ $n_3(8,3;3)\le 20$, see Lemma~\ref{coding_bound_improved_e1}.
In \cite[Section 4.1]{kurz2024optimal} the restrictions on the weights were used to deduce $n_2(8,2;15)\le 55$, $n_2(8,2;28)\le 108$, and $n_2(8,2;29)\le 113$. Lemma~\ref{lemma_linear_code} can also be adjusted to the case of a non-faithful projective system $\cS$ once the dimension distribution of the elements of $\cS$ is known.

For $h>1$ it is an interesting question whether for a given multiset of points
$\cM$ in $\PG(r-1,q)$ a faithful projective $h-(n,r,s,\mu)_q$ system $\cS$ with
$\cP(\cS)=\cM$ exists, see e.g.~\cite{kim2003projections} for some special
cases. In general, this is a very hard problem and we can only
give an {\lq\lq}asymptotic{\rq\rq} answer, see Subsection~\ref{subsec_asymptotic} for definitions, Theorem~\ref{thm_main} for an explicitly parameterized solution in
a special case, and Lemma~\ref{lemma_partitonable_les} for a characterization
based on the solvability of linear equation systems over $\Z$. However, in
the case of existence the parameters of $\cS$ can be computed from $\cM$.
\begin{definition}
  \label{def_partitionable_multiset}
  Let $\cM$ be a multiset of points in $\PG(r-1,q)$. We say that $\cM$
  is \emph{$h$-partitionable} if there exists a faithful projective
  $h-(n,r,s,\mu)_q$ system $\cS$ with $\cP(\cS)=\cM$, i.e.\
  $\cM=\sum_{S\in\cS} \chi_S$. We also say that $\cS$ has \emph{type $\cM$}.
\end{definition}
\begin{lemma}
  \label{lemma_system_parameters_from_multiset}
  Let $\cM$ be a multiset of points in $\PG(r-1,q)$ and $\cS$ be a
  faithful projective $h-(n,r,s,\mu)_q$ system with $\cP(\cS)=\cM$,
  then we have
  \begin{equation}
    \label{eq_n_general}
    n=\#\cM/[h]_q,
  \end{equation}
  \begin{equation}
    \label{eq_s_general}
    s=\frac{\left(\max_{H:\dim(H)=r-1}\cM(H)\right)\cdot[h]_q-\#\cM\cdot[h-1]_q }{q^{h-1}\cdot[h]_q},
  \end{equation}
  \begin{equation}
    \label{eq_d_general}
    d:=n-s=\frac{\#\cM-\max_{H:\dim(H)=r-1} \cM(H)}{q^{h-1}},
  \end{equation}
  \begin{equation}
    \label{eq_mu_general}
    \mu=\max_{P:\dim(P)=1} \cM(P),
  \end{equation}
  and
  \begin{equation}
    \label{eq_div_general}
    \#\cM\equiv \cM(H) \pmod{q^{h-1}}
  \end{equation}
  for each hyperplane $H$ in $\PG(r-1,q)$.
\end{lemma}
\begin{proof}
  Consider $\cM$ as a faithful projective $1-(n',r,s',\mu')_q$ system
  with $s'=\max_{H:\dim(H)=r-1} \cM(H)$, $\mu'=\max_{P:\dim(P)=1} \cM(P)$,
  $n'=\#\cM$, and set $d':=n'-s'$. Then, the stated formulas for the parameters $n$, $s$, $d$, $\mu$ as well as condition~(\ref{eq_div_general}) can be
  easily verified using Lemma~\ref{lemma_linear_code}.
\end{proof}

The geometric lattice $\PG(r-1,q)$ admits duality, i.e., for an $i$-space $S$
we denote the orthogonal subspace with respect to some fixed nondegenerate
bilinear form by $S^\perp$. The $(r-i)$-space $S^\perp$ is also called dual
of $S$ and the dual $\cS^\perp$ of a projective $h-(n,r,s,\mu)_q$ system $\cS$
is obtained from $\cS$ by replacing each element $S\in\cS$ by its dual
$S^\perp$. Directly from the definition we conclude:
\begin{lemma}
  \label{lemma_dual}
  The dual $\cS^\perp$ of a faithful projective $h-(n,r,s,\mu)_q$
  system $\cS$ is a faithful projective $(r-h)-(n,r,\mu,s)_q$ system.
\end{lemma}
\begin{proof}
  The dual of an $h$-space in
  $\PG(r-1,q)$ is an $(r-h)$-space. Let $P$ be an arbitrary point and $H$
  be an arbitrary hyperplane, so that $P^\perp$ is a hyperplane and
  $H^\perp$ is a point. Since at most $\mu$ elements of $\cS$ contain $P$,
  at most $\mu$ elements of $\cS^\perp$ are contained in $P^\perp$.
  Moreover, equality occurs for some point. Similarly, at most $s$ elements
  of $\cS$ are contained in $H$, so that at most $s$ elements of $\cS^\perp$
  contain $H^\perp$. Again, equality occurs for some hyperplane.
\end{proof}

\begin{definition}
  By $n_q(r,h;s)$ we denote the maximum number $n$ such that a projective
  $h-(n,r,s)_q$ system exists.
\end{definition}

\begin{lemma}
  \label{lemma_union}
  Let $\cS_1$ be a projective $h-\left(n_1,r,s_1,\mu_1\right)_q$ system
  and $\cS_2$ be a projective $h-\left(n_2,r,s_2,\mu_2\right)_q$ system.
  Then there exists a projective $h-\left(n_1+n_2,r,\le s_1+s_2,\le \mu_1+\mu_2\right)_q$ system $\cS$. If $\cS_1,\cS_2$ are faithful with types
  $\cM_1$ and $\cM_2$, respectively, then $\cS$ is faithful with type
  $\cM_1+\cM_2$. Moreover, if $r\ge h$, then there exists a faithful
  $h-(1,r,1,1)_q$ system $\cS'$.
\end{lemma}
\begin{proof}
  $\cS$ can be constructed as the multiset union of $\cS_1$ and $\cS_2$.
  If $r\ge h$ then $\cS'$ can be chosen as an arbitrary single $h$-space
  in $\PG(r-1,q)$.
\end{proof}

\begin{corollary}
  \label{cor_union}
  We have $n_q\!\left(r,h;s_1+s_2\right)\ge n_q\!\left(r,h;s_1\right)
  +n_q\!\left(r,h;s_2\right)$. If $r\ge h$, then we have $n_q(r,h;s+1)\ge n_q(r,h;s)+1$.
\end{corollary}

If $r\le h$ we can consider an arbitrary number of copies of the unique $r$-space in
$\PG(r-1,q)$, none lying in a hyperplane, so that $n_q(r,h;s)=\infty$ by definition and
we will always assume $r>h$ in the following.

\begin{lemma}
  \label{lemma_field_reduction}
  Let $\cS$ be a (faithful) projective $h-(n,r,s,\mu)_{q^l}$ system. Then there exists
  a (faithful) projective $hl-(n,rl,s,\mu)_q$ system $\cS'$.
\end{lemma}
\begin{proof}
  The vector space $\F_q^l$ is isomorphic to $\F_{q^l}$ when viewed as a vector space
  over $\F_q$. Under this isomorphism, we get a map $\Psi$ from the $i$-spaces in
  $\PG\!\left(r-1,q^l\right)$ to the $(li)$-spaces in $\PG(rl-1,q)$. Applying $\Psi$
  to the elements of $\cS$ for $i=h$ gives $\cS'$.
\end{proof}
\begin{corollary}
  \label{cor_field_reduction}
    We have $n_q(lr,lh;s)\ge n_{q^l}(r,h;s)$ for all positive integers $l$.
\end{corollary}
The map $\Psi$ is called the field reduction map in \cite{lavrauw2015field} and indeed
very widely used in Galois geometry. A very prominent example is a \emph{Desarguesian spread} of $l$-spaces of $\PG(lr-1,q)$ as the image of the points of
$\PG\!\left(r-1,q^l\right)$ under $\Psi$. Example~\ref{ex_generator_matrix} is
obtained in this way and in Example~\ref{ex_derived_linear_code} we can see
the partition of the $[4]_2=15$ points of $\PG(3,2)$ into $[2]_4=5$ lines. By a
projection argument we can also obtain constructions from field reduction when
the dimension is not divisible by $l$.
\begin{lemma}
  \label{lemma_projection}
  We have $n_q(r-1,h,s)\ge n_q(r,h,s)$ for $r\ge 2$.
\end{lemma}
\begin{proof}
  Let $\cS$ be a projective $h-(n,r,s)_{q}$ system and $P$ an arbitrary point in
  $\PG(r-1,q)$. Projection through a point $P$ yields
  $\PG(r-1,q)/(P)\cong \PG(r-2,q)$ and $\cS$ is mapped to a projective
  $h-(n,r-1,s)_{q}$ system $\cS'$.
\end{proof}
We remark that the elements $S\in \cS$ that contain $P$ are mapped to $(\dim(S)-1)$-spaces $S'$ so that $\cS'$ may not be faithful even if $\cS$ is.
\begin{example}
  \label{ex_projection}
  Consider the faithful projective $2-(5,4,1)_2$ system $\cS$ that corresponds to
  the $[5,4/2,4]_4^2$ additive code from Example~\ref{ex_generator_matrix}. If we
  project $\cS$ through an arbitrary point $P$ of $\PG(3,2)$, we obtain an unfaithful
  projective $2-(5,3,1)_2$ system $\cS'$ consisting of four lines and one point in
  $\PG(2,2)$.
\end{example}

\begin{definition}
  \begin{equation}
    \overline{n}_q(r,h;s):=n_{q^h}(\left\lceil r/h\right\rceil,1;s)
  \end{equation}
\end{definition}
In words, $\overline{n}_q(r,h;s)$ is the size of the largest projective $h-(n,r,s)_{q}$
system that we can obtain starting from a linear code over $\F_{q^h}$ via
Theorem~\ref{thm_connection} and Lemma~\ref{lemma_field_reduction} by an iterative
application of Lemma~\ref{lemma_projection}. Whenever $\overline{n}_q(r,h;s)<n_q(r,h;s)$
we say that additive codes  outperform linear codes for the corresponding parameters,
which is especially interesting if $r/h$ is integral.

\subsection{Basic constructions and bounds}
\label{subsec_basic}
We start to prepare a few basic constructions of projective $h-(n,r,s,\mu)_q$ systems.

\begin{definition}
  A \emph{vector space partition} of $\PG(r-1,q)$ is a multiset $\cV$ of subspaces with dimension at most $(r-1)$ such that every point of $\PG(r-1,q)$ is contained in exactly
  one element of $\cV$.  We say that $\cV$ has type $1^{t_1} 2^{t_2}\dots (r-1)^{t_{r-1}}$ if exactly $t_i$ elements of $\cV$ have dimension $i$ for all $1\le i\le r-1$.
\end{definition}

\begin{example}
  \label{ex_spread}
  The set $\cS$ of all points of $\PG\!\left(r-1,q^h\right)$ is a projective
  $1-\left([r]_{q^h},r,[r-1]_{q^h},1\right)_{q^h}$ system. Applying
  Lemma~\ref{lemma_field_reduction} to $\cS$ gives a vector space partition
  of $\PG(rh-1,q)$ of type $h^{t_h}$ with $t_h=[r]_{q^h}=[rh]_q/[h]_q$. Here one
  also speaks of a \emph{$h$-spread} of $\PG(rh-1,q)$,
  cf.~Example~\ref{ex_derived_linear_code}.
\end{example}

A set of matrices $M\subseteq \F_q^{m\times n}$ with $\operatorname{rk}(A-B)\ge \delta$
for all $A,B\in M$ with $A\neq B$ is called a \emph{rank metric code} with minimum rank
distance $\delta$. A Singleton-type upper bound gives $\# M\le q^{\max\{m,n\}\cdot(\min\{m,n\}-\delta+1)}$. Rank metric codes attaining this bound are called \emph{MRD} codes.
They exist for all parameters with $\delta\le\min\{m,n\}$, even if one additionally
requires that $M$ is linearly closed, see e.g.\ \cite{sheekey201913} for a survey.

\begin{lemma}
  \label{lemma_lifting}
  For $r> 2h$ there exists a vector space partition $\cV$ of $\PG(r-1,q)$ of type
  $h^{t_h} (r-h)^1$ where $t_h=q^{r-h}$.
\end{lemma}
\begin{proof}
  Let $M\subseteq \F_q^{h\times (r-h)}$ be an MRD code with minimum rank distance $h$ and cardinality $q^{r-h}$. Prepending a $h\times h$ unit matrix to the elements of $M$
  gives generator matrices of $h$-spaces in $\PG(r-1,q)$ that are pairwise disjoint and disjoint to an $(r-h)$-space $S$.
\end{proof}
Of course, this lifting type construction is well known. Another construction, based
on field reduction, is e.g.\ given in \cite[Theorem 4.2]{beutelspacher1975partial}. For
$r=2h$ there exists a vector space partition of type $h^{t_h}$, i.e.\ a spread of
$h$-spaces.

\begin{lemma}
  \label{lemma_vsp}
  For $r>a> h$ with $r\equiv a\pmod h$ there exists a vector space partition $\cV$ of $\PG(r-1,q)$ of type $h^{t_h} (a)^1$ where $t_h=q^{a}\cdot\tfrac{q^{r-a}-1}{q^h-1}$.
\end{lemma}
\begin{proof}
  We prove by induction over $r$. Let $\cV$ be the vector space partition
  obtained from Lemma~\ref{lemma_lifting} and let $S\in\cV$ be the unique
  $(r-h)$-dimensional element. If $a=r-h$, which is indeed the case for all
  $r<3h$, then $\cV$ is the desired vector space partition. Otherwise we
  identify $S$ with $\PG(r-h-1,q)$ and replace $S$ by a vector space partition
  of $\PG(r-h-1,q)$ of type $h^{t_h'} a^1$, which exists by induction.
\end{proof}

\begin{lemma}
  \label{lemma_construction_x_preparation}
  For $r>a>h$ with $r\equiv a\pmod h$ let $\cS$ be the set of $h$-dimensional
  elements of a vector space partition $\cV$ of $\PG(r-1,q)$ of type
  $h^{t_h}a^1$ and $A$ be the unique $a$-dimensional element in $\cV$. Then,
  $\cS$ is a faithful projective $h-(t_h,r,s,1)_q$ system where
  $t_h=q^{a}\cdot\tfrac{q^{r-a}-1}{q^h-1}$ and
  $s=q^{a-h}\cdot \tfrac{q^{r-a}-1}{q^h-1}$. Moreover, each hyperplane
  that contains $A$ contains $s-q^{a-h}$ elements from $\cS$.
\end{lemma}
\begin{proof}
  Let $H$ be an arbitrary hyperplane of $\PG(r-1,q)$. Note that every
  $i$-space intersects $H$ in either $[i]_q$ or $[i-1]_q$ points and that
  the elements of $\cS$ partition the points outside of $A$. Counting
  points yields that $H$ contains
  $$
    \frac{[r-1]_q-[a-1]_q-t_h\cdot [h-1]_q}{q^{h-1}}=
    q^{a-h}\cdot \frac{q^{r-a}-1}{q^h-1}=s
  $$
  elements from $\cS$ if $A\not\le H$ and
  $$
    \frac{[r-1]_q-[a]_q-t_h\cdot [h-1]_q}{q^{h-1}}=
    \frac{[r-1]_q-q^{a-1}-[a-1]_q-t_h\cdot [h-1]_q}{q^{h-1}}=s-q^{a-h}
  $$
  elements from $\cS$ if $A\le H$.
\end{proof}

\begin{lemma}
  \label{lemma_construction_x}(Cf.~\cite[Lemma 6]{guan2023some})
  Let $\cS_1$ be a faithful projective $h-\!\left(n_1,r,s_1\right)_q$
  system, $A$ be an $a$-space such that each hyperplane that contains
  $A$ contains at most $s_0$ elements from $\cS_1$, and $\cS_2$ be a
  faithful projective $h-\!\left(n_2,a,s_2\right)_q$
  system. Then, we have
  $n_q\!\left(r,h;\max\!\left\{s_1+s_2,s_0+n_2\right\}\right)\ge n_1+n_2$.
\end{lemma}
\begin{proof}
  We identify $A$ with $\PG(a-1,q)$ and insert $\cS_2$ into the subspace $A$ of $\cS_1$.
\end{proof}
In analogy with linear codes, the authors of \cite{guan2023some,10693309} speak of the {\lq\lq}additive construction X{\rq\rq}. We give an application in
Lemma~\ref{lemma_construction_x_consequence}.

\bigskip

Next we discuss a few basic upper bounds for $n_q(r,h;s)$.

\begin{lemma}
  \label{lemma_indirect_upper_bounds}
  Let $r>h$ and $\cS$ be a projective $h-(n,r,s)_q$ system. Then, we have
  \begin{equation}
    \label{ie_indirect_upper_bounds}
    g_q\!\left(r,q^{h-1}\cdot (n-s)\right)\le [h]_q\cdot n\le n_q\!\left(r,1;n\cdot [h-1]_q+s\cdot q^{h-1}\right).
  \end{equation}
\end{lemma}
\begin{proof}
  W.l.o.g.\ we assume that $\cS$ is faithful and that we have $n>s$. Let
  $n'=[h]_q\cdot n$, $s'=n\cdot[h-1]_q+s\cdot  q^{h-1}$, and
  $d'=q^{h-1}\cdot (n-s)$. By $C$ we denote the $[n',r,d']_q$ code and by
  $\cS'$ the projective $1-(n',r,s')_q$ system constructed in the proof of
  Lemma~\ref{lemma_linear_code}. The parameters of $\cS'$ give the right
  hand side of Inequality~(\ref{ie_indirect_upper_bounds}) and applying
  the Griesmer bound from Inequality~(\ref{eq_griesmer_bound}) to $C$
  yields the left hand side of Inequality~(\ref{ie_indirect_upper_bounds}).
\end{proof}
Indeed, Lemma~\ref{lemma_indirect_upper_bounds} is well-known for $q=h=2$,
see e.g.\ \cite[Lemma 1]{bierbrauer2021optimal} or \cite[Lemma 3]{guan2023some}. The left hand side of Inequality~(\ref{ie_indirect_upper_bounds}) is also equivalent to
\cite[Theorem 12]{ball2024additive} for additive codes.
\begin{lemma}
  Each $[n,r/h,d]_q^h$ code satisfies
  \begin{equation}
    \label{ie_direct_upper_bound}
  n\ge \left\lceil
  \frac{g_q\!\left(r,d\cdot q^{h-1}\right)}{[h]_q}
  \right\rceil
  =
  \left\lceil
  \frac{ \sum\limits_{i=0}^{r-1} \left\lceil d\cdot q^{h-1-i}\right\rceil}{[h]_q}
  \right\rceil
  =d+\left\lceil\frac{\sum\limits_{i=1}^{r-h} \left\lceil\frac{d}{q^i}\right\rceil}{[h]_q}\right\rceil
  =
  d+ \left\lceil \frac{g_q(r-h+1,d)-d   }{[h]_q}\right\rceil
  \end{equation}
  and
  \begin{equation}
    \label{ie_direct_upper_bound_approximation}
    n\ge d\cdot \frac{[r]_q}{q^{r-h}\cdot [h]_q}   .
  \end{equation}
\end{lemma}
\begin{proof}
  Replacing $n-s$ by $d$ and using the fact that $n$ is an integer,
  we can rewrite the left hand side of Inequality~(\ref{ie_indirect_upper_bounds}) to the left hand side of Inequality~(\ref{ie_direct_upper_bound}). Plugging in the definition of $g_q(\cdot,\cdot)$ and using $\sum_{i=0}^{h-1} \left\lceil d\cdot q^{h-1-i}\right\rceil=d\cdot [h]_q$ yields the remaining part of Inequality~(\ref{ie_direct_upper_bound}).  From
  $$
    \left\lceil
  \frac{ \sum\limits_{i=0}^{r-1} \left\lceil d\cdot q^{h-1-i}\right\rceil}{[h]_q}
  \right\rceil \ge \frac{ \sum\limits_{i=0}^{r-1}  d\cdot q^{h-1-i}}{[h]_q}
  =\frac{d\cdot q^{h-1}\cdot\sum\limits_{i=0}^{r-1} q^{-i}}{[h]_q}
  =\frac{d\cdot q^{h-1}}{[h]_q}\cdot\frac{1-q^{-r}}{1-q^{-1}}
  =  d\cdot \frac{[r]_q}{q^{r-h}\cdot [h]_q}
  $$
  we conclude Inequality~(\ref{ie_direct_upper_bound_approximation}).
  For $r=h$ both inequalities are equivalent to $n\ge d$.
\end{proof}

\noindent
In other words, Inequality~(\ref{ie_direct_upper_bound}) is obtained from the sequence of reformulations
\begin{center}
  additive code $\quad\to\quad$ projective $h$-system $\quad\to\quad$  multiset of points $\quad\to\quad$ linear code
\end{center}
and applying the Griesmer bound to the resulting linear code. Inequality~(\ref{ie_direct_upper_bound_approximation}) is obtained from Inequality~(\ref{ie_direct_upper_bound}) by removing all roundings, see Lemma~\ref{lemma_one_weight_bound} for the geometric reformulation. In Section~\ref{sec_generalized_weights} we present a geometric version of the Griesmer bound.

\begin{definition}
  \label{def_griesmer_uppper_bound}
  The \emph{coding upper bound} for $n_q(r,h;s)$ is the largest integer
  $n$ such that
  $[h]_q\cdot n\le n_q(r,1;n\cdot [h-1]_q+s\cdot q^{h-1})$. The
  \emph{Griesmer upper bound} for $n_q(r,h;s)$ is the largest integer
  $n$ such that $g_q\!\left(r,q^{h-1}\cdot (n-s)\right)\le [h]_q\cdot n$.
\end{definition}

\begin{table}[htp]
  \begin{center}
     \begin{tabular}{rrrr}
       \hline
       $s$ & Griesmer upper bound & coding upper bound & $n_2(8,2;s)$ \\
       \hline
        3 &   9 &  7 &  5 \\
        4 &  12 &    & 10 \\
        5 &  17 &    & 17 \\
        6 &  22 & 18 & 18 \\
        7 &  25 & 23 & 23 \\
        8 &  30 & 28 & 28 \\
        9 &  33 &    & 33 \\
       10 &  38 & 36 & 36 \\
       11 &  43 & 40 & 40 \\
       12 &  44 &    & 44 \\
       13 &  49 &    & 49 \\
       14 &  54 &    & 54 \\
       15 &  59 & 57 & 55 \\
       28 & 110 &    & 108 \\ 
       29 & 115 &    & 113 \\ 
       \hline
     \end{tabular}
     \caption{The Griesmer and coding upper bound for $n_2(8,2;s)$.}
     \label{table_griesmer_upper_bound}
  \end{center}
\end{table}

\begin{example}
  The Griesmer upper bound for $n_2(8,2;8)$ is $30$ and the coding upper bound is $28$.
  I.e., the Griesmer bound implies that no $[93,8,46]_2$ code exists but cannot rule
  out the existence of a $[90,8,44]_2$ code, so that $n_2(8,2;8)\le 30$ is the
  sharpest upper bound we can deduce from the Griesmer bound (for linear codes). However,
  since the existence of a $[84,8,40]_2$ code and the non-existence of a $[87,8,42]_2$
  code is known, we obtain $n_2(8,2;8)\le 28$. In Table~\ref{table_griesmer_upper_bound} we list the Griesmer and the coding upper
  bound for $n_2(8,2;s)$ for $s\le 15$ and all cases were either the coding
  upper bound is strictly less than the Griesmer upper bound or the value
  of $n_2(8,2;s)$ is unknown, see \cite{kurz2024optimal} and
  \cite{blokhuis2004small} for $s\in\{3,4\}$. We do not display
  the coding upper bound when it coincides with the Griesmer upper bound.
  The coding upper bound for $n_2(8,2;15)$ is $57$ since a $[171,8,84]_2$ code
  but no $[174,8,86]_2$ code exists. However, Lemma~\ref{lemma_linear_code} gives more fine-grained information on the possible weights occurring
  in the linear code $\cX^{-1}(\cP( \cS))$. The upper bound on the maximum
  weight was used in \cite{kurz2024optimal} to conclude $n_2(8,2;15)\le 55$.

\end{example}
For the current knowledge on $n_q(r,1;s)$ or corresponding bounds for linear codes we
refer to \url{www.codetables.de}, \url{http://mars39.lomo.jp/opu/griesmer.htm}, and \url{http://web.mat.upc.edu/simeon.michael.ball/codebounds.html}.

In \cite{kurz2024optimal} the author refers to the coding upper bound as the weak coding upper bound and similarly defines the strong coding upper bound by additionally assuming the weight restrictions from Lemma~\ref{lemma_linear_code} for the linear code. The (weak) coding upper bound for $n_2(8,2;15)$ is $57$ while the strong coding upper bound is  $55$, see \cite[Theorem 4.9]{kurz2024optimal}. Of course, it can be quite a challenge to evaluate the strong coding upper bound and the results from the literature mostly apply to the (weak) coding upper bound.

\begin{lemma}
   \label{lemma_one_weight_bound}
   We have
  \begin{equation}
    \label{ie_one_weight_bound}
    n_q(r,h;s) \le \frac{[r]_q\cdot s}{[r-h]_q},
  \end{equation}
  where the right hand side is an integer iff $s$ is divisible by $[r-h]_q/[\gcd(r,h)]_q$. In the case of equality each point is contained in exactly
  $\mu=\tfrac{[h]_q\cdot s}{[r-h]_q}$ elements.
\end{lemma}
\begin{proof}
  Let $\cS$ be a faithful projective $h-(n,r,s)_q$ system with
  $n=n_q(r,h;s)$. Since each element $S\in\cS$ is contained in
  $[r-h]_q$ hyperplanes and there are $[r]_q$ hyperplanes in total,
  we conclude $n\le \tfrac{[r]_q\cdot s}{[r-h]_q}$. From Equation~(\ref{eq_gcd}) we
  conclude $\gcd([r]_q,[r-h]_q)=[\gcd(r,h)]_q$, so that the right hand side
  of Inequality~(\ref{ie_one_weight_bound}) is an integer iff $s$ is
  divisible by $[r-h]_q/[\gcd(r,h)]_q$. If $n=\tfrac{[r]_q\cdot s}{[r-h]_q}$,
  then each hyperplane contains exactly $s$ elements from $\cS$, so that
  each point is contained in the same number $\mu$ of elements by duality. Thus, we have $\mu=n\cdot [h]_q/[r]_q=\tfrac{[h]_q\cdot s}{[r-h]_q}$.
\end{proof}
We remark that the Griesmer upper bound is always as least as good as this explicit upper bound and both coincide iff the right hand side of
Inequality~(\ref{ie_one_weight_bound}) is an integer. Of course this bound
is known, see e.g.\ \cite[Proof of Theorem 9]{ball2024additive}. If we know
that a given projective system is unfaithful, then we can typically improve
the stated upper bound a bit, see e.g.\ the proof of  Lemma~\ref{coding_bound_improved_e1}. Via Theorem~\ref{thm_connection} we conclude that
Inequality~(\ref{ie_one_weight_bound}) is equivalent to
Inequality~(\ref{ie_direct_upper_bound_approximation}).

By projection through a subspace $K$ instead of a point $P$ we can
improve Lemma~\ref{lemma_projection}.
\begin{lemma} (Cf.~\cite[Proposition 1]{bierbrauer2009short})
  \label{lemma_projection_subspace}
  Let $\cS$ be a projective $h-(n,r,s)_q$ system and $K$ be an $r'$-space in
  $\PG(r-1,q)$ containing $s'$ elements from $\cS$. Then, projection
  through $K$ yields a projective $h-(n-s',r-r',s-s')_q$ system $\cS'$.
\end{lemma}
Via this projection we essentially just look at all hyperplanes that contain $K$ and can obtain general upper bounds by choosing a suitable subspace $K$.

\begin{lemma}
  \label{lemma_projection_subspace_bound}
  $n_q(r,h;s)\le t + n_q(r-ht,h;s-t)$ for all $t\le r/h-1$.
\end{lemma}
\begin{proof}
  Consider a projective $h-(n,r,s)_q$ system with $n=n_q(r,h;s)$ and let $K$ be a  subspace spanned by $t$ elements of $\cS$. Applying Lemma~\ref{lemma_projection_subspace} to
  $\cS$ and $K$ gives a projective $h-(n-s',r-r',s-s')_q$ system $\cS'$ with $s'\ge t$ and $r'\le ht$, so that
  $$
    n_q(r,h;s)\le s'+n_q(r-r',h;s-s')
    \overset{\text{Lemma~\ref{lemma_projection}}}{\le} s'+n_q(r-ht,h;s-s')
    \overset{\text{Corollary~\ref{cor_union}}}{\le} t + n_q(r-ht,h;s-t).
  $$
\end{proof}
\begin{remark}
  \label{remark_projection_subspace_bound}
  In \cite[Theorem 9]{ball2024additive}
  Lemma~\ref{lemma_projection_subspace_bound} is written down
  as an explicit upper bound for $n_q(r,h;s)$, in terms of additive codes, applying Lemma~\ref{lemma_one_weight_bound} to $\cS'$, together with an analysis of the optimal choice of $t$. The bound from \cite[Theorem 9]{ball2024additive} is quite effective and it is shown that it surpasses the Griesmer upper bound in many cases, see
 \cite[Theorem 13]{ball2024additive}. As an example, we mention that $n_3(7,2;3)\le 15$ can be concluded from $n_3(3,2;1)=13$ via Lemma~\ref{lemma_projection_subspace_bound} or directly
 via \cite[Theorem 9]{ball2024additive} while the Griesmer upper bound is $n_3(7,2;3)\le 18$.\footnote{The Griesmer bound implies that the maximum minimum distance of an $[72,7]_3$-code is
 at most $45$, which is still the best known upper bound and there exists a $[72,7,43]_3$-code.} Similarly, combining Lemma~\ref{lemma_projection_subspace_bound} with $n_2(5,2;3)=11$ and
 $n_2(5,2;4)=16$ yields $n_2(7,2;4)\le 12$ and $n_2(7,2;5)\le 17$, which are both attained with equality. The Griesmer upper bound only gives $n_2(7,2;4)\le 14$, so that it cannot rule out the existence of a $[14,7/2,10]_2^2$ code.
\end{remark}

We can easily characterize asymptotically optimal projective $h-(n,r,s)_{q}$ systems, cf.~\cite{bierbrauer2024asymptotic}.\footnote{For the special case $q=2$, $h=2$, and arbitrary $r$ of Theorem~\ref{thm_partition} we also refer to
\cite[Lemma 3 \& Remark 1]{10693309}.} I.e., if the right
hand side of Inequality~(\ref{ie_one_weight_bound})
is an integer, then the upper bound from Lemma~\ref{lemma_one_weight_bound} can indeed be attained.

\begin{theorem}
  \label{thm_partition}(Cf.~\cite[page 83]{hirschfeld1998projective},  \cite[Corollary 8]{el2011lambda}, or \cite[Lemma 2]{krotov2023multispreads})
  For each pair of integers $r\ge h\ge 1$ a faithful
  projective $h-(n,r,s,\mu)_{q}$ system with $n=\tfrac{[r]_q}{[\gcd(r,h)]_q}$, $s=\tfrac{[r-h]_q}{[\gcd(r,h)]_q}$, and $\mu=\tfrac{[h]_q}{[\gcd(r,h)]_q}$ exists.
\end{theorem}
\begin{proof}
  Due to Lemma~\ref{lemma_field_reduction} we can assume $\gcd(r,h)=1$.
  We prove by double induction on $h$ and $r$. The set of points of
  $\PG(r-1,q)$ gives an example for $h=1$ and arbitrary $r$, so that we
  assume $h\ge 2$ in the following. If $r=h$ or $r=2h$, then we have $h=1$.
  If $h<r<2h$, then a projective $(r-h)-(n,r,\mu,s)_{q}$ system exists by induction and we can apply Lemma~\ref{lemma_dual}. If $r> 2h$, then we apply
  Lemma~\ref{lemma_vsp} to construct a vector space partition $\cV$ of
  $\PG(r-1,q)$ of type $h^{t_r}(r-h)^1$ and denote the special
  $(r-h)$-space by $A$. To construct the desired example we take a
  $\mu$-fold copy of $\cV$ and replace $A$ by a projective $h-(n',r-h,s',\mu')_{q}$ system with $n'=[r-h]_q$, $s'=[r-2h]_q$, and $\mu'=[h]_q$.
\end{proof}
\begin{corollary}
  \label{cor_asymptotic_one_weight}
  Let $\cS$ be a faithful projective $h-(n,r,s,\mu)_q$ system. Then there
  exists a faithful projective
  $$
    h-\left(n+t\cdot \frac{[r]_q}{[\gcd(r,h)]_q},r,s+t\cdot
    \frac{[r-h]_q}{[\gcd(r,h)]_q},\mu+t\cdot\frac{[h]_q}{[\gcd(r,h)]_q}\right)_q
  $$
  system $\cS_t$ for each $t\in\N$.
\end{corollary}
The parameters of the sequence $\cS_t$ tend towards the upper bound from Lemma~\ref{lemma_one_weight_bound}, so that e.g.\
\begin{equation}
  \label{eq_limit_relative_one_weight_bound}
  \lim_{s\to\infty} n_q(r,h;s) \cdot\frac{[r-h]_q}{[r]_q\cdot s}=1.
\end{equation}
In the subsequent section we will
show that the (typically) stronger Griesmer upper bound can always be attained
if $s$ is sufficiently large, see Theorem~\ref{thm_attained_asymptotically}. While we do not know an explicit analytical formula $f_q(r,h;s)$ such that
\begin{equation}
  \lim_{s\to\infty} \left(n_q(r,h;s) -f_q(r,h;s)\right)=0,
\end{equation}
we state such a result in terms of the minimum distance of $[n,r/h,d]_q^h$ codes, see Corollary~\ref{cor_attained_asymptotically}.

We remark that for a faithful projective $h-(n,r,s,\mu)_{q}$ system as in
Theorem~\ref{thm_partition} each hyperplane contains exactly $s$
and each point is contained in exactly $\mu$ elements from $\cS$. The cases
where $\mu=1$ correspond to spreads of $h$-spaces of $\PG(r-1,q)$.

\subsection{Asymptotic formulations}
\label{subsec_asymptotic}

\begin{definition}
  \label{def_partitionable_neg}
  Let $\cM$ be a premultiset of points in $\PG(r-1,q)$ and $V$ be the ambient space.  We say that $\sigma[r]-\cM$ is \emph{$h$-partitionable over $\F_q$} if there exists a faithful projective $h-(n,r,s,\mu)_q$ system $\cS$ with
  $\cP(\cS)=\sigma\cdot\chi_V-\cM$, i.e.\ $\sigma\cdot\chi_V-\cM=
  \sum_{S\in\cS} \chi_S$. We also say that $\cS$ has \emph{type $\sigma[r]-\cM$}.
\end{definition}
\begin{lemma}
  \label{lemma_compute_parameters_from_premultiset}
  Let $\cM$ be a premultiset of points $\cM$ in $\PG(r-1,q)$. If $\cS$ is a
  faithful projective $h-(n,r,s,\mu)_q$ system with type $\sigma[r]-\cM$, then
  we have
  \begin{equation}
     \label{formula_n_premultiset}
     n=\left(\sigma[r]_q-\#\cM\right)/[h]_q,
  \end{equation}
  \begin{equation}
     \label{formula_s_premultiset}
     s=\left(\sigma[r-h]_q-\frac{[h-1]_q\cdot\#\cM-[h]_q\cdot \min_{H:\dim(H)=r-1} \cM(H)}{q^{h-1}}\right)/[h]_q,
  \end{equation}
  \begin{equation}
    d:=n-s=\sigma q^{r-h}-\frac{\#\cM-\min_{H:\dim(H)=r-1} \cM(H)}{q^{h-1}},
  \end{equation}
  \begin{equation}
     \label{formula_mu_premultiset}
     \mu = \sigma- \min_{P:\dim(P)=1} \cM(P),
  \end{equation}
  and
  \begin{equation}
    \#\cM\equiv \cM(H)
  \end{equation}
  for each hyperplane. 
  Moreover, we have $\#\cM\equiv 0\pmod{[\gcd(r,h)]_q}$ and $\sigma\equiv \tfrac{\#\cM}{[\gcd(r,h)]_q}\pmod{\tfrac{[h]_q}{[\gcd(r,h)]_q}}$.
\end{lemma}
\begin{proof}
  Using $\#\chi_V\equiv \chi_V(H)$ for each hyperplane $H$, the stated equations
  follow from Lemma~\ref{lemma_system_parameters_from_multiset}.
  Since $n\in \N$, Equation~(\ref{eq_gcd}) gives
  $\#\cM\equiv 0\pmod{[\gcd(r,h)]_q}$ and $\sigma\equiv
  \tfrac{\#\cM}{[\gcd(r,h)]_q}\pmod{\tfrac{[h]_q}{[\gcd(r,h)]_q}}$.
\end{proof}

\begin{lemma}
  \label{lemma_sigma_constraint_premultiset}
  Let $\cM$ be a premultiset of points in $\PG(r-1,q)$.
  If $x[r]-\cM$ is $h$-partitionable over $\F_q$ for $x\in\{\sigma,\sigma'\}$
  then
  $$\left(\sigma+t\cdot\frac{[h]_q}{[\gcd(r,h)]_q}\right)\cdot[r]-\cM$$ is
  $h$-partitionable over $\F_q$ for all $t\ge 0$
  and we have $\sigma\equiv \sigma'\pmod {\frac{[h]_q}{[\gcd(r,h)]_q}}$.
\end{lemma}
\begin{proof}
  Due to Theorem~\ref{thm_partition} we have that
  $\tfrac{[h]_q}{[\gcd(r,h)]_q}\cdot[r]$ is $h$-partitionable over
  $\F_q$, so that we can apply
  Lemma~\ref{lemma_union}.
  The second statement is obviously true if $\sigma=\sigma'$, so that
  we assume $\sigma'>\sigma$ w.l.o.g. Let $\cS$ be a faithful projective
  $h-(n,r,s)_q$ system with type
  $\sigma[r]-\cM$ and $\cS'$ be a faithful
  projective $h-(n',r,s')_q$ system with type
  $\sigma'[r]-\cM$. Then, we have
  $$
    \Z \ni n'-n=(\sigma'-\sigma)\cdot \frac{[r]_q}{[h]_q},
  $$
  which implies $\sigma\equiv \sigma'\pmod {\frac{[h]_q}{[\gcd(r,h)]_q}}$
  by Equation~(\ref{eq_gcd}).
\end{proof}

For those situations where we are not  interested in the smallest
possible value $\sigma$ such that $\sigma[r]-\cM$ is $h$-partitionable
over $\F_q$ we introduce the following notion:
\begin{definition}
  \label{def_star}
  Let $\cM$ be a premultiset of points in $\PG(r-1,q)$.
  We say that $\star[r]-\cM$ is $h$-partitionable over $\F_q$ if there
  exists an integer $\sigma$ such that $\sigma[r]-\cM$ is $h$-partitionable
  over $\F_q$.
\end{definition}

By using linear algebra we can decide whether $\star[r]-\cM$ is $h$-partitionable
over $\F_q$ for some premultiset of points $\cM$ in $\PG(r-1,q)$. To this end
we utilize the incidence matrix between points and $h$-spaces in $\PG(r-1,q)$
and the set of solutions of a linear equation system over $\Z$.
\begin{definition}
  For each $1\le h\le r$ let $\qbin{r}{h}{q}$ denote the number of $h$-spaces in
  $\PG(r-1,q)$. The \emph{incidence vector} of a faithful projective
  $h-(n,r,s,\mu)_q$ system $\cS$ is a column vector $x\in\N^{\qbin{r}{h}{q}}$
  whose entries $x_K\in \N$ equal the number of occurrences of $h$-spaces $K$
  in $\cS$. Similarly, the \emph{incidence vector} of a premultiset $\cM$ in
  $\PG(r-1,q)$ is a column vector $x\in\Z^{[r]_q}$ whose entries $x_P\in\Z$
  equal the  point multiplicity $\cM(P)$ for every point $P$ in $\PG(r-1,q)$.
  We say that an $a$-space $A$ and a $b$-space $B$ are \emph{incident} if
  either $A\le B$ or $A\ge B$. Let $A^{a,b;r,q}\in\{0,1\}^{\qbin{r}{a}{q}\times \qbin{r}{b}{q}}$ be the \emph{incidence matrix} between the $\qbin{r}{a}{q}$
  $a$-spaces $A$ and the $\qbin{r}{b}{q}$ $b$-spaces $B$ in $\PG(r-1,q)$, i.e.,
  the entries of $A^{a,b;r,q}$ are given by $A^{a,b;r,q}_{A,B} = 1$
  if $A$ is incident with $B$ and $A^{a,b;r,q}_{A,B} = 0$ otherwise.
\end{definition}
The notion $\qbin{r}{h}{q}$ extends $[r]_q$, counting the number of points
or hyperplanes in $\PG(r-1,q)$, since $[r]_q=\qbin{r}{1}{q}=\qbin{r}{r-1}{q}$.
When working with those incidence matrices and incidence vectors for $\PG(r-1,q)$
we will always assume a fixed, but arbitrary, enumeration of the $h$-spaces
for each $1\le h\le r$.

\begin{lemma}
  \label{lemma_partitonable_les}
  Let $\cM$ be a premultiset of points in $\PG(r-1,q)$ and $z$ its incidence vector. Let $v$ be the incidence vector of $\chi_V$ for the ambient space $V$.
  Then, $\star[r]-\cM$ is $h$-partitionable over $\F_q$ iff there exists a unique
  integer  $0\le \sigma<\tfrac{[h]_q}{[\gcd(r,h)]_q}$ with $\sigma\equiv \tfrac{\#\cM}{[\gcd(r,h)]_q}\pmod{\tfrac{[h]_q}{[\gcd(r,h)]_q}}$ and
  $A^{1,h;r,q}\cdot x=\sigma v-z$ admits a solution $x\in\Z^{\qbin{r}{h}{q}}$.
\end{lemma}
\begin{proof}
  If $\cS$ is a faithful projective $h-(n,r,s,\mu)_q$ system with type
  $\sigma'[r]-\cM$, then let $x'\in\N^{\qbin{r}{h}{q}}$ denote the
  incidence vector of $\cS$. Due to Theorem~\ref{thm_partition} there
  exists a faithful projective $h-(n',r,s',\mu')_q$ system $\cS'$ with
  type $\tfrac{[h]_q}{[\gcd(r,h)]_q}\cdot[r]$ and $v'\in\N^{\qbin{r}{h}{q}}$
  its incidence vector. Now let $y$ be the unique integer such that $0\le \sigma<\tfrac{[h]_q}{[\gcd(r,h)]_q}$ for $\sigma:=\sigma'-y\cdot
  \tfrac{[h]_q}{[\gcd(r,h)]_q}$, which also implies uniqueness for $\sigma$.
  Lemma~\ref{lemma_compute_parameters_from_premultiset} gives $\#\cM\equiv 0
  \pmod {[\gcd(r,h)]_q}$ and $\sigma'\equiv \tfrac{\#\cM}{[\gcd(r,h)]_q}\pmod{\tfrac{[h]_q}{[\gcd(r,h)]_q}}$, so that also  $\sigma\equiv
  \tfrac{\#\cM}{[\gcd(r,h)]_q}\pmod{\tfrac{[h]_q}{[\gcd(r,h)]_q}}$. Since
  $A^{1,h;r,q}\cdot x'=\sigma'v-z$ we have $A^{1,h;r,q}\cdot \left(x'-y\cdot v'\right)=\sigma v-z$, so that we set $x:=x'-y\cdot v'\in \Z^{\qbin{r}{h}{q}}$.

  For the other direction we  assume that $x\in\Z^{\qbin{r}{h}{q}}$ is a
  solution of $A^{1,h;r,q}\cdot x=\sigma v-z$ with the stated constraints for
  $\sigma$. If $x\in\N^{\qbin{r}{h}{q}}$, then the corresponding faithful projective system has type $\sigma[r]-\cM$ (noting $\sigma\in\N$). Otherwise
  we let $y<0$ be the
  minimum entry of $x$. Let $\cS''$ be the faithful projective
  $h-(n'',r,s'',\mu'')_q$  system with type $\sigma''[r]$ that consists
  of all 
  $h$-spaces in $\PG(r-1,q)$ and
  $a\in\N^{\qbin{r}{h}{q}}$ its incidence vector. With this we set
  $x'':=x-y\cdot a\in \N^{\qbin{r}{h}{q}}$ and consider the corresponding
  faithful projective system with type $\left(\sigma-y\cdot \sigma''\right)
  \cdot [r]-\cM$.
\end{proof}

Admittedly, Lemma~\ref{lemma_partitonable_les} is not a very strong characterization result and looks rather technical. However, using the
\emph{Smith normal form} of $A^{1,h;r,q}$
we can characterize solvability of linear equation systems
over $\Z$. In order the keep the paper self-contained we give a brief
exposition in Section~\ref{sec_SNF}. In Section~\ref{sec_solomon_stiffler}
we consider a special class of premultisets of points $\cM$ that are
parameterized by parameters $\varepsilon_1,\dots,\varepsilon_{r-1}$, which
are connected to the Griesmer bound and the Solomon--Stiffler construction.
Theorem~\ref{thm_main} gives an explicit characterization of
all parameters $\varepsilon_1,\dots,\varepsilon_{r-1}$ such that
$\star[r]-\cM$ is $h$-partitionable. For a direct application of Lemma~\ref{lemma_partitonable_les} we refer to Example~\ref{ex_lemma_partitonable_les}.

\pagebreak

\section{A generalization of the Solomon--Stiffler construction}
\label{sec_solomon_stiffler}

In \cite{solomon1965algebraically} Solomon and Stiffler constructed
$[n,k,d]_q$ codes with $n=g_q(k,d)$ for all parameters with sufficiently
large minimum distance $d$. Here we want to show the generalization that
the Griesmer upper bound for $n_q(r,h;s)$ can always be attained if $s$
is sufficiently large.\footnote{This property is the reason why we
speak of \textit{the} Griesmer upper bound in Definition~\ref{def_griesmer_uppper_bound} when referring to Lemma~\ref{lemma_indirect_upper_bounds} which is just \textit{an}
application of the Griesmer bound.}

Using a specific parameterization of the minimum distance $d$ the
Griesmer bound in Inequality~(\ref{eq_griesmer_bound}) can be written more
explicitly:
\begin{lemma}
  \label{lemma_parameters_griesmer_code}
  Let $k$ and $d$ be positive integers. Write $d$ as
  \begin{equation}
    \label{eq_griesmer_representation_min_dist}
    d=\sigma q^{k-1}-\sum_{i=1}^{k-1}\varepsilon_iq^{i-1},
  \end{equation}
  where $\sigma\in\N_0$ and the $0\le\varepsilon_i<q$ are integers for all $1\le i\le k-1$. Then, Inequality~(\ref{eq_griesmer_bound})
  is satisfied with equality iff
  \begin{equation}
    \label{eq_griesmer_representation_length}
    n=\sigma[k]_q-\sum_{i=1}^{k-1}\varepsilon_i[i]_q,
  \end{equation}
  which is equivalent to
  \begin{equation}
    \label{eq_griesmer_representation_species}
    n-d=\sigma[k-1]_q-\sum_{i=1}^{k-1}\varepsilon_i[i-1]_q.
  \end{equation}
\end{lemma}

In Subsection~\ref{subsec_griesmer_bound} we give a proof of Lemma~\ref{lemma_parameters_griesmer_code} and state further details on the Griesmer bound, including its geometric reformulation.

\begin{remark}
  Given $k$ and $d$ Equation~(\ref{eq_griesmer_representation_min_dist})
  always determines $\sigma$ and the $\varepsilon_i$ uniquely. This is
  different for Equation~(\ref{eq_griesmer_representation_species}) given
  $k$ and $n-d=s$, cf.~Remark~\ref{remark_nice_vs_compact}. Here it may happen that no solution with
  $0\le \varepsilon_i\le q-1$ exists. By relaxing to
  $0\le \varepsilon_i\le q$ we can ensure existence and uniqueness is
  enforced by additionally requiring $\varepsilon_j=0$ for all $j<i$
  where $\varepsilon_i=q$ for some $i$. The same is true for
  Equation~(\ref{eq_griesmer_representation_length}) given $k$ and $n$.
  For more details we refer to \cite[Chapter 2]{govaerts2003classifications}
  which also gives pointers to Hamada's work on minihypers.
\end{remark}

\begin{lemma}
  \label{lemma_multiset_solomon_stiffler}
  Let $S_1,\dots,S_l$ be a collection of subspaces of $\PG(r-1,q)$ such that
  exactly $\varepsilon_i$ subspaces have dimension $i$ for $1\le i\le r-1$
  and $V$ be the $r$-dimensional ambient space. If
  \begin{equation}
    \label{eq_multiset_solomon_stiffler}
    \cM =\sigma\cdot\chi_V-\sum_{i=1}^l \chi_{S_i}
  \end{equation}
  is a multiset of points, i.e., if we have $\cM(P)\in\N$ for all points $P$, then $\cM$ corresponds to a projective $1-(n,r,\le s,\le \sigma)_q$ system with
  \begin{equation}
    n=\sigma\cdot [r]_q-\sum_{j=1}^{r-1} \varepsilon_j\cdot[j]_q\quad\text{and}\quad s=\sigma\cdot [r-1]_q-\sum_{j=1}^{r-1} \varepsilon_j\cdot[j-1]_q.
  \end{equation}
\end{lemma}
\begin{proof}
  Since $\#\chi_V=[r]_q$ and $\#\chi_{S_i}=[j]_q$ if $S_i$ has dimension $j$,
  we have $n=\#\cM=\sigma\cdot [r]_q-\sum_{j=1}^{r-1} \varepsilon_j\cdot[j]_q$. For each hyperplane $H$ we have
  $\chi_V(H)=[r-1]_q$ and $\chi_{S_i}(H)\in\left\{[j-1]_q,[j]_q\right\}$
  if $S_i$ has dimension $j$, so that $\cM(H)\le \sigma\cdot [r-1]_q-\sum_{j=1}^{r-1} \varepsilon_j\cdot[j-1]_q$.
\end{proof}

Given parameters $k$, $d$, and $q$, the Solomon-Stiffler construction consists of choosing $\varepsilon_1,\dots,\varepsilon_{k-1}$, and $\sigma$ as in
Equation~(\ref{eq_griesmer_representation_min_dist}). If $\sigma$ is sufficiently large then a list of subspaces $S_1,\dots,S_l$ (going in line with $\varepsilon_1,\dots,
\varepsilon_{k-1}$) can be chosen such that $\cM$ as in Equation~(\ref{eq_multiset_solomon_stiffler}), where $r=k$, is a multiset of points. With this, the linear code $C$ corresponding to
$\cM$ is an $[n,k,d]_q$ code with $n=g_q(k,d)$.

For a construction of a general faithful projective $h-(n,r,s)_q$ system $\cS$ with
$n>s$ we aim to reverse Lemma~\ref{lemma_linear_code} where we have
associated a linear code $C$ to $\cS$. To this end we consider a
$q^{h-1}$-divisible linear $[n',r,d']_q$ code $C$ with maximum weight
at most $n\cdot q^{h-1}$,  where $n'=[h]_q\cdot n$ and
$d'=q^{h-1}\cdot (n-s)$. If we can partition the multiset of points
$\cM=\cX(C)$ associated with $C$ into the multiset union of
$h$-spaces, we obtain a faithful projective $h-(n,r,s')_q$ system, where we
hopefully have $s'=s$ (or $s'\le s$). There are a few technical obstacles
to overcome. The existence of a suitable partition of a given multiset
$\cM$ of points into $h$-spaces is formalized in  Definition~\ref{def_partitionable}. Here $\cM$ is described by parameters
$\sigma$ and $\varepsilon_1,\dots,\varepsilon_{r-1}$. In Lemma~\ref{lemma_compute_parameters_from_partition} we show how to compute
the parameters of a possible faithful projective $h-(n,r,s,\mu)_q$ system from this data and deduce necessary conditions for the existence of a partition for the parameters $\varepsilon_1,\dots,\varepsilon_{r-1}$. An
additional condition for $\sigma$ is concluded in Lemma~\ref{lemma_sigma_constraint}. In Theorem~\ref{thm_main} we show
that these conditions are also sufficient.

\begin{definition}
  \label{def_partitionable}
  Let $\sigma\in\N$, $\varepsilon_1,\dots,\varepsilon_{r-1}\in\Z$, and let
  $V$ denote the $r$-dimensional ambient space $\PG(r-1,q)$.
  We say that a faithful projective $h-(n,r,s)_q$ system $\cS$ has \emph{type $\sigma[r]-\sum_{i=1}^{r-1}\varepsilon_i [i]$} if there exist subspaces $S_1\le \dots\le S_{r-1}$ with $\dim\!\left(S_i\right)=i$
  and
  \begin{equation}
    \sum_{S\in\cS} \chi_S=\sigma\cdot\chi_V-\sum_{i=1}^{r-1}\varepsilon_i\cdot \chi_{S_i}.
  \end{equation}
  We say that $\sigma[r]-\sum_{i=1}^{r-1}\varepsilon_i [i]$ is \emph{$h$-partitionable over $\F_q$} if a faithful projective $h-(n,r,s)_q$ system with type $\sigma[r]-\sum_{i=1}^{r-1}\varepsilon_i [i]$ exists for
  suitable parameters $n$ and $s$.
\end{definition}

Note that all chains $S_1\le\dots\le S_{r-1}$ are isomorphic, so that the
notion of being $h$-partitionable does not depend on the choice of the
subspaces $S_1,\dots,S_{r-1}$. However, the chosen restriction usually
causes rather large values of $\sigma$ and does not cover the full generality
of the Solomon-Stiffler construction. More precisely, if $\varepsilon_i\in\N$ for all $1\le i\le r-1$, then we need to choose $\sigma\ge \sum_{i=1}^{r-1}\varepsilon_i$, which is indeed the worst possible choice for the removal of the subspace. On the other hand, a more
general definition causes several technical complications as we will
briefly outline in Appendix~\ref{sec_generalization_type}. We remark that the same restriction, in terms of the Solomon--Stiffler construction, was also used in \cite{pan2025optimal}.

Next we give a few constructions.
\begin{lemma}
  \label{lemma_vsp_type}
  For $r>a\ge h$ with $r\equiv a\pmod h$ and $\sigma\in\N_{\ge 1}$
  we have that $\sigma[r]-\sigma[a]$ is $h$-partitionable over $\F_q$.
\end{lemma}
\begin{proof}
  If $a>h$, then Lemma~\ref{lemma_construction_x_preparation} yields the
  existence of a faithful projective $h-(n,r,s)_q$ system $\cS$ with
  type $[r]-[a]$ and we can use $\sigma$ copies thereof. For $a=h$ we
  replace $\cS$ by a spread of $h$-spaces of $\PG(r-1,q)$ where we remove
  an arbitrary element.
\end{proof}

\begin{lemma}
  \label{lemma_partition_1}
  For $1\le j\le h$ and $r\ge 2h+1-j$
  \begin{equation}
    [j]_q\cdot[r]-1\cdot [r-h-1+j]-\left([j]_q-1\right)\cdot [r-h-1]
  \end{equation}
  is $h$-partitionable over $\F_q$.
\end{lemma}
\begin{proof}
  Let $A$ be an $(r-h-1)$-space and $B\ge A$ be an $(r-h-1+j)$-dimensional subspace of $\PG(r-1,q)$. By $K_1,\dots,K_l$ we denote the $l:=[j]_q$
  $(r-h)$-spaces with $A\le K_i\le B$. For $1\le i\le l$ let $\cV_i$ be a vector space partition of $\PG(r-1,q)$ of type $h^{t_h} (r-h)^1$ where the special $(r-h)$-space coincides with $K_i$, see Lemma~\ref{lemma_lifting}. The desired faithful projective $h-(n,r,s)_q$ system is then given by the union of the $t_h$ $h$-dimensional non-special subspaces of the
  $l$ vector space partitions $\cV_i$.
\end{proof}
For $h=2$ the construction of Lemma~\ref{lemma_partition_1} is stated as construction $\cL^\star$ in \cite{kurz2024optimal}, see also
\cite[Lemma 3]{bierbrauer2021optimal}.


\begin{lemma}
  \label{lemma_partitionable_union}
  If $\sigma[r]-\sum_{i=1}^{r-1}\varepsilon_i[i]$ and
  $\sigma'[r]-\sum_{i=1}^{r-1}\varepsilon_i'[i]$ are
  $h$-partitionable over $\F_q$, then
  $\left(\sigma+\sigma'\right)\cdot[r]-\sum_{i=1}^{r-1}\left(\varepsilon_i+\varepsilon_i'\right)\cdot[i]$ is $h$-partitionable over $\F_q$.
\end{lemma}
\begin{proof}
  Fix some subspaces $S_1\le\dots\le S_{r-1}$ as in Definition~\ref{def_partitionable}.
  Let $\cS$ be a faithful projective $h-(n,r,s,\mu)_q$ system with type
  $\sigma[r]-\sum_{i=1}^{r-1}\varepsilon_i[i]$ and $\cS'$ be a faithful projective $h-(n',r,s',\mu')_q$ system with type
  $\sigma'[r]-\sum_{i=1}^{r-1}\varepsilon_i'[i]$, then the multiset union of
  the elements of $\cS$ and $\cS'$ is a faithful projective $h-(n+n',r,s+s',\mu+\mu')_q$ system with type  $\left(\sigma+\sigma'\right)\cdot[r]-\sum_{i=1}^{r-1}\left(\varepsilon_i+\varepsilon_i'\right)\cdot[i]$.
\end{proof}

\begin{lemma}
  \label{lemma_construction_x_consequence}
  For $r>h\ge 2$ with $r\equiv 1\pmod h$ we have that
  $[h-1]_q\cdot [r]+q^{h-1}\cdot [1]$, $\left([h]_q-1\right)\cdot[r]
  -q\cdot [h-1]$, and $\left([h]_q-1\right)\cdot[r]+[h]-q\cdot [h-1]$
  are $h$-partitionable over $\F_q$.
\end{lemma}
\begin{proof}
  Let us first consider the case $r=h+1$. Let $\cS$ be the faithful projective
  $h-(n,r,s,\mu)_q$ system that consists of all $h$-spaces that contain point
  $S_1$ (as in Definition~\ref{def_partitionable}), so that $n=[h]_q$.
  Then $S_1$ is contained in $[h]_q$ elements and every other point is
  contained in $[h-1]_q$ elements, i.e.\ $\cS$ has type
  $[h-1]_q\cdot[r]+q^{h-1}\cdot [1]$.
  Let $\cS'$ be the faithful projective $h-(n',r,s',\mu')_q$ system that consists of all $h$-spaces that do not contain $S_{h-1}$ (as in
  Definition~\ref{def_partitionable}), so that $\cS'$ has type $\left([h]_q-1\right)\cdot[r]-q\cdot [h-1]$. Adding $S_h$ (as in
  Definition~\ref{def_partitionable}) to $\cS'$ gives type $\left([h]_q-1\right)\cdot[r]+[h]-q\cdot [h-1]$.

  If $r>h+1$, then Lemma~\ref{lemma_vsp_type} shows that $\sigma[r]-\sigma[h+1]$
  is $h$-partitionable for each $\sigma\in \N$, so that the statement follows
  from Lemma~\ref{lemma_partitionable_union}.
\end{proof}
We remark that the first construction of Lemma~\ref{lemma_construction_x_consequence} is also described in e.g.\ \cite[Theorem 4]{bierbrauer2024asymptotic} and \cite[Lemma 10]{10693309} for $q=h=2$. Geometrically it is the smallest covering of the set of points of $\PG(r-1,q)$ by lines.

From Lemma~\ref{lemma_field_reduction}, based on field reduction, we conclude:
\begin{lemma}
  \label{lemma_field_construction_star_partition}
  If $\sigma[r]-\sum_{i=1}^{r-1}\varepsilon_i [i]$ is $h$-partitionable
  over $\F_{q^l}$, so is $\sigma[rl]+\sum_{i=1}^{r-1}\varepsilon_i [il]$.
\end{lemma}

If $\sigma[r]-\sum_{i=1}^{r-1}\varepsilon_i [i]$ is $h$-partitionable over
$\F_q$, then we can compute the parameters of a corresponding faithful projective $h-(n,r,s,\mu)_q$ system $\cS$ as well as of its dual $\cS^\perp$ from this data. Moreover, we obtain some necessary conditions on the parameters $\varepsilon_1,\dots,\varepsilon_{r-1}$. Later on we will see in Theorem~\ref{thm_main} that they are also sufficient for the existence
of $\cS$ for a sufficiently large $\sigma$ satisfying an additional modulo
constraint, see Lemma~\ref{lemma_sigma_constraint}. So, Lemma~\ref{lemma_system_parameters_from_multiset} specializes to:
\begin{lemma}
  \label{lemma_compute_parameters_from_partition}
  If $\cS$ is a faithful projective $h-(n,r,s,\mu)_q$ system with type
  $\sigma[r]-\sum_{i=1}^{r-1}\varepsilon_i [i]$, then we have
  \begin{equation}
    \label{formula_n}
    n=\left(\sigma[r]_q-\sum_{i=1}^{r-1}\varepsilon_i [i]_q\right)/[h]_q,
  \end{equation}
  \begin{equation}
    \label{eq_s}
    s=\max_{1\le j\le r} \left(s_1 -\sum_{i=1}^{j-1} \varepsilon_iq^{i-h}\right),
  \end{equation}
  where
  \begin{equation}
    \label{eq_s1}
    s_1=\left(\sigma[r-h]_q-\sum_{i=h}^{r-1}\varepsilon_i[i-h]_q+\sum_{i=1}^{h-1}\varepsilon_i q^{i-h}[h-i]_q\right)/[h]_q,
  \end{equation}
  and
  \begin{equation}
    \label{eq_mu}
    \mu=\max_{1\le j\le r} \left(\sigma -\sum_{i=1}^{j-1} \varepsilon_i \right).
  \end{equation}
  Moreover, $\varepsilon_i$ is divisible by $q^{h-i}$ for all $1\le i\le h-1$ and
  \begin{equation}
    \label{packing_cond}
    \sum_{i=1}^{r-1}\varepsilon_i[i]_q \equiv 0\pmod {[\gcd(r,h)]_q}.
  \end{equation}
  The dual $\cS^\perp$ of $\cS$ is a faithful projective $h'-(n,r,\mu,s)_q$  system with type
  $\sigma'[r]-\sum_{i=1}^{r-1}\varepsilon_i'[i]$, where $h'=r-h$,
  $\sigma'=s_1$, and
  $\varepsilon'_i=\varepsilon_{r-i} \cdot q^{h'-i}$ for all $1\le i\le r-1$.
\end{lemma}
\begin{proof}
  Let $\cM$ be the multiset of points covered by the elements of $\cS$ and
  $S_1\le\dots\le S_{r-1}$ be subspaces as in
  Definition~\ref{def_partitionable}. Since $\cM$ has cardinality
  $$
    \sigma[r]_q-\sum_{i=1}^{r-1}\varepsilon_i [i]_q
  $$
  and one $h$-space contains $[h]_q$ points, we conclude Equation~(\ref{formula_n}).

  For an arbitrary point $P$ let $1\le j\le r$ denote the minimal
  integer such that $P\not\le S_j$, where we set $j=r$ if $P\le S_{r-1}$.
  With this we have
  \begin{equation}
    \cM(P)=\sigma-\sum_{i=1}^{j-1} \varepsilon_i,
  \end{equation}
  which implies Equation~(\ref{eq_mu}).

  For an arbitrary hyperplane $H$ let $1\le j\le r$ denote the minimal
  integer such that $S_j\not\le H$, where we set $j=r$ if $H=S_{r-1}$.
  Counting points gives
  $$
    \cM(H)= \sigma[r-1]_q-\sum_{i=1}^{j-1}\varepsilon_i [i]_q-
    \sum_{i=j}^{r-1} \varepsilon_i[i-1]_q
    =\sigma[r-1]_q-\sum_{i=1}^{r-1}\varepsilon_i [i-1]_q
    -\sum_{i=1}^{j-1} \varepsilon_iq^{i-1}.
  $$
  The number $s_j$ of elements of $\cS$ contained in $H$ is given by
  $\left(\cM(H)-n\cdot[h-1]_q\right)/q^{h-1}$, so that
  \begin{eqnarray*}
    s_j &=& \left( \sigma[r-1]_q-\sum_{i=1}^{r-1}\varepsilon_i [i-1]_q
    -\sum_{i=1}^{j-1} \varepsilon_iq^{i-1}
    -\left(\sigma[r]_q-\sum_{i=1}^{r-1}\varepsilon_i [i]_q\right)\cdot\frac{[h-1]_q}{[h]_q} \right)/q^{h-1} \\
    &=&  \left(\sigma\cdot\left([h]_q[r-1]_q-[h-1]_q[r]_q\right)-\sum_{i=1}^{r-1} \varepsilon_i\cdot\left([h]_q[i-1]_q-[h-1]_q[i]_q\right) \right)/\left(q^{h-1}\cdot [h]_q\right)\\
    && -\sum_{i=1}^{j-1} \varepsilon_iq^{i-h}\\
    &\overset{\text{(\ref{cross_difference})}}{=} & \left(\sigma[r-h]_q-\sum_{i=h}^{r-1}\varepsilon_i[i-h]_q+\sum_{i=1}^{h-1}\varepsilon_i q^{i-h}[h-i]_q\right)/[h]_q-\sum_{i=1}^{j-1} \varepsilon_iq^{i-h}.
  \end{eqnarray*}
  This verifies Equation~(\ref{eq_s1}) and yields
  \begin{equation}
    s_j=s_1-\sum_{i=1}^{j-1} \varepsilon_iq^{i-h}
  \end{equation}
  for $2\le j\le r$, which implies Equation~(\ref{eq_s}).
  Since $s_j\in\N$ for $1\le j\le r$, we can recursively
  conclude that $\varepsilon_i$ is divisible by $q^{h-i}$ for
  $i=1,\dots,h-1$. By Equation~(\ref{eq_gcd}) we have $\gcd\!\left([r]_q,[h]_q\right)=\left[\gcd(r,h)\right]_q$, so that Equation~(\ref{formula_n})
  implies Equation~(\ref{packing_cond}).

  By Lemma~\ref{lemma_dual} $\cS^\perp$ is a faithful projective
  $h'-(n,r,s,\mu)_q$ system for $h':=r-h$. If $x$ elements of $\cS$ are
  contained in a hyperplane $H$, then $x$ elements of $\cS^\perp$ contain
  the point $H^\perp$. Note that $\left(S_{r-1}\right)^\perp \le \dots\le
  \left(S_1\right)^\perp$ are the subspaces as in
  Definition~\ref{def_partitionable} for $\cS^\perp$ and that we have
  $\operatorname{dim}\!\left(\left(S_{i}\right)^\perp\right)
  =r-i$ for all $1\le i\le r-1$. For every hyperplane $H\not\ge S_1$
  we have $H^\perp \not\le \left(S_1\right)^\perp$, so that $\sigma'=s_1$.
  For every integer $1\le i\le r-1$ we have
  $$
    \varepsilon'_{r-i}=s_i-s_{i+1}=\varepsilon_i q^{i-h},
  $$
  which is equivalent to $\varepsilon_i'=\varepsilon_{r-i} q^{h'-i}$.
\end{proof}
\begin{corollary}
  If all $\varepsilon_i$ are non-negative, then $s=s_1$ and $\mu=\sigma$
  (using the notation from Lemma~\ref{lemma_compute_parameters_from_partition}).
\end{corollary}

\begin{corollary}
  \label{cor_compute_parameters_from_partition_series}
  If $\cS_t$ is a faithful projective $h-\left(n_t,r,s_t\right)_q$ system with type $\left(\sigma+t\cdot\tfrac{[h]_q}{[\gcd(r,h)]_q}\right)\cdot [r]-\sum_{i=1}^{r-1}\varepsilon_i [i]$, where $\varepsilon_1=\dots=\varepsilon_{h-1}=0$ and $\varepsilon_h,\dots,\varepsilon_{r-1}\in \N$, then we have
  \begin{equation}
    n_t=t\cdot\frac{[r]_q}{[\gcd(r,h)]_q}+\left(\sigma[r]_q-\sum_{i=h}^{r-1}\varepsilon_i [i]_q\right)/[h]_q,
  \end{equation}
  \begin{equation}
     s_t=t\cdot\frac{[r-h]_q}{[\gcd(r,h)]_q} +\left(\sigma[r-h]_q-\sum_{i=h}^{r-1}\varepsilon_i[i-h]_q\right)/[h]_q,
  \end{equation}
  and
  \begin{equation}
     n_t-s_t=    t\cdot\frac{[h]_q}{[\gcd(r,h)]_q}\cdot q^{r-h} +
     \sigma\cdot q^{r-h}-\sum_{i=h}^{r-1}\varepsilon_i\cdot q^{i-h}.
  \end{equation}
\end{corollary}

Next we show that the conditions for $\varepsilon_1,\dots,\varepsilon_{r-1}$ imply the corresponding conditions for $\varepsilon_1',\dots,\varepsilon_{r-1}'$ when defined formally, i.e.\ without assuming the existence of a faithful projective system or its dual.
\begin{lemma}
  \label{lemma_formal}
  Let $q$ be a prime power, $r>h\ge 1$, $\varepsilon_1,\dots,\varepsilon_{r-1}\in \Z$ such that $q^{h-i}$ divides $\varepsilon_i$ for
  all $1\le i\le h-1$ and
  $$
    \sum_{i=1}^{r-1}\varepsilon_i[i]_q \equiv 0\pmod {[\gcd(r,h)]_q}.
  $$
  Setting $h':=r-h$ and $\varepsilon_i':=\varepsilon_{r-i} \cdot q^{h'-i}$
  for $1\le i\le r-1$, we have $\varepsilon_1',\dots,\varepsilon_{r-1}'
  \in \Z$ , $q^{h-i}$ divides $\varepsilon_i'$ for all $1\le i\le h-1$,
  and
  \begin{equation}
    \label{eq_mod_gen}
    \sum_{i=1}^{r-1}\varepsilon_i'[i]_q \equiv 0\pmod {[\gcd(r,h')]_q}.
  \end{equation}
\end{lemma}
\begin{proof}
  First we observe that $\varepsilon_i'=\varepsilon_{r-i} \cdot q^{h'-i}$
  is an integer and divisible by $q^{h'-i}$ for all $1\le i\le h'$. For
  all $h'+1\le i\le r-1$ we have $\varepsilon_i'=\varepsilon_{r-i} \cdot q^{h'-i}\in\Z$ since $1\le r-i\le r-h'-1=h-1$ and $\varepsilon_{r-i}$ is divisible by $q^{h-r+i}=q^{i-h'}$.

  In order to verify Equation~(\ref{eq_mod_gen}) we observe
  $g:=\operatorname{gcd}(r,h)=\operatorname{gcd}(r,h')$, so that it
  suffices to show $\varepsilon_i'\cdot[i]_q +\varepsilon_{r-i}
  \cdot[r-i]_q \equiv 0\pmod {[g]_q}$ for all $1\le i\le r-1$.

  For $h'\ge i$ we have
  $$
    \varepsilon_i'\cdot[i]_q +\varepsilon_{r-i}\cdot[r-i]_q = \varepsilon_{r-i} \cdot \left(q^{h'-i}[i]_q+[r-i]_q\right)\equiv \varepsilon_{r-i} \cdot [h']_q \equiv 0 \pmod {[g]_q}.
  $$
  For $i>h'$ we have
  $$
    \varepsilon_i'\cdot[i]_q +\varepsilon_{r-i}\cdot[r-i]_q
    =q^{h'-i}\varepsilon_{r-i} \cdot \left([i]_q+[r-i]_q\cdot q^{i-h'}\right)
    \equiv q^{h'-i}\varepsilon_{r-i} \cdot [h]_q\equiv 0 \pmod {[g]_q}.
  $$
\end{proof}

Lemma~\ref{lemma_sigma_constraint_premultiset} specializes to:
\begin{lemma}
  \label{lemma_sigma_constraint}
  If $x[r]-\sum_{i=1}^{r-1}\varepsilon_i [i]$ is $h$-partitionable over $\F_q$ for $x\in\{\sigma,\sigma'\}$ then
  $$\left(\sigma+t\cdot\frac{[h]_q}{[\gcd(r,h)]_q}\right)\cdot[r]-\sum_{i=1}^{r-1}\varepsilon_i [i]$$ is $h$-partitionable over $\F_q$ for all $t\ge 0$
  and we have $\sigma\equiv \sigma'\pmod {\frac{[h]_q}{[\gcd(r,h)]_q}}$.
\end{lemma}

In other words, if $\sigma[r]-\sum_{i=1}^{r-1}\varepsilon_i [i]$ is
$h$-partitionable over $\F_q$, then
$\sigma'[r]-\sum_{i=1}^{r-1}\varepsilon_i [i]$ is $h$-partitionable
over $\F_q$ for all $\sigma'\ge \sigma$ if the latter yields an integer
in Equation~(\ref{formula_n}). For those situations where we are not
interested in the smallest possible value $\sigma$ such that
$\sigma[r]-\sum_{i=1}^{r-1}\varepsilon_i [i]$ is $h$-partitionable
over $\F_q$, for given parameters $\varepsilon_1,\dots,\varepsilon_{r-1}$,
we introduce the following notion:
\begin{definition}
  We say that $\star[r]-\sum_{i=1}^{r-1}\varepsilon_i [i]$ is
  $h$-partitionable over $\F_q$ if there exists an integer $\sigma$
  such that $\sigma[r]-\sum_{i=1}^{r-1}\varepsilon_i [i]$ is
  $h$-partitionable over $\F_q$.
\end{definition}
\begin{lemma}
  \label{lemma_negation}
  If $\star[r]-\sum_{i=1}^{r-1}\varepsilon_i [i]$ is  $h$-partitionable
  over $\F_q$, so is $\star[r]+\sum_{i=1}^{r-1}\varepsilon_i [i]$.
\end{lemma}
\begin{proof}
  Fix some subspaces $S_1\le\dots\le S_{r-1}$ as in Definition~\ref{def_partitionable}.   Consider a faithful projective
  $h-(n,r,s)_q$ system $\cS$ with type
  $\sigma[r]-\sum_{i=1}^{r-1}\varepsilon_i [i]$ for a suitable $\sigma$.
  Let $\cS'$ be the faithful projective $h-(n',r,s')_q$ system with type
  $\sigma'[r]$ that consists of all $h$-spaces of $\PG(r-1,q)$ and let $\mu'$
  be the maximal number of occurrences of an element in $\cS$. Then, $\mu'$ copies of $\cS'$ give the desired partition after removing the elements
  of $\cS$ (with their respective multiplicity).
\end{proof}

It is an interesting problem to determine for which parameters
$\varepsilon_1,\dots,\varepsilon_{r-1}$ we have that
$\star[r]-\sum_{i=1}^{r-1}\varepsilon_i [i]$ is $h$-partitionable
over $\F_q$. Clearly we need $r\ge h$ and
$\sum_{i=1}^{r-1}\left|\varepsilon_i\right|=0$ if $r=h$. Additionally we
have the packing condition (\ref{packing_cond}) and
that $\varepsilon_i$ has to be divisible by $q^{h-i}$ for all $1\le i\le r-1$, see Lemma~\ref{lemma_compute_parameters_from_partition}.
These conditions are indeed sufficient.

\begin{lemma}
  \label{lemma_vsp_reduction}$\,$\\[-6mm]
  \begin{itemize}
  \item[(i)] Let $\star[r]-\sum_{i=1}^{r-1}\varepsilon_i[i]$ be $h$-partitionable
  over $\F_q$, $h\le a<j$ with $a\equiv j\pmod h$, $\varepsilon_{i}'=\varepsilon_i$ for all $1\le i\le r-1$ with $i\notin\{a,j\}$, $\varepsilon_j'=\varepsilon_j+1$, and
  $\varepsilon_a'=\varepsilon_a-1$. Then,
  $\star[r]-\sum_{i=1}^{r-1}\varepsilon_i'[i]$ is $h$-partitionable
  over $\F_q$.
  \item[(ii)] If $\star[r]-\sum_{i=1}^{r-1}\varepsilon_i[i]$ is
  $h$-partitionable over $\F_q$, then
  $\star[r+th]-\sum_{i=1}^{r-1}\varepsilon_i[i]$ is $h$-partitionable
  over $\F_q$ for all integers $t\ge 1$.
  \end{itemize}
\end{lemma}
\begin{proof}
  Let $\cS$ be a faithful projective $h-(n,r,s)_q$ system with type
  $\sigma[r]-\sum_{i=1}^{r-1}\varepsilon_i[i]$ for a suitably large
  $\sigma$. Fix some subspaces $S_1\le\dots\le S_{r-1}$ as in Definition~\ref{def_partitionable}. By $\cS'$ we denote the faithful projective   $h-(n',r,s')_q$ system with type
  $\sigma'[r]-\sum_{i=1}^{r-1}\varepsilon_i[i]$ that arises as the
  multiset union of the elements of $\cS$ and the set of all
  $h$-spaces in $\PG(r-1,q)$.

  By Lemma~\ref{lemma_vsp_type} $[j]-[a]$ is $h$-partitionable over
  $\F_q$, so that we identify $S_j$ with $\PG(j-1,q)$ to construct
  a faithful projective $h-(n'',r,s'')_q$ system $\cS''$ with
  type $0[r]+[j]-[a]$. Since $\cS'$ contains every $h$-space in
  $\PG(r-1,q)$ at least once and $\cS''$ contains every $h$-space
  in $\PG(r-1,q)$ at most once, removing the elements in $\cS''$
  from the elements of $\cS'$ gives a faithful projective system
  with type $\sigma'[r]-\sum_{i=1}^{r-1}\varepsilon_i'[i]$, which
  shows (i).

  Lemma~\ref{lemma_vsp_type} yields the existence of a faithful
  projective $h-(n''',r,s''')_q$ system $\cS'''$ with type
  $\sigma[r+th]-\sigma[r]$. Fix some subspaces $S_1'\le\dots\le
  S_{r+th-1}'$ as in Definition~\ref{def_partitionable}. Embedding
  $\cS$ in $S_r'$ and taking the multiset union with the elements of
  $\cS'''$ gives a faithful projective system
  with type $\sigma[r+th]-\sum_{i=1}^{r-1}\varepsilon_i[i]$, which
  shows (ii).
\end{proof}

\begin{theorem}
  \label{thm_main}
  Let $q$ be a prime power, $r> h\ge 1$, $g:=\gcd(r,h)$, and
  $\varepsilon_1,\dots,\varepsilon_{r-1}\in \Z$ such that
  $q^{h-i}$ divides $\varepsilon_i$ for all $1\le i<h$ and
  \begin{equation}
    \label{eq_packing_condition_lemma}
    \sum_{i=1}^{r-1} \varepsilon_i\cdot[i]_q \equiv 0\pmod {[g]_q}.
  \end{equation}
  Then $\star[r]-\sum\limits_{i=1}^{r-1}\varepsilon_i[i]$ is
  $h$-partitionable over $\F_q$.
\end{theorem}
\begin{proof}
  We prove by induction over $h$. The statement is obvious for $h=1$, so
  that we assume $h\ge 2$ in the following.

  Due to Theorem~\ref{thm_partition} $\star[r]+[h]_q\cdot [i]$
  is $h$-partitionable over $\F_q$ for each $h+1\le i\le r-1$, so that
  also $\star[r]-[h]_q\cdot [i]$ is $h$-partitionable over $\F_q$ for
  each $h+1\le i\le r-1$ by Lemma~\ref{lemma_negation}. Using
  Lemma~\ref{lemma_partitionable_union} we can assume
  $\varepsilon_i\ge 0$ for all $h+1\le i\le r-1$ since these operations
  do not violate Equation~(\ref{eq_packing_condition_lemma}).

  By iteratively applying Lemma~\ref{lemma_vsp_reduction}.(i) with
  $a=j-h$ we can assume $\varepsilon_j=0$ for all $2h\le j\le r-1$
  because these operations do not violate Equation~(\ref{eq_packing_condition_lemma}) since
  $[j]_q\equiv[a]_q\mod[h]_q$. Due to Lemma~\ref{lemma_vsp_reduction}.(ii)
  we can additionally assume $r<3h$ noting that we still have $\gcd(r,h)=g$.

  If $2h\le r<3h$, then by using Lemma~\ref{lemma_partitionable_union} with  Lemma~\ref{lemma_partition_1} we can assume $\varepsilon_i=0$ for all
  $r-h\le i\le r-1$, since these operations do not violate
  Equation~(\ref{eq_packing_condition_lemma}). Due to Lemma~\ref{lemma_vsp_reduction}.(ii) we can additionally assume
  $r\le 2h$ noting that we still have $\gcd(r,h)=g$.

  For $h<r<2h$ we have $1\le r-h< h$ and we let $\cS$ denote the
  desired faithful projective $(n,r,s)_q$ system with type
  $\sigma[r]-\sum_{i=1}^{r-1} \varepsilon_i[i]$.
  Using Lemma~\ref{lemma_formal} and Lemma~\ref{lemma_compute_parameters_from_partition} we can deduce
  the existence of $\cS^\perp$ for a sufficiently large value of $\sigma$
  from the induction hypothesis.

  It remains to consider the case $r=2h$ where $g=h$. Since $\gcd(r-1,h)=1$
  we can conclude the existence of a faithful projective $h-(n',r-1,s')_q$
  system $\cS'$ with type $t[r-1]-\sum_{i=1}^{r-2} \varepsilon_i[i]$, for some suitably large integer $t\ge -\varepsilon_{r-1}$, from
  the previous part of the proof. By applying  Lemma~\ref{lemma_compute_parameters_from_partition} to
  $\cS'$ we conclude
  \begin{equation}
    \sum_{i=1}^{r-2} \varepsilon_i[i]_q\equiv t[r-1]_q
    \pmod{[h]_q}
  \end{equation}
  from Equation~(\ref{formula_n}). From  Equation~(\ref{eq_packing_condition_lemma}) we conclude
  \begin{equation}
    \sum_{i=1}^{r-2} \varepsilon_i[i]_q\equiv -\varepsilon_{r-1}[r-1]_q
    \pmod{[h]_q},
  \end{equation}
  so that
  $$
    \left(t+\varepsilon_{r-1}\right)\cdot[r-1]_q \equiv 0 \pmod{[h]_q}.
  $$
  Since $\gcd(r-1,h)=1$ Equation~(\ref{eq_gcd}) implies
  $t+\varepsilon_{r-1}\equiv 0 \pmod{[h]_q}$, so that Theorem~\ref{thm_partition} yields that $\left(t+\varepsilon_{r-1}\right)\cdot[r-1]$ as well as $0\cdot[r]+\left(t+\varepsilon_{r-1}\right)\cdot[r-1]$ are $h$-partitionable over $\F_q$. From Lemma~\ref{lemma_negation} we conclude the existence of a faithful projective $h-(n'',r,s'')_q$ system $\cS''$ with type
  $\sigma''[r]-\left(t+\varepsilon_{r-1}\right)\cdot[r-1]$. Fix some subspaces $S_1\le\dots\le S_{r-1}$ as in Definition~\ref{def_partitionable} for $\PG(r-1,q)$. Embedding $\cS'$ in $S_{r-1}$, such that
  $S_1\le\dots\le S_{r-2}$ are the subspaces as in Definition~\ref{def_partitionable}, and taking the multiset union with
  the elements from $\cS''$ gives a faithful projective system
  with type $\sigma''[r]-\sum_{i=1}^{r-1}\varepsilon_i[i]$.
\end{proof}
Note that the stated conditions only ensure that $t[r]-\sum\limits_{i=1}^{r-1}\varepsilon_i[i]$ is $h$-partitionable if $t$ is sufficiently large. Relatively small values of $t$ can sometimes ruled out by some kind of {\lq\lq}dimension
arguments{\rq\rq}, see the example in Remark~\ref{remark_dimension_arguments}.

\begin{example}
  \label{ex_partition_via_theorem}
  We want to explicitly show that $\star[8]-[7]-[6]-[5]-[4]-[3]$ is
  $2$-partitionable over $\F_2$ going along the lines of the proof of
  Theorem~\ref{thm_main}. In our situation we start with $r=8$, $h=2$
  and already  have $\varepsilon_i\ge 0$ for all $1\le i\le r-1$, so
  that we start by reducing to the situation where $\varepsilon_i=0$
  for all $2h\le i\le r-1$, i.e., we need to show that
  $\star[8]-3[3]-2[2]$ is $2$-partitionable over $\F_2$. By reducing
  $r$ by multiples of $h$ it remains to show that $\star[4]-3[3]-2[2]$
  is $2$-partitionable over $\F_2$. From Lemma~\ref{lemma_partition_1}
  we know that $[4]-[2]$ and $3[4]-[3]-2[1]$
  are $2$-partitionable over $\F_2$, so that it remains to show that
  $\star[4]+6[1]$ is $2$-partitionable over $\F_2$. Here we do
  not reduce $r$ by $h$ since we need $r>h$. Indeed, $\star[2]+
  6[1]$ is not $2$-partitionable over $\F_2$. So, we are
  in the situation $r=2h$ and first try to determine a $t$
  such that $t[3]+6[1]$ is $2$-partitionable over $\F_2$,
  where we have $h<r<2h$. Via duality we need to show that
  $\star[3]+3[2]$ is $1$-partitionable over $\F_2$. Clearly, we
  can construct a faithful projective $1-(9,3,9,3)_2$ system $\cS$
  with type $0[3]+3[2]$. By Lemma~\ref{lemma_compute_parameters_from_partition} $\cS^\perp$ is a
  faithful projective $2-(9,3,3,9)_2$ system with type $3[3]+6[1]$, i.e.\
  we have $t=3$. More directly, $\cS^\perp$ arises by taking each line in
  $\PG(2,2)$ that contains point $S_1$ (as in Definition~\ref{def_partitionable}) three times. From Theorem~\ref{thm_partition} we know that $0[4]+3[3]$ is $2$-partitionable
  over $\F_2$. Since $\PG(3,2)$ contains $35$ lines we conclude that
  $105[4]-3[3]$ is $2$-partitionable over $\F_2$, so that also $105[4]+6[1]$
  is $2$-partitionable over $\F_2$. Since $[4]-[2]$ and $3[4]-[3]-2[1]$ are
  $2$-partitionable over $\F_2$ we conclude that also $116[4]-3[3]-2[2]$ is
  $2$-partitionable over $\F_2$. From Lemma~\ref{lemma_vsp_reduction} we get that also $116[8]-3[3]-2[2]$ and $116[8]-[7]-[6]-[5]-[4]-[3]$ are $2$-partitionable over $\F_2$. Lemma~\ref{lemma_sigma_constraint} yields
  that $\sigma[8]-[7]-[6]-[5]-[4]-[3]$ is $2$-partitionable over $\F_2$
  for all $\sigma\ge 116$ with $\sigma\equiv 2\pmod 3$.
\end{example}
\begin{remark}
  \label{remark_partition_via_theorem}
  When we apply the constructive proof of Theorem~\ref{thm_main} to
  compute $\sigma$ such that $\sigma[r]-\sum_{i=1}^{r-1}\varepsilon_i$ is $h$-partitionable over $\F_q$, the obtained value will usually be
  rather large, so that we 
  only give a rough sketch of a relatively crude explicit upper bound. We apply several lemmas that modify the coefficients $\varepsilon_i$ or convert them to $[h]_q-\varepsilon$, so that for each operation at most $\sum_{i=1}^{r-1} \left|\varepsilon_i\right|\,+\,r\cdot [h]_q$ times the full set of all $\qbin{r}{h}{q}$ $h$-spaces in $\PG(r-1,q)$ is needed. Some operations shift the $\varepsilon_i$ along the chain to the lower indices and we apply several lemmas, so that we multiply this number with $10r$. Moving to the dual can blow up the $\varepsilon_i$ by a factor of at most $q^h$. Since we reduce the dimension of the ambient space at most $r$ times, we end up with
  $$
    \qbin{r}{h}{q}\cdot q^{rh}\cdot r\cdot 10 r\cdot \left(\sum_{i=1}^{r-1} \left|\varepsilon_i\right|\,+\,r\cdot [h]_q\right)
    \le 10r^2\cdot q^{r(h+1)}\cdot\left(\sum_{i=1}^{r-1} \left|\varepsilon_i\right|\,+\,r\cdot q^h\right)
  $$
  as a sufficient lower bound for $\sigma$.

  For the case of Example~\ref{ex_partition_via_theorem} we remark that
  Lemma~\ref{lemma_partition_1} also shows that $3[8]-[7]-2[5]$ is $2$-partitionable over $\F_2$, so that Lemma~\ref{lemma_vsp_reduction}
  implies that $3[8]-[7]-[5]-[3]$ is $2$-partitionable over $\F_2$. Since
  $[8]-[6]$ and $[8]-[4]$ are $2$-partitionable over $\F_2$ by
  Lemma~\ref{lemma_vsp_type}, Lemma~\ref{lemma_partitionable_union}
  implies that $5[8]-[7]-[6]-[5]-[4]-[3]$ is $2$-partitionable over $\F_2$.
\end{remark}

\begin{definition}
  For integers $n>s\ge 1$, $r>h\ge 1$, and a prime power $q$ let
  the \emph{surplus} be defined by
  \begin{equation}
    \theta(n,r,s,h,q):= n\cdot[h]_q-g_q\!\left(r,q^{h-1}\cdot(n-s)\right).
  \end{equation}
\end{definition}
So the surplus is negative iff $n$ is larger than the Griesmer upper
bound for $n_q(r,h;s)$.
\begin{lemma}
  \label{lemma_asymptotic_construction}
  Let $n>s\ge 1$, $r>h\ge 1$ be integers and $q$ be a prime power.
  If $\theta(n,r,s,h,q)\ge 0$, then there exists a faithful projective
  $h-\left(n+t\cdot \frac{[r]_q}{[\gcd(r,h)]_q},r,
  s+t\cdot\frac{[r-h]_q}{[\gcd(r,h)]_q}\right)_q$ system $\cS_t$ for all sufficiently large $t$.
\end{lemma}
\begin{proof}
  Setting $d':=q^{h-1}\cdot(n-s)$ and $n':=g_q(r,d')$ we have
  $\theta(n,r,s,h,q)=n[h]_q-n'\ge 0$. Due to
  Lemma~\ref{lemma_parameters_griesmer_code} we can choose
  integers $\sigma$, $\varepsilon_1,\dots,\varepsilon_{r-1}$, with
  $\sigma\ge 0$ and $0\le \varepsilon_i<q$ for all $1\le i\le r-1$,
  such that
  \begin{equation}
    \label{eq_d_prime_asymptotic_construction}
    d'=\sigma q^{r-1}-\sum_{i=1}^{r-1}\varepsilon_iq^{i-1}
  \end{equation}
  and
  \begin{equation}
    \label{eq_n_prime_asymptotic_construction}
    n'=\sigma[r]_q-\sum_{i=1}^{r-1}\varepsilon_i[i]_q.
  \end{equation}
  Since $d'$ is divisible by $q^{h-1}$ we have $\varepsilon_i=0$ for all
  $1\le i\le h-1$. Let $\tau:=\theta(n,r,s,h,q)$, $\sigma':=\sigma+\tau$,
  $\varepsilon'_{r-1}=\varepsilon_{r-1}+\tau q$, and $\varepsilon_i'=\varepsilon_i$
  for all $1\le i\le r-2$, so that $\varepsilon_i'\in\N$ for all $1\le i\le r-1$
  and $\varepsilon_i'=0$ for all $1\le i\le h-1$. Note that
  \begin{equation}
    \label{eq_d_prime_asymptotic_construction2}
    d'=\sigma' q^{r-1}-\sum_{i=1}^{r-1}\varepsilon_i'q^{i-1}
  \end{equation}
  and
  \begin{equation}
    \label{eq_n_prime_asymptotic_construction2}
    n[h]_q=n'+\tau =\sigma'[r]_q-\sum_{i=1}^{r-1}\varepsilon_i'[i]_q,
  \end{equation}
  so that
  \begin{equation}
    \label{eq_cond_asymptotic_construction}
    \sum_{i=1}^{r-1}\varepsilon_i'[i]_q \equiv 0 \pmod{[\gcd(r,h)]_q}.
  \end{equation}
  From Theorem~\ref{thm_main} and Lemma~\ref{lemma_sigma_constraint} we
  conclude that $\left(\sigma'+t\cdot\frac{[h]_q}{[\gcd(r,h)]_q}\right)\cdot[r]-\sum_{i=1}^{r-1}\varepsilon_i'[i]$ is $h$-partitionable over $\F_q$ for all sufficiently large $t$. From Corollary~\ref{cor_compute_parameters_from_partition_series},
  Equation~(\ref{eq_d_prime_asymptotic_construction2}), and
  Equation~(\ref{eq_n_prime_asymptotic_construction2}) we compute the stated parameters of $\cS_t$.
\end{proof}
For other results and notions of asymptotically good additive codes we refer e.g.\ to \cite{shi2018asymptotically}.

\begin{example}
  \label{ex_asymptotic_1}
  Let $q=2$, $r=8$, $h=2$, $n=30$, and $s=8$, so that $d'=q^{h-1}\cdot (n-s)=44$. In order to compute $n':=g_2(8,44)$ we determine  parameters $\sigma$, $\varepsilon_1,\dots,\varepsilon_7$, so that Equation~(\ref{eq_d_prime_asymptotic_construction}) is satisfied, where $0\le \varepsilon_1,\dots,\varepsilon_7<2$. Here the unique solution is
  given by $\sigma=1$, $\varepsilon_7=\varepsilon_5=\varepsilon_3=1$, and
  $\varepsilon_i=0$ otherwise. Via Equation~(\ref{eq_n_prime_asymptotic_construction})
  we can compute $n'=[8]_2-[7]_2-[5]_2-[3]_2=90$, so that the surplus is given by $\theta(30,8,8,2,2)=n[h]_q-n'=0$. Applying Theorem~\ref{thm_main} shows that $\star[8]-[7]-[5]-[3]$ is $2$-partitionable over $\F_2$.
  As mentioned in Remark~\ref{remark_partition_via_theorem}
  $3[8]-[7]-[5]-[3]$ is $2$-partitionable over $\F_2$, so that
  $n_2(8,3;8+21t)\ge 30+85t$ for all $t\ge 2$.
\end{example}
\begin{remark}
\label{remark_asymptotic_one_weight_bound}
The series of projective systems in Example~\ref{ex_asymptotic_1} attains indeed the Griesmer upper bound in relative terms, see Inequality~(\ref{eq_limit_relative_one_weight_bound}). In general we can start with any set of parameters that
do not violate the Griesmer bound in Lemma~\ref{lemma_indirect_upper_bounds} and obtain a series of projective systems that tend towards the upper bound in Lemma~\ref{lemma_one_weight_bound}.
\end{remark}
\begin{remark}
  \label{remark_varepsilon_modification}
  The proof of Lemma~\ref{lemma_asymptotic_construction} provides a more general
  construction. Starting from parameters $\sigma,\varepsilon_1,\dots,\varepsilon_{r-1}$ obtained from Equation~(\ref{eq_d_prime_asymptotic_construction}) and Equation~(\ref{eq_n_prime_asymptotic_construction}) we can modify to
  any non-negative parameters $\sigma',\varepsilon_1',\dots,\varepsilon_{r-1}'$
  satisfying equations (\ref{eq_d_prime_asymptotic_construction2}), (\ref{eq_n_prime_asymptotic_construction2}) and the conditions from Theorem~\ref{thm_main}, where Condition~(\ref{eq_cond_asymptotic_construction}) is automatically satisfied, see e.g.\ Example~\ref{ex_asymptotic_2}.
\end{remark}

\begin{example}
  \label{ex_asymptotic_2}
  Let $q=2$, $r=9$, $h=3$, $n=55$, and $s=7$, so that $d'=q^{h-1}\cdot (n-s)=192$. In order to compute $n':=g_2(9,192)$ we determine  parameters $\sigma$, $\varepsilon_1,\dots,\varepsilon_8$, so that Equation~(\ref{eq_d_prime_asymptotic_construction}) is satisfied, where $0\le \varepsilon_1,\dots,\varepsilon_8<2$. Here the unique solution is
  given by $\sigma=1$, $\varepsilon_7=1$, and $\varepsilon_i=0$ otherwise. Via Equation~(\ref{eq_n_prime_asymptotic_construction})
  we can compute $n'=[9]_2-[7]_2=384$, so that the surplus is given by $\theta(55,9,7,3,2)=n[h]_q-n'=1$. We have that $\star[9]-[7]$ is not
  $3$-partitionable over $\F_2$ since Condition~(\ref{packing_cond}) is violated. Here we use the surplus to modify $\star[9]-[7]$
  to $\star[9]-2[6]$ (instead of $\star[9]-2[8]$).
  Since $2[3]-2[2]$ is $1$-partitionable over
  $\F_8$ we have that $2[9]-2[6]$ is $3$-partitionable over $\F_2$,
  see Lemma~\ref{lemma_field_construction_star_partition}. Thus, we have
  $n_2(9,3;7+9t)\ge 55+73t$ for all $t\ge 1$.
\end{example}

The reason for the modification of the initial parameters $\sigma,\varepsilon_1,\dots,\varepsilon_{r-1}$ in the proof of Lemma~\ref{lemma_asymptotic_construction} is not only to satisfy Condition~(\ref{eq_cond_asymptotic_construction}) but also to obtain the {\lq\lq}right{\rq\rq} $n$ as shown in the following example:
\begin{example}
  \label{ex_asymptotic_3}
  Let $q=2$, $r=7$, $h=3$, $n=13$, and $s=2$, so that $d'=q^{h-1}\cdot (n-s)=44$. In order to compute $n':=g_2(7,44)$ we determine  parameters $\sigma$, $\varepsilon_1,\dots,\varepsilon_6$, so that Equation~(\ref{eq_d_prime_asymptotic_construction}) is satisfied, where $0\le \varepsilon_1,\dots,\varepsilon_6<2$. Here the unique solution is
  given by $\sigma=1$, $\varepsilon_5=\varepsilon_3=1$, and $\varepsilon_i=0$
  otherwise. Via Equation~(\ref{eq_n_prime_asymptotic_construction})
  we can compute $n'=[7]_2-[5]_2-[3]_2=89$, so that the surplus is given
  by $\theta(13,7,2,3,2)=n[h]_q-n'=2$. Here it happens that
  $\star[7]-[5]-[3]$ is $3$-partitionable over $\F_2$ nevertheless the
  surplus is positive. However Lemma~\ref{lemma_compute_parameters_from_partition} implies that whenever
  $\sigma[7]-[5]-[3]$ is $3$-partitionable over $\F_2$, then we have
  $\sigma \equiv 3 \pmod 7$ so that we obtain faithful projective
  $3-(49+127t,7,6+15t)_2$ systems for sufficiently large values of $t$
  instead of the desired faithful projective $3-(13+127t,7,2+15t)_2$ systems.
\end{example}

\begin{theorem}
  \label{thm_attained_asymptotically}
  For all sufficiently large $s$ we have that $n_q(r,h;s)$ attains the
  Griesmer upper bound, see Definition~\ref{def_griesmer_uppper_bound}
  and Inequality~(\ref{ie_direct_upper_bound}).
\end{theorem}
\begin{proof}
  Let $s_i:=\tfrac{[r-h]_q}{[\gcd(r,h)]_q} -i$ for $0\le i<
  \tfrac{[r-h]_q}{[\gcd(r,h)]_q}$ and $n_i$ be the Griesmer upper
  bound for $n_q(r,h;s_i)$, i.e.\ $n_i[h]_q\ge
  g_q\!\left(r,q^{h-1}\cdot(n_i-s_i)\right)$ while $\left(n_i+1\right
  )\cdot[h]_q<g_q\!\left(r,q^{h-1}\cdot(n_i+1-s_i)\right)$. Let $\sigma_i,
  \varepsilon_{1,i},\dots,\varepsilon_{r-1,i}\in\N$ with $\varepsilon_{j,i}<q$ for
  all $1\le j\le r-1$ be uniquely defined by
  \begin{equation}
    d_i:=q^{h-1}\cdot \left(n_i-s_i\right)=\sigma_i\cdot q^{r-1}-\sum_{j=1}^{r-1}\varepsilon_{j,i}\cdot q^{j-1},
  \end{equation}
  so that
  \begin{equation}
    g_q(r,d_i)=\sigma_i\cdot [r]_q-\sum_{j=1}^{r-1}\varepsilon_{j,i}\cdot [j]_q,
  \end{equation}
  using Lemma~\ref{lemma_parameters_griesmer_code}, and $\theta\!\left(n_{i},r,s_{i},h,q\right)\ge 0$. Similarly, let $\sigma_i',
  \varepsilon_{1,i}',\dots,\varepsilon_{r-1,i}'\in\N$ with $\varepsilon_{j,i}'<q$ for
  all $1\le j\le r-1$ be uniquely defined by
  \begin{equation}
    d_i':=q^{h-1} +d_i=\sigma_i'\cdot q^{r-1}-\sum_{j=1}^{r-1}\varepsilon_{j,i}'\cdot q^{j-1},
  \end{equation}
  so that
  \begin{equation}
    g_q(r,d_i')=\sigma_i'\cdot [r]_q-\sum_{j=1}^{r-1}\varepsilon_{j,i}'\cdot [j]_q
  \end{equation}
  and $\theta\!\left(n_{i}+1,r,s_{i},h,q\right)< 0$.
  Now, let $s_{i,t}:=s_i+t\cdot\tfrac{[r-h]_q}{[\gcd(r,h)]_q}$ and
  $n_{i,t}:=n_i+t\cdot \tfrac{[r]_q}{[\gcd(r,h)]_q}$, so that
  Lemma~\ref{lemma_parameters_griesmer_code} implies
  \begin{equation}
    d_{i,t}:=q^{h-1}\cdot \left(n_{i,t}-s_{i,t}\right)=\left(\sigma_i+t\cdot \frac{[h]_q}{[\gcd(r,h)]_q}\right)\cdot q^{r-1}-\sum_{i=1}^{r-i}\varepsilon_i\cdot q^{i-1},
  \end{equation}
  \begin{equation}
    g_q\!\left(r,d_{i,t}\right)=t\cdot \frac{[r]_q}{[\gcd(r,h)]_q}\cdot[h]_q+
    \sigma_i\cdot [r]_q-\sum_{j=1}^{r-1}\varepsilon_{j,i}\cdot [j]_q,
  \end{equation}
  \begin{equation}
    d_{i,t}':=q^{h-1} +d_{i,t}=\left(\sigma_i'+t\cdot \frac{[h]_q}{[\gcd(r,h)]_q}\right)\cdot q^{r-1}-\sum_{i=1}^{r-i}\varepsilon_i'\cdot q^{i-1},
  \end{equation}
  and
  \begin{equation}
    g_q\!\left(r,d_{i,t}'\right)=t\cdot \frac{[r]_q}{[\gcd(r,h)]_q}\cdot[h]_q+
    \sigma_i'\cdot [r]_q-\sum_{j=1}^{r-1}\varepsilon_{j,i}'\cdot [j]_q.
  \end{equation}
  Thus we have $$\theta\!\left(n_{i,t},r,s_{i,t},h,q\right)=\theta\!\left(n_{i},r,s_{i},h,q\right)\ge 0\,\,\text{and}\,\,\theta\!\left(n_{i,t}+1,r,s_{i,t},h,q\right)=\theta\!\left(n_{i}+1,r,s_{i},h,q\right)<0,$$
  i.e.\ the Griesmer upper bound for $n_q\!\left(r,h;s_{i,t}\right)$ is given
  by $n_{i,t}$ for all $t\in \N$. Lemma~\ref{lemma_asymptotic_construction} yields the existence of a faithful $h-\!\left(n_{i,t},r,s_{i,t}\right)_q$ system $\cS_{i,t}$ for all sufficiently large $t$.
\end{proof}

Building up on the techniques used to prove Theorem~\ref{thm_attained_asymptotically}, in \cite{b_symbol} it was shown that the Griesmer bound for linear codes in the $b$-symbol metric can always be attained with equality if the minimum distance is sufficiently large. We remark that there also exists a Griesmer-type bound for finite
quasi-Frobenius rings \cite{shiromoto2001griesmer}.

\begin{corollary}
  \label{cor_attained_asymptotically}
  Each $[n,r/h,d]_q^h$ code satisfies
  \begin{equation}
     \label{ie_direct_upper_bound2}
    n\ge \left\lceil \frac{g_q\!\left(r,d\cdot q^{h-1}\right)}{[h]_q}\right\rceil = \left\lceil\frac{\sum_{i=0}^{r-1}\left\lceil d\cdot q^{h-1-i}\right\rceil}{[h]_q}\right\rceil
    =d+ \left\lceil \frac{g_q(r-h+1,d)-d   }{[h]_q}\right\rceil.
  \end{equation}
  If $d$ is sufficiently large, given $r$, $h$, and $q$, then
  Inequality~(\ref{ie_direct_upper_bound2}) can always be attained with equality. For increasing minimum distance $d$ the relative distance $\delta=\tfrac{d}{n}$ tends to $\tfrac{q^{k-1}}{[k]_q}=1-\tfrac{1}{q}+O\!\left(q^{-2}\right)$ and the rate $R=\tfrac{k}{n}$ tends to zero.
\end{corollary}

\begin{remark}
  \label{remark_nice_vs_compact}
  Note that the values $n_2(7,2;4)=12$ and $n_2(7,2;5)=17$, see Theorem~\ref{thm_n_2_7_2_s}, encode the information that $[12,7/2,8]_2^2$,
  $[14,7/2,9]_2^2$, $[15,7/2,10]_2^2$, $[16,7/2,11]_2^2$, $[17,7/2,12]_2^2$ codes exist and are length-minimal w.r.t.\ the other parameters. Clearly,
  a $[15,7/2,10]_2^2$ code can be obtained from a $[17,7/2,12]_2^2$ code by
  shortening, which favors tabulating the $n_q(r,h;s)$ values. However,
  some inequalities have a nicer representation in terms of $[n,r/h,d]_q^h$
  codes.
\end{remark}

\medskip

The proof of Theorem~\ref{thm_attained_asymptotically} suggests the following algorithm to determine explicit formulas for $n_q(r,h;s)$ assuming that $s$
is sufficiently large. For all $0\le i< \tfrac{[r-h]_q}{[\gcd(r,h)]_q}$ compute
the Griesmer upper bound $n_i$ for $n_q(r,h;s_i)$ where $s_i=\tfrac{[r-h]_q}{[\gcd(r,h)]_q} -i$. Then we have
\begin{equation}
  n_q\!\left(r,h;t\cdot \frac{[r-h]_q}{[\gcd(r,h)]_q} -i \right)=
  t\cdot \frac{[r]_q}{[\gcd(r,h)]_q} -\left(\frac{[r]_q}{[\gcd(r,h)]_q}-n_i\right)
\end{equation}
for all sufficiently large $t$. As an example we mention:
\begin{proposition} (Cf.~\cite[Table I]{guan2023some},\cite[Table II]{10693309})
  \label{prop_q_2_h_2_r_7}
  For all sufficiently large $t$ we have\\[-10mm]
  \begin{itemize}
    \item $n_2(7,2;31t)=127t$;\\[-11mm]
    \item $n_2(7,2;31t-1)=127t-5$;\\[-11mm]
    \item $n_2(7,2;31t-2)=127t-10$;\\[-11mm]
    \item $n_2(7,2;31t-3)=127t-15$;\\[-11mm]
    \item $n_2(7,2;31t-4)=127t-20$;\\[-11mm]
    \item $n_2(7,2;31t-5)=127t-21$;\\[-11mm]
    \item $n_2(7,2;31t-6)=127t-26$;\\[-11mm]
    \item $n_2(7,2;31t-7)=127t-31$;\\[-11mm]
    \item $n_2(7,2;31t-8)=127t-36$;\\[-11mm]
    \item $n_2(7,2;31t-9)=127t-41$;\\[-11mm]
    \item $n_2(7,2;31t-10)=127t-42$;\\[-11mm]
    \item $n_2(7,2;31t-11)=127t-47$;\\[-11mm]
    \item $n_2(7,2;31t-12)=127t-52$;\\[-11mm]
    \item $n_2(7,2;31t-13)=127t-55$;\\[-11mm]
    \item $n_2(7,2;31t-14)=127t-60$;\\[-11mm]
    \item $n_2(7,2;31t-15)=127t-63$;\\[-11mm]
    \item $n_2(7,2;31t-16)=127t-68$;\\[-11mm]
    \item $n_2(7,2;31t-17)=127t-73$;\\[-11mm]
    \item $n_2(7,2;31t-18)=127t-76$;\\[-11mm]
    \item $n_2(7,2;31t-19)=127t-81$;\\[-11mm]
    \item $n_2(7,2;31t-20)=127t-84$;\\[-11mm]
    \item $n_2(7,2;31t-21)=127t-87$;\\[-11mm]
    \item $n_2(7,2;31t-22)=127t-92$;\\[-11mm]
    \item $n_2(7,2;31t-23)=127t-95$;\\[-11mm]
    \item $n_2(7,2;31t-24)=127t-100$;\\[-11mm]
    \item $n_2(7,2;31t-25)=127t-105$;\\[-11mm]
    \item $n_2(7,2;31t-26)=127t-108$;\\[-11mm]
    \item $n_2(7,2;31t-27)=127t-113$;\\[-11mm]
    \item $n_2(7,2;31t-28)=127t-116$;\\[-11mm]
    \item $n_2(7,2;31t-29)=127t-121$;\\[-11mm]
    \item $n_2(7,2;31t-30)=127t-126$.
  \end{itemize}
\end{proposition}
In \cite{10693309} the stated formulas of Proposition~\ref{prop_q_2_h_2_r_7}
were indeed shown to be true for all
$t\ge 2$ and $n_2(7,2;31-i)$ was determined for all
$i\in\{0,\dots,31\}\backslash\{19,24,25\}$, referring to \cite{kurz2024computer} for $i=18$ and \cite{guan2023some} for the previous state of the art. The final
three missing cases were resolved in \cite{kurz2024optimal}. Having Theorem~\ref{thm_attained_asymptotically} at hand, the determination of $n_q(r,h;s)$ can be reduced to a finite set of cases for $s$ for each given
list of parameters $q$, $r$, and $h$, as the results of Griesmer \cite{griesmer1960bound}, Solomon and Stiffler \cite{solomon1965algebraically} imply for $h=1$ a.k.a.\ linear codes. For the ease of the reader we collect
such results in Section~\ref{generic_results} in the appendix and mention that
it is an important and difficult task to determine $n_q(r,h;s)$ for the
remaining \textit{small} values of $s$, even for $h=1$. We try to summarize
the current state of knowledge for small parameters, as we are aware of, in
Section~\ref{sec_small_parameters}.

\begin{table}[htp]
  \begin{center}
    \begin{tabular}{rrrrrrl}
      \hline
      $q$ & $r$ & $h$ & $i$ & $s_{i,t}:=t\cdot \tfrac{[r-h]_q}{[\gcd(r,h)]_q}-i$
      & $n_q\left(r,h;s_{i,t}\right)$ & $n_q\left(r,h;s_{i,t}\right)-\overline{n}_q\left(r,h;s_{i,t}\right)$\\
      \hline
      2 & 8 & 2 & 13 & $21t-13$ & $85t-55$ & 2 \\
      2 & 8 & 2 & 14 & $21t-14$ & $85t-60$ & 2 \\
      2 & 8 & 2 & 18 & $21t-18$ & $85t-76$ & 2 \\
      2 & 8 & 2 & 19 & $21t-19$ & $85t-81$ & 2 \\
      \hline
      3 & 6 & 2 &  7 & $10t-7$  & $91t-67$ & 3 \\
      3 & 6 & 2 &  8 & $10t-8$  & $91t-77$ & 3 \\
      3 & 6 & 2 &  9 & $10t-9$  & $91t-87$ & 3 \\
      \hline
      2 & 9 & 3 &  5 & $9t-5$ & $73t-43$ & 2 \\
      2 & 9 & 3 &  6 & $9t-6$ & $73t-52$ & 2 \\
      2 & 9 & 3 &  7 & $9t-7$ & $73t-59$ & 4 \\
      2 & 9 & 3 &  8 & $9t-8$ & $73t-68$ & 4 \\
      \hline
      4 & 6 & 2 & 9 & $17t-9$ & $273t-149$ & 4 \\
      4 & 6 & 2 & 10 & $17t-10$ & $273t-166$ & 4 \\
      4 & 6 & 2 & 11 & $17t-11$ & $273t-183$ & 4 \\
      4 & 6 & 2 & 12 & $17t-12$ & $273t-200$ & 4 \\
      4 & 6 & 2 & 13 & $17t-13$ & $273t-213$ & 8 \\
      4 & 6 & 2 & 14 & $17t-14$ & $273t-230$ & 8 \\
      4 & 6 & 2 & 15 & $17t-15$ & $273t-247$ & 8 \\
      4 & 6 & 2 & 16 & $17t-16$ & $273t-264$ & 8 \\
      \hline
    \end{tabular}
    \caption{Parameterized series of improvements for additive codes.}
    \label{table_improvements}
  \end{center}
\end{table}

As observed in Inequality~(\ref{eq_limit_relative_one_weight_bound}), cf.~Remark~\ref{remark_asymptotic_one_weight_bound}, the upper bound of
Lemma~\ref{lemma_one_weight_bound} can be reached asymptotically, i.e.\ we
have
\begin{equation}
  \lim_{s\to \infty} \frac{n_q(r,h;s)\cdot [r-h]_q}{s\cdot[r]_q}=1
\end{equation}
for all parameters. Thus, we have
\begin{equation}
  n_q(r,h;s)>\overline{n}_q(r,h;s)
\end{equation}
for all sufficiently large $s$ whenever $r/h$ is not an integer, i.e., additive codes outperform linear codes for large enough $s$ if $r/h$ is fractional.
For the case $r/h\in \N$ we provide:
\begin{theorem}
  \label{thm_parameters_improvements_integral}
  For all sufficiently large $t$ we have the following improvements
  of additive codes over linear codes listed in Table~\ref{table_improvements}
  and the tables in Section~\ref{sec_parameterized_outperform}.
\end{theorem}

The four series of improvements for $n_2(8,2;s)$ were already mentioned in \cite{kurz2024optimal}. For $(q,r,h)=(2,10,2)$ we refer to Table~\ref{table_improvements_q_2_r_10_h_2}. An example for $h=4$ is given
by Table~\ref{table_improvements_q_2_r_12_h_4}, which also shows
that $n_q\left(r,h;s_{i,t}\right)-\overline{n}_q\left(r,h;s_{i,t}\right)$ does
not need to be a power of the characteristic of $\F_q$. An example for $q=5$ is given by Table~\ref{table_improvements_q_5_r_6_h_2}. In Theorem~\ref{thm_r_3h_h_2}
we will determine the parameterized series of improvements $n_q(6,2;s)>\overline{n}_q(6,2;s)$ for all field sizes $q$. It turns out
that there are $q(q-2)$ series and the improvement is divisible by $q$.

\begin{lemma}
  For $k=r/h\in\N$ we have
  \begin{equation}
     n_q(r,h;s)-\overline{n}_q(r,h;s)< (k-1)\cdot q^{h}
  \end{equation}
  for large $s$.
\end{lemma}
\begin{proof}
  First we observe that the upper bounds from Lemma~\ref{lemma_one_weight_bound} coincide for $n_q(r,h;s)$ and $n_{q^h}(k,1;s)$, so that it suffices to upper bound the difference between the Griesmer bound and its version without rounding in Inequality~(\ref{ie_direct_upper_bound_approximation}). So, from
  $$
    d\cdot \frac{[k]_{q^h}}{\left(q^h\right)^{k-1}}-g_{q^h}(k,d)\le k-1
  $$
  and
  $$
    \frac{[k]_{q^h}}{[k-1]_{q^h}}\cdot s- \frac{[k]_{q^h}}{[k-1]_{q^h}}\cdot (s-k+1)\le
    (k-1)\cdot \frac{[k]_{q^h}}{[k-1]_{q^h}}< (k-1)\cdot q^{h}
  $$
  we conclude the statement.
\end{proof}
\noindent
We remark that $\lim_{s\to\infty} n_q(r,h;s)-\overline{n}_q(r,h;s)=\infty$ if $r/h\notin\N$.

\begin{table}[htp]
  \begin{center}
    \begin{tabular}{rrrrrrl}
      \hline
      $q$ & $r$ & $h$ & $s$ & $n_q(r,h;s)$ & $\overline{n}_q(r,h;s)$ \\
      \hline
      2 & 8 & 2 &  9 & 33 & 31 \\
      2 & 8 & 2 & 10 & 36 & 34 \\
      2 & 8 & 2 & 11 & 40 & 39 \\
      2 & 8 & 2 & 14 & 54 & 50 \\
      2 & 8 & 2 & 27 & 107 & 103 \\
      2 & 9 & 3 &  7 & 51--55 & 49 \\
      2 & 9 & 3 & 10 & 76--78 & 74 \\
      3 & 6 & 2 &  3 & 21 & 17 \\
      3 & 6 & 2 &  8 & 66--68 & 65 \\
      \hline
    \end{tabular}
    \caption{Sporadic improvements for additive codes.}
    \label{table_sporadic}
  \end{center}
\end{table}

In the introduction we have mentioned the improvement
$n_3(6,2;3)=21$ over $\overline{n}_3(6,2;3)=17$ from \cite{clerck2001perp}.
Corollary~\ref{cor_asymptotic_one_weight} yields $n_3(6,2;10t-7)\ge 91t-70$
for all $t\ge 1$. Note that we have $\overline{n}_3(6,2;10t-7)= 91t-70$ for all $t\ge 2$. In Table~\ref{table_sporadic} we collect the known sporadic
improvements where $n_q(r,h;s)>\overline{n}_q(r,h;s)$ which are not part of a
parametric series of improvements. So e.g.\ we do not mention
$n_2(8,2;23)=89>87=\overline{n}_2(8,2;23)$ since
Corollary~\ref{cor_asymptotic_one_weight} gives $n_2(8,2;21t-19)=85t-81$ for
$t\ge 2$, which is contained in Table~\ref{table_improvements}.
The data for $(q,r,h)=(2,8,2)$ descends from \cite{kurz2024optimal}.

\begin{table}[htp]
  \begin{center}
    \begin{tabular}{rrrrrrl}
      \hline
      $q$ & $r$ & $h$ & $s$ & $n_q(r,h;s)$ & $\overline{n}_q(r,h;s)$ \\
      \hline
      2 & 10 & 2 & 18 & 64--68 & 62--66 \\
      2 & 22 & 2 & 26 & 56--90 & 55--74 \\
      \hline
    \end{tabular}
    \caption{Temporary sporadic improvements for additive codes.}
    \label{table_temporary_sporadic}
  \end{center}
\end{table}

Sometimes one can construct additive codes which have better parameters than
the best known linear codes. The lower bound $n_2(70,2;40)\ge 47$
and the lower bounds listed in Table~\ref{table_temporary_sporadic} were
obtained in \cite{guan2023some}.

\begin{remark}
   \label{remark_outperform_blockcode}
   Let $b_q(k,d)$ the minimum length of a block code with an alphabet $\cA$ of size $q$, minimum Hamming distance $d$, and code size $q^k$, where $k\in\mathbb{Q}$. Clearly, we have $b_q(k,d)\le g_q(k,d)$ for sufficiently large $d$ due to the existence of linear Griesmer codes. It is an interesting problem to determine parameters where $b_q(k,d)< g_q(k,d)$. In \cite{guerrini2016optimal} $b_2(2,d)=g_2(2,d)$ and $b_2(3,d)=g_2(3,d)$ was shown. For $(q,k)=(2,2)$ all optimal codes even have to be linear. There exists a binary block code of size $20$, length $19$, and minimum Hamming distance $10$ attaining the Plotkin bound, see e.g.\ \cite{agrell2001table}. So, we have $b_2(4,10+15t)<g_2(4,10+15t)$ and $b_2(4,9+15t)<g_2(4,9+15t)$ for all $t\in\N$ using $4$-dimensional simplex codes. Given the existence of a \emph{Hadamard matrix} of order
   $2^k+4$ \cite[Theorem 37]{guerrini2016optimal} yields the existence of
   some minimum distance $d$ with $b_2(k,d)< g_2(k,d)$ for $k\ge 4$. Indeed, it is well known that Hadamard matrices can be used to construct good binary block codes, see e.g.\ \cite{levenshtein1964application}. A good recursive construction for block codes is the \emph{Plotkin sum} a.k.a.\
   \emph{$(u|u+v)$-construction} \cite{plotkin1960binary,sloane1970new}. Other parameters $(k,d)$ with $b_2(k,d)< g_2(k,d)$ are e.g.\ given by
   $(5,6)$, $(5,10)$, $(6,6)$, $(7,6)$, and $(8,6)$. Due to our incomplete information on $b_2(k,d)$ it is not know if also $(6,10)$ and $(7,10)$ can be included in the previous list.
\end{remark}

\pagebreak

\section{Linear equation systems over $\mathbf{\mathbb{Z}}$ and the Smith normal form}
\label{sec_SNF}
For each field $\F$ the Gauss--Jordan algorithm computes the set of
all $x\in\F^n$ satisfying $Ax=b$ for a matrix $A\in\F^{m\times n}$ and a
vector $b\in\F^m$. The Hermite normal form \cite{hermite1851introduction}
of $A$ is a triangular matrix and the Smith normal form (SNF) \cite{smith1861xv}
for $A$ is a diagonal matrix. Both allow to solve linear equation systems over
$\Z$ and can be computed efficiently, see e.g.\ \cite{bradley1971algorithms}.
For a brief introduction into the underlying theory we refer e.g.\ to
\cite[chapters 4,5]{schrijver1998theory}. Known results and general techniques
to (theoretically) compute SNFs for incidence matrices, e.g.\ in $\PG(r-1,q)$,
are surveyed in  \cite{sin2013smith}. Here we want to show the relation  to
our topic in an informal way by merely considering examples.

\begin{theorem}
  Let $A$ be a non-zero  $m\times n$ matrix over a principal ideal domain $R$.
  There exist invertible $m\times m$ and $n\times n$-matrices $S$, $T$
  (with entries in $R$) such that
  \begin{equation}\label{eq_SNF}
     SAT=
     \begin{pmatrix}
        \alpha_1 & 0 & 0 & \dots & 0 & \dots & 0 \\
        0 & \alpha_2 & 0 \\
        0 & 0 & \ddots & & \vdots & & \vdots \\
        \vdots & & & \alpha_r \\
        0 & & \dots & & 0 & \dots & 0 \\
        \vdots & &  & & \vdots & & \vdots \\
        0 & & \dots & & 0 & \dots & 0
     \end{pmatrix},
  \end{equation}
  $\alpha_i$ divides $\alpha_{i+1}$ for all $1\le i<r$, and the $\alpha_i$ are
  unique up to multiplication by a unit (in $R$).
\end{theorem}
The right hand side of Equation~(\ref{eq_SNF}) is called the \emph{Smith
normal form (SNF)} of matrix $A$. The elements $\alpha_i$ are called the
\emph{elementary divisors}, \emph{invariants}, or \emph{invariant factors}.
The matrices $S$ and $T$ can be algorithmically obtained by recursively
applying invertible row and column operations to $A$ till it reaches the
desired diagonal form.

\begin{example}
  \label{ex_SNF}
  Let $$A=\begin{pmatrix}
    1 & 1 & 5 & 7 \\
    2 & 8 & 10 & 20 \\
    3 & 3 & 45 & 51 \\
    1 & 7 & 5 & 13 \\
    2 & 2 & 40 & 44
  \end{pmatrix}$$
  and $M$ be the $\Z$-module generated by the rows of $A$. The Smith normal form
  of $A$ is given by
  $$
    D:=SAT=
    \begin{pmatrix}
    1 & 0 &  0 & 0 \\
    0 & 6 &  0 & 0 \\
    0 & 0 & 30 & 0 \\
    0 & 0 &  0 & 0 \\
    0 & 0 &  0 & 0
    \end{pmatrix},
  $$
  where
  $$
    S=\begin{pmatrix}
       1 &  0 &  0 & 0 & 0 \\
      -2 &  1 &  0 & 0 & 0 \\
      -3 &  0 &  1 & 0 & 0 \\
       1 & -1 &  0 & 1 & 0 \\
       1 &  0 & -1 & 0 & 1
    \end{pmatrix}
    \quad\text{and}\quad
    T=\begin{pmatrix}
      1 & -1 & -6 &  1 \\
      0 &  1 & -1 &  1 \\
      0 &  0 &  0 &  1 \\
      0 &  0 &  1 & -1
    \end{pmatrix}.
  $$
  We have
  \begin{eqnarray*}
    M' 
    &\!\!=\!\!& \left\{ \left(z_1,z_2,z_3,z_4\right)\in \Z^{4}\,:\, z_1\equiv 0\!\!\!\!\pmod 1, z_2\equiv 0\!\!\!\!\pmod 6,z_3\equiv 0\!\!\!\!\pmod {30}, z_4=0\right\}
  \end{eqnarray*}
  for the $\Z$-module $M'$ generated by the rows of $D$ and
  \begin{eqnarray*}
    M&=&\big\{ \left(z_1,z_2,z_3,z_4\right)\in \Z^{4}\,:\,z_1\equiv 0\!\!\!\!\pmod 1,  -z_1+z_2\equiv 0\!\!\!\!\pmod 6, \\
    &&  -6z_1-z_2+z_4\equiv 0\!\!\!\! \pmod {30}, z_1+z_2+z_3-z_4=0\big\}.
  \end{eqnarray*}
  Equivalently, the linear equation system
  $D^\top x=b=\left(b_1,\dots,b_5\right)^\top$ has solutions over
  $\Z$ iff $b_2$ is divisible by $6$, $b_3$ is divisible by $30$, and we
  have $b_4=b_5=0$. The full set of solutions is given by
  $$
    \left\{(x_1,\dots,x_4)^\top \in\Z^4\,:\, x_1=b_1/1, x_2=b_2/6,
    x_3=b_3/30\right\}.
  $$
\end{example}

\begin{proposition}
  Let $A\in\Z^{m\times n}$ and $D=SAT$ its SNF. Then $D^\top \tilde{x}=\tilde{b}$ has a solution $\tilde{x}\in\Z^m$ iff $\tilde{b}_i\equiv 0\pmod {\alpha_i}$ for all $1\le i\le r$ and $\tilde{b}_i=0$ for all $r<i\le n$. The set of all solutions is given by
  $$
    \left\{\left(\tilde{x}_1,\dots,\tilde{x}_m\right)^\top\in\Z^m\,:\, \tilde{x}_i=\tilde{b}_i/\alpha_i\,\forall 1\le i\le r\right\}.
  $$
  Moreover, we have $A^\top x=b$ with $x\in\Z^m$ iff $D^\top \tilde{x}=\tilde{b}$ with $\tilde{x}\in\Z^m$, where $\tilde{b}=T^\top b$ and $\tilde{x}^\top=x^\top S^{-1}$.
\end{proposition}

\begin{remark}
  The case of some zero rows in the SNF of a matrix, see Example~\ref{ex_SNF}, is typical for our situation of the incidence matrix $A^{1,h;r,q}$,
  since there are more $h$-spaces than points in $\PG(r-1,q)$ when $r-1>h>1$. So,
  there exist $h$-spaces $S_1,\dots,S_l$ and integers $x_1,\dots,x_l$, not all equal to zero, such that $\cM:=\sum_{i=1}^l x_i\cdot\chi_{S_i}$ is the empty
  multiset of points, i.e.\ $\cM(P)=0$ for every point $P$. In the literature such
  solutions are known as (subspace) \emph{trades}, see e.g.\ \cite{krotov2017minimum}. In $\PG(3,2)$ each set of pairwise disjoint lines
  can be completed to a line spread, so that there clearly exist two different
  line spreads $\cL_1=\left\{L_1,\dots,L_5\right\}$, $\cL_2=\left\{L_6,\dots,L_{10}\right\}$ with $\sum_{i=1}^5 \chi_{L_i}=\sum_{i=6}^{10} \chi_{L_i}$.
\end{remark}

\begin{example}
  \label{example_incidences_lines_points_pg_2_3}
  Let $B$ be the incidence matrix between lines and points in $\PG(2,3)$ and $T$ be
  the column transformation matrix of the Smith normal form of $B$, i.e.\
  $$
  B=
  {\tiny\arraycolsep=0.3\arraycolsep\ensuremath{\left(\begin{array}{rrrrrrrrrrrrr}
  1 & 0 & 0 & 1 & 0 & 0 & 1 & 0 & 0 & 1 & 0 & 0 & 0 \\
  1 & 0 & 0 & 0 & 1 & 0 & 0 & 0 & 1 & 0 & 1 & 0 & 0 \\
  1 & 0 & 0 & 0 & 0 & 1 & 0 & 1 & 0 & 0 & 0 & 1 & 0 \\
  0 & 1 & 0 & 0 & 1 & 0 & 0 & 1 & 0 & 1 & 0 & 0 & 0 \\
  0 & 1 & 0 & 0 & 0 & 1 & 1 & 0 & 0 & 0 & 1 & 0 & 0 \\
  0 & 1 & 0 & 1 & 0 & 0 & 0 & 0 & 1 & 0 & 0 & 1 & 0 \\
  0 & 0 & 1 & 0 & 0 & 1 & 0 & 0 & 1 & 1 & 0 & 0 & 0 \\
  0 & 0 & 1 & 1 & 0 & 0 & 0 & 1 & 0 & 0 & 1 & 0 & 0 \\
  0 & 0 & 1 & 0 & 1 & 0 & 1 & 0 & 0 & 0 & 0 & 1 & 0 \\
  1 & 1 & 1 & 0 & 0 & 0 & 0 & 0 & 0 & 0 & 0 & 0 & 1 \\
  0 & 0 & 0 & 1 & 1 & 1 & 0 & 0 & 0 & 0 & 0 & 0 & 1 \\
  0 & 0 & 0 & 0 & 0 & 0 & 1 & 1 & 1 & 0 & 0 & 0 & 1 \\
  0 & 0 & 0 & 0 & 0 & 0 & 0 & 0 & 0 & 1 & 1 & 1 & 1
  \end{array}\right)}}
  \quad\text{and}\quad
   T=
    {\tiny\arraycolsep=0.3\arraycolsep\ensuremath{\left(\begin{array}{rrrrrrrrrrrrr}
    1 & 0 & 0 & 0 & 0 & 0 & 0 &  0 & -2 & -1 & -1 & -2 &  -3 \\
    0 & 1 & 0 & 0 & 0 & 0 & 0 & -1 & -1 & -1 & -2 & -1 & -11 \\
    0 & 0 & 1 & 0 & 0 & 0 & 0 & -2 &  0 & -1 &  0 &  0 & -11 \\
    0 & 0 & 0 & 1 & 0 & 0 & 0 & -2 &  0 & -2 & -1 &  0 &  -7 \\
    0 & 0 & 0 & 0 & 1 & 0 & 0 &  0 & -2 & -2 &  0 & -1 &  -3 \\
    0 & 0 & 0 & 0 & 0 & 1 & 0 & -1 & -1 & -2 & -2 & -2 &  -3 \\
    0 & 0 & 0 & 0 & 0 & 1 & 1 & -1 & -1 & -3 & -3 & -3 &  -7 \\
    0 & 0 & 0 & 0 & 0 & 0 & 0 &  1 &  0 &  0 &  0 &  0 &  -3 \\
    0 & 0 & 0 & 0 & 0 & 0 & 0 &  0 &  1 &  0 &  0 &  0 &  -3 \\
    0 & 0 & 0 & 0 & 0 & 0 & 1 &  0 &  0 &  0 & -1 & -1 &  -7 \\
    0 & 0 & 0 & 0 & 0 & 0 & 0 &  0 &  0 &  0 &  1 &  0 &  -3 \\
    0 & 0 & 0 & 0 & 0 & 0 & 0 &  0 &  0 &  0 &  0 &  1 &  -3 \\
    0 & 0 & 0 & 0 & 0 & 0 & 0 &  0 &  0 &  0 &  0 &  0 &   1
    \end{array}\right)}}.
  $$
  Here five invariant factors of $B$ equal $3$, the last equals $12$, and the first seven equal $1$.
  So, for $z=\left(z_1,\dots,z_{13}\right)\in\Z^{1\times [3]_3}$ we have  $z\in M:=\left\{x^\top B\,:\, x\in \Z^{[3]_3}\right\}$ iff
  {\footnotesize\begin{eqnarray*}
    2z_2+z_3+z_4+2z_6+2z_7+z_8 & \equiv& 0\pmod 3,\\
    z_1+2z_2+z_5+2z_6+2z_7+z_9 & \equiv& 0\pmod 3,\\
    2z_1+2z_2+2z_3+z_4+z_5+z_6 & \equiv& 0\pmod 3,\\
    2z_1+z_2+2z_4+z_6+2z_{10}+z_{11} & \equiv& 0\pmod 3,\\
    z_1+2z_2+2z_5+z_6+2z_{10}+z_{12} & \equiv& 0\pmod 3,\\
    z_2+z_3+2z_4+2z_7+2z_{10}+z_{13} & \equiv& 0\pmod 3,\\
    \sum_{i=1}^{13} z_i &\equiv& 0\pmod 4,
  \end{eqnarray*}}
\noindent
where we have broken up the condition of the last column of $T$ modulo $12$ into two conditions modulo $3$ and modulo $4$, respectively.
\end{example}

\begin{example}
  Let $C$ be the $\Z_3$-code of the incidence matrix between lines and points in $\PG(2,3)$ and $B$ as in Example~\ref{example_incidences_lines_points_pg_2_3}. Generator matrices
  of $C$ and its dual code $C^\perp$ are e.g.\ given by
  $$
  {\tiny\arraycolsep=0.3\arraycolsep\ensuremath{\left(\begin{array}{rrrrrrrrrrrrr}
  1 & 0 & 0 & 0 & 0 & 1 & 0 & 1 & 0 & 0 & 0 & 1 & 0 \\
  0 & 1 & 0 & 0 & 0 & 1 & 0 & 2 & 2 & 0 & 1 & 0 & 2 \\
  0 & 0 & 1 & 0 & 0 & 1 & 0 & 0 & 1 & 0 & 2 & 2 & 2 \\
  0 & 0 & 0 & 1 & 0 & 2 & 0 & 1 & 2 & 0 & 2 & 1 & 1 \\
  0 & 0 & 0 & 0 & 1 & 2 & 0 & 2 & 1 & 0 & 1 & 2 & 0 \\
  0 & 0 & 0 & 0 & 0 & 0 & 1 & 1 & 1 & 0 & 0 & 0 & 1 \\
  0 & 0 & 0 & 0 & 0 & 0 & 0 & 0 & 0 & 1 & 1 & 1 & 1 \\
  \end{array}\right)}}
  \quad\text{and}\quad
  {\tiny\arraycolsep=0.3\arraycolsep\ensuremath{\left(\begin{array}{rrrrrrrrrrrrr}
  1 & 0 & 0 & 0 & 0 & 1 & 0 & 1 & 0 & 2 & 2 & 0 & 2 \\
  0 & 1 & 0 & 0 & 0 & 1 & 0 & 2 & 2 & 0 & 1 & 0 & 2 \\
  0 & 0 & 1 & 0 & 0 & 1 & 0 & 0 & 1 & 0 & 2 & 2 & 2 \\
  0 & 0 & 0 & 1 & 0 & 2 & 0 & 1 & 2 & 2 & 1 & 0 & 0 \\
  0 & 0 & 0 & 0 & 1 & 2 & 0 & 2 & 1 & 0 & 1 & 2 & 0 \\
  0 & 0 & 0 & 0 & 0 & 0 & 1 & 1 & 1 & 2 & 2 & 2 & 0
  \end{array}\right)}},
  $$
  respectively. Clearly, we have $\dim(C)+\dim(C^\perp)=[3]_3=13$ since $C$ and $C^\perp$ are vector spaces. Now consider
  the incidence matrix
  $$
  B'=
  {\tiny\arraycolsep=0.3\arraycolsep\ensuremath{\left(\begin{array}{rrrrrrrrrrrrr}
  0 & 1 & 1 & 0 & 1 & 1 & 0 & 1 & 1 & 0 & 1 & 1 & 1 \\
  0 & 1 & 1 & 1 & 0 & 1 & 1 & 1 & 0 & 1 & 0 & 1 & 1 \\
  0 & 1 & 1 & 1 & 1 & 0 & 1 & 0 & 1 & 1 & 1 & 0 & 1 \\
  1 & 0 & 1 & 1 & 0 & 1 & 1 & 0 & 1 & 0 & 1 & 1 & 1 \\
  1 & 0 & 1 & 1 & 1 & 0 & 0 & 1 & 1 & 1 & 0 & 1 & 1 \\
  1 & 0 & 1 & 0 & 1 & 1 & 1 & 1 & 0 & 1 & 1 & 0 & 1 \\
  1 & 1 & 0 & 1 & 1 & 0 & 1 & 1 & 0 & 0 & 1 & 1 & 1 \\
  1 & 1 & 0 & 0 & 1 & 1 & 1 & 0 & 1 & 1 & 0 & 1 & 1 \\
  1 & 1 & 0 & 1 & 0 & 1 & 0 & 1 & 1 & 1 & 1 & 0 & 1 \\
  0 & 0 & 0 & 1 & 1 & 1 & 1 & 1 & 1 & 1 & 1 & 1 & 0 \\
  1 & 1 & 1 & 0 & 0 & 0 & 1 & 1 & 1 & 1 & 1 & 1 & 0 \\
  1 & 1 & 1 & 1 & 1 & 1 & 0 & 0 & 0 & 1 & 1 & 1 & 0 \\
  1 & 1 & 1 & 1 & 1 & 1 & 1 & 1 & 1 & 0 & 0 & 0 & 0 \\
  \end{array}\right)}}
  $$
  between affine planes and points in $\PG(2,3)$. By $C'$ we denote the corresponding $\Z_3$-code and observe that every row of $B'$ is orthogonal to every row of $B$ w.r.t.\ $\Z/3\Z$,
  i.e.\ $C'\subseteq C^\perp$. By e.g.\ computing the Hermite normal form of $B'$ we can verify $\dim(C')=6$, so that indeed $C'=C^\perp$. In the context of
  Example~\ref{example_incidences_lines_points_pg_2_3} this means that we can replace the six $\!\!\mod 3$-conditions by $B'z\equiv 0\pmod 3$, which corresponds to thirteen
  single $\!\!\mod 3$-conditions. Of course we can also select six of these such that the corresponding rows of $B'$ generate $C'$.
\end{example}

\begin{theorem}(E.g.~\cite[Theorem 3.1]{chandler2006invariant}.)
  \label{theorem_invariant_factors}
  The invariant factors of $A^{1,h;r,q}$ are all $p$-powers
  except the last, which is a $p$-power times $[h]_q$, where $p$ is the characteristic of $\F_q$.
\end{theorem}

The part $[h]_q$ of the last invariant factor has an easy explanation. Let $S_1,\dots,S_l$ be a list of $h$-spaces in $\PG(r-1,q)$, $x_1, \dots,x_l\in \Z$, and $\cM=\sum_{i=1}^{l} x_i\cdot\chi_{S_i}$ be a premultiset of points. Then clearly,
$\#\cM$ has to be divisible by $[h]_q$ since each $h$-space consists of $[h]_q$ points. For the other invariant factors and their multiplicities we refer to
\cite[Theorem 3.3]{chandler2006invariant}. A corresponding basis over the
\emph{$p$-adic integers} is described in \cite[Section 5]{chandler2006invariant}
and \cite[Theorem 7.2]{chandler2006invariant}. For a more geometric
description, using (generalized) \emph{Reed--Muller codes}, we refer to
\cite{assmus1998polynomial}. If the field size $q$ itself is a prime, then the additional conditions on the premultiset of points $\cM$ are equivalent to
$\cM(A)\equiv 0\pmod {q^{h-1}}$ for all \emph{affine subspaces} $A$ of codimension $1$, i.e., for all $A$ that can be written as $H\backslash K$ where $H$ is a hyperplane of $\PG(r-1,q)$ and $K\le H$ is a $(r-2)$-space. It can be checked easily that those conditions are equivalent to $\cM(H)\equiv \#\cM \pmod{q^{h-1}}$ for every hyperplane $H$, i.e., $q^{h-1}$-divisibility. If $q$ is a proper $p$-power, then those conditions are only necessary and there are further conditions arising from so-called subfield subcodes (or Baer subspaces in geometrical terms), see \cite[Subsection 5.8]{assmus1998polynomial}.\footnote{In \cite{blokhuis1993size} a certain weight function was introduced to capture the extra conditions.} Except the first condition $\
\cM\equiv 0\pmod{[h]_q}$ all conditions are also satisfied for the set of all points of the ambient space $\cV=\PG(r-1,q)$, so that we can also apply them to $\cM$ directly when  studying $\sigma\cdot\chi_V-\cM$ for some  $\sigma\in\N$.

Given these reformulations one can turn the conditions from Lemma~\ref{lemma_partitonable_les} into fairly explicit ones and decide whether $\star[r]-\cM$ is $h$-partitionable over $\F_q$ for a given premultiset of points $\cM$ in $\PG(r-1,q)$. However, determining the smallest $\sigma\in \N$ such that
$\sigma[r]-\cM$ is $h$-partitionable over $\F_q$ is a significant and hard challenge. One example is given by partial spreads of $h$-spaces. To this end let $A_q(r,2h;h)$ denote the maximum number of $h$-spaces in $\PG(r-1,q)$ such that
each point is covered at most once. For a given partial spread $\cP$ let
$\cM$ denote the set of uncovered points. In our notation $\cP$ is a faithful projective $h-(\#\cP,r,s,1)_q$ system with type $1\cdot [r]-\cM$, where
$\#\cP$ and $s$ can be computed from $\cM$ via Lemma~\ref{lemma_compute_parameters_from_premultiset}. In the other direction
we have the conditions $\#\cM\equiv [r]_q\pmod{[h]_q]}$, $\cM(P)\le 1$ for
every point $P$, and $\cM(H)\equiv \#\cM\pmod{q^{h-1}}$ for every hyperplane
$H$ of $\PG(r-1,q)$. As an example we state $129\le A_2(11,8;4)\le 132$
\cite{kurz2017packing}. There cannot be a partial spread of $133$ solids
($4$-spaces) in $\PG(10,2)$ since there is no $8$-divisible set of $52$ points
in $\PG(10,2)$ while there are $8$-divisible sets of cardinality $67$, see
e.g.\ \cite{honold2018partial,honold2019lengths}. So, the open question is
whether $[11]-\cM$
can be partitioned into $132$ solids for one of these choices for $\cM$.
Another example is given by $244\le A_3(8,6;3)\le 248$. If $A_3(8,6;3)=248$,
then the set $\cM$ of uncovered points has to be the unique $9$-divisible set of
points with cardinality $56$ over $\F_3$, which is known as the
\emph{Hill cap} \cite{hill1978caps}, see e.g.\
\cite[Section 6]{honold2018partial} for more details. With e.g.\ $25\,095\,280$ planes in $\PG(7,3)$ the parameters for the Hill cap are already rather large, so that we consider another example for the application of Lemma~\ref{lemma_partitonable_les}.
\begin{example}
  \label{ex_lemma_partitonable_les}
  Let us consider multisets of lines in $\PG(4,2)$, so that the premultiset of points $\cM$ has to be $2$-divisible. Let us choose $\cM$ as a set of $10$ points. Up to symmetry there are $26$ choices for $\cM$, see e.g.\ \cite{projective_divisible_binary_codes}. The smallest possible spanned dimension is four and here $\cM$ is given as the complement of a projective base a.k.a.\ frame of a solid. In a first step we compute
  a solution $x\in\Z^{\qbin{5}{2}{2}}$ of
  $A^{1,2;5,2}\cdot x=\sigma v-z$, where $v$ is the all-one vector of size $\qbin{5}{1}{2}=31$ and $z$ the incidence vector of $\cM$. To this end we compute the Smith normal form of $A^{1,2;5,2}$, i.e.\ integral matrices $S$ and $T$ such that $S \cdot A^{1,2;5,2} \cdot T$ is an integral diagonal matrix. In our situation $D=S \cdot A^{1,2;5,2} \cdot T$ consists of $26$ ones, four twos, and a six in the main diagonal. Setting $x=Tx'$, $v'=Sv$, $z'=Sz$, and multiplying our equation by $S$ from the left we obtain $Dx'=\sigma v'-z'$. Since $\sigma\equiv \#\cM=10\pmod 3$, we choose $\sigma=1$ and easily solve the equations for the first $31$ components of $x'$. The other $155-31=124$ entries of $x'$ can be chosen as arbitrary integers, so that we set them to zero. Using $x=Tx'$ we obtain a solution  $x\in\Z^{\qbin{5}{2}{2}}$ of our initial equation system for $\sigma=1$ by a simple matrix multiplication. All those computations can be performed by a small \texttt{GAP} script using the \texttt{FinInG} package \cite{fining}. Here we only remark that the most negative entry of the computed solution $x$ has a value of $-424$. So, we consider the modified solution $\widetilde{x}:=x+424\cdot v\in\N^{155}$ and $\widetilde{\sigma}=\sigma+15\cdot 424=6361$ as corresponding multiplier for the ambient space, i.e.\ $\left(\widetilde{\sigma}+3t\right)\cdot [5]-\cM$ is $2$-partitionable over $\F_2$ for all $t\in\N$. Instead of setting the unspecified $124$ entries of $x'$ to zero we can also consider them as integer variables and replace $\sigma=1$ by $\sigma=1+3t$ for another variable $t\in\Z$. With this, an ILP model for the minimum possible $\sigma$ such that $\sigma[5]-\cM$ is $2$-partitionable over $\F_2$ is given by the minimization of $t$ subject to $Tx'\ge 0$. It turns out that $4[5]-\cM$ is $2$-partionable into a set of lines over $\F_2$ while $[5]-\cM$ is not. For the latter statement we give an easy theoretical proof. To this end, let $S$ be the solid spanned by the points $P$ with $\cM(P)>0$. Each of the seven lines, for the case $t=0$, intersects $S$ it at least a point. Since only $15-\#\cM=5$ points of $S$ need to be covered, we obtain a contradiction.
\end{example}

As an open problem we pose the question if the described ILP model has computational advantages when combined with the prescription of automorphism groups?

\begin{definition}
  Let $p$ be a prime, $l$ be a positive integer, and $B\in R^{m\times n}$ a
  matrix where $R=\Z$ or $R=\Z/p^l\Z$. The \emph{(linear) $\Z_{p^l}$-code}
  $C$ of $B$ is given by the row span of $B$ w.r.t.\  $\Z/p^l\Z$. The matrix
  $B$ is called a \emph{generator matrix} of $C$. The \emph{dual code}
  $C^\perp$ consists of all row vectors that are orthogonal to all elements
  in $C$ (w.r.t.\ $\Z/p^l\Z$). We also call $C^\perp$ the
  \emph{$\Z_{p^l}$-kernel} of $B$.
\end{definition}

\begin{example}
  \label{example_incidences_planes_points_pg_3_2}
  Let $B$ be the incidence matrix between planes and points in $\PG(3,2)$ and $T$ be
  the column transformation matrix of the Smith normal form of $B$, i.e.\
  $$
  B=
  {\tiny\arraycolsep=0.3\arraycolsep\ensuremath{\left(\begin{array}{rrrrrrrrrrrrrrr}
  1 & 0 & 1 & 0 & 1 & 0 & 1 & 0 & 1 & 0 & 1 & 0 & 1 & 0 & 0 \\
  1 & 0 & 0 & 1 & 1 & 0 & 0 & 1 & 1 & 0 & 0 & 1 & 0 & 1 & 0 \\
  1 & 0 & 1 & 0 & 0 & 1 & 0 & 1 & 0 & 1 & 0 & 1 & 1 & 0 & 0 \\
  1 & 0 & 0 & 1 & 0 & 1 & 1 & 0 & 0 & 1 & 1 & 0 & 0 & 1 & 0 \\
  0 & 1 & 0 & 1 & 0 & 1 & 0 & 1 & 1 & 0 & 1 & 0 & 1 & 0 & 0 \\
  0 & 1 & 1 & 0 & 0 & 1 & 1 & 0 & 1 & 0 & 0 & 1 & 0 & 1 & 0 \\
  0 & 1 & 0 & 1 & 1 & 0 & 1 & 0 & 0 & 1 & 0 & 1 & 1 & 0 & 0 \\
  0 & 1 & 1 & 0 & 1 & 0 & 0 & 1 & 0 & 1 & 1 & 0 & 0 & 1 & 0 \\
  1 & 1 & 0 & 0 & 1 & 1 & 0 & 0 & 1 & 1 & 0 & 0 & 0 & 0 & 1 \\
  1 & 1 & 0 & 0 & 0 & 0 & 1 & 1 & 0 & 0 & 1 & 1 & 0 & 0 & 1 \\
  0 & 0 & 1 & 1 & 0 & 0 & 1 & 1 & 1 & 1 & 0 & 0 & 0 & 0 & 1 \\
  0 & 0 & 1 & 1 & 1 & 1 & 0 & 0 & 0 & 0 & 1 & 1 & 0 & 0 & 1 \\
  1 & 1 & 1 & 1 & 0 & 0 & 0 & 0 & 0 & 0 & 0 & 0 & 1 & 1 & 1 \\
  0 & 0 & 0 & 0 & 1 & 1 & 1 & 1 & 0 & 0 & 0 & 0 & 1 & 1 & 1 \\
  0 & 0 & 0 & 0 & 0 & 0 & 0 & 0 & 1 & 1 & 1 & 1 & 1 & 1 & 1 \\
  \end{array}\right)}}
  \quad\text{and}\quad
  T=
  {\tiny\arraycolsep=0.3\arraycolsep\ensuremath{\left(\begin{array}{rrrrrrrrrrrrrrr}
  1 & 0 & 0 & 0 & 0 & 0 & 0 & 0 & -1 & 0 & 0 & -1 & -3 & -3 & -13 \\
  0 & 1 & 0 & 0 & 0 & 0 & -1 & -1 & -1 & -1 & -1 & -3 & -1 & -1 & -6 \\
  0 & 0 & 1 & 0 & 0 & -1 & 0 & -1 & 0 & -1 & -1 & -3 & -1 & -2 & -20 \\
  0 & 0 & 0 & 1 & 0 & -1 & -1 & 0 & 0 & 0 & -1 & -1 & -3 & -2 & -27 \\
  0 & 0 & 0 & 1 & 1 & -1 & -1 & -1 & 0 & 0 & -1 & -4 & -5 & -2 & -41 \\
  0 & 0 & 0 & 0 & 0 & 1 & 0 & 0 & 0 & 0 & 0 & 0 & -1 & 0 & -6 \\
  0 & 0 & 0 & 0 & 0 & 0 & 1 & 0 & 0 & 0 & 0 & 0 & -1 & -1 & -6 \\
  0 & 0 & 0 & 0 & 0 & 0 & 0 & 1 & 0 & 0 & 0 & 0 & -1 & -1 & -13 \\
  0 & 0 & 0 & 0 & 1 & 0 & 0 & 0 & 1 & 0 & 0 & -3 & -4 & -1 & -34 \\
  0 & 0 & 0 & 0 & 0 & 0 & 0 & 0 & 1 & 1 & 0 & -1 & -2 & -1 & -13 \\
  0 & 0 & 0 & 0 & 0 & 0 & 0 & 0 & 0 & 1 & 1 & -1 & -2 & -2 & -13 \\
  0 & 0 & 0 & 0 & 0 & 0 & 0 & 0 & 0 & 0 & 0 & 1 & 0 & 0 & -6 \\
  0 & 0 & 0 & 0 & 0 & 0 & 0 & 0 & 0 & 0 & 1 & 0 & 0 & -1 & -13 \\
  0 & 0 & 0 & 0 & 0 & 0 & 0 & 0 & 0 & 0 & 0 & 0 & 0 & 1 & -6 \\
  0 & 0 & 0 & 0 & 0 & 0 & 0 & 0 & 0 & 0 & 0 & 0 & 0 & 0 & 1
  \end{array}\right)}}.
  $$
  Here six invariant factors of $B$ equal $2$, three equal $4$, the last equals $28$, and the first five equal $1$. Generator matrices
  for the $\Z_4$-code $C$ of $B$ and its dual code $C^\perp$ are e.g.\ given by
  $$
  {\tiny\arraycolsep=0.3\arraycolsep\ensuremath{\left(\begin{array}{rrrrrrrrrrrrrrr}
  \textbf{1} & 0 & 0 & 1 & 0 & 1 & 1 & 0 & 0 & 1 & 1 & 0 & 0 & 1 & 0 \\
  0 & \textbf{1} & 0 & 1 & 0 & 1 & 0 & 1 & 0 & 1 & 0 & 1 & 0 & 1 & 1 \\
  0 & 0 & \textbf{1} & 1 & 0 & 0 & 1 & 1 & 0 & 0 & 1 & 1 & 1 & 1 & 0 \\
  0 & 0 & 0 & \textbf{2} & 0 & 0 & 0 & 2 & 0 & 0 & 0 & 2 & 0 & 2 & 0 \\
  0 & 0 & 0 & 0 & \textbf{1} & 1 & 1 & 1 & 0 & 0 & 0 & 0 & 1 & 1 & 1 \\
  0 & 0 & 0 & 0 & 0 & \textbf{2} & 0 & 2 & 0 & 0 & 0 & 0 & 0 & 2 & 2 \\
  0 & 0 & 0 & 0 & 0 & 0 & \textbf{2} & 2 & 0 & 0 & 0 & 0 & 0 & 0 & 2 \\
  0 & 0 & 0 & 0 & 0 & 0 & 0 & 0 & 0 & 0 & 0 & 0 & 0 & 0 & 0 \\
  0 & 0 & 0 & 0 & 0 & 0 & 0 & 0 & \textbf{1} & 1 & 1 & 1 & 1 & 1 & 1 \\
  0 & 0 & 0 & 0 & 0 & 0 & 0 & 0 & 0 & \textbf{2} & 0 & 2 & 0 & 2 & 2 \\
  0 & 0 & 0 & 0 & 0 & 0 & 0 & 0 & 0 & 0 & \textbf{2} & 2 & 0 & 0 & 2 \\
  0 & 0 & 0 & 0 & 0 & 0 & 0 & 0 & 0 & 0 & 0 & 0 & 0 & 0 & 0 \\
  0 & 0 & 0 & 0 & 0 & 0 & 0 & 0 & 0 & 0 & 0 & 0 & \textbf{2} & 2 & 2 \\
  0 & 0 & 0 & 0 & 0 & 0 & 0 & 0 & 0 & 0 & 0 & 0 & 0 & 0 & 0 \\
  0 & 0 & 0 & 0 & 0 & 0 & 0 & 0 & 0 & 0 & 0 & 0 & 0 & 0 & 0 \\
  \end{array}\right)}}
  \quad\text{and}\quad
  {\tiny\arraycolsep=0.3\arraycolsep\ensuremath{\left(\begin{array}{rrrrrrrrrrrrrrr}
  \textbf{1} & 0 & 0 & 1 & 0 & 1 & 1 & 0 & 1 & 0 & 0 & 1 & 1 & 0 & 1 \\
  0 & \textbf{1} & 0 & 1 & 0 & 1 & 0 & 1 & 0 & 1 & 0 & 1 & 0 & 1 & 1 \\
  0 & 0 & \textbf{1} & 1 & 0 & 0 & 1 & 1 & 0 & 0 & 1 & 1 & 1 & 1 & 0 \\
  0 & 0 & 0 & \textbf{2} & 0 & 0 & 0 & 2 & 0 & 0 & 0 & 2 & 0 & 2 & 0 \\
  0 & 0 & 0 & 0 & \textbf{1} & 1 & 1 & 1 & 1 & 1 & 1 & 1 & 0 & 0 & 0 \\
  0 & 0 & 0 & 0 & 0 & \textbf{2} & 0 & 2 & 0 & 0 & 0 & 0 & 0 & 2 & 2 \\
  0 & 0 & 0 & 0 & 0 & 0 & \textbf{2} & 2 & 0 & 0 & 0 & 0 & 2 & 2 & 0 \\
  0 & 0 & 0 & 0 & 0 & 0 & 0 & 0 & 0 & 0 & 0 & 0 & 0 & 0 & 0 \\
  0 & 0 & 0 & 0 & 0 & 0 & 0 & 0 & \textbf{2} & 0 & 0 & 2 & 2 & 0 & 2 \\
  0 & 0 & 0 & 0 & 0 & 0 & 0 & 0 & 0 & \textbf{2} & 0 & 2 & 0 & 2 & 2 \\
  0 & 0 & 0 & 0 & 0 & 0 & 0 & 0 & 0 & 0 & \textbf{2} & 2 & 2 & 2 & 0 \\
  0 & 0 & 0 & 0 & 0 & 0 & 0 & 0 & 0 & 0 & 0 & 0 & 0 & 0 & 0 \\
  0 & 0 & 0 & 0 & 0 & 0 & 0 & 0 & 0 & 0 & 0 & 0 & 0 & 0 & 0 \\
  0 & 0 & 0 & 0 & 0 & 0 & 0 & 0 & 0 & 0 & 0 & 0 & 0 & 0 & 0 \\
  0 & 0 & 0 & 0 & 0 & 0 & 0 & 0 & 0 & 0 & 0 & 0 & 0 & 0 & 0 \\
  \end{array}\right)}}
  $$
  with invariant factors $[1,1,1,1,1,2,2,2,2,2,2,0,0,0,0]$ and $[1,1,1,1,2,2,2,2,2,2,0,0,0,0]$, which we abbreviate as $1^52^6 0^4$ and $1^42^60^5$.
  The rows of the stated generator matrix for $C^\perp$ correspond to ten $\!\!\!\mod 4$-conditions, where six can be rewritten to $\!\!\!\mod 2$-conditions.
\end{example}

\pagebreak

\section{Small parameters}
\label{sec_small_parameters}

While the minimum possible length of a $k$-dimensional linear code over $\F_q$
is given by the Griesmer bound if the desired minimum Hamming distance $d$ is
sufficiently large, the gap to the Griesmer bound is unbounded in $k$
\cite{dodunekov1984note,rousseva2019linear}. So, in this section we want
to collect results on $n_q(r,h;s)$ for
small parameters. For results for $q=2$ and $q=3$ we refer to Subsection~\ref{subsec_q_2} and Subsection~\ref{subsec_q_3}, respectively.
For results for $(r,h)=(5,2)$ we refer to Subsection~\ref{subsec_dimension_2_5}. Here
we  start with some parametric results. As mentioned
before, we have $n_q(r,h;s)= \infty$ for $r\le h$, so that we assume
$r\ge h+1$ in the following. We are especially interested in the integral
cases where $r/h\in\N$. Ignoring the degenerated case $r=h$, the first
interesting case $r=2h$ can be completely solved:
\begin{theorem}
  \label{thm_r_h_integral_2}
  \begin{equation}
    n_q(2h,h;s)=\overline{n}_q(2h,h;s)=\frac{[2h]_q}{[h]_q}\cdot s
    =\left(q^h+1\right)\cdot s
  \end{equation}
\end{theorem}
\begin{proof}
  The lower bound follows from Theorem~\ref{thm_partition} and the upper bound from Lemma~\ref{lemma_one_weight_bound}.
\end{proof}

Clearly we have $n_q(3,1;1)=1$. For odd $q$ the upper bound $n_q(3,1;2)\le q+1$
was shown in \cite{bose1947mathematical} and an \emph{oval} (conic) in $\PG(2,q)$
shows $n_q(3,1;2)\ge q+1$ for all field sizes $q$. Segre \cite{segre1955ovals}
has shown that all examples attaining $n_q(3,1;2)=q+1$ are equivalent for odd
$q$, cf.~\cite{jarnefelt1949observation}. For even $q$ each faithful projective
$1-(q+1,3,2)_q$ system can be extended to a faithful projective $1-(q+2,3,2)_q$
system, which is called \emph{hyperoval}. The upper bound $n_q(3,1;2)\le q+2$
follows from the Griesmer bound, as shown below when setting $t=1$,
$i=q-1$.

\begin{proposition}
  \label{prop_r_3_h_1}
  We have
  \begin{equation}
    n_q\!\left(3,1;(q+1)t-i\right)=[3]_q \cdot t-[2]_q \cdot i
    =\left(q^2+q+1\right)\cdot t-\left(q+1\right)\cdot i
  \end{equation}
  for $0\le i\le 1$, $t\ge 1$ and for $2\le i\le q$, $t\ge 2$.
\end{proposition}
\begin{proof}
  The statement for $n_q(3,1;(q+1)t)$, i.e.\ $i=0$, follows from
  Theorem~\ref{thm_partition}.
  Let $n_{i,t}=[3]_qt-[2]_qi$ and $s_{i,t}=[2]_qt-[1]_qi$ for $0\le i\le q$ and
  $t\ge 1$. Since $d_{i,t}:=n_{i,t}-s_{i,t}=q^2t-qi$ we can apply the Solomon--Stiffler construction if we can find $i$ lines in $\PG(2,q)$ that cover each point at most twice for $2\le i\le q$ or at most once for $i=1$. By duality
  and using $n_q(3,1;2)\ge q+1$ this is indeed
  possible. For $t\ge 1$ and $1\le i\le q$ the Griesmer bound,
  see Lemma~\ref{lemma_parameters_griesmer_code}, shows that the length
  of $n'$ of each $[n',3,d']_q$ code with minimum distance
  $$
    d':=d_{i,t}+1=\left(n_{i,t}+1\right)-s_{i,t}=t\cdot q^2-(i-1)\cdot q-(q-1)\cdot 1
  $$
  is at least
  $$
   t\cdot[3]_q -(i-1)\cdot [2]_q-(q-1)\cdot[1]_q =[3]_qt-[2]_qi +2>n_{i,t},
  $$
  i.e.\ $n_q\!\left(3,1;(q+1)t-i\right)\le [3]_q \cdot t-[2]_q \cdot i$ follows
  from the Griesmer upper bound.
\end{proof}
We remark that Lemma~\ref{lemma_one_weight_bound} implies
\begin{equation}
  n_q\!\left(3,1;(q+1)t-i\right) \le \left\lfloor \frac{[3]_q}{[3-1]_q}\cdot \left((q+1)t-i\right)\right\rfloor
  =[3]_qt -qi-\left\lceil\frac{i}{q+1}\right\rceil
\end{equation}
for $0\le i\le q$ and $t\ge 1$, which is tight iff $i=0,1$. For $t=i=1$ the $q^2$ lines of $\PG(2,q)$ that are
disjoint to an arbitrary but fixed point yield a projective faithful  $2-\left(q^2,3,1,q\right)_q$ system whose dual is a projective faithful $1-\left(q^2,3,q,1\right)_q$ system. The determination of $n_q(3,1;s)$ is a
challenging problem for $3\le s\le q-1$ and was solved completely for $q\le 9$
only, see e.g.\ \url{http://mars39.lomo.jp/opu/griesmer.htm}. In
Table~\ref{table_n_q_3_1} we summarize the known values for $q\le 9$. We remark
that in all cases $1-\left(n_q(3,1;s),3,s,1\right)_q$ systems do exist, see \url{http://web.mat.upc.edu/simeon.michael.ball/codebounds.html}.

\begin{table}[htp]
  \begin{center}
    \begin{tabular}{rrrrr}
      \hline
      $q$ & $s$ & $n_q(3,1;s)$ & Gub & gap  \\
      \hline
      2 & 1 & 1 & 1 \\
        & 2 & 4 & 4 \\
        & 3 & 7 & 7 \\
      \hline
      3 & 1 &  1 &  1 \\
        & 2 &  4 &  5 & 1 \\
        & 3 &  9 &  9 \\
        & 4 & 13 & 13 \\
      \hline
      4 & 1 &  1 &  1 \\
        & 2 &  6 &  6 \\
        & 3 &  9 & 11 & 2 \\
        & 4 & 16 & 16 \\
        & 5 & 21 & 21 \\
      \hline
      5 & 1 &  1 &  1 \\
        & 2 &  6 &  7 & 1 \\
        & 3 & 11 & 13 & 2 \\
        & 4 & 16 & 19 & 3 \\
        & 5 & 25 & 25 \\
        & 6 & 31 & 31 \\
      \hline
      7 & 1 &  1 &  1 \\
        & 2 &  8 &  9 & 1 \\
        & 3 & 15 & 17 & 2 \\
        & 4 & 22 & 25 & 3 \\
      \hline
    \end{tabular}
    \quad\quad\quad
    \begin{tabular}{rrrrr}
      \hline
      $q$ & $s$ & $n_q(3,1;s)$ & Gub & gap \\
      \hline
      7 & 5 & 29 & 33 & 4 \\
        & 6 & 36 & 41 & 5 \\
        & 7 & 49 & 49 \\
        & 8 & 57 & 57 \\
      \hline
      8 & 1 &  1 &  1 \\
        & 2 & 10 & 10 \\
        & 3 & 15 & 19 & 4 \\
        & 4 & 28 & 28 \\
        & 5 & 33 & 37 & 4 \\
        & 6 & 42 & 46 & 4 \\
        & 7 & 49 & 55 & 6 \\
        & 8 & 64 & 64 \\
        & 9 & 73 & 73 \\
      \hline
      9 &  1 &  1 &  1 \\
        &  2 & 10 & 11 & 1 \\
        &  3 & 17 & 21 & 4 \\
        &  4 & 28 & 31 & 3 \\
        &  5 & 37 & 41 & 4 \\
        &  6 & 48 & 51 & 3 \\
        &  7 & 55 & 61 & 6 \\
        &  8 & 65 & 71 & 6 \\
        &  9 & 81 & 81 \\
        & 10 & 91 & 91 \\
      \hline
    \end{tabular}
    \caption{Griesmer upper bounds and exact values for $n_q(3,1;s)$ for $1\le s\le q+1$ and $q\le 9$.}
    \label{table_n_q_3_1}
  \end{center}
\end{table}

\begin{lemma}
  \label{lemma_r_3h_h_2}
  We have $\overline{n}_q\!\left(6,2;\left(q^2+1\right)t-i\right)=
  \left(q^4+q^2+1\right)t-\left(q^2+1\right)i$ for all $0\le i\le q^2$
  and all $t\ge 2$.
\end{lemma}
\begin{proof}
  With $\tilde{q}:=q^2$ we compute $\gcd(r,h)=2$, $[\gcd(r,h)]_q=q+1$,
  $\tfrac{[r-h]_q}{[\gcd(r,h)]_q}=q^2+1=\tilde{q}+1$, and
  $\tfrac{[r]_q}{[\gcd(r,h)]_q}=q^4+q^2+1=\tilde{q}^2+\tilde{q}+1$. For
  the lower bounds we refer to Proposition~\ref{prop_r_3_h_1}. Due to
  Theorem~\ref{thm_partition} we can assume $1\le i\le q^2$. Let
  $n_{i,t}:=\left(q^4+q^2+1\right)t-\left(q^2+1\right)i$ and
  $s_{i,t}:=\left(q^2+1\right)t-i$, so that
  $$
    \left( n_{i,t}+1\right)-s_{i,t}=t\cdot\tilde{q}^2-(i-1)\cdot \tilde{q}
    -\left(\tilde{q}-1\right)\cdot 1.
  $$
  Since $t\cdot [3]_{\tilde{q}}-(i-1)\cdot [2]_{\tilde{q}}
  -\left(\tilde{q}-1\right)\cdot [1]_{\tilde{q}}=n_{i,t}+2$
  Lemma~\ref{lemma_parameters_griesmer_code} yields
  $n_{\tilde{q}}\!\left(3,1;\left(\tilde{q}+1\right)t-i\right)\le
  \left(\tilde{q}^2+\tilde{q}+1\right)\cdot t-\left(\tilde{q}+1\right)i$
  for $1\le i\le \tilde{q}$.
\end{proof}

\begin{theorem}
  \label{thm_r_3h_h_2}
  Let $0\le i\le q^2$ and $i=aq-b$ where $a,b\in\N$ and $b\le q-1$. For each
  $t\ge q^2+q$ we have
  \begin{equation}
    \label{eq_r_3h_h_2}
    n_q(6,2;(q^2+1)t-i)=(q^4+q^2+1)t-(q^2+1)i+\max\{a-2,0\}\cdot q
  \end{equation}
  and $n_q(6,2;(q^2+1)t-i)=\overline{n}_q(6,2;(q^2+1)t-i)+
  \max\{\left\lceil i/q\right\rceil-2,0\}\cdot q$.
\end{theorem}
\begin{proof}
  First we note $a=\left\lceil i/q\right\rceil$, so that the second part
  follows from Lemma~\ref{lemma_r_3h_h_2} and Equation~(\ref{eq_r_3h_h_2}).
  Due to Theorem~\ref{thm_partition} we can assume $1\le i\le q^2$. Let
  $n_{i,t}:=\left(q^4+q^2+1\right)t-\left(q^2+1\right)i$ and
  $s_{i,t}:=\left(q^2+1\right)t-i$, so that $n_{i,t}=\overline{n}_q(6,2;s_{i,t})$ for all $t\ge 2$ by Lemma~\ref{lemma_r_3h_h_2} and
  \begin{equation}
    q\cdot\big(\left( n_{i,t}+1\right)-s_{i,t}\big)=
    t\cdot q^5-(i-1)\cdot q^3-(q-1)\cdot q^2
    -(q-1)\cdot q.
  \end{equation}
  For $1\le i\le q$ we have
  $$
    t[6]_q-(i-1)[4]-(q-1)[3]-(q-1)[2]=t[6]_q-i[5]_q
    +q+3=n_{i,t}(q+1)+2>n_{i,t}(q+1),
  $$
  so that Lemma~\ref{lemma_parameters_griesmer_code} gives $n_q(6,2;s_{i,t})\le n_{i,t}$.
  For $q+1\le i\le 2q$ we have
  $$
    q\cdot\big(\left( n_{i,t}+1\right)-s_{i,t}\big)=
    t\cdot q^5-q^4-(i-q-1)\cdot q^3-(q-1)\cdot q^2
    -(q-1)\cdot q
  $$
  and
  $$
    t[6]_q-[5]_q-(i-q-1)[4]-(q-1)[3]-(q-1)[2]
    =t[6]_q-i[4]_q+q+2=n_{i,t}(q+1)+1>n_{i,t}(q+1)
,  $$
  so that Lemma~\ref{lemma_parameters_griesmer_code} gives $n_q(6,2;s_{i,t})\le n_{i,t}$. For $1\le i\le 2q$ we have $n_q(6,2;s_{i,t})\ge
  \overline{n}_q(6,2;s_{i,t})= n_{i,t}$. So far we have shown $n_q(6,2;s_{i,t})=\overline{n}_q(6,2;s_{i,t}) =
  n_{i,t}+\max\{a-2,0\}\cdot q$ for $a\le 2$, i.e. $0\le i\le 2q$, and
  $t\ge q^2+q>2$.

  For $3\le a\le q$, $0\le b\le q-1$ , and $i=aq-b$ we have
  $$
    q\cdot\big(\left( n_{i,t}+(a-2)q\right)-s_{i,t}\big)=
    t\cdot q^5-(a-1)\cdot q^4-(q-1-b)\cdot q^3-(q+2-a)q^2
  $$
  and
  \begin{eqnarray*}
    t\cdot [6]_q-(a-1)\cdot [5]_q-(q-1-b)\cdot [4]_q-(q+2-a)[3]_q
    &=&t\cdot [6]_q-i[4]_q+(a-2)q(q+1)\\
    &=& \left( n_{i,t}+(a-2)q\right) \cdot(q+1).
  \end{eqnarray*}
  By Theorem~\ref{thm_partition} we have that $(q+1)[3]$ is $2$-partitionable
  over $\F_q$. Lemma~\ref{lemma_partition_1} shows
  that $[6]-[4]$ and $(q+1)[6]-[5]-q[3]$ are $2$-partitionable over $\F_q$.
  Consider the corresponding constructions. Taking $(a-2)$ copies of the
  first, $(q-1-b)$ copies of the second, and $(a-1)$ copies of the third
  construction shows that $t[6]-(a-1)[5]-(q-1-b) [4]-
  (q+2-a)[3]$ is $2$-partitionable over $\F_q$ for
  $$
    t=(q+1)(a-1)+(q+1-b)=aq-b +a=i+a\le q^2+q.
  $$
  Since $[6]$ is $2$-partitionable over $\F_q$ by Theorem~\ref{thm_partition} we
  have $n_q(6,2;s_{i,t})\ge n_{i,t}+(a-2)q$ for all $2q+1\le i\le q^2$ and
  all $t\ge q^2+q$. For the upper bound we consider
  $$
    q\cdot\big(\left( n_{i,t}+(a-2)q+1\right)-s_{i,t}\big)=
    t\cdot q^5-(a-1)\cdot q^4-(q-1-b)\cdot q^3-(q+1-a)q^2-(q-1)q
  $$
  and
  \begin{eqnarray*}
    && t\cdot [6]_q-(a-1)\cdot [5]_q-(q-1-b)\cdot [4]_q-(q+1-a)[3]_q-(q-1)[2]_q \\
    &=&t\cdot [6]_q-i[4]_q+(a-2)q(q+1)+q+2 =\left( n_{i,t}+(a-2)q+1\right) \cdot(q+1)+1\\
    &>& \left( n_{i,t}+(a-2)q+1\right) \cdot(q+1),
  \end{eqnarray*}
  so that Lemma~\ref{lemma_parameters_griesmer_code} gives
  $n_q(6,2;s_{i,t})\le n_{i,t}+(a-2)q$.
\end{proof}

\begin{theorem}
  \label{thm_hp1_h_fractional}
  We have $n_q(h+1,h;s)=s\cdot [h+1]_q$,
  $\overline{n}_q(h+1,h;s)=s\cdot \left(q^h+1\right)$, and
  $n_q(h+1,h;s)-\overline{n}_q(h+1,h;s)=s\cdot q[h-1]_q$, which is positive for
  $h>1$.
\end{theorem}
\begin{proof}
  For $n_q(h+1,h;s)$ the upper bound follows from
  Lemma~\ref{lemma_one_weight_bound} and a construction is given by choosing
  each $h$-space $s$ times. The value of $\overline{n}_q(h+1,h;s)$ follows
  from Theorem~\ref{thm_r_h_integral_2}.
\end{proof}

\begin{proposition}
  \label{prop_hp2_h_fractional}
  For even $h$ we have $n_q(h+2,h;s)=\tfrac{[h+2]_q}{[2]_q}\cdot s=[h/2+1]_{q^2}\cdot s$ for all $s\in \N$. Let $h=2h'+1$ where $h'\in\N_{\ge 1}$. We have
  \begin{equation}
    \label{eq_hp2_h}
    n_q\!\left(h+2,h;(q+1)t-i\right)
    =[h+2]_q\cdot t-\frac{[h+2]_q+q}{[2]_q}\cdot i
    =[h+2]_q\cdot t-\left(q\cdot [h'+1]_{q^2}+1\right)\cdot i
  \end{equation}
  for $0\le i\le 1$, $t\ge 1$ and for $2\le i\le q$, $t\ge 2$.
\end{proposition}
\begin{proof}
  The statements for $n_q(h+2,h;s)$ with even $h$ and for $n_q(h+2,h;(q+1)t)$
  with odd $h$ follow from Theorem~\ref{thm_partition}, so that we assume
  $h=2h'+1$ for a fixed $h'\in\N_{\ge 1}$ and $i\ge 1$ in the following. Let
  \begin{equation}
    n_{i,t}=[h+2]_qt-\frac{[h+2]_q+q}{[2]_q}\cdot i= [h+2]_q\cdot t
  -\left(q\cdot [h'+1]_{q^2}+1\right)\cdot i
  \end{equation}
  and
  \begin{equation}
    s_{i,t}=(q+1)t-i
  \end{equation}
  for $0\le i\le q$ and $t\ge 1$. Lemma~\ref{lemma_vsp} yields the
  existence of a faithful projective $2-\left(\tfrac{q^{h+2}-q^{3}}{q^2-1},h+2,\star,1\right)_q$ system where all elements are disjoint to a special
  $3$-space $A$. Taking $s_{i,t}$ copies thereof and embedding the dual of a
  faithful projective $1-\left([3]_qt-(q+1)i,3,(q+1)t-i\right)_q$ system
  from Proposition~\ref{prop_r_3_h_1} into $A$ gives a faithful projective
  $2-\left(n_{i,t},h+2,\star,s_{i,t}\right)_q$ system $\cS_{i,t}$ for
  $t\ge 2$ or $t=i=1$ since
  $$
   \frac{q^{h+2}-q^{3}}{q^2-1}\cdot s_{i,t}+\left([3]_qt-(q+1)i\right)=
   [h+2]_q\cdot t-\frac{q^{h+2}+q^2-q-1}{q^2-1}\cdot i
   =[h+2]_q\cdot t +\frac{[h+2]_q+q}{[2]_q}\cdot i.
  $$
  So, the dual $\cS_{i,t}^\perp$ of $\cS_{i,t}$ yields
  $n_q\!\left(h+2,h;s_{i,t}\right)
  \ge n_{i,t}$ for $0\le i\le 1$, $t\ge 1$ and for $2\le i\le q$, $t\ge 2$. Let
  $$
    d_{i,t}:=q^{h-1}\cdot\left(n_{i,t}-s_{i,t}\right)=
    q^{h-1}\cdot\left(q^2[h]_qt-q\cdot [h'+1]_{q^2}\cdot i\right)
    =\left([h]_qt-qi[h']_{q^2}\right)\cdot q^{h+1}-i\cdot q^h,
  $$
  so that
  $$
    q^{h-1}\cdot\left(n_{i,t}+1-s_{i,t}\right)=
    \left([h]_qt-qi[h']_{q^2}\right)\cdot q^{h+1}-(i-1)\cdot q^h-(q-1)\cdot q^{h-1}.
  $$
  So, for $t\ge 1$ and $1\le i\le q$ the Griesmer bound,
  see Lemma~\ref{lemma_parameters_griesmer_code}, shows that the length
  of $n'$ of each $\left[n',h+2,d'\right]_q$ code with minimum distance
  $d'=d_{i,t}+q^{h-1}$ is at least
  \begin{eqnarray*}
    &&\left([h]_qt-qi[h']_{q^2}\right)\cdot [h+2]_q-(i-1)\cdot [h+1]_q-(q-1)\cdot [h]_q \\
    &=& [h+2]_qt\cdot[h]_q- q^3i[h']_{q^2}\cdot[h]_q-q(q+1)i[h']_{q^2}-
    (iq-1)[h]_q \\
    &=& [h]_q\cdot\left([h+2]_qt-\left(q[h'+1]_{q^2}+1\right)\cdot i\right)
      +(i+1)[h]_q-q(q+1)i[h']_{q^2} \\
    &=& [h]_q\cdot n_{i,t}+[h]_q-1>[h]_q\cdot n_{i,t},
  \end{eqnarray*}
  i.e.\ the Griesmer upper bound from Definition~\ref{def_griesmer_uppper_bound}
  is $n_q\!\left(h+2,h;s_{i,t}\right) \le n_{i,t}$, see
  Lemma~\ref{lemma_indirect_upper_bounds}.
\end{proof}
We remark that Lemma~\ref{lemma_one_weight_bound} implies
$$
  n_q\!\left(h+2,h;(q+1)t-i\right) \le \left\lfloor \frac{[h]_q}{[2]_q}\cdot \left((q+1)t-i\right)\right\rfloor
    =[h+2]_qt -q\cdot [(h+1)/2]_{q^2}\cdot i-\left\lceil\frac{i}{q+1}\right\rceil
$$
for odd $h$, $0\le i\le q$, and $t\ge 1$, which is tight iff $i=0,1$, see Proposition~\ref{prop_r_3_h_1} for $h=1$.

\begin{table}[htp]
  \begin{center}
    \begin{tabular}{rrrrr}
      \hline
      $q$ & $s$ & $n_q(5,3;s)$ & Gub & gap  \\
      \hline
      2 & 1 &  9 &  9 \\
        & 2 & 20 & 20 \\
        & 3 & 31 & 31 \\
      \hline
      3 & 1 &  28 &  28 \\
        & 2 &  \textbf{58} &  59 & 1 \\
        & 3 &  90 &  90 \\
        & 4 & 121 & 121 \\
      \hline
      4 & 1 &  65 &  65 \\
        & 2 & 134 & 134 \\
        & 3 & \textbf{201} & 203 & 2 \\
        & 4 & 272 & 272 \\
        & 5 & 341 & 341 \\
      \hline
      5 & 1 & 126 & 126 \\
        & 2 & \textbf{256} & 257 & 1 \\
        & 3 & \textbf{386} & 388 & 2 \\
        & 4 & \textbf{516} & 519 & 3 \\
        & 5 & 650 & 650 \\
        & 6 & 781 & 781 \\
      \hline
      7 & 1 &  344 &  344 \\
        & 2 &  \textbf{694} &  695 & 1 \\
        & 3 & \textbf{1044} & 1046 & 2 \\
        & 4 & \textbf{1394} & 1397 & 3 \\
      \hline
    \end{tabular}
    \quad\quad\quad
    \begin{tabular}{rrrr}
      \hline
      $q$ & $s$ & $n_q(5,3;s)$ & Gub \\
      \hline
      7 & 5 & \textbf{1744} & 1748 \\
        & 6 & \textbf{2094} & 2099 \\
        & 7 & 2450 & 2450 \\
        & 8 & 2801 & 2801 \\
      \hline
      8 & 1 &  513 & 513 \\
        & 2 & 1034 & 1034 \\
        & 3 & 1551--1555 & 1555 \\
        & 4 & 2076 & 2076 \\
        & 5 & 2593--2597 & 2597 \\
        & 6 & 3114--3118 & 3118 \\
        & 7 & 3633--3639 & 3639 \\
        & 8 & 4160 & 4160 \\
        & 9 & 4681 & 4681 \\
      \hline
      9 &  1 &  730 &  730 \\
        &  2 & 1468--1469 & 1469 \\
        &  3 & 2204--2208 & 2208 \\
        &  4 & 2944--2947 & 2947 \\
        &  5 & 3682--3686 & 3686 \\
        &  6 & 4422--4425 & 4425 \\
        &  7 & 5158--5164 & 5164 \\
        &  8 & 5897--5903 & 5903 \\
        &  9 & 6642 & 6642 \\
        & 10 & 7381 & 7381 \\
      \hline
    \end{tabular}
    \caption{Griesmer upper bounds and values for $n_q(5,3;s)$ for $1\le s\le q+1$ and $q\le 9$.}
    \label{table_n_q_5_3}
  \end{center}
\end{table}

\begin{proposition}
  \label{prop_partial_line_spread}
  For odd $h\ge 3$ we have $n_q(h+2,h;1)=\tfrac{[h+2]_q-q^2}{[2]_q}=q^3\cdot [(h-1)/2]_{q^2}+1$.
\end{proposition}
\begin{proof}
  The dual of a faithful projective $h-(n,h+2,1,\mu)_q$ system is a faithful
  projective $2-(n,h+2,\mu,1)_q$ system, i.e.\ a partial line spread in
  $\PG(h+1,q)$, so that \cite[Theorem 4.1 \& 4.2]{beutelspacher1975partial}
  yields the stated formula.
\end{proof}

\begin{proposition}
  \label{prop_r_5_h_3}
  Bounds for $n_q(5,3;s)$, where $1\le s\le q+1$ and $q\le9$,
  are summarized in Table~\ref{table_n_q_5_3}.
\end{proposition}
\begin{proof}
  For $q\le s\le q+1$ we refer to Proposition~\ref{prop_hp2_h_fractional}
  and for $s=1$ Proposition~\ref{prop_partial_line_spread} yields
  $n_q(5,3;1)=q^3+1$. As in the proof of Proposition~\ref{prop_hp2_h_fractional}
  we can use Proposition~\ref{lemma_vsp} to deduce
  $n_q(5,3;s)\ge q^3s+n_q(3,1;s)$. If not stated otherwise we use this lower
  bound with the values from Table~\ref{table_n_q_3_1}. The dual of a
  faithful projective $3-(n,5,s,\mu)_q$ system is a faithful projective
  $2-(n,5,\mu,s)_q$ system $\cS$. Due to Lemma~\ref{lemma_linear_code} $\cP(\cS)$
  is $q$-divisible. By Lemma~\ref{lemma_divisible_properties}
  $\cM:=s\cdot\chi_V-\cP(\cS)$ is also $q$-divisible with
  cardinality $\#\cM=s[5]_q-n[2]_q$ and maximum point multiplicity
  at most $s$, where $V$ is the $5$-dimensional ambient space. So, non-existence
  results for $q$-divisible multisets of points in $\PG(4,q)$ with suitable
  cardinalities and maximum point multiplicities can imply upper bounds
  for $n_q(5,3;s)$.

  As a example we will use the fact that there is no $3$-divisible
  multiset of points in $\PG(4,q)$ with cardinality $6$ and maximum point
  multiplicity at most $2$, which can be verified by exhaustive enumeration
  or proven theoretically, see Lemma~\ref{lemma_3_div_card_6}. Thus, we have $n_3(5,3;2)<59$ since $2\cdot [5]_3-59\cdot[2]_3=6$.

  Using the software package \texttt{LinCode} \cite{bouyukliev2021computer}
  we have shown by exhaustive enumeration that the following multisets of
  points do not exist:
  \begin{itemize}
   \item There is no $3$-divisible multiset of points in $\PG(4,3)$ with
         cardinality $\#\cM=6$ and maximum point multiplicity at most $2$.\\[-10mm]
   \item There is no $4$-divisible multiset of points in $\PG(4,4)$ with
         cardinality $\#\cM=13$ 
         and maximum point multiplicity at most $3$.\\[-10mm]
   \item There is no $5$-divisible multiset of points in $\PG(4,5)$ with
         cardinality $\#\cM\in\{20,21,22\}$
         and maximum point multiplicity at most $4$.\\[-10mm]
   \item There is no $7$-divisible multiset of points in $\PG(4,7)$ with
         cardinality $\#\cM\in\{35,36,37,38,42,$ $43,44,45,46\}$
         and maximum point multiplicity at most $6$.
  \end{itemize}
  In Table~\ref{table_n_q_5_3} we have marked the corresponding improved
  upper bounds for $n_q(5,3;s)$ in bold font. For more details we refer to
  Section~\ref{sec_divisible_multisets}.
\end{proof}

\bigskip

So far we have tried to determine parametric formulas or bounds for $n_q(r,h;s)$ for small parameters of $r$ and $h$ in terms of $s$. We may also consider the situation for fixed small values of $s$. Since $n_q(r,h;s)=\infty$ for $r\le h$
we assume $r\ge h+1$ in the following.
\begin{lemma}
  \label{lemma_small}
  We have $n_q(r,h;s)=s$ for $hs<r$.
\end{lemma}
\begin{proof}
  Let $\cS$ be a faithful projective $h-(n,r,s)_q$ system with $n=n_q(r,h;s)$.
  The span of $n$ elements from $\cS$ has dimension at most $nh$, i.e.\ all
  elements of $\cS$ are contained in a hyperplane of $\PG(r-1,q)$.
\end{proof}
So, we will mostly assume $s\ge r/h$ in the following, i.e.\ $s=2$ is the first interesting case.

\begin{theorem}
  \label{thm_r_3h_s_2}
  We have $n_q(3h,h;2)=\overline{n}_q(3h,h;2)$, i.e.\ $n_q(3h,h;2)=q^h+1$ if
  $q$ is odd and $n_q(3h,h;2)=q^h+2$ for even $q$.
\end{theorem}
\begin{proof}
  Let $\cS$ be a faithful projective $h-(n,3h;2)_q$ system with $n=n_q(3h,h;2)$.
  Any two elements of $\cS$ span a $2h$-space since otherwise we find three
  elements contained in a hyperplane. Denote the number of hyperplanes with
  $i$ elements from $\cS$ by $a_i$.
  Double-counting yields the equations
  \begin{eqnarray}
    a_0+a_1+a_2&=&[3h]_q,\\
    a_1+2a_2&=&n\cdot[2h]_q,\text{ and}\\
    a_2&=&\tfrac{n(n-1)}{2}\cdot[h]_q,
  \end{eqnarray}
  so that the second equation minus twice the third equation gives
  $$
    a_1=n[h]_q\cdot \left(q^h+2-n\right).
  $$
  Since $a_1$ is non-negative, we have $n_q(3h,h;2)\le q^h+2$. If $n=q^h+2$,
  then $a_1=0$. However, the corresponding
  code would be a projective two-weight code with weight difference $2q^{h-1}$, which is not a power of the characteristic of $\F_q$ if $q$ is odd -- contradiction to Lemma~\ref{lemma_two_weight} from ~\cite{delsarte1972weights}.
  Ovals and hyperovals in $\PG\!\left(2,q^h\right)$ give the corresponding
  constructions for odd $q$ and even $q$, respectively.
\end{proof}
We remark that $n_3(6,2;2)=10$ was shown in \cite{ball2024additive} by exhaustive enumeration.

\begin{remark}
  Let $\cS$ be a faithful $h-(n,lh;l-1)_q$ system with $n=n_q(lh,h;l-1)$ for
  some $l\ge 2$. If there exist $i$ elements of $\cS$ that span a subspace $S$
  of dimension strictly less than $hi$ for some $1\le i\le l$, then adding any further $l-i$ elements yields the existence of a hyperplane with at least
  $l$ elements, which is a contradiction. Thus, the dimension spanned by any subset of elements of $\cS$ is congruent to $0$ modulo $h$. In
  \cite[Proposition 3.1]{blokhuis2004small} is was shown for $q=h=2$ that
  this conditions ensures that $\cS$ can be obtained from a
  faithful projective $1-(n,l;l-1)_{q^h}$ system by the subfield construction.
  The existence of non-Desarguesian spreads of $h$-spaces for $q>2$ or $h>2$
  shows that further conditions are needed in order to conclude linearity,
  cf.\ \cite[Theorem 13]{ball2023additive} and \cite{adriaensen2023additive}.
  For characterizations of Desarguesian spreads we refer to \cite{rottey2017geometric}. So, it is an interesting question, whether
  $n_q(lh,h;l-1)>\overline{n}_q(lh,h;l-1)$ is possible for $l>3$.
\end{remark}

\begin{lemma}
  For each odd prime power $q$ we have
  \begin{equation}
    n_q(5,2;2)\le q^2+q+1.
  \end{equation}
\end{lemma}
\begin{proof}
  Consider a faithful projective $2-(n,5,2,\mu)_q$ system $\cS$ and denote the number of hyperplanes that contain $i$ elements from $\cS$ by $a_i$. We will show $\mu=1$ if $n>q^2+2$. The dual $\cS^\perp$ is a faithful projective
  $3-(n,5,\mu,2)_q$ system. Note that two planes ($3$-spaces) in $\PG(4,q)$
  intersect in at least a point. If $\cS^\perp$ contains a plane with
  multiplicity $2$, then we have $n=2$. If $\cS^\perp$ contains two elements
  $E_1,E_2$ such that their intersection is a line $L$, then the elements in
  $\cS^\perp\backslash\{E_1,E_2\}$ need to intersect $E_1$ outside of $L$, so
  that $n\le q^2+2$.

  If $n>q^2+2$, then the elements of $\cS$ form a partial line spread, i.e.\ we
  have $\mu=1$. Double-counting gives
  \begin{eqnarray}
    a_0+a_1+a_2 &=&[5]_q,\\
    a_1+2a_2 &=& n[3]_q,\text{ and}\\
    2a_2&=&n(n-1).
  \end{eqnarray}
  If $n>[3]_q+1$, then $a_1<0$, which is impossible. If $n=[3]_q+1=q^2+q+2$,
  then we have $a_1=0$, so that
  the corresponding code is a projective two-weight code with difference $2q$ of the occurring non-zero weights. However, the weight difference of a projective two-weight
  code has to be a power of the characteristic of $\F_q$, see Lemma~\ref{lemma_two_weight} or \cite[Corollary 2]{delsarte1972weights}, which
  gives a contradiction if $q$ is odd.
\end{proof}

\begin{remark}
For $q=2$ a vector space partition of $\PG(4,2)$ of type $2^6 3^1$ gives a faithful projective $2-(8,5,2)_2$ system (which is unique) and there also exists a faithful projective $2-(6,5,2)_2$ system that contains a
pair of lines intersecting in a point. For $q\in \{3,5\}$ there exist a faithful projective $2-(q^2+2,5,2)_q$ system that contains a pair of lines intersecting in a point.
For $q=3$ the maximum size $n$ of a faithful projective $2-(n,5,2)_q$ system is $12=q^2+q$ \cite{ball2024additive}.
\end{remark}

A subclass of special interest are so-called \emph{MDS codes} attaining the Singleton bound with equality, see e.g.\ \cite{ball2023additive,ball2024additive}. Many of these codes fall into the
class of Reed--Solomon codes, but there are also other constructions see
e.g.\ \cite[Remark 27]{ball2024additive} and \cite{jin2024new}.

\begin{proposition}
  \label{prop_r_6_h_2_s_q}
  We have $n_q(6,2;q)\le q\cdot \left(q^2-q+1\right)$. If $\cS$ is a
  faithful projective $2-(n,6,q,\mu)_q$ system attaining equality, then we
  have $\mu=1$ and each $4$-space contains either $0$ or $q$ elements
  from $\cS$ and each subset of elements of $\cS$ spans an even-dimensional
  subspace.
\end{proposition}
\begin{proof}
  If $L_1,L_2$ are two different elements in $\cS$ that intersect in a
  point, then let $\pi$ be the $3$-space spanned by $L_1$, $L_2$ and
  consider the projection of $\cS$ through $\pi$, see
  Lemma~\ref{lemma_projection}, so that Lemma~\ref{lemma_one_weight_bound} gives $n\le (q-2)\cdot \left(q^2+q+1\right)+2=q\cdot\left(q^2-q-1\right)$. Thus, if
  $n>q\cdot\left(q^2-q-1\right)$, then we have $\mu=1$ (noting that the
  case $L_1=L_2$ leads to even stronger upper bounds for $n$). Assuming
  $\mu=1$, let $\cS'$ be the projection of $\cS$ through $L_1$, so that
  $\cS'$ is a faithful projective $2-(n-1,4;q-1,\mu')_q$ system.
  Lemma~\ref{lemma_one_weight_bound} gives $\#\cS'=n-1\le \left(q^2+1\right)
  \cdot (q-1)$, so that $\#\cS=n\le q\cdot \left(q^2-q+1\right)$. In the
  case of equality we have $\mu'=q-1$ and $\cS'$ has type $(q-1)\cdot[4]$.
  Moreover, $\cS'$ is faithful, i.e., any subset of elements of $\cS$
  spans an even-dimensional subspace. If $S$ is an arbitrary $4$-space
  that contains an element $L$ from $\cS$, then projection through $L$
  yields that the elements of $\cS\backslash\{L\}$ cover $(q-1)\cdot(q+1)$
  points from $S$. Since no subset of elements from $\cS$ spans a $5$-space,
  every element of $\cS$ that intersects $S$ is fully contained in $S$, i.e.,
  $S$ contains exactly $q$ elements from $\cS$.
\end{proof}

\begin{remark}
\label{remark_perp_system}
A (multi-)set $\cS_q$ of $2$-spaces in $\PG(5,q)$ with cardinality $q(q^2-q+1)$
such that each $4$-space contains either $0$ or $q$ elements from $\cS$ is a
special case of a so-called \emph{perp-system}, see \cite{clerck2001perp}
for details. They do indeed exist for even field sizes $q$ \cite[Lemma 5.1]{clerck2001perp}. The construction is based on maximal arcs
in $\PG(2,q^2)$ -- Denniston arcs to be more precise \cite{denniston1969some} -- i.e.\ the corresponding codes are linear over $\F_{q^2}$. For odd $q$ we
cannot obtain such examples from maximal arcs in $\PG(2,q^2)$ \cite{ball1997maximal}, so that $\overline{n}_q(6,2;q)<q\cdot \left(q^2-q+1\right)$. For $q=3$ an example attaining the upper bound from
Proposition~\ref{prop_r_6_h_2_s_q} was found by a computer search, see
\cite[Example 2]{clerck2001perp}, and for odd $q>3$ no such example is known.
\end{remark}

\subsection{Additive codes that are linear over the binary field}
\label{subsec_q_2}
As mentioned before, we have $n_q(r,h;s)=\infty$ for $h\le r$, so that we assume
$r\ge h+1$. For $n_2(3,2;s)$ we refer to Theorem~\ref{thm_hp1_h_fractional}
and for $n_2(4,2;s)$ we refer to Theorem~\ref{thm_r_h_integral_2}.
\begin{theorem} (\cite{bierbrauer2021optimal})
  We have
  \begin{itemize}
   \item $n_2(5,2;7t)=31t$ for $t\ge 1$;\\[-10mm]
   \item $n_2(5,2;7t-1)=31t-5$ for $t\ge 1$;\\[-10mm]
   \item $n_2(5,2;7t-2)=31t-10$ for $t\ge 1$;\\[-10mm]
   \item $n_2(5,2;7t-3)=31t-15$ for $t\ge 1$;\\[-10mm]
   \item $n_2(5,2;7t-4)=31t-20$ for $t\ge 1$;\\[-10mm]
   \item $n_2(5,2;7t-5)=31t-23$ for $t\ge 1$;\\[-10mm]
   \item $n_2(5,2;7t-6)=31t-28$ for $t\ge 2$ and $n_2(5,2;1)=1$.
  \end{itemize}
\end{theorem}
\begin{proof}
  Lemma~\ref{lemma_small} gives $n_2(5,2;1)=1$. The other upper bounds
  follow from the Griesmer upper bound. Due to Corollary~\ref{cor_asymptotic_one_weight} it suffices to give a construction
  for the first elements of the seven sequences. Theorem~\ref{thm_partition} gives
  $n_2(5,2;7)\ge 31$. Lemma~\ref{lemma_vsp_type} shows that $[5]-[3]$ is $2$-partitionable over $\F_2$, so that $n_2(5,2;2)\ge 8$.
  From small linear codes we conclude $n_2(5,2;3)\ge \overline{n}_2(5,2;3)=11$,
  $n_2(5,2;4)\ge \overline{n}_2(5,2;4)=16$, and $n_2(5,2;5)\ge
  \overline{n}_2(5,2;5)=21$. From Lemma~\ref{lemma_partition_1} we conclude
  that $3[5]-[4]-2[2]$ is $2$-partitionable over $\F_2$, so that also
  $3[5]-[4]$ is $2$-partitionable over $\F_2$ and $n_2(5,2;6)\ge 26$.
  Corollary~\ref{cor_union} gives $n_2(5,2;8)\ge n_2(5,2;2)+n_2(5,2;6)\ge 34$.
\end{proof}

\begin{remark}
  For $r>2h$ Lemma~\ref{lemma_vsp_type} gives that $[r]-[r-h]$ is $h$-partitionable over $\F_q$, so that $n_q\!\left(r,h;q^{r-2h}\right)\ge q^{r-h}$.
  From the Griesmer upper bound we can conclude that indeed $n_q\!\left(r,h;q^{r-2h}\right)= q^{r-h}$. For $r>2h$ with $r\equiv 1\pmod h$
  Lemma~\ref{lemma_construction_x_consequence} gives that $[h-1]_q\cdot [r]+q^{h-1}\cdot[1]$ is $h$-partitionable over $\F_q$, so that
  $n_q\!\left(r,h;1+[h-1]_q\cdot \sum_{i=1}^{\left\lfloor (r-h)/h\right\rfloor} q^{r-h-ih}\right)\ge 1+[h-1]_q\cdot \sum_{i=1}^{\left\lfloor r/h\right\rfloor} q^{r-ih}$. From the Griesmer upper bound we can conclude that this lower bound is
  indeed tight.
\end{remark}

\begin{theorem} (\cite[Theorem 1]{bierbrauer2021optimal})
  We have $n_2(6,2;s)=\overline{n}_2(6,2;s)$ for all $s$, i.e.\
  \begin{itemize}
   \item $n_2(6,2;5t)=21t$ for $t\ge 1$;\\[-10mm]
   \item $n_2(6,2;5t-1)=21t-5$  for $t\ge 1$;\\[-10mm]
   \item $n_2(6,2;5t-2)=21t-10$  for $t\ge 1$ and $n_2(5,2;3)=9$;\\[-10mm]
   \item $n_2(6,2;5t-3)=21t-15$ for $t\ge 1$;\\[-10mm]
   \item $n_2(6,2;5t-4)=21t-20$ for $t\ge 1$.
  \end{itemize}
\end{theorem}
\begin{proof}
  The lower bounds follow from $n_2(6,2;s)\ge \overline{n}_2(6,2;s)$ and
  Corollary~\ref{cor_asymptotic_one_weight}. The upper bound
  $n_2(6,2;3)\le 9$ was shown in \cite[Section 4.2]{blokhuis2004small}.
  All other upper bounds follow from the Griesmer upper bound.
\end{proof}

\begin{theorem}\label{thm_n_2_7_2_s} (\cite{kurz2024optimal}; cf.~\cite[Table I]{guan2023some},\cite[Table II]{10693309})
  We have
  \begin{itemize}
    \item $n_2(7,2;31t)=127t$ for $t\ge 1$;\\[-10mm]
    \item $n_2(7,2;31t-1)=127t-5$ for $t\ge 1$;\\[-10mm]
    \item $n_2(7,2;31t-2)=127t-10$ for $t\ge 1$;\\[-10mm]
    \item $n_2(7,2;31t-3)=127t-15$ for $t\ge 1$;\\[-10mm]
    \item $n_2(7,2;31t-4)=127t-20$ for $t\ge 1$;\\[-10mm]
    \item $n_2(7,2;31t-5)=127t-21$ for $t\ge 1$;\\[-10mm]
    \item $n_2(7,2;31t-6)=127t-26$ for $t\ge 1$;\\[-10mm]
    \item $n_2(7,2;31t-7)=127t-31$ for $t\ge 1$;\\[-10mm]
    \item $n_2(7,2;31t-8)=127t-36$ for $t\ge 1$;\\[-10mm]
    \item $n_2(7,2;31t-9)=127t-41$ for $t\ge 1$;\\[-10mm]
    \item $n_2(7,2;31t-10)=127t-42$ for $t\ge 1$;\\[-10mm]
    \item $n_2(7,2;31t-11)=127t-47$ for $t\ge 1$;\\[-10mm]
    \item $n_2(7,2;31t-12)=127t-52$ for $t\ge 1$;\\[-10mm]
    \item $n_2(7,2;31t-13)=127t-55$ for $t\ge 1$;\\[-10mm]
    \item $n_2(7,2;31t-14)=127t-60$ for $t\ge 1$;\\[-10mm]
    \item $n_2(7,2;31t-15)=127t-63$ for $t\ge 1$;\\[-10mm]
    \item $n_2(7,2;31t-16)=127t-68$ for $t\ge 1$;\\[-10mm]
    \item $n_2(7,2;31t-17)=127t-73$ for $t\ge 1$;\\[-10mm]
    \item $n_2(7,2;31t-18)=127t-76$ for $t\ge 1$;\\[-10mm]
    \item $n_2(7,2;31t-19)=127t-81$ for $t\ge 1$;\\[-10mm]
    \item $n_2(7,2;31t-20)=127t-84$ for $t\ge 1$;\\[-10mm]
    \item $n_2(7,2;31t-21)=127t-87$ for $t\ge 1$;\\[-10mm]
    \item $n_2(7,2;31t-22)=127t-92$ for $t\ge 1$;\\[-10mm]
    \item $n_2(7,2;31t-23)=127t-95$ for $t\ge 1$;\\[-10mm]
    \item $n_2(7,2;31t-24)=127t-100$ for $t\ge 1$;\\[-10mm]
    \item $n_2(7,2;31t-25)=127t-105$ for $t\ge 1$;\\[-10mm]
    \item $n_2(7,2;31t-26)=127t-108$ for $t\ge 2$ and $n_2(7,2;5)=17$;\\[-10mm]
    \item $n_2(7,2;31t-27)=127t-113$ for $t\ge 2$ and $n_2(7,2;4)=12$;\\[-10mm]
    \item $n_2(7,2;31t-28)=127t-116$ for $t\ge 2$ and $n_2(7,2;3)=7$;\\[-10mm]
    \item $n_2(7,2;31t-29)=127t-121$ for $t\ge 2$ and $n_2(7,2;2)=2$;\\[-10mm]
    \item $n_2(7,2;31t-30)=127t-126$ for $t\ge 1$.
  \end{itemize}
\end{theorem}
\begin{proof}
  Lemma~\ref{lemma_small} gives $n_2(7,2;1)=1$ and $n_2(7,2;2)=2$. Theorem~\ref{thm_partition} yields $n_2(7,2;31t)=127t$ for $t\ge 1$.
  In \cite{blokhuis2004small} $n_2(7,2;3)\le 7$ was shown. The coding
  upper bound implies $n_2(7,2;4)\le 12$ and $n_2(7,2;5)\le 17$, cf.\ Remark~\ref{remark_projection_subspace_bound}. All other
  upper bounds follow from the Griesmer upper bound. Due to Corollary~\ref{cor_asymptotic_one_weight} and Corollary~\ref{cor_union}
  it suffices to give constructions for
  $
    s\in \{3,\dots, 13, 15, 21, 25, 26, 30\}
  $.
  Constructions for $s=3,4$ were given in \cite{blokhuis2004small} and
  for $s=5$ we can use $n_2(7,2;5)\ge \overline{n}_2(7,2;5)=17$. For
  $s\in\{ 6, 7, 12, 13\}$ examples were found using ILP searches, see
  \cite{kurz2024optimal} and Section~\ref{sec_searches}. For $s=9$ an example is given by a vector space
  partition of $\PG(6,2)$ of type $2^{35} 3^1 4^1$.
  For $s=15$ an example is given in \cite[Example 2]{10693309}.
  For
  $
    s\in \{8,10, 11, 21, 25, 26, 30\}
  $ examples
  can be easily constructed using the general tools provided in
  Section~\ref{sec_solomon_stiffler}, see
  Proposition~\ref{prop_q_2_r_7_h_2_generic} for the details.
\end{proof}

\begin{theorem} (\cite{kurz2024optimal})\label{thm_n_2_8_2_s}
  For $s\ge 30$ the Griesmer upper bound for $n_2(8,2;s)$ can always be attained.
  For all $s\in \N_{>0}$ with $s\not\equiv 2,3,7,8\pmod {21}$ and $s\notin\{9,10,11,14,15,24,27\}$ we have $n_2(8,2;s)=\overline{n}_2(8,2;s)$.
  More concretely:
  \begin{itemize}
    \item $n_2(8,2;21t)=85t$ for $t\ge 1$;\\[-10mm]
    \item $n_2(8,2;21t-1)=85t-5$ for $t\ge 1$;\\[-10mm]
    \item $n_2(8,2;21t-2)=85t-10$ for $t\ge 1$;\\[-10mm]
    \item $n_2(8,2;21t-3)=85t-15$ for $t\ge 1$;\\[-10mm]
    \item $n_2(8,2;21t-4)=85t-20$ for $t\ge 1$;\\[-10mm]
    \item $n_2(8,2;21t-5)=85t-21$ for $t\ge 1$;\\[-10mm]
    \item $n_2(8,2;21t-6)=85t-26$ for $t\ge 2$ and $n_2(8,2;15)=55$;\\[-10mm]
    \item $n_2(8,2;21t-7)=85t-31$ for $t\ge 1$;\\[-10mm]
    \item $n_2(8,2;21t-8)=85t-36$ for $t\ge 1$;\\[-10mm]
    \item $n_2(8,2;21t-9)=85t-41$ for $t\ge 1$;\\[-10mm]
    \item $n_2(8,2;21t-10)=85t-42$ for $t\ge 2$ and $n_2(8,2;11)=40$;\\[-10mm]
    \item $n_2(8,2;21t-11)=85t-47$ for $t\ge 2$ and $n_2(8,2;10)=36$;\\[-10mm]
    \item $n_2(8,2;21t-12)=85t-52$ for $t\ge 1$;\\[-10mm]
    \item $n_2(8,2;21t-13)=85t-55$ for $t\ge 3$, $n_2(8,2;8)=28$, and $n_2(8,2;29)=113$;\\[-10mm]
    \item $n_2(8,2;21t-14)=85t-60$ for $t\ge 3$, $n_2(8,2;7)=23$, and $n_2(8,2;28)=108$;\\[-10mm]
    \item $n_2(8,2;21t-15)=85t-63$ for $t\ge 2$ and $n_2(8,2;6)=18$;\\[-10mm]
    \item $n_2(8,2;21t-16)=85t-68$ for $t\ge 1$;\\[-10mm]
    \item $n_2(8,2;21t-17)=85t-73$ for $t\ge 2$ and $n_2(8,2;4)=10$;\\[-10mm]
    \item $n_2(8,2;21t-18)=85t-76$ for $t\ge 2$ and $n_2(8,2;3)=5$;\\[-10mm]
    \item $n_2(8,2;21t-19)=85t-81$ for $t\ge 2$ and $n_2(8,2;2)=2$;\\[-10mm]
    \item $n_2(8,2;21t-20)=85t-84$ for $t\ge 1$.
  \end{itemize}
\end{theorem}
\begin{proof}
  Lemma~\ref{lemma_small} gives $n_2(8,2;1)=1$ and $n_2(8,2;2)=2$. Theorem~\ref{thm_partition} yields $n_2(8,2;21t)=85t$ for $t\ge 1$.
  In \cite{blokhuis2004small} $n_2(8,2;3)\le 5$ and $n_2(8,2;4)\le 10$
  were shown. The coding upper bound implies $n_2(8,2;6)\le 18$,
  $n_2(8,2;7)\le 23$, $n_2(8,2;8)\le 28$, $n_2(8,2;10)\le 36$,
  $n_2(8,2;11)\le 40$. The strong coding upper bound implies $n_2(8,2;15)\le 55$, $n_2(8,2;28)\le 108$, and $n_2(8,2;29)\le 113$, see  \cite[Section 4.1]{kurz2024optimal} for details. All other
  upper bounds follow from the Griesmer upper bound. The lower bound
  $n_2(8,2;s)\ge \overline{n}_2(8,2;s)$ matches the upper bound for all
  $s\in\{5,\dots,48\}\backslash\{9,10,11,14,15,23,24,28,29,44,45\}$.
  For $s\in\{9,10,11,14,23,24,27,49,50\}$ we refer to
  \cite{kurz2024optimal} and Section~\ref{sec_searches} for explicit examples obtained using ILP searches.
  For $s=50$ also the tools from Section~\ref{sec_solomon_stiffler} can be
  used, see Proposition~\ref{prop_q_2_r_8_h_2_generic} for the details.
  With this, all remaining constructions can be obtained using  Corollary~\ref{cor_asymptotic_one_weight} and Corollary~\ref{cor_union}.
\end{proof}

\begin{remark}
  \label{remark_z4}
  One might wonder how the famous $\Z_4$-linear octacode, starting the series of Kerdock and Preparata codes, see e.g.\ \cite{hammons1994z}, fits into the setting of Theorem~\ref{thm_n_2_8_2_s}. In general, every linear code
  over $\Z_{p^r}$ yields an additive code over $\F_{p^r}$ but not every
  additive code over $\F_{p^r}$ comes from a linear code
  over $\Z_{p^r}$. For $(p,r)=(2,2)$ we write $\Z_4=\{0,1,2,3\}$ and
  $\F_4=\left\{0,1,\omega,\omega^2\right\}$ with $\omega^2=1+\omega$. The standard additive isomorphism $\psi\colon\Z_4\to\F_4$ is defined by
  $0\mapsto 0$, $1\mapsto 1$, $2\mapsto \omega$, and $3\mapsto \omega^2$.
  Starting from a $\Z_4$ generator matrix
  $$
    \begin{pmatrix}
      1&1&1&1&1&1&1&1\\
      0&1&2&3&0&1&2&3\\
      0&0&1&1&2&2&3&3\\
      0&0&0&1&1&2&2&3
    \end{pmatrix}
  $$
  of the octacode, we obtain an additive $\F_4$ generator matrix by applying $\psi$ to the rows and the doubled rows:
  $$
    \begin{pmatrix}
      1&1&1&1&1&1&1&1\\
      \omega&\omega&\omega&\omega&\omega&\omega&\omega&\omega\\
      0&1&\omega&1+\omega&0&1&\omega&1+\omega\\
      0&\omega&1+\omega&1&0&\omega&1+\omega&1\\
      0&0&1&1&\omega&\omega&1+\omega&1+\omega\\
      0&0&\omega&\omega&1+\omega&1+\omega&1&1\\
      0&0&0&1&1&\omega&\omega&1+\omega\\
      0&0&0&\omega&\omega&1+\omega&1+\omega&1
    \end{pmatrix}.
  $$
  Geometrically this is a (multi-)set $\cL$ of eight lines in $\PG(7,2)$. By replacing the lines by their three contained points we obtain a
  (multi-) set of points in $\PG(7,2)$ that corresponds to a
  $[24,8]_2$-code $C$ with generator matrix
  $$
    \begin{pmatrix}
    101&101&101&101&101&101&101&101\\
    011&011&011&011&011&011&011&011\\
    000&101&011&110&000&101&011&110\\
    000&011&110&101&000&011&110&101\\
    000&000&101&101&011&011&110&110\\
    000&000&011&011&110&110&101&101\\
    000&000&000&101&101&011&011&110\\
    000&000&000&011&011&110&110&101
    \end{pmatrix}.
  $$
  More precisely, $C$ is an even projective $[24,8,6]_2$-code with maximum
  weight $16$ and an automorphism group of order $12$. The weight distribution is given by $0^1 6^6   8^{24}   10^{36}   12^{108}   14^{54}   16^{27}$, so that up to five lines of $\cL$ are contained in a
  hyperplane of $\PG(7,2)$. Note that $n_2(8,2;4)=10$ and $n_2(8,2;5)=17$, i.e.\ there exist additive $\F_4$ codes with better parameters. However, applying the Gray map to the octacodes gives a non-linear (block) code with size $256=2^8$, length $16$, alphabet $\F_2$, and minimum Hamming distance $6$ known as the Nordstrom-Robinson or Semakov-Zinovev code \cite{forney1992nordstrom}. Each linear $[n,8,6]_2$-code has length $n\ge 17$, i.e.\ the Nordstrom-Robinson code outperforms linear binary codes.
\end{remark}

\medskip

\noindent
A few further constructions and upper bounds for $n_2(r,2;s)$ can be
found in the literature:\\[-10mm]
\begin{itemize}
 \item $n_2(14,2;7)\le 11$ \cite{bierbrauer2009short};\\[-10mm]
 \item $n_2(9,2;5)\le 11$ \cite{bierbrauer2009short}, \cite{bierbrauer2010geometric};\\[-10mm]
 \item $n_2(15,2;8)\le 13$ \cite{bierbrauer2009short};\\[-10mm]
 \item $n_2(10,2;6)\le 14$ \cite{bierbrauer2015nonexistence};\\[-10mm]
 \item $n_2(28,2;14)\le 17$ \cite{bierbrauer2019additive};\\[-10mm]
 \item $n_2(35,2;18)\ge 22$ \cite{bierbrauer2019additive}.
\end{itemize}

\bigskip

For $n_2(4,3;s)$ we refer to Theorem~\ref{thm_hp1_h_fractional}, for
$n_2(5,3;s)$ we refer to Proposition~\ref{prop_hp2_h_fractional} in
combination with Table~\ref{table_n_q_5_3}, and for $n_2(6,3;s)$ we
refer to Theorem~\ref{thm_r_h_integral_2}.

\begin{theorem}
  We have\\[-10mm]
\begin{itemize}
\item $n_2(7,3;15t)=127t$ for $t\ge 1$;\\[-10mm]
\item $n_2(7,3;15t-1)=127t-9$ for $t\ge 1$;\\[-10mm]
\item $n_2(7,3;15t-2)=127t-18$ for $t\ge 1$;\\[-10mm]
\item $n_2(7,3;15t-3)=127t-27$ for $t\ge 1$;\\[-10mm]
\item $n_2(7,3;15t-4)=127t-36$ for $t\ge 1$;\\[-10mm]
\item $n_2(7,3;15t-5)=127t-45$ for $t\ge 1$;\\[-10mm]
\item $n_2(7,3;15t-6)=127t-54$ for $t\ge 1$;\\[-10mm]
\item $n_2(7,3;15t-7)=127t-61$ for $t\ge 1$;\\[-10mm]
\item $n_2(7,3;15t-8)=127t-70$ for $t\ge 1$;\\[-10mm]
\item $n_2(7,3;15t-9)=127t-77$ for $t\ge 1$;\\[-10mm]
\item $n_2(7,3;15t-10)=127t-86$ for $t\ge 1$;\\[-10mm]
\item $n_2(7,3;15t-11)=127t-95$ for $t\ge 1$;\\[-10mm]
\item $n_2(7,3;15t-12)=127t-104$ for $t\ge 1$;\\[-10mm]
\item $n_2(7,3;15t-13)=127t-111$ for $t\ge 1$;\\[-10mm]
\item $n_2(7,3;15t-14)=127t-120$ for $t\ge 2$ and $n_2(7,3;1)=1$.
\end{itemize}
\end{theorem}
\begin{proof}
  Lemma~\ref{lemma_small} gives $n_2(7,3;1)=1$. Theorem~\ref{thm_partition}
  yields $n_2(7,3;15)=127$. All other
  upper bounds follow from the Griesmer upper bound. The lower bound
  $n_2(7,3;s)\ge \overline{n}_2(7,3;s)$ matches the upper bound for $s=9$.
  For $s\in\{3,5\}$ examples have been found using ILP searches, see
  Section~\ref{sec_searches}. For $s\in\{2,6,13,14\}$ the constructions from
  Section~\ref{sec_solomon_stiffler} can be used, see Proposition~\ref{prop_q_2_r_7_h_3_generic} for the details. With this, all remaining constructions can be obtained using  Corollary~\ref{cor_asymptotic_one_weight} and Corollary~\ref{cor_union}.
\end{proof}

\begin{remark}
  We have $n_2(7,3;2)=16$ and it is not hard to show that any projective $3-(16,7,2,\mu)_2$ system $\cS$ is indeed faithful and we have $\mu=1$, i.e.\
  $\cS$ is a partial plane spread of cardinality $16$ in $\PG(6,2)$. Those
  have been classified in \cite{honold2019classification} and there are
  exactly $37+3988=4025$ isomorphism types. For more details on the equivalence of linear or additive codes we refer e.g.\ to \cite{ball2022equivalence}.
\end{remark}

\begin{lemma}
  \label{coding_bound_improved_e1}
  We have $n_3(8,3;3)\le 20$.
\end{lemma}
\begin{proof}
  Assume that $\cS$ is a faithful projective $3-(21,8,3)_2$ system and $C:=\cX^{-1}(\cS)$ is the linear code that corresponds to the multiset
  of points covered by the elements of $\cS$. By Lemma~\ref{lemma_linear_code}
  $C$ is a $[147,8,\{72,76,80,84\}]_2$ code.

  If a plane $S$ would be contained at least twice, then $\#\cS\le 2+ [8-3]_2/[8-6]_2\le 13$ -- contradiction.  If two different planes $S_1,S_2\in\cS$ intersect in at least a point then projection through $S_1$ gives
  a $3-(20,5,2)_2$ system $\cS'$ by Lemma~\ref{lemma_projection_subspace}, where
  $S_2$ is projected to a line or a point, i.e.\ $\cS'$ is unfaithful. Let
  $S'\in \cS'$ be an element of dimension at most $2$, so that $S'$ is
  contained in at least $7$ hyperplanes of $\PG(4,2)$. These hyperplanes
  contain at most one other element from $\cS'$. Since every hyperplane
  contains at most two elements from $\cS'\backslash\{S'\}$ we have
  $$
    \#\cS'-1 \le \frac{7\cdot 1+24\cdot 2}{3}<19,
  $$
  which is a contradiction. Thus, no two elements of $\cS$ intersect in a point
  and the linear code $C$ is projective. Let $A_i$ denote the number of codewords
  of weight $i$ in $C$. From the first three MacWilliams identities we compute
  $A_{84}=234-A_{72}$, $A_{80}=3A_{72}-609$, and $A_{76}=630-3A_{72}$.

  Let $c\in C$ be a codeword of weight $84$ and $C'$ be the corresponding residual
  code,\footnote{Our coding theoretic proof of $n_3(8,3;3)\le 20$ is rather concise and we did not introduce all used concepts properly. However, all
  of this is quite standard when the aim is to show the non-existence of
  certain linear codes see e.g.\ \cite{bouyukhev2000smallest} or
  \cite[Chapter 3]{guritman2000restrictions}.} so that $C'$ is a
  projective $[63,7,\{30,32,34,36,38\}]_2$ code. Using
  the software package \texttt{LinCode} \cite{bouyukliev2021computer}
  we have verified by exhaustive enumeration that there is a unique
  such code and a generator matrix is given by
  $$
    \begin{pmatrix}
      111111111111111111111111000000001111000000001110000000001000000\\
      111111111111111100000000111111110000111100000001110000000100000\\
      000000001111111111111111111111110000000011110000001110000010000\\
      000011110000111100001111000011111111111111110000000001110001000\\
      011100110011000100010011000100110111011100010110110110110000100\\
      101101010101011100100101001001010001000101111011011011010000010\\
      000110100110101101100001011011011010001100101110100111100000001
    \end{pmatrix}.
  $$
  The corresponding weight enumerator is given by
  $$
    1+76x^{30}+31x^{32}+20x^{38}.
  $$
  With this we can compute the number $A_i$ of codewords of weight $i$ of
  $C$ as $A_{72}=183$, $A_{76}=51-y$, $A_{80}=2y$, and $A_{84}=21-y$,
  where $0\le y\le 20$. 
  Plugging $A_{72}=183$ into the previous equations gives $A_{84}=51$,
  $A_{80}=-60$, and $A_{76}=81$, which is a contradiction since $A_{80}$
  cannot be negative.

  If $C$ does not contain a codeword of weight $84$, then we have $A_{84}=0$,
  so that $A_{72}=234$, $A_{80}=93$, and $A_{76}=-72$, which is a contradiction
  since $A_{76}$ cannot be negative.
\end{proof}
We remark that an $[147,8,\{72,76,80,128\}]_2$ code exists and the coding upper
bound gives $n_3(8,3;3)\le 21$.


\begin{theorem}
  We have\\[-10mm]
\begin{itemize}
\item $n_2(8,3;31t)=255t$ for $t\ge 1$;\\[-10mm]
\item $n_2(8,3;31t-1)=255t-9$ for $t\ge 1$;\\[-10mm]
\item $n_2(8,3;31t-2)=255t-18$ for $t\ge 1$;\\[-10mm]
\item $n_2(8,3;31t-3)=255t-27$ for $t\ge 1$;\\[-10mm]
\item $n_2(8,3;31t-4)=255t-36$ for $t\ge 1$;\\[-10mm]
\item $n_2(8,3;31t-5)=255t-43$ for $t\ge 1$;\\[-10mm]
\item $n_2(8,3;31t-6)=255t-52$ for $t\ge 1$;\\[-10mm]
\item $n_2(8,3;31t-7)=255t-59$ for $t\ge 1$;\\[-10mm]
\item $n_2(8,3;31t-8)=255t-68$ for $t\ge 1$;\\[-10mm]
\item $n_2(8,3;31t-9)=255t-75$ for $t\ge 1$;\\[-10mm]
\item $n_2(8,3;31t-10)=255t-84$ for $t\ge 1$;\\[-10mm]
\item $n_2(8,3;31t-11)=255t-91$ for $t\ge 1$;\\[-10mm]
\item $n_2(8,3;31t-12)=255t-100$ for $t\ge 1$;\\[-10mm]
\item $n_2(8,3;31t-13)=255t-109$ for $t\ge 1$;\\[-10mm]
\item $n_2(8,3;31t-14)=255t-118$ for $t\ge 1$;\\[-10mm]
\item $n_2(8,3;31t-15)=255t-127$ for $t\ge 1$;\\[-10mm]
\item $n_2(8,3;31t-16)=255t-134$ for $t\ge 2$ and $n_2(8,3;15)\in \{119,\dots,121\}$;\\[-10mm]
\item $n_2(8,3;31t-17)=255t-143$ for $t\ge 1$;\\[-10mm]
\item $n_2(8,3;31t-18)=255t-150$ for $t\ge 1$;\\[-10mm]
\item $n_2(8,3;31t-19)=255t-159$ for $t\ge 1$;\\[-10mm]
\item $n_2(8,3;31t-20)=255t-166$ for $t\ge 2$ and $n_2(8,3;11)\in \{87,\dots,89\}$;\\[-10mm]
\item $n_2(8,3;31t-21)=255t-175$ for $t\ge 1$;\\[-10mm]
\item $n_2(8,3;31t-22)=255t-182$ for $t\ge 1$;\\[-10mm]
\item $n_2(8,3;31t-23)=255t-191$ for $t\ge 1$;\\[-10mm]
\item $n_2(8,3;31t-24)=255t-200$ for $t\ge 1$;\\[-10mm]
\item $n_2(8,3;31t-25)=255t-209$ for $t\ge 1$;\\[-10mm]
\item $n_2(8,3;31t-26)=255t-218$ for $t\ge 1$;\\[-10mm]

\item $n_2(8,3;31t-27)=255t-223$ for $t\ge 1$;\\[-10mm]
\item $n_2(8,3;31t-28)=255t-232$ for $t\ge 2$ and $n_2(8,3;3)\in \{18,\dots,20\}$;\\[-10mm]
\item $n_2(8,3;31t-29)=255t-241$ for $t\ge 2$ and $n_2(8,3;2)=10$;\\[-10mm]
\item $n_2(8,3;31t-30)=255t-250$ for $t\ge 2$ and $n_2(8,3;1)=1$.
\end{itemize}
\end{theorem}
\begin{proof}
  Lemma~\ref{lemma_small} gives $n_2(8,3;1)=1$. Theorem~\ref{thm_partition}
  yields $n_2(8,3;31)=255$. From Lemma~\ref{lemma_projection_subspace_bound}
  and $n_2(5,3;2)=9$ we conclude $n_2(8,3;2)\le 10$. Lemma~\ref{coding_bound_improved_e1} gives $n_2(8,8;3)\le 20$.
  All other upper bounds follow from the Griesmer
  upper bound. The lower bound $n_2(8,3;s)\ge \overline{n}_2(8,3;s)$ matches the upper bound for $s\in\{2,9\}$. For $s\in\{3,5,6,7,10,19,20,22\}$ examples
  have been found using ILP searches, see Section~\ref{sec_searches}.
  For $s\in\{4,21,25,30,41\}$
  the constructions from Section~\ref{sec_solomon_stiffler} can be used, see
  Proposition~\ref{prop_q_2_r_8_h_3_generic} for the details. With this, all remaining constructions can be obtained using  Corollary~\ref{cor_asymptotic_one_weight} and Corollary~\ref{cor_union}.
\end{proof}

\begin{remark}
  By a heuristic search we have constructed more than eighty thousand $[132,7,\{64,68,72,76\}]_2$
  codes. After extending nine thousand of these codes we have found two
  $[133,8,\{64,68,72,76\}]_2$ codes. The corresponding (multi-) set of
  points allows to choose only between $15$ and $17$ planes instead of $19$.
\end{remark}


\begin{remark}
  \label{remark_dimension_arguments}
  Assume that $\cS$ is a faithful projective $3-(n,8,3)_2$ system that
  matches the upper bound $n_2(8,3;11)\le 89=:n$. The linear code $C$
  corresponding to the multiset of points $\cM$ covered by the elements of $\cS$,
  see Lemma~\ref{lemma_linear_code}, would be a $4$-divisible $[623,8,\ge 312]_2$
  code. In Proposition~\ref{prop_q_2_r_8_h_3_generic} we showed that $(7t-4)\cdot[8]-[7]-[4]$ is $3$-partitionable over $\F_2$ for all $t\ge 3$. Plugging
  in $t=1$ we easily see that $3[8]-[7]-[4]$ is $1$-partitionable over
  $\F_2$, which is essentially the Solomon--Stiffler construction for the
  code linear $C$. However, this code is ruled out due to the condition on
  the maximum weight in Lemma~\ref{lemma_linear_code}. Framed differently,
  if $H=S_7$ is the hyperplane that is removed (according to the notation in
  Definition~\ref{def_partitionable}), then we have
  $\cM(H)\le (3-1)\cdot[7]_2=254$. However, each element of $\cS$ intersects $H$
  in at least $[2]_2=3$ points, which yields the contradiction
  $3\cdot 89=267>254$. Thus, we conclude that $3[8]-[7]-[4]$ is
  $3$-partitionable over $\F_2$ nevertheless all conditions of
  Theorem~\ref{thm_main} are satisfied. Such {\lq\lq}dimension
  arguments{\rq\rq} are quite  common in non-existence proofs of vector
  space partitions of a certain type, see e.g.\ \cite{el2009partitions}.
\end{remark}

\begin{theorem}
  We have\\[-10mm]
\begin{itemize}
\item $n_2(9,3;9t)=73t$ for $t\ge 1$;\\[-10mm]
\item $n_2(9,3;9t-1)=73t-9$ for $t\ge 1$;\\[-10mm]
\item $n_2(9,3;9t-2)=73t-18$ for $t\ge 2$ and $n_2(9,3;7)\in\{51,\dots,55\}$;\\[-10mm]
\item $n_2(9,3;9t-3)=73t-27$ for $t\ge 2$ and $n_2(9,3;6)\in\{42,\dots,46\}$;\\[-10mm]
\item $n_2(9,3;9t-4)=73t-36$ for $t\ge 2$ and $n_2(9,3;5)\in\{33,\dots,37\}$;\\[-10mm]
\item $n_2(9,3;9t-5)=73t-43$ for $t\ge 7$, $n_2(9,3;4)=28$, and $n_2(9,3;13)\in\{101,\dots,103\}$;\\[-10mm]
\item $n_2(9,3;9t-6)=73t-52$ for $t\ge 8$, $n_2(9,3;3)\in\{15,\dots,19\}$, and $n_2(9,3;12)\in\{92,\dots,94\}$;\\[-10mm]
\item $n_2(9,3;9t-7)=73t-59$ for $t\ge 10$, $n_2(9,3;2)=10$, and $n_2(9,3;11)\in\{83,\dots,87\}$;\\[-10mm]
\item $n_2(9,3;9t-8)=73t-68$ for $t\ge 11$, $n_2(9,3;1)=1$, and $n_2(9,3;10)\in\{76,\dots,78\}$.
\end{itemize}
\end{theorem}
\begin{proof}
  Lemma~\ref{lemma_small} gives $n_2(9,3;1)=1$. Theorem~\ref{thm_partition}
  yields $n_2(9,3;9)=73$. We have $n_2(9,3;2)\le n_2(8,3;1)=10$. The coding
  upper bound implies $n_2(9,3;3)\le 19$ and $n_2(9,3;4)\le 28$.
  More precisely, in \cite{bouyukhev2000smallest} $d\le 33$ for each $[73,8,d]_2$
  code was shown, so that $d\le 66$ for each $[140,9,d]_2$ code and $n_2(9,3;3)<20$.
  In \cite{van1979proof} the non-existence of a $[56,7,26]_2$ code was shown, so that no $[104,8,50]_2$ and no $[203,9,99]_2$ codes exist.
  With this, we conclude $n_2(9,3;4)< 29$. All other upper bounds follow from
  the Griesmer upper bound. For lower bounds for $n_2(9,3;9t-i)$, where
  $i\in\{5,6,7,8\}$, and large $t$ we refer to
  Proposition~\ref{prop_q_2_r_9_h_3_generic}. By ILP searches the
  lower bounds $n_2(9,3;7)\ge 51$ and $n_2(9,3;10)\ge 76$ have been
  found in \cite{krotovkurz2025}, see also Section~\ref{sec_searches} for generator matrices. All other lower bounds are
  obtained from $n_2(9,3;s)\ge \overline{n}_2(9,3;s)$.
\end{proof}
As already remarked in Table~\ref{table_improvements}, we have
$n_2(9,3;9t-i)>\overline{n}_2(9,3;9t-i)$ for all $i\in\{5,6,7,8\}$ and
sufficiently large $t$. We conjecture that the lower bounds on $t$ from
Proposition~\ref{prop_q_2_r_9_h_3_generic} can be lowered substantially.

\subsection{Additive codes that are linear over the ternary field}
\label{subsec_q_3}
As mentioned before, we have $n_q(r,h;s)=\infty$ for $h\le r$, so that we assume
$r\ge h+1$. For $n_3(3,2;s)$ we refer to Theorem~\ref{thm_hp1_h_fractional}
and for $n_3(4,2;s)$ we refer to Theorem~\ref{thm_r_h_integral_2}.

\begin{theorem}
  We have
  \begin{itemize}
\item $n_3(5,2;13t)=121t$ for $t\ge 1$;\\[-10mm]
\item $n_3(5,2;13t-1)=121t-10$ for $t\ge 1$;\\[-10mm]
\item $n_3(5,2;13t-2)=121t-20$ for $t\ge 1$;\\[-10mm]
\item $n_3(5,2;13t-3)=121t-30$ for $t\ge 1$;\\[-10mm]
\item $n_3(5,2;13t-4)=121t-40$ for $t\ge 1$;\\[-10mm]
\item $n_3(5,2;13t-5)=121t-50$ for $t\ge 1$;\\[-10mm]
\item $n_3(5,2;13t-6)=121t-60$ for $t\ge 1$;\\[-10mm]
\item $n_3(5,2;13t-7)=121t-67$ for $t\ge 1$;\\[-10mm]
\item $n_3(5,2;13t-8)=121t-77$ for $t\ge 1$;\\[-10mm]
\item $n_3(5,2;13t-9)=121t-87$ for $t\ge 1$;\\[-10mm]
\item $n_3(5,2;13t-10)=121t-94$ for $t\ge 1$;\\[-10mm]
\item $n_3(5,2;13t-11)=121t-104$ for $t\ge 2$ and $n_3(5,2;2)=12$;\\[-10mm]
\item $n_3(5,2;13t-12)=121t-114$ for $t\ge 2$ and  $n_3(5,2;1)=1$.
\end{itemize}
\end{theorem}
\begin{proof}
  Lemma~\ref{lemma_small} gives $n_3(5,2;1)=1$ and $n_3(5,2;2)=12$ was shown
  in \cite{ball2024additive} by exhaustive enumeration. The other upper
  bounds follow from the Griesmer upper bound. For $s\in\{9,10,20\}$ the
  lower bound $n_3(5,2;s)\ge \overline{n}_3(5,2;s)$ matches the upper bound.
  Theorem~\ref{thm_partition} gives $n_3(5,2;13)\ge 121$. Lemma~\ref{lemma_vsp_type} shows that $[5]-[3]$ is $2$-partitionable over $\F_3$, so that $n_3(5,2;3)\ge 27$.
  From Lemma~\ref{lemma_partition_1} we conclude that $4[5]-[4]-3[2]$ is
  $2$-partitionable over $\F_3$, so that also $4[5]-[4]$ is $2$-partitionable
  over $\F_3$ and $n_3(5,2;12)\ge 111$. Via ILP searches we found the following
  lower bounds: $n_3(5,2;4)\ge 34$, $n_3(5,2;5)\ge 44$, $n_3(5,2;7)\ge 61$,
  $n_3(5,2;8)\ge 71$, and $n_3(5,2;11)\ge 101$, see Section~\ref{sec_searches}. The examples for $s\in\{4,5\}$ were first found in \cite{krotovkurz2025}.
  The other lower bounds follow from Corollary~\ref{cor_asymptotic_one_weight} and Corollary~\ref{cor_union}.
\end{proof}

\medskip

\begin{theorem}
  We have\\[-10mm]
\begin{itemize}
\item $n_3(6,2;10t)=91t$ for $t\ge 1$;\\[-10mm]
\item $n_3(6,2;10t-1)=91t-10$ for $t\ge 1$;\\[-10mm]
\item $n_3(6,2;10t-2)=91t-20$ for $t\ge 2$ and $n_3(6,2;8)\in\{66,\dots,68\}$;\\[-10mm]
\item $n_3(6,2;10t-3)=91t-30$ for $t\ge 2$ and $n_3(6,2;7)\in\{55,\dots,61\}$;\\[-10mm]
\item $n_3(6,2;10t-4)=91t-40$ for $t\ge 2$ and $n_3(6,2;6)\in\{48,49\}$;\\[-10mm]
\item $n_3(6,2;10t-5)=91t-50$ for $t\ge 2$ and $n_3(6,2;5)\in\{37,\dots,41\}$;\\[-10mm]
\item $n_3(6,2;10t-6)=91t-60$ for $t\ge 2$ and $n_3(6,2;4)\in\{28,\dots,31\}$;\\[-10mm]
\item $n_3(6,2;10t-7)=91t-67$ for $t\ge 8$, $n_3(6,2;3)=21$, $n_3(6,2;13)\in\{112,113\}$, and $n_3(6,2;23)\in\{203,\dots,206\}$;\\[-10mm]
\item $n_3(6,2;10t-8)=91t-77$ for $t\ge 2$ and $n_3(6,2;2)=10$;\\[-10mm]
\item $n_3(6,2;10t-9)=91t-87$ for $t\ge 2$ and $n_3(6,2;1)=1$.\\[-10mm]
\end{itemize}
\end{theorem}
\begin{proof}
  Lemma~\ref{lemma_small} gives $n_3(6,3;1)=1$ and Theorem~\ref{thm_r_3h_s_2} gives $n_3(6,2;2)=10$. The coding upper
  bound yields $n_3(6,2;3)\le 21$, $n_3(6,2;6)\le 49$, $n_3(6,2;8)\le 68$,
  and $n_3(6,2;13)\le 113$. The other upper bounds follow from the Griesmer
  upper bound. For $s\in\{9,14,\dots,19\}$ the lower bound $n_3(6,2;s)\ge
  \overline{n}_3(6,2;s)$ matches the upper bound.
  Theorem~\ref{thm_partition} gives $n_3(6,2;10)\ge 91$ and
  $n_3(6,2;3)\ge 21$ was shown in \cite{clerck2001perp}. For the three
  series of improvements $n_3(6,2;10t-i)>\overline{n}_3(6,2;10t-i)$ for
  $i\in\{7,8,9\}$ we refer to Proposition~\ref{prop_q_3_r_6_h_2_generic}.
  Via ILP searches we found the following
  lower bounds: $n_3(6,2;8)\ge 66$, $n_3(6,2;11)\ge 95$, and $n_3(6,2;12)\ge 105$, see Section~\ref{sec_searches}.
  The other lower bounds follow from $n_3(6,2;s)\ge
  \overline{n}_3(6,2;s)$, Corollary~\ref{cor_asymptotic_one_weight}, and Corollary~\ref{cor_union}.
\end{proof}
We remark that the function $n_3(6,1;s)$ is only partially known. By IlP computations we have checked that $3$-divisible $[196,6,\{129,\dots,147\}]_3$- and $[224,6,\{147,\dots,168\}]_3$ codes cannot have an element of order seven in their automorphism group.


\subsection{Additive codes for dimension 2.5}
\label{subsec_dimension_2_5}
In \cite{krotovkurz2025} the following results for $n_4(5,2;s)$ and $n_5(5,2;s)$ were obtained.

\begin{theorem}
  \label{thm_n_4_5_2}
  We have\\[-10mm]
  \begin{itemize}
    \item $n_4(5,2;21t)=341t$ for $t\ge 1$;\\[-10mm]
    \item $n_4(5,2;21t-1)=341t-17$ for $t\ge 1$;\\[-10mm]
    \item $n_4(5,2;21t-2)=341t-34$ for $t\ge 1$;\\[-10mm]
    \item $n_4(5,2;21t-3)=341t-51$ for $t\ge 1$;\\[-10mm]
    \item $n_4(5,2;21t-4)=341t-68$ for $t\ge 1$;\\[-10mm]
    \item $n_4(5,2;21t-5)=341t-85$ for $t\ge 1$;\\[-10mm]
    \item $n_4(5,2;21t-6)=341t-102$ for $t\ge 1$;\\[-10mm]
    \item $n_4(5,2;21t-7)=341t-119$ for $t\ge 1$;\\[-10mm]
    \item $n_4(5,2;21t-8)=341t-136$ for $t\ge 1$;\\[-10mm]
    \item $n_4(5,2;21t-9)=341t-149$ for $t\ge 1$;\\[-10mm]
    \item $n_4(5,2;21t-10)=341t-166$ for $t\ge 1$;\\[-10mm]
    \item $n_4(5,2;21t-11)=341t-183$ for $t\ge 2$ and $n_4(5,2;10)\in\{154,\dots,158\}$;\\[-10mm]
    \item $n_4(5,2;21t-12)=341t-200$ for $t\ge 1$;\\[-10mm]
    \item $n_4(5,2;21t-13)=341t-213$ for $t\ge 1$;\\[-10mm]
    \item $n_4(5,2;21t-14)=341t-230$ for $t\ge 2$ and $n_4(5,2;7)\in\{107,\dots,111\}$;\\[-10mm]
    \item $n_4(5,2;21t-15)=341t-247$ for $t\ge 2$ and $n_4(5,2;6)\in\{90,\dots,94\}$;\\[-10mm]
    \item $n_4(5,2;21t-16)=341t-264$ for $t\ge 2$ and $n_4(5,2;5)\in\{75,\dots,77\}$;\\[-10mm]
    \item $n_4(5,2;21t-17)=341t-277$ for $t\ge 1$;\\[-10mm]
    \item $n_4(5,2;21t-18)=341t-294$ for $t\ge 2$ and $n_4(5,2;3)\in\{39,\dots,42\}$;\\[-10mm]
    \item $n_4(5,2;21t-19)=341t-311$ for $t\ge 2$ and $n_4(5,2;2)\in\{20,21,22\}$;\\[-10mm]
    \item $n_4(5,2;21t-20)=341t-328$ for $t\ge 2$.
  \end{itemize}
\end{theorem}

\begin{theorem}
  \label{thm_n_5_5_2}
  We have\\[-10mm]
  \begin{itemize}
\item $n_5(5,2;31t)=781t$ for $t\ge 1$;\\[-10mm]
\item $n_5(5,2;31t-1)=781t-26$ for $t\ge 1$;\\[-10mm]
\item $n_5(5,2;31t-2)=781t-52$ for $t\ge 2$ and $n_5(5,2;29)\in \{719,\dots,729\}$;\\[-10mm]
\item $n_5(5,2;31t-3)=781t-78$ for $t\ge 2$ and $n_5(5,2;28)\in \{693,\dots,703\}$;\\[-10mm]
\item $n_5(5,2;31t-4)=781t-104$ for $t\ge 2$ and $n_5(5,2;27)\in \{667,\dots,677\}$;\\[-10mm]
\item $n_5(5,2;31t-5)=781t-130$ for $t\ge 1$;\\[-10mm]
\item $n_5(5,2;31t-6)=781t-156$ for $t\ge 1$;\\[-10mm]
\item $n_5(5,2;31t-7)=781t-182$ for $t\ge 2$ and $n_5(5,2;24)\in \{594,\dots,599\}$;\\[-10mm]
\item $n_5(5,2;31t-8)=781t-208$ for $t\ge 2$ and $n_5(5,2;23)\in \{568,\dots,573\}$;\\[-10mm]
\item $n_5(5,2;31t-9)=781t-234$ for $t\ge 2$ and $n_5(5,2;22)\in \{542,\dots,547\}$;\\[-10mm]
\item $n_5(5,2;31t-10)=781t-260$ for $t\ge 1$;\\[-10mm]
\item $n_5(5,2;31t-11)=781t-281$ for $t\ge 1$;\\[-10mm]
\item $n_5(5,2;31t-12)=781t-307$ for $t\ge 2$ and $n_5(5,2;19)\in \{455,\dots,474\}$;\\[-10mm]
\item $n_5(5,2;31t-13)=781t-333$ for $t\ge 3$, $n_5(5,2;18)\in \{433,\dots,448\}$, and $n_5(5,2;49)\in \{1224,\dots,1229\}$;\\[-10mm]
\item $n_5(5,2;31t-14)=781t-359$ for $t\ge 2$ and $n_5(5,2;17)\in \{412,\dots,422\}$;\\[-10mm]
\item $n_5(5,2;31t-15)=781t-385$ for $t\ge 2$ and $n_5(5,2;16)\in \{391,\dots,396\}$;\\[-10mm]
\item $n_5(5,2;31t-16)=781t-406$ for $t\ge 1$;\\[-10mm]
\item $n_5(5,2;31t-17)=781t-432$ for $t\ge 2$ and $n_5(5,2;14)\in \{330,\dots,349\}$;\\[-10mm]
\item $n_5(5,2;31t-18)=781t-458$ for $t\ge 3$, $n_5(5,2;13)\in \{308,\dots,323\}$, and $n_5(5,2;44)\in \{1096,\dots,1104\}$;\\[-10mm]
\item $n_5(5,2;31t-19)=781t-484$ for $t\ge 3$, $n_5(5,2;12)\in \{286,\dots,297\}$, and $n_5(5,2;43)\in \{1067,\dots,1078\}$;\\[-10mm]
\item $n_5(5,2;31t-20)=781t-510$ for $t\ge 2$ and $n_5(5,2;11)\in \{260,\dots,271\}$;\\[-10mm]
\item $n_5(5,2;31t-21)=781t-531$ for $t\ge 1$;\\[-10mm]
\item $n_5(5,2;31t-22)=781t-557$ for $t\ge 2$ and $n_5(5,2;9)\in \{202,\dots,224\}$;\\[-10mm]
\item $n_5(5,2;31t-23)=781t-583$ for $t\ge 3$, $n_5(5,2;8)\in \{176,\dots,198\}$, and $n_5(5,2;39)\in \{969,\dots,979\}$;\\[-10mm]
\item $n_5(5,2;31t-24)=781t-609$ for $t\ge 3$, $n_5(5,2;7)\in \{157,\dots,172\}$, and $n_5(5,2;38)\in \{941,\dots,953\}$;\\[-10mm]
\item $n_5(5,2;31t-25)=781t-635$ for $t\ge 2$ and $n_5(5,2;6)\in \{132,\dots,146\}$;\\[-10mm]
\item $n_5(5,2;31t-26)=781t-656$ for $t\ge 1$;\\[-10mm]
\item $n_5(5,2;31t-27)=781t-682$ for $t\ge 2$ and $n_5(5,2;4)\in \{77,\dots,93\}$;\\[-10mm]
\item $n_5(5,2;31t-28)=781t-708$ for $t\ge 3$, $n_5(5,2;3)\in \{50,\dots,62\}$, and $n_5(5,2;34)\in \{844,\dots,854\}$;\\[-10mm]
\item $n_5(5,2;31t-29)=781t-734$ for $t\ge 3$, $n_5(5,2;2)\in \{27,\dots,31\}$, and $n_5(5,2;33)\in \{818,\dots,828\}$;\\[-10mm]
\item $n_5(5,2;31t-30)=781t-760$ for $t\ge 3$ and $n_5(5,2;32)\in \{792,\dots,802\}$.
  \end{itemize}
\end{theorem}

\pagebreak


\pagebreak

\appendix

\section{Generalized type of a faithful projective system}
\label{sec_generalization_type}
Since Definition~\ref{def_partitionable} is too restricted to cover the
full generality of the Solomon--Stiffler construction we provide a
generalization and analyze its implications for corresponding theoretical
results in this section. In the other direction we state an application of this restricted Definition to few-weight codes at the end of this section, see Lemma~\ref{lemma_few_weight_application}.

\begin{definition}
  \label{def_partitionable_generalized}
  We say that a faithful projective $h-(n,r,s)_q$ system $\cS$ has
  \emph{generalized type $\sigma[r]-\sum_{i=1}^{r-1}\varepsilon_i [i]$}
  if there exist subspaces $T_1,\dots,T_{l}$ with \#$\left\{1\le j\le l\mid \dim\!\left(T_j\right)=i\right\}=\left|\varepsilon_i\right|$ for $1\le i\le r-1$
  and
  \begin{equation}
    \sum_{S\in\cS} \chi_S=\sigma\cdot\chi_V-\sum_{i=j}^{l}\operatorname{sgn}(\varepsilon_{\dim(T_j)})\cdot \chi_{T_i},
  \end{equation}
  where $\operatorname{sgn}$ denotes the sign function.
  We say that $\sigma[r]-\sum_{i=1}^{r-1}\varepsilon_i [i]$ is \emph{weakly $h$-partitionable over $\F_q$} if a faithful projective
  $h-(n,r,s)_q$ system with generalized type
  $\sigma[r]-\sum_{i=1}^{r-1}\varepsilon_i [i]$ exists
  for suitable subspaces $T_1,\dots,T_l$ and parameters $n,s$.
\end{definition}

While the notion of being $h$-partitionable over $\F_q$ does not depend on
the choice of the subspaces $S_1,\dots,S_{r-1}$, the notion of being
weakly $h$-partitionable over $\F_q$ can depend on a careful selection
of the subspaces $T_1,\dots,T_l$.
Evaluating $\#\cM$, $\min_H \cM(H)$, $\min_P\cM(P)$ in Lemma~\ref{lemma_compute_parameters_from_premultiset} (and an application
of Equation~(\ref{cross_difference})) yields:
\begin{lemma}
  \label{lemma_compute_parameters_from_partition_generalized}
    If $\cS$ is a faithful projective $h-(n,r,s,\mu)_q$ system with
  generalized type $\sigma[r]-\sum_{i=1}^{r-1}\varepsilon_i [i]$,
  then we have
  \begin{equation}
    \label{formula_n_g}
    n=\left(\sigma[r]_q-\sum_{i=1}^{r-1}\varepsilon_i [i]_q\right)/[h]_q
  \end{equation}
  and
  \begin{equation}
    \label{packing_cond_generalized}
    \sum_{i=1}^{r-1}\varepsilon_i[i]_q \equiv 0\pmod {[\gcd(r,h)]_q}.
  \end{equation}
  If additionally all $\varepsilon_i$ are non-negative, then
  \begin{equation}
    \label{formula_s_g}
    s\le \left(\sigma[r-h]_q-\sum_{i=h}^{r-1}\varepsilon_i[i-h]_q+\sum_{i=1}^{h-1}\varepsilon_i q^{i-h}[h-i]_q\right)/[h]_q
    \quad\text{and}\quad \mu\le\sigma.
  \end{equation}
\end{lemma}



From Lemma~\ref{lemma_field_reduction}, based on field reduction,
we conclude:
\begin{lemma}
  \label{lemma_field_construction_star_partition_generalized}
  If $\sigma[r]-\sum_{i=1}^{r-1}\varepsilon_i [i]$ is weakly $h$-partitionable over $\F_{q^l}$, so is $\sigma[rl]+\sum_{i=1}^{r-1}\varepsilon_i [il]$.
\end{lemma}

Lemma~\ref{lemma_sigma_constraint_premultiset} implies:
\begin{corollary}
  \label{cor_sigma_constraint_gen}
  If $x[r]-\sum_{i=1}^{r-1}\varepsilon_i [i]$ is weakly $h$-partitionable over $\F_q$ for $x\in\{\sigma,\sigma'\}$ then
  $$\left(\sigma+t\cdot\frac{[h]_q}{[\gcd(r,h)]_q}\right)\cdot[r]-\sum_{i=1}^{r-1}\varepsilon_i [i]$$ is weakly $h$-partitionable over $\F_q$ for all $t\ge 0$
  and we have $\sigma\equiv \sigma'\pmod {\frac{[h]_q}{[\gcd(r,h)]_q}}$.
\end{corollary}

For those situations where we are not  interested in the smallest
possible value $\sigma$ such that $\sigma[r]-\cM$ is weakly $h$-partitionable
over $\F_q$ we specialize Definition~\ref{def_star}:
\begin{definition}
  We say that $\star[r]-\sum_{i=1}^{r-1}\varepsilon_i [i]$ is weakly
  $h$-partitionable over $\F_q$ if there exists a $\sigma\in\N$ such that
  $\sigma[r]-\sum_{i=1}^{r-1}\varepsilon_i [i]$ is weakly
  $h$-partitionable over $\F_q$.
\end{definition}
\begin{lemma}
  If $\star[r]-\cM$ is $h$-partitionable over $\F_q$, so is $\star[r]+\cM$.
\end{lemma}
\begin{corollary}
  \label{cor_negation_generalized}
  If $\star[r]-\sum_{i=1}^{r-1}\varepsilon_i [i]$ is weakly
  $h$-partitionable over $\F_q$, so is
  $\star[r]+\sum_{i=1}^{r-1}\varepsilon_i [i]$.
\end{corollary}

It is an interesting, but possibly very hard, problem to determine for
which parameters $\varepsilon_1,\dots,\varepsilon_{r-1}$ we have that
$\star[r]-\sum_{i=1}^{r-1}\varepsilon_i [i]$ is weakly $h$-partitionable
over $\F_q$. Clearly we need $r\ge h$ and
$\sum_{i=1}^{r-1}\left|\varepsilon_i\right|=0$ if $r=h$. Additionally we
have the packing condition (\ref{packing_cond_generalized}), see Lemma~\ref{lemma_compute_parameters_from_partition_generalized}.
For $h=1$ this condition is trivially satisfied and we indeed have that
$\star[r]-\sum_{i=1}^{r-1}\varepsilon_i [i]$ is weakly $1$-partitionable
over $\F_q$ for all parameters.

\begin{lemma}
  \label{lemma_divisible_condition_weak}
  Let $\star[r]-\sum_{i=1}^{r-1}\varepsilon_i [i]$ be weakly
  $h$-partitionable over $\F_q$.
  \begin{enumerate}
    \item[(i)]  If $\varepsilon_i\ge 0$ for all $1\le i\le h-1$, then
                there exists a $q^{h-1}$-divisible multiset of points
                $\cM_1$ in $\PG(r-1,q)$ with cardinality
                $\#\cM_1=\sum_{i=1}^{h-1} \varepsilon_i[i]_q$.
    \item[(ii)] If $\varepsilon_i\le 0$ for all $1\le i\le h-1$, then
                there exists a $q^{h-1}$-divisible multiset of points
                $\cM_2$ in $\PG(r-1,q)$ with cardinality
                $\#\cM_2=-\sum_{i=1}^{h-1} \varepsilon_i[i]_q$.
  \end{enumerate}
\end{lemma}
\begin{proof}
  Let $\cS$ be a faithful projective $h-(n,r,s)_q$ system of generalized type
  $\sigma[r]-\sum_{i=1}^{r-1}\varepsilon_i [i]$ for suitable parameters
  and $T_1,\dots,T_l$ denote the subspaces as in
  Definition~\ref{def_partitionable}. The multiset of points covered by the
  elements of $\cS$ is given by
  $$
    \cM=\sigma\chi_V-\sum_{i=1}^l \operatorname{sgn}(\varepsilon_{\dim(T_i)})\cdot\chi_T,
  $$
  where $V$ denotes the $r$-dimensional ambient space $\PG(r-1,q)$, so
  that $\cM$ is $q^{h-1}$-divisible by
  Lemma~\ref{lemma_divisible_properties}. In case (i) Lemma~\ref{lemma_divisible_properties} implies that also
  $\cM_1:=\sum_{1\le i\le l\,:\, \dim(T_i)\le h-1} \chi_{T_i}$ is
  $q^{h-1}$-divisible. For case (ii) we consider the $\mu$-complement
  $\cM^{C_\mu}$ of $\cM$, defined by $\cM^{C_\mu}(P)=\mu-\cM(P)$ for every
  point $P$, for a suitably large $\mu\in \N$. Applying Lemma~\ref{lemma_divisible_properties} yields that $\cM^{C_\mu}$ as well
  as $\cM_2:=\sum_{1\le i\le l\,:\, \dim(T_i)\le h-1} \chi_{T_i}$ is
  $q^{h-1}$-divisible.
\end{proof}
\begin{remark}
  The possible lengths of $q^{h-1}$-divisible codes over $\F_q$ have
  been completely characterized in \cite[Theorem 1]{kiermaier2020lengths}.
  The condition in Lemma~\ref{lemma_divisible_condition_weak} is only
  necessary and far from being sufficient. The distribution of the
  $\varepsilon_i$ with $i<h$ is not used at all. E.g.\ for $h=3$ and $q=2$
  there exists a $4$-divisible multiset of points $\cM$ with
  cardinality $7$. However, $\cM$ has to equal the characteristic function
  of a $3$-space, see e.g.\ \cite[Corollary 4]{korner2024lengths}, so that
  $\left(\varepsilon_2,\varepsilon_1\right)=(2,1)$ is impossible. Partitions
  of $\Delta$-divisible multisets of points into subspaces are e.g.\ briefly discussed in \cite[Section 10.1]{kurz2021divisible}, \cite{kurz2022vector}
  and have e.g.\ applications for the so-called supertail of a vector space partition \cite{heden2013supertail,nuastase2018complete}.

  If the $\varepsilon_i$ with $i<h$ have different signs the cardinality
  $\sum_{i=1}^{h-1} \varepsilon_i[i]_q$ is not sufficient to conclude
  claims on the nonexistence. E.g., we will see later on that
  $\star[r]-q\cdot [2]$ is $3$-partitionable over $\F_q$ if $r>3$ and
  $r\not\equiv 0\pmod 3$, so that $\star[r]-q\cdot [2]+\left(q^2+q+1\right)\cdot[1]$ is weakly $3$-partitionable over $\F_q$ by replacing one $3$-space by its $q^2+q+1$ points, while there clearly is no $q^2$-divisible multiset of points of cardinality $-q\cdot[2]_q+\left(q^2+q+1\right)\cdot[1]_q=1$.
\end{remark}

\bigskip

At the end of this section we want to give of an application of the technical restriction in Definition~\ref{def_partitionable} to additive few-weight codes that can be obtained from the Solomon--Stiffler construction.
\begin{lemma}
   \label{lemma_few_weight_application}
    Let $\cS$ be a faithful projective $h-(n,r,s)_q$ system with type
    $\sigma[r]-\sum_{i=1}^{r-1}\varepsilon_i[i]$, where $\varepsilon_i\in\N$ for all $1\le i\le r-1$. Then, the number of elements from $\cS$ that are not contained
    in some hyperplane $H$ is given by
    \begin{equation}
      \label{eq_weight_ss_construcion_type}
      \sigma q^{r-h} -\sum_{i=j}^{r-1} \varepsilon_i q^{i-h}
    \end{equation}
    for some integer $1\le j\le r$ and all values are indeed attained.
\end{lemma}
\begin{proof}
  Let $\cM$ be the multiset of points covered by the elements of $\cS$ and $S_1\le \dots\le S_{r-1}$ be subspaces as in Definition~\ref{def_partitionable}. First we observe that $\cM$ has cardinality
  $$
    \sigma[r]_q-\sum_{i=1}^{r-1} \varepsilon_i[i]_q,
  $$
  so that
  \begin{equation}
    n=\#\cS=\left( \sigma[r]_q-\sum_{i=1}^{r-1} \varepsilon_i[i]_q, \right)/[h]_q.
  \end{equation}
  For an arbitrary hyperplane $H$ let $1\le j\le r$ denote the minimal integer such that $S_j\not\le H$, where we set $j=r$ is $H=S_{r-1}$. Counting points gives
  $$
    \cM(H)=\sigma[r-1]_q-\sum_{i=1}^{j-1} \varepsilon_i[i]_q-\sum_{i=j}^{r-1}\varepsilon_i[i-1]_q
    =\sigma[r-1]_q-\sum_{i=1}^{r-1} [i-1]_q-\sum_{i=1}^{j-1}\varepsilon_i q^{i-1}.
  $$
  The number $s_j$ of elements of $\cS$ contained in $H$ is given by $\left(\cM(H)-n\cdot [h-1]_q\right)/q^{h-1}$, so that
  \begin{equation}
     s_j =\left(\sigma[r-h]_q-\sum_{i=h}^{r-1} \varepsilon_i[i-h]_q
     +\sum_{i=1}^{h-1}\varepsilon_i q^{i-h}[h-i]_q\right)/[h]_q
     -\sum_{i=1}^{j-1} \varepsilon_i q^{i-h},
  \end{equation}
  see the proof of Lemma~\ref{lemma_compute_parameters_from_partition} for details. With this we compute
  $
    n-s_j=\sigma q^{r-h} -\sum_{i=j}^{r-1} \varepsilon_i q^{i-h}
  $.
\end{proof}

\begin{corollary}
  \label{cor_few_weight_application}
  Let $\cS$ be a faithful projective $h-(n,r,s)_q$ system with type
  $\sigma[r]-\sum_{i=1}^{r-1}\varepsilon_i[i]$, where $\varepsilon_i\in\N$ for all $1\le i\le r-1$. Then, the corresponding additive code $C$ is a
  $t$-weight code, where $t=1+\#\left \{ 1\le i\le r-1\,:\, \varepsilon_i>0\right\}$.
\end{corollary}
Note that Expression~(\ref{eq_weight_ss_construcion_type}) is attained for exactly $q^{r-j}$ hyperplanes, so that one can also easily compute the weight distribution of the corresponding additive code $C$ from the parameters $\sigma$, $\varepsilon_1,\dots,\varepsilon_{r-1}$. Cf.~\cite{pan2025optimal} where the full subcode support distribution, see Section~\ref{sec_generalized_weights}, was computed for the special case $h=1$, i.e.\ linear codes, and $\left|\varepsilon_i\right|\in\{0,1\}$ for all $1\le i\le r-1$. In Section~\ref{sec_two_weight} we state a few preliminary observations on additive two-weight codes.

\pagebreak

\section{Parameterized series of additive codes that outperform linear codes}
\label{sec_parameterized_outperform}

\begin{table}[htp]
  \begin{center}
    \begin{tabular}{rrrrrrl}
      \hline
      $q$ & $r$ & $h$ & $i$ & $s_{i,t}:=t\cdot \tfrac{[r-h]_q}{[\gcd(r,h)]_q}-i$
      & $n_q\left(r,h;s_{i,t}\right)$ & $n_q\left(r,h;s_{i,t}\right)-\overline{n}_q\left(r,h;s_{i,t}\right)$\\
      \hline
      2 & 10 & 2 & 13 & $85t-13$ & $341t-55$ & 2 \\
      2 & 10 & 2 & 14 & $85t-14$ & $341t-60$ & 2 \\
      2 & 10 & 2 & 18 & $85t-18$ & $341t-76$ & 2 \\
      2 & 10 & 2 & 19 & $85t-19$ & $341t-81$ & 2 \\
      2 & 10 & 2 & 34 & $85t-34$ & $341t-140$ & 2 \\
      2 & 10 & 2 & 35 & $85t-35$ & $341t-145$ & 2 \\
      2 & 10 & 2 & 39 & $85t-39$ & $341t-161$ & 2 \\
      2 & 10 & 2 & 40 & $85t-40$ & $341t-166$ & 2 \\
      2 & 10 & 2 & 45 & $85t-45$ & $341t-183$ & 2 \\
      2 & 10 & 2 & 46 & $85t-46$ & $341t-188$ & 2 \\
      2 & 10 & 2 & 50 & $85t-50$ & $341t-204$ & 2 \\
      2 & 10 & 2 & 51 & $85t-51$ & $341t-209$ & 2 \\
      2 & 10 & 2 & 53 & $85t-53$ & $341t-215$ & 2 \\
      2 & 10 & 2 & 54 & $85t-54$ & $341t-220$ & 2 \\
      2 & 10 & 2 & 55 & $85t-55$ & $341t-223$ & 4 \\
      2 & 10 & 2 & 56 & $85t-56$ & $341t-228$ & 4 \\
      2 & 10 & 2 & 58 & $85t-58$ & $341t-236$ & 2 \\
      2 & 10 & 2 & 59 & $85t-59$ & $341t-241$ & 2 \\
      2 & 10 & 2 & 60 & $85t-60$ & $341t-244$ & 4 \\
      2 & 10 & 2 & 61 & $85t-61$ & $341t-249$ & 4 \\
      2 & 10 & 2 & 66 & $85t-66$ & $341t-268$ & 2 \\
      2 & 10 & 2 & 67 & $85t-67$ & $341t-273$ & 2 \\
      2 & 10 & 2 & 71 & $85t-71$ & $341t-289$ & 2 \\
      2 & 10 & 2 & 72 & $85t-72$ & $341t-294$ & 2 \\
      2 & 10 & 2 & 74 & $85t-74$ & $341t-300$ & 2 \\
      2 & 10 & 2 & 75 & $85t-75$ & $341t-305$ & 2 \\
      2 & 10 & 2 & 76 & $85t-76$ & $341t-308$ & 4 \\
      2 & 10 & 2 & 77 & $85t-77$ & $341t-313$ & 4 \\
      2 & 10 & 2 & 79 & $85t-79$ & $341t-321$ & 2 \\
      2 & 10 & 2 & 80 & $85t-80$ & $341t-326$ & 2 \\
      2 & 10 & 2 & 81 & $85t-81$ & $341t-329$ & 4 \\
      2 & 10 & 2 & 82 & $85t-82$ & $341t-334$ & 4 \\
      \hline
    \end{tabular}
    \caption{Parameterized series of improvements for additive codes with $q=2$, $r=10$, and $h=2$.}
    \label{table_improvements_q_2_r_10_h_2}
  \end{center}
\end{table}

In this appendix we want to extend Table~\ref{table_improvements} on
parameterized series of improvements $n_q(r,h;s)>\overline{n}_q(r,h;s)$
for additive codes in the integral case $r/h\in \N$. By
Theorem~\ref{thm_attained_asymptotically} it suffices to compare the
Griesmer upper bounds for $n_q\!\left(r,h;\tfrac{[r-h]_q}{[h]_q}\cdot t-i\right)$
and $\overline{n}_q\!\left(r,h;\tfrac{[r-h]_q}{[h]_q}\cdot t-i\right)=
n_{q^h}\!\left(\tfrac{r}{h},1;[r/h-1]_{q^h}\cdot t-i\right)$ in terms of $t\in\N$
for all $0\le i< \tfrac{[r-h]_q}{[h]_q}$. This can be easily done by a small
computer program. Here we just list the cases of all such $i$ for small
parameters $q$, $r$, and $h$. Since $n_q(h,h;s)=s$ and
$n_q(2h,h;s)=\overline{n}_q(2h,h;s)$ for all $s\in \N$, see
Theorem~\ref{thm_r_h_integral_2}, improvements can only occur if $r/h\ge 3$. The
case $(r,h)=(6,2)$ is completely settled in Theorem~\ref{thm_r_3h_h_2}. There
are exactly $q(q-2)$ parametric improvements for each field size $q$. So, especially none for $q=2$. As shown in \cite{bierbrauer2021optimal}, cf.\
Subsection~\ref{subsec_q_2}, we even have $n_2(6,2;s)=\overline{n}_2(6,2;s)$.
The parametric improvements for $n_q(6,2;s)$ are already stated in
Table~\ref{table_improvements} when $q\in\{3,4\}$. For $q=5$ we refer
to Table~\ref{table_improvements_q_5_r_6_h_2}. Since there would be at least $35$
parametric improvements for $q\ge 7$ and we have an analytic solution
in Theorem~\ref{thm_r_3h_h_2}, we abstain from listing further tables for $n_q(6,2;s)$.

\begin{table}[htp]
  \begin{center}
    \begin{tabular}{rrrrrrl}
      \hline
      $q$ & $r$ & $h$ & $i$ & $s_{i,t}:=t\cdot \tfrac{[r-h]_q}{[\gcd(r,h)]_q}-i$
      & $n_q\left(r,h;s_{i,t}\right)$ & $n_q\left(r,h;s_{i,t}\right)-\overline{n}_q\left(r,h;s_{i,t}\right)$\\
      \hline
      5 & 6 & 2 & 11 & $26t-11$ & $651t-281$ & 5 \\
      5 & 6 & 2 & 12 & $26t-12$ & $651t-307$ & 5 \\
      5 & 6 & 2 & 13 & $26t-13$ & $651t-333$ & 5 \\
      5 & 6 & 2 & 14 & $26t-14$ & $651t-359$ & 5 \\
      5 & 6 & 2 & 15 & $26t-15$ & $651t-385$ & 5 \\
      5 & 6 & 2 & 16 & $26t-16$ & $651t-406$ & 10 \\
      5 & 6 & 2 & 17 & $26t-17$ & $651t-432$ & 10 \\
      5 & 6 & 2 & 18 & $26t-18$ & $651t-458$ & 10 \\
      5 & 6 & 2 & 19 & $26t-19$ & $651t-484$ & 10 \\
      5 & 6 & 2 & 20 & $26t-20$ & $651t-510$ & 10 \\
      5 & 6 & 2 & 21 & $26t-21$ & $651t-531$ & 15 \\
      5 & 6 & 2 & 22 & $26t-22$ & $651t-557$ & 15 \\
      5 & 6 & 2 & 23 & $26t-23$ & $651t-583$ & 15 \\
      5 & 6 & 2 & 24 & $26t-24$ & $651t-609$ & 15 \\
      5 & 6 & 2 & 25 & $26t-25$ & $651t-635$ & 15 \\
      \hline
    \end{tabular}
    \caption{Parameterized series of improvements for additive codes with $q=5$, $r=6$, and $h=2$.}
    \label{table_improvements_q_5_r_6_h_2}
  \end{center}
\end{table}

Parametric improvements for $n_2(8,2;s)$ are contained in Table~\ref{table_improvements} and for $n_3(8,2;s)$ we refer to Table~\ref{table_improvements_q_3_r_8_h_2_part1} and Table~\ref{table_improvements_q_3_r_8_h_2_part2}. The latter are so numerous that
we do not give further tables for larger field sizes. Parametric improvements for $n_2(10,2;s)$ are listed in Table~\ref{table_improvements_q_2_r_10_h_2}.

\begin{table}[htp]
  \begin{center}
    \begin{tabular}{rrrrrrl}
      \hline
      $q$ & $r$ & $h$ & $i$ & $s_{i,t}:=t\cdot \tfrac{[r-h]_q}{[\gcd(r,h)]_q}-i$
      & $n_q\left(r,h;s_{i,t}\right)$ & $n_q\left(r,h;s_{i,t}\right)-\overline{n}_q\left(r,h;s_{i,t}\right)$\\
      \hline
      2 & 12 & 4 & 7 & $17t-7$ & $273t-117$ & 2 \\
      2 & 12 & 4 & 8 & $17t-8$ & $273t-134$ & 2 \\
      2 & 12 & 4 & 9 & $17t-9$ & $273t-147$ & 6 \\
      2 & 12 & 4 & 10 & $17t-10$ & $273t-164$ & 6 \\
      2 & 12 & 4 & 11 & $17t-11$ & $273t-179$ & 8 \\
      2 & 12 & 4 & 12 & $17t-12$ & $273t-196$ & 8 \\
      2 & 12 & 4 & 13 & $17t-13$ & $273t-213$ & 8 \\
      2 & 12 & 4 & 14 & $17t-14$ & $273t-230$ & 8  \\
      2 & 12 & 4 & 15 & $17t-15$ & $273t-245$ & 10 \\
      2 & 12 & 4 & 16 & $17t-16$ & $273t-262$ & 10 \\
      \hline
     \end{tabular}
    \caption{Parameterized series of improvements for additive codes with $q=2$, $r=12$, and $h=4$.}
    \label{table_improvements_q_2_r_12_h_4}
  \end{center}
\end{table}

\begin{table}[htp]
  \begin{center}
    \begin{tabular}{rrrrrrl}
      \hline
      $q$ & $r$ & $h$ & $i$ & $s_{i,t}:=t\cdot \tfrac{[r-h]_q}{[\gcd(r,h)]_q}-i$
      & $n_q\left(r,h;s_{i,t}\right)$ & $n_q\left(r,h;s_{i,t}\right)-\overline{n}_q\left(r,h;s_{i,t}\right)$\\
      \hline
      3 & 9 & 3 & 10 & $28t-10$ & $757t-274$ & 6 \\
      3 & 9 & 3 & 11 & $28t-11$ & $757t-302$ & 6 \\
      3 & 9 & 3 & 12 & $28t-12$ & $757t-330$ & 6 \\
      3 & 9 & 3 & 13 & $28t-13$ & $757t-355$ & 9 \\
      3 & 9 & 3 & 14 & $28t-14$ & $757t-383$ & 9 \\
      3 & 9 & 3 & 15 & $28t-15$ & $757t-411$ & 9 \\
      3 & 9 & 3 & 16 & $28t-16$ & $757t-439$ & 9 \\
      3 & 9 & 3 & 17 & $28t-17$ & $757t-467$ & 9 \\
      3 & 9 & 3 & 18 & $28t-18$ & $757t-495$ & 9 \\
      3 & 9 & 3 & 19 & $28t-19$ & $757t-517$ & 15 \\
      3 & 9 & 3 & 20 & $28t-20$ & $757t-545$ & 15 \\
      3 & 9 & 3 & 21 & $28t-21$ & $757t-573$ & 15 \\
      3 & 9 & 3 & 22 & $28t-22$ & $757t-598$ & 18 \\
      3 & 9 & 3 & 23 & $28t-23$ & $757t-626$ & 18 \\
      3 & 9 & 3 & 24 & $28t-24$ & $757t-654$ & 18 \\
      3 & 9 & 3 & 25 & $28t-25$ & $757t-682$ & 18 \\
      3 & 9 & 3 & 26 & $28t-26$ & $757t-710$ & 18 \\
      3 & 9 & 3 & 27 & $28t-27$ & $757t-738$ & 18 \\
      \hline
    \end{tabular}
    \caption{Parameterized series of improvements for additive codes with $q=3$, $r=9$, and $h=3$.}
    \label{table_improvements_q_3_r_9_h_3}
  \end{center}
\end{table}

For $h=3$ Table~\ref{table_improvements} contains the
parametric improvements for $n_2(9,3;s)$ and for $n_3(9,3;s)$ we list them
in Table~\ref{table_improvements_q_3_r_9_h_3}. For $n_4(9,3;65t-i)$ there are
parametric improvements for all $13\le i\le 64$ and for $n_2(12,3;73t-i)$ there
are parametric improvements for all $i\in \{5,\dots,8\}\cup\{14,\dots17\}\cup\{23,\dots,26\}\cup\{32,\dots,35\}\cup\{37,\dots,44\}\cup\{46,\dots,53\}\cup\{55,\dots,62\},\cup\{64,\dots,71\}$. The only other case being reasonably small to be
fully included are the parametric improvements for $n_2(12,4;s)$, see
Table~\ref{table_improvements_q_2_r_12_h_4}.

\begin{table}[htp]
  \begin{center}
    \begin{tabular}{rrrrrrl}
      \hline
      $q$ & $r$ & $h$ & $i$ & $s_{i,t}:=t\cdot \tfrac{[r-h]_q}{[\gcd(r,h)]_q}-i$
      & $n_q\left(r,h;s_{i,t}\right)$ & $n_q\left(r,h;s_{i,t}\right)-\overline{n}_q\left(r,h;s_{i,t}\right)$\\
      \hline
      3 & 8 & 2 & 7 & $91t-7$ & $820t-67$ & 3 \\
      3 & 8 & 2 & 8 & $91t-8$ & $820t-77$ & 3 \\
      3 & 8 & 2 & 9 & $91t-9$ & $820t-87$ & 3 \\
      3 & 8 & 2 & 17 & $91t-17$ & $820t-158$ & 3 \\
      3 & 8 & 2 & 18 & $91t-18$ & $820t-168$ & 3 \\
      3 & 8 & 2 & 19 & $91t-19$ & $820t-178$ & 3 \\
      3 & 8 & 2 & 27 & $91t-27$ & $820t-249$ & 3 \\
      3 & 8 & 2 & 28 & $91t-28$ & $820t-259$ & 3 \\
      3 & 8 & 2 & 29 & $91t-29$ & $820t-269$ & 3 \\
      3 & 8 & 2 & 34 & $91t-34$ & $820t-310$ & 3 \\
      3 & 8 & 2 & 35 & $91t-35$ & $820t-320$ & 3 \\
      3 & 8 & 2 & 36 & $91t-36$ & $820t-330$ & 3 \\
      3 & 8 & 2 & 37 & $91t-37$ & $820t-337$ & 6 \\
      3 & 8 & 2 & 38 & $91t-38$ & $820t-347$ & 6 \\
      3 & 8 & 2 & 39 & $91t-39$ & $820t-357$ & 6 \\
      3 & 8 & 2 & 44 & $91t-44$ & $820t-401$ & 3 \\
      3 & 8 & 2 & 45 & $91t-45$ & $820t-411$ & 3 \\
      3 & 8 & 2 & 46 & $91t-46$ & $820t-421$ & 3 \\
      3 & 8 & 2 & 47 & $91t-47$ & $820t-428$ & 6 \\
      3 & 8 & 2 & 48 & $91t-48$ & $820t-438$ & 6 \\
      3 & 8 & 2 & 49 & $91t-49$ & $820t-448$ & 6 \\
      3 & 8 & 2 & 54 & $91t-54$ & $820t-492$ & 3 \\
      3 & 8 & 2 & 55 & $91t-55$ & $820t-502$ & 3 \\
      3 & 8 & 2 & 56 & $91t-56$ & $820t-512$ & 3 \\
      3 & 8 & 2 & 57 & $91t-57$ & $820t-519$ & 6 \\
      3 & 8 & 2 & 58 & $91t-58$ & $820t-529$ & 6\\
       \hline
    \end{tabular}
    \caption{Parameterized series of improvements for additive codes with $q=3$, $r=8$, and $h=2$ -- part 1.}
    \label{table_improvements_q_3_r_8_h_2_part1}
  \end{center}
\end{table}

The collected data suggests:
\begin{conjecture}
  For each prime power $q$ and $r,h\in \N$ with $r\ge 3h$, $h\ge 2$, $r\equiv 0\pmod h$ and $(q,r,h)\neq (2,6,2)$ there exist infinitely many $s\in\N$
  with $n_q(r,h;s)>\overline{n}_q(r,h;s)$.
\end{conjecture}
Looking at the data one might also conjecture that the improvement
$n_q(r,h;s)-\overline{n}_q(r,h;s)$ is always divisible by the characteristic $p$ of $\F_q$ when $r/h\in\N$ and $s$ is sufficiently large. It is also conspicuous that the improvements seem to come in blocks of consecutive values of $s$ with
equal improvement $n_q(r,h;s)-\overline{n}_q(r,h;s)$. So far it seems that the length of those blocks is always divisible by the field size $q$, but this might also be an artifact due to too few observations.

\begin{table}[htp]
  \begin{center}
    \begin{tabular}{rrrrrrl}
      \hline
      $q$ & $r$ & $h$ & $i$ & $s_{i,t}:=t\cdot \tfrac{[r-h]_q}{[\gcd(r,h)]_q}-i$
      & $n_q\left(r,h;s_{i,t}\right)$ & $n_q\left(r,h;s_{i,t}\right)-\overline{n}_q\left(r,h;s_{i,t}\right)$\\
      \hline
      3 & 8 & 2 & 59 & $91t-59$ & $820t-539$ & 6 \\
      3 & 8 & 2 & 61 & $91t-61$ & $820t-553$ & 3 \\
      3 & 8 & 2 & 62 & $91t-62$ & $820t-563$ & 3 \\
      3 & 8 & 2 & 63 & $91t-63$ & $820t-573$ & 3 \\
      3 & 8 & 2 & 64 & $91t-64$ & $820t-580$ & 6 \\
      3 & 8 & 2 & 65 & $91t-65$ & $820t-590$ & 6 \\
      3 & 8 & 2 & 66 & $91t-66$ & $820t-600$ & 6 \\
      3 & 8 & 2 & 67 & $91t-67$ & $820t-607$ & 9 \\
      3 & 8 & 2 & 68 & $91t-68$ & $820t-617$ & 9 \\
      3 & 8 & 2 & 69 & $91t-69$ & $820t-627$ & 9 \\
      3 & 8 & 2 & 71 & $91t-71$ & $820t-644$ & 3 \\
      3 & 8 & 2 & 72 & $91t-72$ & $820t-654$ & 3 \\
      3 & 8 & 2 & 73 & $91t-73$ & $820t-664$ & 3 \\
      3 & 8 & 2 & 74 & $91t-74$ & $820t-671$ & 6 \\
      3 & 8 & 2 & 75 & $91t-75$ & $820t-681$ & 6 \\
      3 & 8 & 2 & 76 & $91t-76$ & $820t-691$ & 6 \\
      3 & 8 & 2 & 77 & $91t-77$ & $820t-698$ & 9 \\
      3 & 8 & 2 & 78 & $91t-78$ & $820t-708$ & 9 \\
      3 & 8 & 2 & 79 & $91t-79$ & $820t-718$ & 9 \\
      3 & 8 & 2 & 81 & $91t-81$ & $820t-735$ & 3 \\
      3 & 8 & 2 & 82 & $91t-82$ & $820t-745$ & 3 \\
      3 & 8 & 2 & 83 & $91t-83$ & $820t-755$ & 3 \\
      3 & 8 & 2 & 84 & $91t-84$ & $820t-762$ & 6 \\
      3 & 8 & 2 & 85 & $91t-85$ & $820t-772$ & 6 \\
      3 & 8 & 2 & 86 & $91t-86$ & $820t-782$ & 6 \\
      3 & 8 & 2 & 87 & $91t-87$ & $820t-789$ & 9 \\
      3 & 8 & 2 & 88 & $91t-88$ & $820t-799$ & 9 \\
      3 & 8 & 2 & 89 & $91t-89$ & $820t-809$ & 9 \\
      \hline
    \end{tabular}
    \caption{Parameterized series of improvements for additive codes with $q=3$, $r=8$, and $h=2$ -- part 2.}
    \label{table_improvements_q_3_r_8_h_2_part2}
  \end{center}
\end{table}

While it is interesting to know that the Griesmer upper bound can always be reached, see Theorem~\ref{thm_attained_asymptotically}, and that infinitely many parametric improvements
$n_q\!\left(r,h;\tfrac{[r-h]_q}{[h]_q}\cdot t-i\right)>
\overline{n}_q\!\left(r,h;\tfrac{[r-h]_q}{[h]_q}\cdot t-i\right)$ do indeed
exist, explicit lower bounds on $t$ would be desirable. In Section~\ref{generic_results} we compute such lower bounds for $t$ in a generic manner by using a few basic general constructions. The import and very hard problem of determining $n_q(r,h;s)$ for relatively small values of $s$ remains widely open. We provide a few general results in Section~\ref{sec_small_parameters} and study the cases of field sizes $q=2$ and $q=3$ in Subsection~\ref{subsec_q_2} and Subsection~\ref{subsec_q_3}, respectively.

\pagebreak

\section{Generic results}
\label{generic_results}

The aim of this section is to determine $n_q(r,h;s)$ for small parameters $q$, $r$, and $s$ explicitly, assuming that $s$ is sufficiently large. As upper bound we will always utilize the Griesmer upper bound, which can be reached due to Theorem~\ref{thm_attained_asymptotically}. For the corresponding lower bounds we will use a few basic general constructions, see Theorem~\ref{thm_partition},
Lemma~\ref{lemma_partition_1}, Lemma~\ref{lemma_construction_x_consequence}, and Lemma~\ref{lemma_vsp_type}. With these we will show that $\sigma[r]-\sum_{i=1}^{r-1}\varepsilon_i[i]$ is $h$-partitionable over $\F_q$ for suitable
parameters $\sigma$ and $\varepsilon_1,\dots,\varepsilon_{r-1}$. The stated general results for
$n_q\!\left(r,h;\tfrac{[r-h_]q}{[\gcd(rh,)]_q}\cdot t-i\right)$ will then follow
from Corollary~\ref{cor_asymptotic_one_weight} and Corollary~\ref{cor_union} (which we will not mention explicitly in the subsequent proofs). All those
reasonings may be automated so that we speak of generic results. We will restrict ourselves to small parameters, not treating those that are covered by the general results from Section~\ref{sec_small_parameters}. Especially we restrict to field sizes $q\in\{2,3\}$. Cf.\ the results for $n_q(r,h;s)$ for $q\in\{2,3\}$ and small values of $s$ in Subsection~\ref{subsec_q_2} and Subsection~\ref{subsec_q_3}, respectively.

\begin{proposition}
  \label{prop_q_2_r_5_h_2_generic}
  The Griesmer upper bound for $n_2(5,2;s)$ is attained
  for all $s\ge 2$, i.e.\ we have\\[-10mm]
\begin{itemize}
\item $n_2(5,2;7t)=31t$ for $t\ge 1$ via $3t\cdot[5]$;\\[-10mm]
\item $n_2(5,2;7t-1)=31t-5$ for $t\ge 1$ via $3t\cdot[5]-[4]$;\\[-10mm]
\item $n_2(5,2;7t-2)=31t-10$ for $t\ge 1$ via $(3t-1)\cdot[5]+[2]-2[1]$;\\[-10mm]
\item $n_2(5,2;7t-3)=31t-15$ for $t\ge 1$ via $(3t-1)\cdot[5]-2[3]$;\\[-10mm]
\item $n_2(5,2;7t-4)=31t-20$ for $t\ge 1$ via $(3t-2)\cdot[5]+2[1]$;\\[-10mm]
\item $n_2(5,2;7t-5)=31t-23$ for $t\ge 1$ via $(3t-2)\cdot[5]-[3]$;\\[-10mm]
\item $n_2(5,2;7t-6)=31t-28$ for $t\ge 2$ via $(3t-2)\cdot[5]-[4]-[3]$.\\[-10mm]
\end{itemize}
\end{proposition}
\begin{proof}
  By Theorem~\ref{thm_partition} we have that $3[5]$, $[4]$, and $[2]$ are $2$-partitionable over $\F_2$. Lemma~\ref{lemma_partition_1} shows that $[5]-[3]$
  and $3[5]-[4]-2[2]$ are $2$-partitionable over $\F_2$. By
  Lemma~\ref{lemma_construction_x_consequence} we have that $2[5]+[2]-2[1]$
  and $[5]+2[1]$ are $2$-partitionable over $\F_2$. So, also $3[5]-[4]$ is $2$-partitionable
  over $\F_2$.
\end{proof}

\begin{proposition}
  \label{prop_q_2_r_6_h_2_generic}
  The Griesmer upper bound for $n_2(6,2;s)$ is attained
  for all $s\ge 8$, i.e.\ we have\\[-10mm]
\begin{itemize}
\item $n_2(6,2;5t)=21t$ for $t\ge 1$ via $t\cdot[6]$;\\[-10mm]
\item $n_2(6,2;5t-1)=21t-5$ for $t\ge 1$ via $t\cdot[6]-[4]$;\\[-10mm]
\item $n_2(6,2;5t-2)=21t-10$ for $t\ge 1$ via $t\cdot[6]-[4]+[3]+2[1]$;\\[-10mm]
\item $n_2(6,2;5t-3)=21t-15$ for $t\ge 3$ via $t\cdot[6]-3[4]$;\\[-10mm]
\item $n_2(6,2;5t-4)=21t-20$ for $t\ge 1$ via $(t-1)\cdot[6]+[2]$.\\[-10mm]
\end{itemize}
\end{proposition}
\begin{proof}
  By Theorem~\ref{thm_partition} we have that $[6]$, $[4]$, and $[2]$ are
  $2$-partitionable over $\F_2$. Lemma~\ref{lemma_partition_1} shows that
  $[6]-[4]$ is $2$-partitionable over $\F_2$. By
  Lemma~\ref{lemma_construction_x_consequence} we have that $[3]+2[1]$
  is $2$-partitionable over $\F_2$.
\end{proof}


\begin{proposition}
  \label{prop_q_2_r_7_h_2_generic}
  The Griesmer upper bound for $n_2(7,2;s)$ is attained
  for all $s\ge 24$, i.e.\ we have\\[-10mm]
\begin{itemize}
\item $n_2(7,2;31t)=127t$ for $t\ge 1$ via $3t\cdot[7]$;\\[-10mm]
\item $n_2(7,2;31t-1)=127t-5$ for $t\ge 1$ via $3t\cdot[7]-[4]$;\\[-10mm]
\item $n_2(7,2;31t-2)=127t-10$ for $t\ge 1$ via $3t\cdot[7]-2[4]$;\\[-10mm]
\item $n_2(7,2;31t-3)=127t-15$ for $t\ge 1$ via $3t\cdot[7]-3[4]$;\\[-10mm]
\item $n_2(7,2;31t-4)=127t-20$ for $t\ge 1$ via $3t\cdot[7]-2[5]+2[1]$;\\[-10mm]
\item $n_2(7,2;31t-5)=127t-21$ for $t\ge 1$ via $3t\cdot [7]-[6]$;\\[-10mm]
\item $n_2(7,2;31t-6)=127t-26$ for $t\ge 1$ via $3t\cdot [7]-[6]-[4]$;\\[-10mm]
\item $n_2(7,2;31t-7)=127t-31$ for $t\ge 1$ via $3t\cdot [7]-[6]-2[4]$;\\[-10mm]
\item $n_2(7,2;31t-8)=127t-36$ for $t\ge 2$ via $3t\cdot [7]-[6]-[5]-[4]+[2]-2[1]$;\\[-10mm]
\item $n_2(7,2;31t-9)=127t-41$ for $t\ge 1$ via $(3t-1)\cdot [7]+4[1]$;\\[-10mm]
\item $n_2(7,2;31t-10)=127t-42$ for $t\ge 1$ via $(3t-1)\cdot[7]+[2]-2[1]$;\\[-10mm]
\item $n_2(7,2;31t-11)=127t-47$ for $t\ge 1$ via $(3t-1)\cdot[7]-2[3]$;\\[-10mm]
\item $n_2(7,2;31t-12)=127t-52$ for $t\ge 1$ via $(3t-1)\cdot [7]-[5]+2[1]$;\\[-10mm]
\item $n_2(7,2;31t-13)=127t-55$ for $t\ge 1$ via $(3t-1)\cdot[7]-[5]-[3]$;\\[-10mm]
\item $n_2(7,2;31t-14)=127t-60$ for $t\ge 2$ via $(3t-1)\cdot[7]-[5]-[4]-[3]$;\\[-10mm]
\item $n_2(7,2;31t-15)=127t-63$ for $t\ge 1$ via $(3t-1)\cdot[7]-2[5]$;\\[-10mm]
\item $n_2(7,2;31t-16)=127t-68$ for $t\ge 1$;\\[-10mm]
\item $n_2(7,2;31t-17)=127t-73$ for $t\ge 2$ via $(3t-2)\cdot[7]-[6]-[4]-2[3]$;\\[-10mm]
\item $n_2(7,2;31t-18)=127t-76$ for $t\ge 2$ via $(3t-1)\cdot [7]-[6]-[5]-[3]$;\\[-10mm]
\item $n_2(7,2;31t-19)=127t-81$ for $t\ge 2$ via $(3t-1)\cdot [7]-[6]-[5]-[4]-[3]$;\\[-10mm]
\item $n_2(7,2;31t-20)=127t-84$ for $t\ge 1$ via $(3t-2)\cdot [7]+2[1]$;\\[-10mm]
\item $n_2(7,2;31t-21)=127t-87$ for $t\ge 1$ via $(3t-2)\cdot [7]-[3]$;\\[-10mm]
\item $n_2(7,2;31t-22)=127t-92$ for $t\ge 2$ via $(3t-2)\cdot [7]-[4]-[3]$;\\[-10mm]
\item $n_2(7,2;31t-23)=127t-95$ for $t\ge 1$ via $(3t-2)\cdot[7]-[5]$;\\[-10mm]
\item $n_2(7,2;31t-24)=127t-100$ for $t\ge 2$ via $(3t-2)\cdot[7]-[5]-[4]$;\\[-10mm]
\item $n_2(7,2;31t-25)=127t-105$ for $t\ge 2$ via $(3t-2)\cdot[7]-[5]-2[4]$;\\[-10mm]
\item $n_2(7,2;31t-26)=127t-108$ for $t\ge 2$ via $(3t-2)\cdot[7]-[6]-[3]$;\\[-10mm]
\item $n_2(7,2;31t-27)=127t-113$ for $t\ge 2$ via $(3t-2)\cdot[7]-[6]-[4]-[3]$;\\[-10mm]
\item $n_2(7,2;31t-28)=127t-116$ for $t\ge 2$ via $(3t-2)\cdot[7]-[6]-[5]$;\\[-10mm]
\item $n_2(7,2;31t-29)=127t-121$ for $t\ge 2$ via $(3t-2)\cdot[7]-[6]-[5]-[4]$;\\[-10mm]
\item $n_2(7,2;31t-30)=127t-126$ for $t\ge 1$ via $3t\cdot[7]-[2]$.
\end{itemize}
\end{proposition}
\begin{proof}
  By Theorem~\ref{thm_partition} we have that $3[7]$, $3[7]-[2]$, and
  $[4]$ are $2$-partitionable over $\F_2$. Lemma~\ref{lemma_partition_1}
  shows that $[7]-[5]$
  and $3[7]-[6]-2[4]$ are $2$-partitionable over $\F_2$. By
  Lemma~\ref{lemma_vsp_type} $[7]-[3]$ and $[6]-[4]$ are $2$-partitionable
  over $\F_2$. By Lemma~\ref{lemma_construction_x_consequence} we have
  that $2[7]+[2]-2[1]$ and $[7]+2[1]$ are $2$-partitionable over $\F_2$.
  So, also $3[7]-[6]$ and $3[7]-3[4]$ are $2$-partitionable over $\F_2$.
\end{proof}



\begin{proposition}
  \label{prop_q_2_r_8_h_2_generic}
  The Griesmer upper bound for $n_2(8,2;s)$ is attained
  for all $s\ge 260$, i.e.\ we have\\[-10mm]
\begin{itemize}
\item $n_2(8,2;21t)=85t$ for $t\ge 1$ via $t\cdot[8]$;\\[-10mm]
\item $n_2(8,2;21t-1)=85t-5$ for $t\ge 1$ via $t\cdot[8]-[4]$;\\[-10mm]
\item $n_2(8,2;21t-2)=85t-10$ for $t\ge 2$ via $t\cdot[8]-2[4]$;\\[-10mm]
\item $n_2(8,2;21t-3)=85t-15$ for $t\ge 3$ via $t\cdot[8]-3[4]$;\\[-10mm]
\item $n_2(8,2;21t-4)=85t-20$ for $t\ge 1$ via $t\cdot[8]-[6]+[2]$;\\[-10mm]
\item $n_2(8,2;21t-5)=85t-21$ for $t\ge 1$ via $t\cdot[8]-[6]$;\\[-10mm]
\item $n_2(8,2;21t-6)=85t-26$ for $t\ge 2$ via $t\cdot[8]-[6]-[4]$;\\[-10mm]
\item $n_2(8,2;21t-7)=85t-31$ for $t\ge 3$ via $t\cdot[8]-[6]-2[4]$ ;\\[-10mm]
\item $n_2(8,2;21t-8)=85t-36$ for $t\ge 4$ via $t\cdot[8]-[6]-3[4]$;\\[-10mm]
\item $n_2(8,2;21t-9)=85t-41$ for $t\ge 3$ via $t\cdot[8]-[7]+4[1]$;\\[-10mm]
\item $n_2(8,2;21t-10)=85t-41$ for $t\ge 2$ via $t\cdot[8]-2[6]$;\\[-10mm]
\item $n_2(8,2;21t-11)=85t-47$ for $t\ge 3$ via $t\cdot[8]-2[6]-[4]$;\\[-10mm]
\item $n_2(8,2;21t-12)=85t-52$ for $t\ge 3$ via $t\cdot[8]-[7]-[5]+2[1]$;\\[-10mm]
\item $n_2(8,2;21t-13)=85t-55$ for $t\ge 3$ via $t\cdot[8]-[7]-[5]-[3]$;\\[-10mm]
\item $n_2(8,2;21t-14)=85t-60$ for $t\ge 4$ via $t\cdot[8]-[7]-[5]-[4]-[3]$;\\[-10mm]
\item $n_2(8,2;21t-15)=85t-63$ for $t\ge 3$ via $t\cdot[8]-3[6]$;\\[-10mm]
\item $n_2(8,2;21t-16)=85t-68$ for $t\ge 4$ via $t\cdot[8]-3[6]-[4]$;\\[-10mm]
\item $n_2(8,2;21t-17)=85t-73$ for $t\ge 4$ via $t\cdot[8]-1[7]-[6]-[5]+2[1]$;\\[-10mm]
\item $n_2(8,2;21t-18)=85t-76$ for $t\ge 4$ via $t\cdot[8]-[7]-[6]-[5]-[3]$;\\[-10mm]
\item $n_2(8,2;21t-19)=85t-81$ for $t\ge 5$ via $t\cdot[8]-[7]-[6]-[5]-[4]-[3]$;\\[-10mm]
\item $n_2(8,2;21t-20)=85t-84$ for $t\ge 1$ via $(t-1)\cdot[8]+[2]$.\\[-10mm]
\end{itemize}
\end{proposition}
\begin{proof}
  By Theorem~\ref{thm_partition} we have that $[8]$, $[6]$, and $[4]$ are
  $2$-partitionable over $\F_2$.  Lemma~\ref{lemma_partition_1} shows that
  $[8]-[6]$ and $3[8]-[7]-2[5]$ are $2$-partitionable over $\F_2$.
  By Lemma~\ref{lemma_vsp_type} $[8]-[4]$ and $[5]-[3]$ are $2$-partitionable
  over $\F_2$.
  By Lemma~\ref{lemma_construction_x_consequence} we have
  that $[5]+2[1]$ is $2$-partitionable over $\F_2$.
\end{proof}

\begin{proposition}
  \label{prop_q_2_r_9_h_2_generic}
  The Griesmer upper bound for $n_2(9,2;s)$ is attained
  for all $s\ge 156$, i.e.\ we have\\[-10mm]
\begin{itemize}
\item $n_2(9,2;127t)=511t$ for $t\ge 1$ via $3t\cdot[9]$;\\[-10mm]
\item $n_2(9,2;127t-1)=511t-5$ for $t\ge 1$ via $3t\cdot[9]-[4]$;\\[-10mm]
\item $n_2(9,2;127t-2)=511t-10$ for $t\ge 1$ via $3t\cdot[9]-2[4]$;\\[-10mm]
\item $n_2(9,2;127t-3)=511t-15$ for $t\ge 1$ via $3t\cdot[9]-3[4]$;\\[-10mm]
\item $n_2(9,2;127t-4)=511t-20$ for $t\ge 2$ via $3t\cdot[9]$-4[4];\\[-10mm]
\item $n_2(9,2;127t-5)=511t-21$ for $t\ge 1$ via $3t\cdot[9]-[6]$;\\[-10mm]
\item $n_2(9,2;127t-6)=511t-26$ for $t\ge 1$ via $3t\cdot[9]-[6]-[4]$;\\[-10mm]
\item $n_2(9,2;127t-7)=511t-31$ for $t\ge 1$ via $3t\cdot[9]-[6]-2[4]$;\\[-10mm]
\item $n_2(9,2;127t-8)=511t-36$ for $t\ge 2$ via $3t\cdot[9]-[6]-3[4]$;\\[-10mm]
\item $n_2(9,2;127t-9)=511t-41$ for $t\ge 2$ via $3t\cdot[9]-[6]-4[4]$;\\[-10mm]
\item $n_2(9,2;127t-10)=511t-42$ for $t\ge 1$ via $3t\cdot[9]-2[6]$;\\[-10mm]
\item $n_2(9,2;127t-11)=511t-47$ for $t\ge 1$ via $3t\cdot[9]-2[6]-[4]$;\\[-10mm]
\item $n_2(9,2;127t-12)=511t-52$ for $t\ge 2$ via $3t\cdot[9]-2[6]-2[4]$;\\[-10mm]
\item $n_2(9,2;127t-13)=511t-55$ for $t\ge 1$ via $3t\cdot[9]-[7]-[5]-[3]$;\\[-10mm]
\item $n_2(9,2;127t-14)=511t-60$ for $t\ge 2$ via $3t\cdot[9]-[7]-[5]-[4]-[3]$;\\[-10mm]
\item $n_2(9,2;127t-15)=511t-63$ for $t\ge 1$ via $3t\cdot[9]-[7]-2[5]$;\\[-10mm]
\item $n_2(9,2;127t-16)=511t-68$ for $t\ge 2$ via $3t\cdot[9]-[7]-[6]-2[3]$;\\[-10mm]
\item $n_2(9,2;127t-17)=511t-73$ for $t\ge 2$ via $3t\cdot[9]-3[6]-2[4]$;\\[-10mm]
\item $n_2(9,2;127t-18)=511t-76$ for $t\ge 2$ via $3t\cdot[9]-[7]-[6]-[5]-[3]$;\\[-10mm]
\item $n_2(9,2;127t-19)=511t-81$ for $t\ge 2$ via $3t\cdot[9]-[7]-[6]-[5]-[4]-[3]$;\\[-10mm]
\item $n_2(9,2;127t-20)=511t-84$ for $t\ge 2$ via $3t\cdot[9]-4[6]$;\\[-10mm]
\item $n_2(9,2;127t-21)=511t-85$ for $t\ge 1$ via $3t\cdot[9]-[8]$;\\[-10mm]
\item $n_2(9,2;127t-22)=511t-90$ for $t\ge 1$ via $3t\cdot[9]-[8]-[4]$;\\[-10mm]
\item $n_2(9,2;127t-23)=511t-95$ for $t\ge 1$ via $3t\cdot[9]-[8]-2[4]$;\\[-10mm]
\item $n_2(9,2;127t-24)=511t-100$ for $t\ge 2$ via $3t\cdot[9]-[8]-3[4]$;\\[-10mm]
\item $n_2(9,2;127t-25)=511t-105$ for $t\ge 2$ via $3t\cdot[9]-[8]-4[4]$;\\[-10mm]
\item $n_2(9,2;127t-26)=511t-106$ for $t\ge 1$ via $3t\cdot[9]-[8]-[6]$;\\[-10mm]
\item $n_2(9,2;127t-27)=511t-111$ for $t\ge 1$ via $3t\cdot[9]-[8]-[6]-[4]$;\\[-10mm]
\item $n_2(9,2;127t-28)=511t-116$ for $t\ge 2$ via $3t\cdot[9]-[8]-[6]-2[4]$;\\[-10mm]
\item $n_2(9,2;127t-29)=511t-121$ for $t\ge 2$ via $3t\cdot[9]-[8]-[6]-3[4]$;\\[-10mm]
\item $n_2(9,2;127t-30)=511t-126$ for $t\ge 2$ via $3t\cdot[9]-6[6]$;\\[-10mm]
\item $n_2(9,2;127t-31)=511t-127$ for $t\ge 1$ via $3t\cdot[9]-[8]-2[6]$;\\[-10mm]
\item $n_2(9,2;127t-32)=511t-132$ for $t\ge 2$ via $3t\cdot[9]-[8]-2[6]-[4]$;\\[-10mm]
\item $n_2(9,2;127t-33)=511t-137$ for $t\ge 2$ via $3t\cdot[9]-[8]-2[6]-2[4]$;\\[-10mm]
\item $n_2(9,2;127t-34)=511t-140$ for $t\ge 2$ via $3t\cdot[9]-[8]-[7]-[5]-[3]$;\\[-10mm]
\item $n_2(9,2;127t-35)=511t-145$ for $t\ge 2$ via $3t\cdot[9]-[8]-[7]-[5]-[4]-[3]$;\\[-10mm]
\item $n_2(9,2;127t-36)=511t-148$ for $t\ge 2$ via $3t\cdot[9]-[8]-3[6]$;\\[-10mm]
\item $n_2(9,2;127t-37)=511t-153$ for $t\ge 2$ via $3t\cdot[9]-[8]-3[6]-[4]$;\\[-10mm]
\item $n_2(9,2;127t-38)=511t-158$ for $t\ge 2$ via $3t\cdot[9]-[8]-3[6]-2[4]$;\\[-10mm]
\item $n_2(9,2;127t-39)=511t-161$ for $t\ge 2$ via $3t\cdot[9]-[8]-[7]-[6]-[5]-[3]$;\\[-10mm]
\item $n_2(9,2;127t-40)=511t-166$ for $t\ge 2$ via $3t\cdot[9]-[8]-[7]-[6]-[5]-[4]-[3]$;\\[-10mm]
\item $n_2(9,2;127t-41)=511t-169$ for $t\ge 1$ via $(3t-1)\cdot[9]-[2]+[1]$;\\[-10mm]
\item $n_2(9,2;127t-42)=511t-170$ for $t\ge 1$ via $(3t-1)\cdot[9]+[1]$;\\[-10mm]
\item $n_2(9,2;127t-43)=511t-175$ for $t\ge 1$ via $(3t-1)\cdot[9]-2[3]$;\\[-10mm]
\item $n_2(9,2;127t-44)=511t-180$ for $t\ge 2$ via $(3t-1)\cdot[9]-[4]-2[3]$;\\[-10mm]
\item $n_2(9,2;127t-45)=511t-183$ for $t\ge 1$ via $(3t-1)\cdot[9]-[5]-[3]$;\\[-10mm]
\item $n_2(9,2;127t-46)=511t-188$ for $t\ge 2$ via $(3t-1)\cdot[9]-[5]-[4]-[3]$;\\[-10mm]
\item $n_2(9,2;127t-47)=511t-191$ for $t\ge 1$ via $(3t-1)\cdot[9]-2[5]$;\\[-10mm]
\item $n_2(9,2;127t-48)=511t-196$ for $t\ge 2$ via $(3t-1)\cdot[9]-[6]-2[3]$;\\[-10mm]
\item $n_2(9,2;127t-49)=511t-201$ for $t\ge 2$ via $(3t-1)\cdot[9]-[6]-[4]-2[3]$;\\[-10mm]
\item $n_2(9,2;127t-50)=511t-204$ for $t\ge 2$ via $(3t-1)\cdot[9]-[6]-[5]-[3]$;\\[-10mm]
\item $n_2(9,2;127t-51)=511t-209$ for $t\ge 2$ via $(3t-1)\cdot[9]-[6]-[5]-[4]-[3]$;\\[-10mm]
\item $n_2(9,2;127t-52)=511t-212$ for $t\ge 2$ via $(3t-1)\cdot[9]-[6]-2[5]$;\\[-10mm]
\item $n_2(9,2;127t-53)=511t-215$ for $t\ge 1$ via $(3t-1)\cdot[9]-[7]-[3]$;\\[-10mm]
\item $n_2(9,2;127t-54)=511t-220$ for $t\ge 2$ via $(3t-1)\cdot[9]-[7]-[4]-[3]$;\\[-10mm]
\item $n_2(9,2;127t-55)=511t-223$ for $t\ge 1$ via $(3t-1)\cdot[9]-[7]-[5]$;\\[-10mm]
\item $n_2(9,2;127t-56)=511t-228$ for $t\ge 2$ via $(3t-1)\cdot[9]-[7]-[5]-[4]$;\\[-10mm]
\item $n_2(9,2;127t-57)=511t-233$ for $t\ge 2$ via $(3t-1)\cdot[9]-[7]-[5]-2[4]$;\\[-10mm]
\item $n_2(9,2;127t-58)=511t-236$ for $t\ge 2$ via $(3t-1)\cdot[9]-[7]-[6]-[3]$;\\[-10mm]
\item $n_2(9,2;127t-59)=511t-241$ for $t\ge 2$ via $(3t-1)\cdot[9]-[7]-[6]-[4]-[3]$;\\[-10mm]
\item $n_2(9,2;127t-60)=511t-244$ for $t\ge 2$ via $(3t-1)\cdot[9]-[7]-[6]-[5]$;\\[-10mm]
\item $n_2(9,2;127t-61)=511t-249$ for $t\ge 2$ via $(3t-1)\cdot[9]-[7]-[6]-[5]-[4]$;\\[-10mm]
\item $n_2(9,2;127t-62)=511t-254$ for $t\ge 2$ via $(3t-1)\cdot[9]-[7]-4[5]$;\\[-10mm]
\item $n_2(9,2;127t-63)=511t-255$ for $t\ge 1$ via $(3t-1)\cdot[9]-2[7]$;\\[-10mm]
\item $n_2(9?,2;127t-64)=511t-260$ for $t\ge 2$ via $(3t-1)\cdot[9]-[8]-2[3]$;\\[-10mm]
\item $n_2(9,2;127t-65)=511t-265$ for $t\ge 2$ via $(3t-1)\cdot[9]-[8]-[4]-2[3]$;\\[-10mm]
\item $n_2(9,2;127t-66)=511t-268$ for $t\ge 2$ via $(3t-1)\cdot[9]-[8]-[5]-[3]$;\\[-10mm]
\item $n_2(9,2;127t-67)=511t-273$ for $t\ge 2$ via $(3t-1)\cdot[9]-[8]-[5]-[4]-[3]$;\\[-10mm]
\item $n_2(9,2;127t-68)=511t-276$ for $t\ge 2$ via $(3t-1)\cdot[9]-[8]-2[5]$;\\[-10mm]
\item $n_2(9,2;127t-69)=511t-281$ for $t\ge 2$ via $(3t-1)\cdot[9]-[8]-[6]-2[3]$;\\[-10mm]
\item $n_2(9,2;127t-70)=511t-286$ for $t\ge 2$ via $(3t-1)\cdot[9]-2[7]-3[5]$;\\[-10mm]
\item $n_2(9,2;127t-71)=511t-289$ for $t\ge 2$ via $(3t-1)\cdot[9]-[8]-[6]-[5]-[3]$;\\[-10mm]
\item $n_2(9,2;127t-72)=511t-294$ for $t\ge 2$ via $(3t-1)\cdot[9]-[8]-[6]-[5]-[4]-[3]$;\\[-10mm]
\item $n_2(9,2;127t-73)=511t-297$ for $t\ge 2$ via $(3t-1)\cdot[9]-[8]-[6]-2[5]$;\\[-10mm]
\item $n_2(9,2;127t-74)=511t-300$ for $t\ge 2$ via $(3t-1)\cdot[9]-[8]-[7]-[3]$;\\[-10mm]
\item $n_2(9,2;127t-75)=511t-305$ for $t\ge 2$ via $(3t-1)\cdot[9]-[8]-[7]-[4]-[3]$;\\[-10mm]
\item $n_2(9,2;127t-76)=511t-308$ for $t\ge 2$ via $(3t-1)\cdot[9]-[8]-[7]-[5]$;\\[-10mm]
\item $n_2(9,2;127t-77)=511t-313$ for $t\ge 2$ via $(3t-1)\cdot[9]-[8]-[7]-[5]-[4]$;\\[-10mm]
\item $n_2(9,2;127t-78)=511t-318$ for $t\ge 2$ via $(3t-1)\cdot[9]-3[7]-2[5]$;\\[-10mm]
\item $n_2(9,2;127t-79)=511t-321$ for $t\ge 2$ via $(3t-1)\cdot[9]-[8]-[7]-[6]-[3]$;\\[-10mm]
\item $n_2(9,2;127t-80)=511t-326$ for $t\ge 2$ via $(3t-1)\cdot[9]-[8]-[7]-[6]-[4]-[3]$;\\[-10mm]
\item $n_2(9,2;127t-81)=511t-329$ for $t\ge 2$ via $(3t-1)\cdot[9]-[8]-[7]-[6]-[5]$;\\[-10mm]
\item $n_2(9,2;127t-82)=511t-334$ for $t\ge 2$ via $(3t-1)\cdot[9]-[8]-[7]-[6]-[5]-[4]$;\\[-10mm]
\item $n_2(9,2;127t-83)=511t-339$ for $t\ge 1$ via $(3t-2)\cdot[9]-[2]+2[1]$;\\[-10mm]
\item $n_2(9,2;127t-84)=511t-340$ for $t\ge 1$ via $(3t-2)\cdot[9]+2[1]$;\\[-10mm]
\item $n_2(9,2;127t-85)=511t-343$ for $t\ge 1$ via $(3t-2)\cdot[9]-[3]$;\\[-10mm]
\item $n_2(9,2;127t-86)=511t-348$ for $t\ge 2$ via $(3t-2)\cdot[9]-[4]-[3]$;\\[-10mm]
\item $n_2(9,2;127t-87)=511t-351$ for $t\ge 1$ via $(3t-2)\cdot[9]-[5]$;\\[-10mm]
\item $n_2(9,2;127t-88)=511t-356$ for $t\ge 2$ via $(3t-2)\cdot[9]-[5]-[4]$;\\[-10mm]
\item $n_2(9,2;127t-89)=511t-361$ for $t\ge 2$ via $(3t-2)\cdot[9]-[5]-2[4]$;\\[-10mm]
\item $n_2(9,2;127t-90)=511t-364$ for $t\ge 2$ via $(3t-2)\cdot[9]-[6]-[3]$;\\[-10mm]
\item $n_2(9,2;127t-91)=511t-369$ for $t\ge 2$ via $(3t-2)\cdot[9]-[6]-[4]-[3]$;\\[-10mm]
\item $n_2(9,2;127t-92)=511t-372$ for $t\ge 2$ via $(3t-2)\cdot[9]-[6]-[5]$;\\[-10mm]
\item $n_2(9,2;127t-93)=511t-377$ for $t\ge 2$ via $(3t-2)\cdot[9]-[6]-[5]-[4]$;\\[-10mm]
\item $n_2(9,2;127t-94)=511t-382$ for $t\ge 2$ via $(3t-2)\cdot[9]-4[5]$;\\[-10mm]
\item $n_2(9,2;127t-95)=511t-383$ for $t\ge 1$ via $(3t-2)\cdot[9]-[7]$;\\[-10mm]
\item $n_2(9,2;127t-96)=511t-388$ for $t\ge 2$ via $(3t-2)\cdot[9]-[7]-[4]$;\\[-10mm]
\item $n_2(9,2;127t-97)=511t-393$ for $t\ge 2$ via $(3t-2)\cdot[9]-[7]-2[4]$;\\[-10mm]
\item $n_2(9,2;127t-98)=511t-398$ for $t\ge 2$ via $(3t-2)\cdot[9]-[7]-3[4]$;\\[-10mm]
\item $n_2(9,2;127t-99)=511t-403$ for $t\ge 3$ via $(3t-2)\cdot[9]-[7]-4[4]$;\\[-10mm]
\item $n_2(9,2;127t-100)=511t-404$ for $t\ge 2$ via $(3t-2)\cdot[9]-[7]-[6]$;\\[-10mm]
\item $n_2(9,2;127t-101)=511t-409$ for $t\ge 2$ via $(3t-2)\cdot[9]-[7]-[6]-[4]$;\\[-10mm]
\item $n_2(9,2;127t-102)=511t-414$ for $t\ge 2$ via $(3t-2)\cdot[9]-[7]-[6]-2[4]$;\\[-10mm]
\item $n_2(9,2;127t-103)=511t-419$ for $t\ge 3$ via $(3t-2)\cdot[9]-[7]-[6]-[5]-2[3]$;\\[-10mm]
\item $n_2(9,2;127t-104)=511t-424$ for $t\ge 3$ via $(3t-2)\cdot[9]-3[6]-[5]-2[4]$;\\[-10mm]
\item $n_2(9,2;127t-105)=511t-425$ for $t\ge 2$ via $(3t-2)\cdot[9]-[7]-2[6]$;\\[-10mm]
\item $n_2(9,2;127t-106)=511t-428$ for $t\ge 2$ via $(3t-2)\cdot[9]-[8]-[3]$;\\[-10mm]
\item $n_2(9,2;127t-107)=511t-433$ for $t\ge 2$ via $(3t-2)\cdot[9]-[8]-[4]-[3]$;\\[-10mm]
\item $n_2(9,2;127t-108)=511t-436$ for $t\ge 2$ via $(3t-2)\cdot[9]-[8]-[5]$;\\[-10mm]
\item $n_2(9,2;127t-109)=511t-441$ for $t\ge 2$ via $(3t-2)\cdot[9]-[8]-[5]-[4]$;\\[-10mm]
\item $n_2(9,2;127t-110)=511t-446$ for $t\ge 2$ via $(3t-2)\cdot[9]-2[7]-2[5]$;\\[-10mm]
\item $n_2(9,2;127t-111)=511t-449$ for $t\ge 2$ via $(3t-2)\cdot[9]-[8]-[6]-[3]$;\\[-10mm]
\item $n_2(9,2;127t-112)=511t-454$ for $t\ge 2$ via $(3t-2)\cdot[9]-[8]-[6]-[4]-[3]$;\\[-10mm]
\item $n_2(9,2;127t-113)=511t-457$ for $t\ge 2$ via $(3t-2)\cdot[9]-[8]-[6]-[5]$;\\[-10mm]
\item $n_2(9,2;127t-114)=511t-462$ for $t\ge 2$ via $(3t-2)\cdot[9]-[8]-[6]-[5]-[4]$;\\[-10mm]
\item $n_2(9,2;127t-115)=511t-467$ for $t\ge 3$ via $(3t-2)\cdot[9]-[8]-4[5]$;\\[-10mm]
\item $n_2(9,2;127t-116)=511t-468$ for $t\ge 2$ via $(3t-2)\cdot[9]-[8]-[7]$;\\[-10mm]
\item $n_2(9,2;127t-117)=511t-473$ for $t\ge 2$ via $(3t-2)\cdot[9]-[8]-[7]-[4]$;\\[-10mm]
\item $n_2(9,2;127t-118)=511t-478$ for $t\ge 2$ via $(3t-2)\cdot[9]-[8]-[7]-2[4]$;\\[-10mm]
\item $n_2(9,2;127t-119)=511t-483$ for $t\ge 3$ via $(3t-2)\cdot[9]-[8]-[7]-[5]-2[3]$;\\[-10mm]
\item $n_2(9,2;127t-120)=511t-488$ for $t\ge 3$ via $(3t-2)\cdot[9]-3[7]-[5]-2[4]$;\\[-10mm]
\item $n_2(9,2;127t-121)=511t-489$ for $t\ge 2$ via $(3t-2)\cdot[9]-[8]-[7]-[6]$;\\[-10mm]
\item $n_2(9,2;127t-122)=511t-494$ for $t\ge 2$ via $(3t-2)\cdot[9]-[8]-[7]-[6]-[4]$;\\[-10mm]
\item $n_2(9,2;127t-123)=511t-499$ for $t\ge 3$ via $(3t-2)\cdot[9]-[8]-[7]-3[5]$;\\[-10mm]
\item $n_2(9,2;127t-124)=511t-504$ for $t\ge 3$ via $(3t-2)\cdot[9]-[8]-[7]-[6]-[5]-2[3]$;\\[-10mm]
\item $n_2(9,2;127t-125)=511t-509$ for $t\ge 1$ via $3t\cdot[9]-2[2]$;\\[-10mm]
\item $n_2(9,2;127t-126)=511t-510$ for $t\ge 1$ via $3t\cdot[9]-[2]$.\\[-10mm]
\end{itemize}
\end{proposition}
\begin{proof}
  By Theorem~\ref{thm_partition} we have that $3[9]$, $[8]$, $[6]$, and $[4]$ are
  $2$-partitionable over $\F_2$.  Lemma~\ref{lemma_partition_1} shows that
  $[9]-[7]$ and $3[9]-[8]-2[6]$ are $2$-partitionable over $\F_2$.
  By Lemma~\ref{lemma_vsp_type} $[9]-[5]$, $[7]-[5]$, $[7]-[3]$, $[8]-[6]$,
  $[8]-[4]$, and $[6]-[4]$ are $2$-partitionable over $\F_2$.
  By Lemma~\ref{lemma_construction_x_consequence} we have
  that $[5]+2[1]$ and $2[5]+[1]$ are $2$-partitionable over $\F_2$.
\end{proof}

\begin{proposition}
  \label{prop_q_2_r_10_h_2_generic}
  The Griesmer upper bound for $n_2(10,2;s)$ is attained
  for all $s\ge 396$, i.e.\ we have\\[-10mm]
\begin{itemize}
\item $n_2(10,2;85t)=341t$ for $t\ge 1$ via $t\cdot[10]$;\\[-10mm]
\item $n_2(10,2;85t-1)=341t-5$ for $t\ge 1$ via $t\cdot[10]-[4]$;\\[-10mm]
\item $n_2(10,2;85t-2)=341t-10$ for $t\ge 2$ via $t\cdot[10]-2[4]$;\\[-10mm]
\item $n_2(10,2;85t-3)=341t-15$ for $t\ge 3$ via $t\cdot[10]-3[4]$;\\[-10mm]
\item $n_2(10,2;85t-4)=341t-20$ for $t\ge 4$ via $t\cdot[10]-4[4]$;\\[-10mm]
\item $n_2(10,2;85t-5)=341t-21$ for $t\ge 1$ via $t\cdot[10]-[6]$;\\[-10mm]
\item $n_2(10,2;85t-6)=341t-26$ for $t\ge 2$ via $t\cdot[10]-[6]-[4]$;\\[-10mm]
\item $n_2(10,2;85t-7)=341t-31$ for $t\ge 3$ via $t\cdot[10]-[6]-2[4]$;\\[-10mm]
\item $n_2(10,2;85t-8)=341t-36$ for $t\ge 4$ via $t\cdot[10]-[6]-3[4]$;\\[-10mm]
\item $n_2(10,2;85t-9)=341t-41$ for $t\ge 5$ via $t\cdot[10]-[6]-4[4]$;\\[-10mm]
\item $n_2(10,2;85t-10)=341t-42$ for $t\ge 2$ via $t\cdot[10]-2[6]$;\\[-10mm]
\item $n_2(10,2;85t-11)=341t-47$ for $t\ge 3$ via $t\cdot[10]-2[6]-[4]$;\\[-10mm]
\item $n_2(10,2;85t-12)=341t-52$ for $t\ge 4$ via $t\cdot[10]-2[6]-2[4]$;\\[-10mm]
\item $n_2(10,2;85t-13)=341t-55$ for $t\ge 3$ via $t\cdot[10]-[7]-[5]-[3]$;\\[-10mm]
\item $n_2(10,2;85t-14)=341t-60$ for $t\ge 4$ via $t\cdot[10]-[7]-[5]-[4]-[3]$;\\[-10mm]
\item $n_2(10,2;85t-15)=341t-63$ for $t\ge 3$ via $t\cdot[10]-[7]-2[5]$;\\[-10mm]
\item $n_2(10,2;85t-16)=341t-68$ for $t\ge 4$ via $t\cdot[10]-[7]-[6]-2[3]$;\\[-10mm]
\item $n_2(10,2;85t-17)=341t-73$ for $t\ge 5$ via $t\cdot[10]-[7]-2[5]-2[4]$;\\[-10mm]
\item $n_2(10,2;85t-18)=341t-76$ for $t\ge 4$ via $t\cdot[10]-[7]-[6]-[5]-[3]$;\\[-10mm]
\item $n_2(10,2;85t-19)=341t-81$ for $t\ge 5$ via $t\cdot[10]-[7]-[6]-[5]-[4]-[3]$;\\[-10mm]
\item $n_2(10,2;85t-20)=341t-84$ for $t\ge 4$ via $t\cdot[10]-[7]-[6]-2[5]$;\\[-10mm]
\item $n_2(10,2;85t-21)=341t-85$ for $t\ge 1$ via $t\cdot[10]-[8]$;\\[-10mm]
\item $n_2(10,2;85t-22)=341t-90$ for $t\ge 2$ via $t\cdot[10]-[8]-[4]$;\\[-10mm]
\item $n_2(10,2;85t-23)=341t-95$ for $t\ge 3$ via $t\cdot[10]-[8]-2[4]$;\\[-10mm]
\item $n_2(10,2;85t-24)=341t-100$ for $t\ge 4$ via $t\cdot[10]-[8]-3[4]$;\\[-10mm]
\item $n_2(10,2;85t-25)=341t-105$ for $t\ge 5$ via $t\cdot[10]-[8]-4[4]$;\\[-10mm]
\item $n_2(10,2;85t-26)=341t-106$ for $t\ge 2$ via $t\cdot[10]-[8]-[6]$;\\[-10mm]
\item $n_2(10,2;85t-27)=341t-111$ for $t\ge 3$ via $t\cdot[10]-[8]-[6]-[4]$;\\[-10mm]
\item $n_2(10,2;85t-28)=341t-116$ for $t\ge 4$ via $t\cdot[10]-[8]-[6]-2[4]$;\\[-10mm]
\item $n_2(10,2;85t-29)=341t-121$ for $t\ge 5$ via $t\cdot[10]-[8]-[6]-3[4]$;\\[-10mm]
\item $n_2(10,2;85t-30)=341t-126$ for $t\ge 6$ via $t\cdot[10]-[8]-[6]-[5]-[4]-2[3]$;\\[-10mm]
\item $n_2(10,2;85t-31)=341t-127$ for $t\ge 3$ via $t\cdot[10]-[8]-2[6]$;\\[-10mm]
\item $n_2(10,2;85t-32)=341t-132$ for $t\ge 4$ via $t\cdot[10]-[8]-2[6]-[4]$;\\[-10mm]
\item $n_2(10,2;85t-33)=341t-137$ for $t\ge 5$ via $t\cdot[10]-[8]-2[6]-2[4]$;\\[-10mm]
\item $n_2(10,2;85t-34)=341t-140$ for $t\ge 4$ via $t\cdot[10]-[8]-[7]-[5]-[3]$;\\[-10mm]
\item $n_2(10,2;85t-35)=341t-145$ for $t\ge 5$ via $t\cdot[10]-[8]-[7]-[5]-[4]-[3]$;\\[-10mm]
\item $n_2(10,2;85t-36)=341t-148$ for $t\ge 4$ via $t\cdot[10]-[8]-3[6]$;\\[-10mm]
\item $n_2(10,2;85t-37)=341t-153$ for $t\ge 5$ via $t\cdot[10]-[8]-3[6]-[4]$;\\[-10mm]
\item $n_2(10,2;85t-38)=341t-158$ for $t\ge 6$ via $t\cdot[10]-[8]-3[6]-2[4]$;\\[-10mm]
\item $n_2(10,2;85t-39)=341t-161$ for $t\ge 5$ via $t\cdot[10]-[8]-[7]-[6]-[5]-[3]$;\\[-10mm]
\item $n_2(10,2;85t-40)=341t-166$ for $t\ge 6$ via $t\cdot[10]-[8]-[7]-[6]-[5]-[4]-[3]$;\\[-10mm]
\item $n_2(10,2;85t-41)=341t-169$ for $t\ge 5$ via $t\cdot[10]-[8]-4[6]$;\\[-10mm]
\item $n_2(10,2;85t-42)=341t-170$ for $t\ge 2$ via $t\cdot[10]-2[8]$;\\[-10mm]
\item $n_2(10,2;85t-43)=341t-175$ for $t\ge 3$ via $t\cdot[10]-2[8]-[4]$;\\[-10mm]
\item $n_2(10,2;85t-44)=341t-180$ for $t\ge 4$ via $t\cdot[10]-[9]-[4]-2[3]$;\\[-10mm]
\item $n_2(10,2;85t-45)=341t-183$ for $t\ge 3$ via $t\cdot[10]-[9]-[5]-[3]$;\\[-10mm]
\item $n_2(10,2;85t-46)=341t-188$ for $t\ge 4$ via $t\cdot[10]-[9]-[5]-[4]-[3]$;\\[-10mm]
\item $n_2(10,2;85t-47)=341t-191$ for $t\ge 3$ via $t\cdot[10]-2[8]-[6]$;\\[-10mm]
\item $n_2(10,2;85t-48)=341t-196$ for $t\ge 4$ via $t\cdot[10]-2[8]-[6]-[4]$;\\[-10mm]
\item $n_2(10,2;85t-49)=341t-201$ for $t\ge 5$ via $t\cdot[10]-2[8]-[6]-2[4]$;\\[-10mm]
\item $n_2(10,2;85t-50)=341t-204$ for $t\ge 4$ via $t\cdot[10]-[9]-[6]-[5]-[3]$;\\[-10mm]
\item $n_2(10,2;85t-51)=341t-209$ for $t\ge 5$ via $t\cdot[10]-[9]-[6]-[5]-[4]-[3]$;\\[-10mm]
\item $n_2(10,2;85t-52)=341t-212$ for $t\ge 4$ via $t\cdot[10]-2[8]-2[6]$;\\[-10mm]
\item $n_2(10,2;85t-53)=341t-215$ for $t\ge 3$ via $t\cdot[10]-[9]-[7]-[3]$;\\[-10mm]
\item $n_2(10,2;85t-54)=341t-220$ for $t\ge 4$ via $t\cdot[10]-[9]-[7]-[4]-[3]$;\\[-10mm]
\item $n_2(10,2;85t-55)=341t-223$ for $t\ge 3$ via $t\cdot[10]-[9]-[7]-[5]$;\\[-10mm]
\item $n_2(10,2;85t-56)=341t-228$ for $t\ge 4$ via $t\cdot[10]-[9]-[7]-[5]-[4]$;\\[-10mm]
\item $n_2(10,2;85t-57)=341t-233$ for $t\ge 5$ via $t\cdot[10]-[9]-[7]-[5]-2[4]$;\\[-10mm]
\item $n_2(10,2;85t-58)=341t-236$ for $t\ge 4$ via $t\cdot[10]-[9]-[7]-[6]-[3]$;\\[-10mm]
\item $n_2(10,2;85t-59)=341t-241$ for $t\ge 5$ via $t\cdot[10]-[9]-[7]-[6]-[4]-[3]$;\\[-10mm]
\item $n_2(10,2;85t-60)=341t-244$ for $t\ge 4$ via $t\cdot[10]-[9]-[7]-[6]-[5]$;\\[-10mm]
\item $n_2(10,2;85t-61)=341t-249$ for $t\ge 5$ via $t\cdot[10]-[9]-[7]-[6]-[5]-[4]$;\\[-10mm]
\item $n_2(10,2;85t-62)=341t-254$ for $t\ge 6$ via $t\cdot[10]-2[8]-4[6]$;\\[-10mm]
\item $n_2(10,2;85t-63)=341t-255$ for $t\ge 3$ via $t\cdot[10]-3[8]$;\\[-10mm]
\item $n_2(10,2;85t-64)=341t-260$ for $t\ge 4$ via $t\cdot[10]-3[8]-[4]$;\\[-10mm]
\item $n_2(10,2;85t-65)=341t-265$ for $t\ge 5$ via $t\cdot[10]-3[8]-2[4]$;\\[-10mm]
\item $n_2(10,2;85t-66)=341t-268$ for $t\ge 4$ via $t\cdot[10]-[9]-[8]-[5]-[3]$;\\[-10mm]
\item $n_2(10,2;85t-67)=341t-273$ for $t\ge 5$ via $t\cdot[10]-[9]-[8]-[5]-[4]-[3]$;\\[-10mm]
\item $n_2(10,2;85t-68)=341t-276$ for $t\ge 4$ via $t\cdot[10]-[9]-[8]-2[7]$;\\[-10mm]
\item $n_2(10,2;85t-69)=341t-281$ for $t\ge 5$ via $t\cdot[10]-[9]-[8]-[6]-2[3]$;\\[-10mm]
\item $n_2(10,2;85t-70)=341t-286$ for $t\ge 6$ via $t\cdot[10]-[9]-[8]-[6]-[4]-2[3]$;\\[-10mm]
\item $n_2(10,2;85t-71)=341t-289$ for $t\ge 5$ via $t\cdot[10]-[9]-[8]-[6]-[5]-[3]$;\\[-10mm]
\item $n_2(10,2;85t-72)=341t-294$ for $t\ge 6$ via $t\cdot[10]-[9]-[8]-[6]-[5]-[4]-[3]$;\\[-10mm]
\item $n_2(10,2;85t-73)=341t-297$ for $t\ge 5$ via $t\cdot[10]-[9]-[8]-[6]-2[5]$;\\[-10mm]
\item $n_2(10,2;85t-74)=341t-300$ for $t\ge 4$ via $t\cdot[10]-[9]-[8]-[7]-[3]$;\\[-10mm]
\item $n_2(10,2;85t-75)=341t-305$ for $t\ge 5$ via $t\cdot[10]-[9]-[8]-[7]-[4]-[3]$;\\[-10mm]
\item $n_2(10,2;85t-76)=341t-308$ for $t\ge 4$ via $t\cdot[10]-[9]-[8]-[7]-[5]$;\\[-10mm]
\item $n_2(10,2;85t-77)=341t-313$ for $t\ge 5$ via $t\cdot[10]-[9]-[8]-[7]-[5]-[4]$;\\[-10mm]
\item $n_2(10,2;85t-78)=341t-318$ for $t\ge 6$ via $t\cdot[10]-[9]-[8]-2[6]-2[5]$;\\[-10mm]
\item $n_2(10,2;85t-79)=341t-321$ for $t\ge 5$ via $t\cdot[10]-[9]-[8]-[7]-[6]-[3]$;\\[-10mm]
\item $n_2(10,2;85t-80)=341t-326$ for $t\ge 6$ via $t\cdot[10]-[9]-[8]-[7]-[6]-[4]-[3]$;\\[-10mm]
\item $n_2(10,2;85t-81)=341t-329$ for $t\ge 5$ via $t\cdot[10]-[9]-[8]-[7]-[6]-[5]$;\\[-10mm]
\item $n_2(10,2;85t-82)=341t-334$ for $t\ge 6$ via $t\cdot[10]-[9]-[8]-[7]-[6]-[5]-[4]$;\\[-10mm]
\item $n_2(10,2;85t-83)=341t-339$ for $t\ge 1$ via $(t-1)\cdot[10]+2[2]$;\\[-10mm]
\item $n_2(10,2;85t-84)=341t-340$ for $t\ge 1$ via $(t-1)\cdot[10]+[2]$.\\[-10mm]
\end{itemize}
\end{proposition}
\begin{proof}
  By Theorem~\ref{thm_partition} we have that $[10]$, $[8]$, $[6]$, and $[4]$ are
  $2$-partitionable over $\F_2$.  Lemma~\ref{lemma_partition_1} shows that
  $[10]-[8]$ and $3[10]-[9]-2[7]$ are $2$-partitionable over $\F_2$.
  By Lemma~\ref{lemma_vsp_type} $[10]-[6]$, $[8]-[6]$, $[6]-[4]$, $[9]-[7]$,
  $[7]-[5]$, and $[5]-[3]$ are $2$-partitionable
  over $\F_2$.

  By Lemma~\ref{lemma_construction_x_consequence} we have
  that $[5]+2[1]$ is $2$-partitionable over $\F_2$.
\end{proof}

\begin{proposition}
  \label{prop_q_2_r_7_h_3_generic}
  The Griesmer upper bound for $n_2(7,3;s)$ is attained
  for all $s\ge 12$, i.e.\ we have\\[-10mm]
\begin{itemize}
\item $n_2(7,3;15t)=127t$ for $t\ge 1$ via $7t\cdot[7]$;\\[-10mm]
\item $n_2(7,3;15t-1)=127t-9$ for $t\ge 1$ via $7t\cdot[7]-[6]$;\\[-10mm]
\item $n_2(7,3;15t-2)=127t-18$ for $t\ge 1$ via $(7t-1)\cdot[7]+[3]-2[2]$;\\[-10mm]
\item $n_2(7,3;15t-3)=127t-27$ for $t\ge 1$ via $(7t-1)\cdot[7]-2[5]$;\\[-10mm]
\item $n_2(7,3;15t-4)=127t-36$ for $t\ge 2$ via $(7t-1)\cdot[7]-[6]-2[5]$;\\[-10mm]
\item $n_2(7,3;15t-5)=127t-45$ for $t\ge 1$ via $(7t-2)\cdot[7]-[5]-2[4]$;\\[-10mm]
\item $n_2(7,3;15t-6)=127t-54$ for $t\ge 2$ via $(7t-1)\cdot[7]-4[5]$;\\[-10mm]
\item $n_2(7,3;15t-7)=127t-61$ for $t\ge 1$ via $(7t-3)\cdot[7]-[5]-[4]$;\\[-10mm]
\item $n_2(7,3;15t-8)=127t-70$ for $t\ge 2$ via $(7t-3)\cdot[7]-[6]-[5]-[4]$;\\[-10mm]
\item $n_2(7,3;15t-9)=127t-77$ for $t\ge 1$ via $(7t-4)\cdot[7]-[5]$;\\[-10mm]
\item $n_2(7,3;15t-10)=127t-86$ for $t\ge 2$ via $(7t-4)\cdot[7]-[6]-[5]$;\\[-10mm]
\item $n_2(7,3;15t-11)=127t-95$ for $t\ge 1$ via $(7t-5)\cdot[7]-2[4]$;\\[-10mm]
\item $n_2(7,3;15t-12)=127t-104$ for $t\ge 2$ via $(7t-5)\cdot[7]-[6]-2[4]$;\\[-10mm]
\item $n_2(7,3;15t-13)=127t-111$ for $t\ge 1$ via $(7t-6)\cdot[7]-[4]$;\\[-10mm]
\item $n_2(7,3;15t-14)=127t-120$ for $t\ge 2$ via $(7t-6)\cdot[7]-[6]-[4]$.
\end{itemize}
\end{proposition}
\begin{proof}
  By Theorem~\ref{thm_partition} we have that $7[7]$, $[6]$, and $[3]$ are
  $3$-partitionable over $\F_2$. Lemma~\ref{lemma_partition_1} shows that
  $[7]-[4]$, $3[7]-[5]-2[3]$, and $7[7]-[6]-6[3]$ are $3$-partitionable
  over $\F_2$. By Lemma~\ref{lemma_construction_x_consequence} we have
  that $6[7]+[3]-2[2]$ is $3$-partitionable over $\F_2$.
\end{proof}

\begin{proposition}
  \label{prop_q_2_r_8_h_3_generic}
  The Griesmer upper bound for $n_2(8,3;s)$ is attained
  for all $s\ge 72$, i.e.\ we have\\[-10mm]
\begin{itemize}
\item $n_2(8,3;31t)=255t$ for $t\ge 1$ via $7t\cdot[8]$;\\[-10mm]
\item $n_2(8,3;31t-1)=255t-9$ for $t\ge 1$ via $7t\cdot[8]-[6]$;\\[-10mm]
\item $n_2(8,3;31t-2)=255t-18$ for $t\ge 2$ via $7t\cdot[8]-2[6]$;\\[-10mm]
\item $n_2(8,3;31t-3)=255t-27$ for $t\ge 3$ via $7t\cdot[8]-3[6]$;\\[-10mm]
\item $n_2(8,3;31t-4)=255t-36$ for $t\ge 2$ via $7t\cdot[8]-4[6]$;\\[-10mm]
\item $n_2(8,3;31t-5)=255t-43$ for $t\ge 3$ via $(7t-1)\cdot[8]-[5]-[4]$;\\[-10mm]
\item $n_2(8,3;31t-6)=255t-52$ for $t\ge 1$ via $(7t-1)\cdot[8]-[6]-[5]-[4]$;\\[-10mm]
\item $n_2(8,3;31t-7)=255t-59$ for $t\ge 3$ via $(7t-1)\cdot[8]-[7]-[5]$;\\[-10mm]
\item $n_2(8,3;31t-8)=255t-68$ for $t\ge 2$ via $(7t-1)\cdot[8]-[7]-[6]-[5]$;\\[-10mm]
\item $n_2(8,3;31t-9)=255t-75$ for $t\ge 3$ via $(7t-2)\cdot[8]-[4]$;\\[-10mm]
\item $n_2(8,3;31t-10)=255t-84$ for $t\ge 1$ via $(7t-2)\cdot[8]-[6]-[4]$;\\[-10mm]
\item $n_2(8,3;31t-11)=255t-91$ for $t\ge 3$ via $(7t-2)\cdot[8]-[7]$;\\[-10mm]
\item $n_2(8,3;31t-12)=255t-100$ for $t\ge 2$ via $(7t-2)\cdot[8]-[7]-[6]$;\\[-10mm]
\item $n_2(8,3;31t-13)=255t-109$ for $t\ge 3$ via $(7t-2)\cdot[8]-[7]-2[6]$;\\[-10mm]
\item $n_2(8,3;31t-14)=255t-118$ for $t\ge 3$ via $(7t-3)\cdot[8]-[5]-2[4]$;\\[-10mm]
\item $n_2(8,3;31t-15)=255t-127$ for $t\ge 1$ via $(7t-3)\cdot[8]-4[5]$;\\[-10mm]
\item $n_2(8,3;31t-16)=255t-134$ for $t\ge 3$ via $(7t-3)\cdot[8]-[7]-[5]-[4]$;\\[-10mm]
\item $n_2(8,3;31t-17)=255t-143$ for $t\ge 2$ via $(7t-3)\cdot[8]-[7]-[6]-[5]-[4]$;\\[-10mm]
\item $n_2(8,3;31t-18)=255t-150$ for $t\ge 3$ via $(7t-4)\cdot[8]-2[4]$;\\[-10mm]
\item $n_2(8,3;31t-19)=255t-159$ for $t\ge 1$ via $(7t-4)\cdot[8]-[6]-2[4]$;\\[-10mm]
\item $n_2(8,3;31t-20)=255t-166$ for $t\ge 3$ via $(7t-4)\cdot[8]-[7]-[4]$;\\[-10mm]
\item $n_2(8,3;31t-21)=255t-175$ for $t\ge 2$ via $(7t-4)\cdot[8]-[7]-[6]-[4]$;\\[-10mm]
\item $n_2(8,3;31t-22)=255t-182$ for $t\ge 4$ via $(7t-4)\cdot[8]-2[7]$;\\[-10mm]
\item $n_2(8,3;31t-23)=255t-191$ for $t\ge 1$ via $(7t-5)\cdot[8]-2[5]$;\\[-10mm]
\item $n_2(8,3;31t-24)=255t-200$ for $t\ge 2$ via $(7t-5)\cdot[8]-[6]-2[5]$;\\[-10mm]
\item $n_2(8,3;31t-25)=255t-209$ for $t\ge 3$ via $(7t-5)\cdot[8]-[7]-[5]-2[4]$;\\[-10mm]
\item $n_2(8,3;31t-26)=255t-218$ for $t\ge 4$ via $(7t-5)\cdot[8]-3[6]-2[5]$;\\[-10mm]

\item $n_2(8,3;31t-27)=255t-223$ for $t\ge 1$ via $(7t-6)\cdot[8]-[5]$;\\[-10mm]
\item $n_2(8,3;31t-28)=255t-232$ for $t\ge 2$ via $(7t-6)\cdot[8]-[6]-[5]$;\\[-10mm]
\item $n_2(8,3;31t-29)=255t-241$ for $t\ge 3$ via $(7t-6)\cdot[8]-[7]-2[4]$;\\[-10mm]
\item $n_2(8,3;31t-30)=255t-250$ for $t\ge 4$ via $(7t-6)\cdot[8]-[7]-[6]-2[4]$.
\end{itemize}
\end{proposition}
\begin{proof}
  By Theorem~\ref{thm_partition} we have that $7[8]$, $7[4]$, $[6]$, and $[3]$ are
  $3$-partitionable over $\F_2$. Lemma~\ref{lemma_partition_1} shows that
  $[8]-[5]$, $3[8]-[6]-2[4]$, and $7[8]-[7]-6[4]$ are $3$-partitionable
  over $\F_2$. By Lemma~\ref{lemma_construction_x_consequence} we have
  that $6[8]+[3]-2[2]$ is $3$-partitionable over $\F_2$. Moreover,
  Lemma~\ref{lemma_construction_x_consequence} implies that $[5]+2[1]$ is
  $2$-partitionable over $\F_2$. With this, the corresponding dual shows
  that $2[5]+[4]$ is $3$-partitionable over $\F_2$, so that also
  $2[8]+[4]$ is $3$-partitionable over $\F_2$.
\end{proof}

\begin{proposition}
  \label{prop_q_2_r_9_h_3_generic}
  The Griesmer upper bound for $n_2(9,3;s)$ is attained
  for all $s\ge 663$, i.e.\ we have\\[-10mm]
\begin{itemize}
\item $n_2(9,3;9t)=73t$ for $t\ge 1$ via $t\cdot[9]$;\\[-10mm]
\item $n_2(9,3;9t-1)=73t-9$ for $t\ge 1$ via $t\cdot[9]-[6]$;\\[-10mm]
\item $n_2(9,3;9t-2)=73t-18$ for $t\ge 2$ via $t\cdot[9]-2[6]$;\\[-10mm]
\item $n_2(9,3;9t-3)=73t-27$ for $t\ge 3$ via $t\cdot[9]-3[6]$;\\[-10mm]
\item $n_2(9,3;9t-4)=73t-36$ for $t\ge 4$ via $t\cdot[9]-4[6]$;\\[-10mm]
\item $n_2(9,3;9t-5)=73t-43$ for $t\ge 7$ via $t\cdot[9]-[8]-[5]-[4]$;\\[-10mm]
\item $n_2(9,3;9t-6)=73t-52$ for $t\ge 8$ via $t\cdot[9]-[8]-[6]-[5]-[4]$;\\[-10mm]
\item $n_2(9,3;9t-7)=73t-59$ for $t\ge 10$ via $t\cdot[9]-[8]-[7]-[5]$;\\[-10mm]
\item $n_2(9,3;9t-8)=73t-68$ for $t\ge 11$ via $t\cdot[9]-[8]-[7]-[6]-[5]$.
\end{itemize}
\end{proposition}
\begin{proof}
  By Theorem~\ref{thm_partition} we have that $[9]$ and $7[5]$ are
  $3$-partitionable over $\F_2$. Lemma~\ref{lemma_partition_1} shows that
  $[9]-[6]$, $3[9]-[7]-2[5]$, and $7[9]-[8]-6[5]$ are $3$-partitionable
  over $\F_2$. By Lemma~\ref{lemma_construction_x_consequence} we have
  that $2[5]-2[1]$ is $2$-partitionable over $\F_2$. With this, the
  corresponding dual shows that $5[5]-[4]$ is $3$-partitionable over $\F_2$.
\end{proof}
For $i\in \{5,6,7,8\}$ this gives four infinite series of improvements
$n_2(9,3,9t-i)>\overline{n}_2(9,3,9t-i)$ for sufficiently large $t\in\N$, see
Table~\ref{table_improvements}.

\begin{proposition}
  \label{prop_q_3_r_5_h_2_generic}
  The Griesmer upper bound for $n_3(5,2;s)$ is attained
  for all $s\ge 20$, i.e.\ we have\\[-10mm]
\begin{itemize}
\item $n_3(5,2;13t)=121t$ for $t\ge 1$ via $4t\cdot[5]$;\\[-10mm]
\item $n_3(5,2;13t-1)=121t-10$ for $t\ge 1$ via $4t\cdot[5]-[4]$;\\[-10mm]
\item $n_3(5,2;13t-2)=121t-20$ for $t\ge 2$ via $4t\cdot[5]-2[4]$;\\[-10mm]
\item $n_3(5,2;13t-3)=121t-30$ for $t\ge 1$ via $(4t-1)\cdot[5]+[2]-3[1]$;\\[-10mm]
\item $n_3(5,2;13t-4)=121t-40$ for $t\ge 1$ via $(4t-1)\cdot[5]-3[3]$;\\[-10mm]
\item $n_3(5,2;13t-5)=121t-50$ for $t\ge 2$ via $(4t-1)\cdot[5]-[4]-3[3]$;\\[-10mm]
\item $n_3(5,2;13t-6)=121t-60$ for $t\ge 3$ via $(4t-1)\cdot[5]-2[4]-3[3]$;\\[-10mm]
\item $n_3(5,2;13t-7)=121t-67$ for $t\ge 1$ via $(4t-2)\cdot[5]-2[3]$;\\[-10mm]
\item $n_3(5,2;13t-8)=121t-77$ for $t\ge 2$ via $(4t-2)\cdot[5]-[4]-2[3]$;\\[-10mm]
\item $n_3(5,2;13t-9)=121t-87$ for $t\ge 3$ via $(4t-2)\cdot[5]-2[4]-2[3]$;\\[-10mm]
\item $n_3(5,2;13t-10)=121t-94$ for $t\ge 1$ via $(4t-3)\cdot[5]-[3]$;\\[-10mm]
\item $n_3(5,2;13t-11)=121t-104$ for $t\ge 2$ via $(4t-3)\cdot[5]-[4]-[3]$;\\[-10mm]
\item $n_3(5,2;13t-12)=121t-114$ for $t\ge 3$ via $(4t-3)\cdot[5]-2[4]-[3]$.\\[-10mm]
\end{itemize}
\end{proposition}
\begin{proof}
  By Theorem~\ref{thm_partition} we have that $4[5]$, $[4]$, and $[2]$ are $2$-partitionable over $\F_3$. Lemma~\ref{lemma_partition_1} shows that $[5]-[3]$
  and $4[5]-[4]-3[2]$ are $2$-partitionable over $\F_3$. By
  Lemma~\ref{lemma_construction_x_consequence} we have that $3[5]+[2]-3[1]$
  and $[5]+3[1]$ are $2$-partitionable over $\F_3$.
\end{proof}

\begin{proposition}
  \label{prop_q_3_r_6_h_2_generic}
  The Griesmer upper bound for $n_3(6,2;s)$ is attained
  for all $s\ge 90$, i.e.\ we have\\[-10mm]
\begin{itemize}
\item $n_3(6,2;10t)=91t$ for $t\ge 1$ via $t\cdot[6]$;\\[-10mm]
\item $n_3(6,2;10t-1)=91t-10$ for $t\ge 1$ via $t\cdot[6]-[4]$;\\[-10mm]
\item $n_3(6,2;10t-2)=91t-20$ for $t\ge 2$ via $t\cdot[6]-2[4]$;\\[-10mm]
\item $n_3(6,2;10t-3)=91t-30$ for $t\ge 3$ via $t\cdot[6]-3[4]$;\\[-10mm]
\item $n_3(6,2;10t-4)=91t-40$ for $t\ge 4$ via $t\cdot[6]-4[4]$;\\[-10mm]
\item $n_3(6,2;10t-5)=91t-50$ for $t\ge 5$ via $t\cdot[6]-5[4]$;\\[-10mm]
\item $n_3(6,2;10t-6)=91t-60$ for $t\ge 6$ via $t\cdot[6]-6[4]$;\\[-10mm]
\item $n_3(6,2;10t-7)=91t-67$ for $t\ge 8$ via $t\cdot[6]-2[5]-2[3]$;\\[-10mm]
\item $n_3(6,2;10t-8)=91t-77$ for $t\ge 9$ via $t\cdot[6]-2[5]-[4]-2[3]$;\\[-10mm]
\item $n_3(6,2;10t-9)=91t-87$ for $t\ge 10$ via $t\cdot[6]-2[5]-2[4]-2[3]$.\\[-10mm]
\end{itemize}
\end{proposition}
\begin{proof}
  By Theorem~\ref{thm_partition} we have that $[6]$, $[4]$, and $4[3]$
  are $2$-partitionable over $\F_3$. Lemma~\ref{lemma_partition_1} shows
  that $[6]-[4]$ and $4[6]-[5]-3[3]$ are $2$-partitionable over $\F_3$.
\end{proof}

\begin{proposition}
  \label{prop_q_4_r_5_h_2_generic}
  The Griesmer upper bound for $n_4(5,2;s)$ is attained
  for all $s\ge 55$, i.e.\ we have\\[-10mm]
\begin{itemize}
\item $n_4(5,2;21t)=341t$ for $t\ge 1$ via $5t\cdot[5]$;\\[-10mm]
\item $n_4(5,2;21t-1)=341t-17$ for $t\ge 1$ via $5t\cdot[5]-[4]$;\\[-10mm]
\item $n_4(5,2;21t-2)=341t-34$ for $t\ge 2$ via $5t\cdot[5]-2[4]$;\\[-10mm]
\item $n_4(5,2;21t-3)=341t-51$ for $t\ge 3$ via $5t\cdot[5]-3[4]$;\\[-10mm]
\item $n_4(5,2;21t-4)=341t-68$ for $t\ge 1$ via $(5t-1)\cdot[5]+[2]-4[1]$;\\[-10mm]
\item $n_4(5,2;21t-5)=341t-85$ for $t\ge 1$ via $(5t-1)\cdot[5]-4[3]$;\\[-10mm]
\item $n_4(5,2;21t-6)=341t-102$ for $t\ge 2$ via $5t\cdot[5]-[4]-4[3]$;\\[-10mm]
\item $n_4(5,2;21t-7)=341t-119$ for $t\ge 3$ via $5t\cdot[5]-2[4]-4[3]$;\\[-10mm]
\item $n_4(5,2;21t-8)=341t-136$ for $t\ge 4$ via $5t\cdot[5]-3[4]-4[3]$;\\[-10mm]
\item $n_4(5,2;21t-9)=341t-149$ for $t\ge 1$ via $(5t-2)\cdot[5]-3[3]$;\\[-10mm]
\item $n_4(5,2;21t-10)=341t-166$ for $t\ge 2$ via $(5t-2)\cdot[5]-[4]-3[3]$;\\[-10mm]
\item $n_4(5,2;21t-11)=341t-183$ for $t\ge 3$ via $(5t-2)\cdot[5]-2[4]-3[3]$;\\[-10mm]
\item $n_4(5,2;21t-12)=341t-200$ for $t\ge 4$ via $(5t-2)\cdot[5]-3[4]-3[3]$;\\[-10mm]
\item $n_4(5,2;21t-13)=341t-213$ for $t\ge 1$ via $(5t-3)\cdot[5]-2[3]$;\\[-10mm]
\item $n_4(5,2;21t-14)=341t-230$ for $t\ge 2$ via $(5t-3)\cdot[5]-[4]-2[3]$;\\[-10mm]
\item $n_4(5,2;21t-15)=341t-247$ for $t\ge 3$ via $(5t-3)\cdot[5]-2[4]-2[3]$;\\[-10mm]
\item $n_4(5,2;21t-16)=341t-264$ for $t\ge 4$ via $(5t-3)\cdot[5]-3[4]-2[3]$;\\[-10mm]
\item $n_4(5,2;21t-17)=341t-277$ for $t\ge 1$ via $(5t-4)\cdot[5]-[3]$;\\[-10mm]
\item $n_4(5,2;21t-18)=341t-294$ for $t\ge 2$ via $(5t-4)\cdot[5]-[4]-[3]$;\\[-10mm]
\item $n_4(5,2;21t-19)=341t-311$ for $t\ge 3$ via $(5t-4)\cdot[5]-2[4]-[3]$;\\[-10mm]
\item $n_4(5,2;21t-20)=341t-328$ for $t\ge 4$ via $(5t-4)\cdot[5]-3[4]-[3]$.\\[-10mm]
\end{itemize}
\end{proposition}
\begin{proof}
  By Theorem~\ref{thm_partition} we have that $5[5]$, $[4]$, and $[2]$ are $2$-partitionable over $\F_4$. Lemma~\ref{lemma_partition_1} shows that $[5]-[3]$
  and $5[5]-[4]-4[2]$ are $2$-partitionable over $\F_4$. By
  Lemma~\ref{lemma_construction_x_consequence} we have that $4[5]+[2]-4[1]$
  and $[5]+4[1]$ are $2$-partitionable over $\F_4$.
\end{proof}

\begin{proposition}
  \label{prop_q_5_r_5_h_2_generic}
  The Griesmer upper bound for $n_5(5,2;s)$ is attained
  for all $s\ge 114$, i.e.\ we have\\[-10mm]
\begin{itemize}
\item $n_5(5,2;31t)=781t$ for $t\ge 1$ via $6t\cdot[5]$;\\[-10mm]
\item $n_5(5,2;31t-1)=781t-26$ for $t\ge 1$ via $6t\cdot[5]-[4]$;\\[-10mm]
\item $n_5(5,2;31t-2)=781t-52$ for $t\ge 2$ via $6t\cdot[5]-2[4]$;\\[-10mm]
\item $n_5(5,2;31t-3)=781t-78$ for $t\ge 3$ via $6t\cdot[5]-3[4]$;\\[-10mm]
\item $n_5(5,2;31t-4)=781t-104$ for $t\ge 4$ via $6t\cdot[5]-4[4]$;\\[-10mm]
\item $n_5(5,2;31t-5)=781t-130$ for $t\ge 1$ via $(6t-1)\cdot[5]+[2]-5[1]$;\\[-10mm]
\item $n_5(5,2;31t-6)=781t-156$ for $t\ge 1$ via $(6t-1)\cdot[5]-5[3]$;\\[-10mm]
\item $n_5(5,2;31t-7)=781t-182$ for $t\ge 2$ via $(6t-1)\cdot[5]-[4]-5[3]$;\\[-10mm]
\item $n_5(5,2;31t-8)=781t-208$ for $t\ge 3$ via $(6t-1)\cdot[5]-2[4]-5[3]$;\\[-10mm]
\item $n_5(5,2;31t-9)=781t-234$ for $t\ge 4$ via $(6t-1)\cdot[5]-3[4]-5[3]$;\\[-10mm]
\item $n_5(5,2;31t-10)=781t-260$ for $t\ge 5$ via $(6t-1)\cdot[5]-4[4]-5[3]$;\\[-10mm]
\item $n_5(5,2;31t-11)=781t-281$ for $t\ge 1$ via $(6t-2)\cdot[5]-4[3]$;\\[-10mm]
\item $n_5(5,2;31t-12)=781t-307$ for $t\ge 2$ via $(6t-2)\cdot[5]-[4]-4[3]$;\\[-10mm]
\item $n_5(5,2;31t-13)=781t-333$ for $t\ge 3$ via $(6t-2)\cdot[5]-2[4]-4[3]$;\\[-10mm]
\item $n_5(5,2;31t-14)=781t-359$ for $t\ge 4$ via $(6t-2)\cdot[5]-3[4]-4[3]$;\\[-10mm]
\item $n_5(5,2;31t-15)=781t-385$ for $t\ge 5$ via $(6t-2)\cdot[5]-4[4]-4[3]$;\\[-10mm]
\item $n_5(5,2;31t-16)=781t-406$ for $t\ge 1$ via $(6t-3)\cdot[5]-3[3]$;\\[-10mm]
\item $n_5(5,2;31t-17)=781t-432$ for $t\ge 2$ via $(6t-3)\cdot[5]-[4]-3[3]$;\\[-10mm]
\item $n_5(5,2;31t-18)=781t-458$ for $t\ge 3$ via $(6t-3)\cdot[5]-2[4]-3[3]$;\\[-10mm]
\item $n_5(5,2;31t-19)=781t-484$ for $t\ge 4$ via $(6t-3)\cdot[5]-3[4]-3[3]$;\\[-10mm]
\item $n_5(5,2;31t-20)=781t-510$ for $t\ge 5$ via $(6t-3)\cdot[5]-4[4]-3[3]$;\\[-10mm]
\item $n_5(5,2;31t-21)=781t-531$ for $t\ge 1$ via $(6t-4)\cdot[5]-2[3]$;\\[-10mm]
\item $n_5(5,2;31t-22)=781t-557$ for $t\ge 2$ via $(6t-4)\cdot[5]-[4]-2[3]$;\\[-10mm]
\item $n_5(5,2;31t-23)=781t-583$ for $t\ge 3$ via $(6t-4)\cdot[5]-2[4]-2[3]$;\\[-10mm]
\item $n_5(5,2;31t-24)=781t-609$ for $t\ge 4$ via $(6t-4)\cdot[5]-3[4]-2[3]$;\\[-10mm]
\item $n_5(5,2;31t-25)=781t-635$ for $t\ge 5$ via $(6t-4)\cdot[5]-4[4]-2[3]$;\\[-10mm]
\item $n_5(5,2;31t-26)=781t-656$ for $t\ge 1$ via $(6t-5)\cdot[5]-[3]$;\\[-10mm]
\item $n_5(5,2;31t-27)=781t-682$ for $t\ge 2$ via $(6t-5)\cdot[5]-[4]-[3]$;\\[-10mm]
\item $n_5(5,2;31t-28)=781t-708$ for $t\ge 3$ via $(6t-5)\cdot[5]-2[4]-[3]$;\\[-10mm]
\item $n_5(5,2;31t-29)=781t-734$ for $t\ge 4$ via $(6t-5)\cdot[5]-3[4]-[3]$;\\[-10mm]
\item $n_5(5,2;31t-30)=781t-760$ for $t\ge 5$ via $(6t-5)\cdot[5]-4[4]-[3]$.\\[-10mm]
\end{itemize}
\end{proposition}
\begin{proof}
  By Theorem~\ref{thm_partition} we have that $6[5]$, $[4]$, and $[2]$ are $2$-partitionable over $\F_5$. Lemma~\ref{lemma_partition_1} shows that $[5]-[3]$
  and $6[5]-[4]-5[2]$ are $2$-partitionable over $\F_5$. By
  Lemma~\ref{lemma_construction_x_consequence} we have that $6[5]+[2]-5[1]$
  and $[5]+5[1]$ are $2$-partitionable over $\F_5$.
\end{proof}

\pagebreak

\section{$\mathbf{q}$-divisible multisets of points}
\label{sec_divisible_multisets}

In Proposition~\ref{prop_r_5_h_3} we have used nonexistence results
for $q$-divisible multisets of points in $\PG(4,q)$ with maximum point
multiplicity at most $s$ to conclude upper bounds for $n_q(5,3;s)$,
see Table~\ref{table_n_q_5_3}.

\begin{table}[htp]
  \begin{center}
    \begin{tabular}{|r|rrrrr|}
      \hline
      n / r & 1 & 2 & 3 & 4 & 5 \\
      \hline
       2 & 1 & 0 & 0 & 0 & 0 \\
       3 & 0 & 1 & 0 & 0 & 0 \\
      \hline
    \end{tabular}
    \caption{Number of nonisomorphic $2$-divisible spanning multisets of points
    in $\PG(r-1,2)$ with cardinality $n$.}
    \label{table_2_div_multiset_card}
  \end{center}
\end{table}

We have used the software package \texttt{LinCode} \cite{bouyukliev2021computer}
to enumerate $q$-divisible codes over $\F_q$, which are in one-to-one
correspondence to spanning multisets of points. The obtained enumeration
results for $q\in\{2,3,4,5,7\}$ are given in tables~\ref{table_2_div_multiset_card}-\ref{table_7_div_multiset_card}.

\begin{table}[htp]
  \begin{center}
    \begin{tabular}{|r|rrrrr|}
      \hline
      n / r & 1 & 2 & 3 & 4 & 5 \\
      \hline
       3 & 1 & 0 & 0 & 0 & 0 \\
       4 & 0 & 1 & 0 & 0 & 0 \\
       6 & 1 & 1 & 0 & 0 & 0 \\
       7 & 0 & 1 & 1 & 0 & 0 \\
       8 & 0 & 1 & 1 & 1 & 0 \\
      \hline
    \end{tabular}
    \caption{Number of nonisomorphic $3$-divisible spanning multisets of points
    in $\PG(r-1,3)$ with cardinality $n$.}
    \label{table_3_div_multiset_card}
  \end{center}
\end{table}

The results for $q=2$ in Table~\ref{table_2_div_multiset_card} can be easily
verified theoretically. For $q=3$ one can show that a $3$-divisible multiset
of points with cardinality $6$ is given by $3\chi_P+3\chi_Q$ for two points
$P$ and $Q$ that may coincide:
\begin{lemma}
  \label{lemma_3_div_card_6}
  There is no $3$-divisible multiset of points in $\PG(v-1,3)$ with maximum
  point multiplicity $2$ and cardinality $6$.
\end{lemma}
\begin{proof}
   Let $\cM$ be a $3$-divisible multiset of points in $\PG(v-1,3)$ with cardinality
  $\#\cM=6$, so that $\cM(H)\in\{0,3,6\}$ for every hyperplane $H$. Due to e.g.\   \cite[Lemma 7.7]{kurz2021divisible} not all points can have multiplicity at most
  one. Let $r$ denote the dimension of the span of the points with positive
  multiplicity $\cM(P)>0$. Let $\cM'$ denote the embedding in a corresponding
  $r$-space which is isomorphic to $\PG(r-1,q)$, so that $\cM'(H)\in\{0,3\}$
  for every hyperplane of $\PG(r-1,q)$. The following observation can be
  easily verified: If $r\le 2$ or we have $\cM(P)\ge 3$ for some point  then
  $\cM=6\chi_{P}$ or $\cM=3\chi_{P}+3\chi_{P'}$ for some point $P'$.  So, let
  $P$ be a point  with multiplicity $\cM'(P)=2$. Since $\cM'(H)\le 3$ we have
  $\cM'(Q)\le 1$ for every point $Q\neq P$, so that there are exactly four
  points of multiplicity one and $r=3$. However, for some point $Q$ with
  multiplicity $\cM(Q)=1$ the four lines through $Q$ each have multiplicity
  $3$, so that $\#\cM=1+4\cdot 2=9\neq 6$ -- contradiction.
\end{proof}

\begin{table}[htp]
  \begin{center}
    \begin{tabular}{|r|rrrrr|}
      \hline
      n / r & 1 & 2 & 3 & 4 & 5 \\
      \hline
       4 & 1 & 0 & 0 & 0 & 0 \\
       5 & 0 & 1 & 0 & 0 & 0 \\
       8 & 1 & 1 & 0 & 0 & 0 \\
       9 & 0 & 1 & 1 & 0 & 0 \\
      10 & 0 & 1 & 1 & 1 & 0 \\
      12 & 1 & 2 & \textbf{2} & 0 & 0 \\
      13 & 0 & 2 & 3 & 1 & 0 \\
      14 & 0 & 1 & 5 & 3 & 1 \\
      15 & 0 & 1 & \textbf{3} & \textbf{6} & \textbf{2} \\
      \hline
    \end{tabular}
    \caption{Number of nonisomorphic $4$-divisible spanning multisets of points
    in $\PG(r-1,4)$ with cardinality $n$.}
    \label{table_4_div_multiset_card}
  \end{center}
\end{table}

The possible cardinalities of $q^l$-divisible multisets of points over $\F_q$, where $l\in \N$, have been completely characterized in \cite[Theorem 1]{kiermaier2020lengths}. In our situation $l=1$ there exists a $q$-divisible
multiset of points with cardinality $n$ in $\PG(r-1,q)$, where $r$ is sufficiently large, iff $n$ can be written as $n=a\cdot q+b\cdot(q+1)$ with $a,b\in \N$. If $n\ge q^2-q$ this is always possible and $q^2-q-1$ does not admit such a
representation. The union of $a$ $q$-fold points and $b$ lines gives
a $q$-divisible multiset of points and the question arises if there are other
examples. For cardinality $q^2$ there is e.g.\ the affine plane, i.e.\
$\chi_E-\chi_L$ for some plane $E$ and some line $L\le E$. In tables~\ref{table_2_div_multiset_card}-\ref{table_7_div_multiset_card} we mark
those entries which contain examples that do not consist of unions of $q$-fold points and lines in bold font and briefly discuss the obtained additional constructions in the following.

There is a unique spanning $4$-divisible multiset of points in $\PG(2,4)$ that has cardinality $12$ and maximum point multiplicity at most $3$. It is a doubled oval with an automorphism group of order $138240$.
There is a unique spanning $5$-divisible multiset of points in $\PG(2,5)$ that
has cardinality $23$ and maximum point multiplicity at most $4$. A corresponding
matrix is given by
$$
  \begin{pmatrix}
   11111111111111111110100\\
   00000011223333333441010\\
   01133311441112333231001
  \end{pmatrix}.
$$
This object is also known as a \emph{(strong) $(3\mod 5)$ arc} in $\PG(2,5)$ and a
nice picture can be found in \cite{rousseva2016structure}. By
\cite[Theorem 10]{landjev2019divisible} it is in one-to-one correspondence
to a \emph{$(9,1)$-blocking set} in $\PG(2, 5)$ with line multiplicities
contained in $\{1, 2, 3, 4\}$, i.e., (trivial) blocking sets containing a
full line are excluded. It is well known that the \emph{projective triangle}
is the only possibility over $\F_5$, see e.g.\
\cite{blokhuis1994size,di1969minimum}. In general $(t\mod q)$ arcs are the combinatorial objects that characterize when linear codes are extendable
without increasing the minimum distance, which can be utilized for
nonexistence results for linear codes with certain parameters, see
e.g.\ \cite{kurz2023classification}.

There is a unique spanning $5$-divisible multiset of points in $\PG(3,5)$ that
has cardinality $24$ and does not consist of a union of four lines. A
corresponding matrix is given by
$$
  \begin{pmatrix}
    111111111111111111101000\\
    000000112233333334410100\\
    011333002200033343440010\\
    101013012401301333300001
  \end{pmatrix}.
$$
Although, cardinalities of $q$-divisible multisets of points in $\PG(4,q)$ whose
cardinalities are multiples of $q+1$ do not play a role in the context of Proposition~\ref{prop_r_5_h_3}, they are interesting from another point of view.
A maximal partial line spread is a set of lines which covers each point at most
once and which cannot be extended by a further line without destroying this property. The set of points that are not covered by some given partial line spread $\cL$ of $\PG(r-1,q)$, called  holes, form a $q$-divisible set of points
$\cM$ in $\PG(r-1,q)$ with cardinality $[r]_q-\#\cL\cdot[2]_q$. The property that $\cL$ is maximal is equivalent to the property that $\cM$ does not contain a full line in its support (set of points $P$ with $\cM(P)\ge 1$).

\begin{table}[htp]
  \begin{center}
    \begin{tabular}{|r|rrrrr|}
      \hline
      n / r & 1 & 2 & 3 & 4 & 5 \\
      \hline
       5 & 1 & 0 &   0 &  0 & 0 \\
       6 & 0 & 1 &   0 &  0 & 0 \\
      10 & 1 & 1 &   0 &  0 & 0 \\
      11 & 0 & 1 &   1 &  0 & 0 \\
      12 & 0 & 1 &   1 &  1 & 0 \\
      15 & 1 & 2 &   1 &  0 & 0 \\
      16 & 0 & 2 &   3 &  1 & 0 \\
      17 & 0 & 1 &   4 &  3 & 1 \\
      18 & 0 & 1 &   3 &  5 & 2 \\
      20 & 1 & 4 &   2 &  1 & 0 \\
      21 & 0 & 3 &   9 &  5 & 1 \\
      22 & 0 & 2 &  14 & 18 & 6 \\
      23 & 0 & 1 & \textbf{109} & 27 & 18 \\
      24 & 0 & 1 & 64 & \textbf{35} & 29 \\
      \hline
    \end{tabular}
    \caption{Number of nonisomorphic $5$-divisible spanning multisets of points
    in $\PG(r-1,5)$ with cardinality $n$.}
    \label{table_5_div_multiset_card}
  \end{center}
\end{table}

\begin{theorem}(\cite[Theorem 13]{govaerts2003particular})
  Let $\cM$ be a $q^r$-divisible multiset of points with cardinality
  $\delta[r+1]_q$ in $\PG(v-1,q)$ where $v>r$. If $q>2$ and
  $1\le \delta<\varepsilon$, where $q+\varepsilon$ is the
  size of the smallest nontrivial blocking sets in $\PG(2,q)$, then there exists $(r+1)$-spaces $S_1,\dots,S_\delta$ such that
  $$
    \cM=\sum_{i=1}^\delta \chi_{S_i},
  $$
  i.e., $\cM$ is the sum of $(r+1)$-spaces.
\end{theorem}
There are two $7$-divisible spanning multisets of points in $\PG(2,7)$ with cardinality $39$ and maximum point multiplicity at most $6$. Generator matrices are given by
$$
  \left(\begin{smallmatrix}
  111111111111111111111111111111111100100\\
  000000000111111111122233344455566611010\\
  000335566000023555504412255611405522001
  \end{smallmatrix}\right)
  \,\,\text{and}\,\,
  \left(\begin{smallmatrix}
  111111111111111111111111111111111100100\\
  000000000111111111122233344455566611010\\
  001444666000344455513400022244401315001
  \end{smallmatrix}\right).
$$
As the previous example, it can be obtained from a $(12,1)$ blocking set in $\PG(2,7)$. Again one example is given by the projective triangle and there
is another (sporadic) example, see e.g.\ \cite{blokhuis2003blocking}.
There are seventeen $7$-divisible spanning multisets of points in
$\PG(2,7)$ with cardinality $47$ and maximum point multiplicity at most $6$.
Clearly, we can start from one of the two stated examples for cardinality
$39$ and add the points of a line. However, there are several other
examples. In \cite{heden2001maximal} a maximal partial line spread in $\PG(3,7)$ of cardinality $45$ and a corresponding $7$-divisible set of points of cardinality $40$, that is not the union of five disjoint lines, was constructed.

The described computer enumeration results imply that each $7$-divisible
multiset of points $\cM$ in $\PG(4,7)$ with cardinality at most $38$ can be
written as the sum of characteristic functions of lines and seven-fold points.
By restricting the maximum point multiplicity to six we can slightly extend
the outlined computer enumerations of linear codes to conclude:
\begin{proposition}
  \label{prop_sum_of_lines}
  Let $\cM$ be a $7$-divisible multiset of points in $\PG(4,7)$ with maximum
  point multiplicity at most $6$. If $\#\cM\le 38$ or $\#\cM\in\{41,\dots,46\}$,
  then there exist lines $L_1,\dots,L_l$ such that $\cM=\sum_{i=1}^l \chi_{L_i}$.
\end{proposition}

\begin{table}[htp]
  \begin{center}
    \begin{tabular}{|r|rrrrr|}
      \hline
      n / r & 1 & 2 & 3 & 4 & 5 \\
      \hline
       7 & 1 &  0 &  0 &   0 &   0 \\
       8 & 0 &  1 &  0 &   0 &   0 \\
      14 & 1 &  1 &  0 &   0 &   0 \\
      15 & 0 &  1 &  1 &   0 &   0 \\
      16 & 0 &  1 &  1 &   1 &   0 \\
      21 & 1 &  2 &  1 &   0 &   0 \\
      22 & 0 &  2 &  3 &   1 &   0 \\
      23 & 0 &  1 &  4 &   3 &   1 \\
      24 & 0 &  1 &  3 &   5 &   2 \\
      28 & 1 &  5 &  3 &   1 &   0 \\
      29 & 0 &  3 & 10 &   5 &   1 \\
      30 & 0 &  2 & 15 &  19 &   6 \\
      31 & 0 &  1 & 12 &  29 &  19 \\
      32 & 0 &  1 &  8 &  43 &  32 \\
      35 & 1 &  7 & 11 &   4 &   1 \\
      36 & 0 &  6 & 34 &  32 &   7 \\
      37 & 0 &  3 & 56 & 124 &  57 \\
      38 & 0 &  2 & 57 & 329 & 261 \\
      39 & 0 &  1 & \textbf{41} & 584 & 973 \\
      40 & 0 &  1 & 18 & \textbf{720} & 3639 \\
      \hline
    \end{tabular}
    \caption{Number of nonisomorphic $7$-divisible spanning multisets of points
    in $\PG(r-1,7)$ with cardinality $n$.}
    \label{table_7_div_multiset_card}
  \end{center}
\end{table}

For small dimensions $r$ the $q$-divisible multisets of points in $\PG(r-1,q)$ can be easily characterized. In $\PG(0,q)$ the ambient space itself is just a point so that the only $q$-divisible multisets of points are given
by $m\cdot q\cdot \chi_P$ for any $m\in\N$ and the unique point $P$. Now let $\cM$ be a $q$-divisible multiset of points in $\PG(1,q)$. If $P$ is a point with
multiplicity $\cM(P)\ge q$, then $\cM-q\chi_P$ is also a $q$-divisible
multiset of points in $\PG(1,q)$, so that we can assume $\cM(P)\le q-1$ for
all points $P$ w.l.o.g. In $\PG(1,q)$ hyperplanes coincide with points, so
that $\cM(P)\equiv \#\cM\pmod q$ for every point $P$, so that all points
have the same multiplicity and $\cM$ consists of copies of the entire
ambient space, which is a line. Thus, for $r\le 2$ we have that $q$-divisible
multisets of points in $\PG(r-1,q)$ are given as the union of $q$-fold points
and lines. As already mentioned, in $\PG(2,q)$ we have the affine plane with
cardinality $q^2$ and maximum point multiplicity one.

For a given multiset of points let $p_i$ denote the number of points of
multiplicity $i$ and $a_i$ denote the number of hyperplanes with multiplicity
$i$. By double-counting we obtain:
\begin{lemma}
  \label{lemma_qp1_hyperplanes}
  Let $\cM$ be a multiset of points in $\PG(r-1,q)$, where $r\ge 3$,
  and $S$ an arbitrary $(r-2)$-space. For the $q+1$ hyperplanes
  $H_0,\dots,H_q$ that contain $S$ we have
  \begin{equation}
    \label{eq_qp1_hyperplanes}
    \sum_{i=0}^q \cM\!\left(H_i\right)=\#\cM+q\cdot\cM(S).
  \end{equation}
  Moreover, we have
  \begin{eqnarray*}
    \sum_{i=0}^{\#\cM} p_i = [r]_q,&&
    \sum_{i=0}^{\#\cM} a_i = [r]_q,\\
    \sum_{i=1}^{\#\cM} ip_i = \#\cM\cdot[r-1]_q, &&
    \sum_{i=1}^{\#\cM} ia_i = \#\cM\cdot[r-1]_q\text{,\quad and}\\
    {{\#\cM} \choose 2}\cdot[r-2]_q+q^{r-2} \cdot \left(\sum_{i=2}^{\#\cM} {i \choose 2} \cdot  p_i \right) &=& \sum_{i=2}^{\#\cM} {i \choose 2} \cdot a_i.
  \end{eqnarray*}
\end{lemma}

\begin{lemma}
  \label{lemma_subspace_mult_bound}
  Let $\cM$ be a multiset of points in $\PG(r-1,q)$, $1\le s\le r-2$, and
  $S$ an arbitrary $s$-space. If each hyperplane $H\ge S$ satisfies $\cM(H)\ge x$,
  then we have $\#\cM> qx-(q-1)\cdot\cM(S)$.
\end{lemma}
\begin{proof}
  Since each $s$-space is contained in $[r-s]_q$ hyperplanes and each
  $(s+1)$-space is contained in $[r-s-1]_q$ hyperplanes, double-counting
  points yields
  $\#\cM\ge \cM(S)+\tfrac{[r-s]_q\cdot(x-\cM(S))}{[r-s-1]_q}$. Using
  $[r-s]_q/[r-s-1]_q>q$ yields the stated inequality.
\end{proof}

We can easily check that divisibility is preserved by projection through a point:
\begin{lemma}
  \label{lemma_projection_divisible}
  Let $\cM$ be a (spanning) $\Delta$-divisible multiset of points in $\PG(r-1,q)$, where $r\ge 2$. For an arbitrary point $P$ let $\cM'$ be the multiset of
  points obtained from $\cM$ by projection through $P$. Then, $\cM'$ is a
  (spanning) $\Delta$-divisible multiset of points in $\PG(r-2,q)$ with cardinality
  $\#\cM-\cM(P)$.
\end{lemma}
The multiplicities of the lines through $P$ w.r.t\ $\cM$ are in one-to-to correspondence to the point multiplicities w.r.t.\ $\cM'$.

\begin{lemma}
  \label{lemma_existence_subspace}
  Let $\cM$ be a multiset of points in $\PG(r-1,q)$ with cardinality $\#\cM\le [3]_q$, where $r\ge 3$. If $P$ is a point with positive multiplicity, then
  there exists an $(r-2)$-space $S$ with $\cM(S)=\cM(P)$.
\end{lemma}
\begin{proof}
  For $r=3$ the statement is trivial, so that we assume $r>3$. Let $T$ be a subspace with maximum possible dimension satisfying $\cM(T)=\cM(P)$ and
  $\dim(T)\le r-2$. If $\dim(T)=r-2$ we are done. Otherwise consider the
  $[r-\dim(T)]_q\ge [3]_q$ subspaces that have dimension $\dim(T)+1$ and
  contain $T$. Since there are at most $[3]_q$ points with positive multiplicity
  we conclude the existence of a subspace $T'>T$ with $\dim(T')=\dim(T+1)$
  and $\cM(T')=\cM(T)=\cM(P)$ -- contradiction.
\end{proof}

\begin{proposition}
  \label{prop_7_div_card_42}
  Each $7$-divisible multiset of points in $\PG(4,7)$ with cardinality $42$ has
  maximum point multiplicity at least $7$.
\end{proposition}
\begin{proof}
  Let $\cM$ be a spanning $7$-divisible multiset of points in $\PG(r-1,7)$
  with cardinality $42$ and maximum point multiplicity at most $6$. Due to
  the classification of $q$-divisible multisets of points in $\PG(0,q)$ and
  $\PG(1,q)$, we can assume $r\ge 3$. Since no $7$-divisible multiset of
  points with cardinality $41$ exists over $\F_7$, see e.g.\ \cite[Theorem 1]{kiermaier2020lengths}, we can use Lemma~\ref{lemma_projection_divisible}
  to conclude $\cM(P)\neq 1$ for all points $P$ in $\PG(r-1,q)$.\footnote{Alternatively, one can consider the eight hyperplanes through
  an $(r-2)$-space with multiplicity $1$. We
  remark that there is some research on sets without tangents \cite{van2011sets} and generalizations thereof \cite{heger2024avoiding}.} Using the
  notation from Lemma~\ref{lemma_qp1_hyperplanes} we state
  $p_0+\sum_{i=2}^6 p_i=[r]_7$, $\sum_{i=2}^6 i p_i=42$, and $p_i=0$ for
  $i=1$ or $i\ge 7$. Let $Q$ be a point with positive but minimum
  possible positive multiplicity, i.e.\ $\cM(Q)\ge 2$. From
  Lemma~\ref{lemma_subspace_mult_bound} we conclude $\cM(Q)\in \{2,3\}$.

  Let us first consider the case $r=3$ and $\cM(Q)=2$. All eight lines
  through $Q$ have multiplicity exactly $7$, see
  Equation~(\ref{eq_qp1_hyperplanes}),  and either
  contain a point with multiplicity $5$ or a point with multiplicity $3$
  and another point with multiplicity $2$. Thus, we have $p_3+p_5\ge 8$.
  If $p_5\ge8$, then we have $p_5=8$, $p_2=1$, and $p_i=0$ otherwise.
  However, considering the eight lines through a point of multiplicity $5$
  gives a contradiction. Thus, we have $p_3\ge 1$ and let $R$ denote point
  of multiplicity $3$. From the eight lines through $R$, seven have
  multiplicity $7$ and one has multiplicity $14$, so that $p_5\le 1$
  and $p_3\le 4$ (using $p_1=0$), which contradicts $p_3+p_5\ge 8$.

  Assume $\cM(Q)=2$ and let $\cM'$ denote the projection of $\cM$ through $Q$, so
  that $\cM'$ is a spanning $7$-divisible multiset of points in $\PG(r-2,7)$
  with cardinality $40$ and without a point with multiplicity $1$, see
  Lemma~\ref{lemma_projection_divisible}. From
  our computer enumeration we know that $\cM'$ is the sum of the characteristic
  functions of five lines. Since $\cM$ does not contain a point with multiplicity
  $1$, we have $\cM'=2\cdot \chi_{L_1} +3\cdot\chi_{L_2}$ for  two lines
  $L_1,L_2$. If $L_1=L_2$ then we have $r=3$, which we already have excluded
  before. If $\left|L_1\cap L_2\right|=1$,
  then we have $r=4$ and $\cM(L)\in\{2,4,5,7\}$ for every line $L$ that
  contains $Q$, so that no point has multiplicity $4$ or $6$. Seven of these
  lines have multiplicity $4$, one has multiplicity $7$, and seven have
  multiplicity $5$, so that $8\le p_2\le 9$, $7\le p_3\le 8$, $0\le p_5\le 1$,
  and $p_3+p_5=8$.

  If $P_5$ is a point with multiplicity $5$, $P_3$ a point with multiplicity $3$,
  and $L$ the line spanned by $P_3$ and $P_5$, then $L$ cannot contain a point
  with multiplicity $2$, so that $\cM(L)\in\{8,11,14,17,20,23,26\}$ using
  Equation~(\ref{eq_qp1_hyperplanes}) and the $7$-divisibility of $\cM$. If
  $\cM(L)=26$, then $\cM-3\cdot \chi_{L}$  would be a $7$-divisible multiset
  of points in $\PG(3,7)$ with cardinality $21$ and maximum point multiplicity
  at most $6$ -- contradiction. If $\cM(L)\in\{8,17,20,23\}$, then Equation~(\ref{eq_qp1_hyperplanes}) cannot be satisfied. So, there exist
  three lines $L_1'$, $L_2'$, $L_3'$ through $P_5$ with $\cM(L_1')=\cM(L_2')=11$,
  $\cM(L_3')=14$, and all seven points of multiplicity $3$ are contained on
  $L_1'\cup L_2'\cup L_3'$. Thus, no hyperplane $H\ge L_1'$ can satisfy
  $\cM(H)=14$, so that Equation~(\ref{eq_qp1_hyperplanes}) cannot
  be satisfied -- contradiction. It remains to consider the case $p_2=9$, $p_3=8$,
  and $p_5=0$ for $\left|L_1\cap L_2\right|=1$. Let $L'$ be a line spanned by two
  points of multiplicity $3$, so that $L'$ does not contain a point of
  multiplicity $2$. From Equation~(\ref{eq_qp1_hyperplanes}) we
  conclude $\cM(L')\in\{12,18,21\}$. If $\cM(L')=21$, then $L'$ consists of
  seven points of multiplicity $3$. Let $L''$ be a line spanned by the unique
  point of multiplicity $3$ outside of $L'$ and a point of multiplicity $3$ on $L'$. As before,  we conclude $\cM(L'')\in\{12,18,21\}$, which is impossible
  since there are only two points of multiplicity $3$ on $L''$. If $\cM(L')=18$,
  then $L'$ consists of six points of multiplicity $3$. Since all eight
  hyperplanes through $L'$ have multiplicity $21$, we have $p_3\ge 6+8=14$,
  which is a contradiction. So, all lines contain either $0$, $1$, or $4$ points
  of multiplicity $3$, so that $p_3\equiv 1\pmod 3$ -- contradiction. It
  remains to consider the case $\left|L_1\cap L_2\right|=0$ where $r=5$, $p_2=9$,
  $p_3=8$, and $p_5=0$. Moreover, each line $L$ that contains a point of multiplicity $2$ satisfies $\cM(L)\in\{2,4,5\}$. For a plane $\pi$
  we can use Lemma~\ref{eq_qp1_hyperplanes} and the $7$-divisibility of $\cM$
  to deduce $\cM(\pi)\in\{0,2,3,4,5,7,10,11,12,14,18,19,21,26,28\}$.
  If $\cM(\pi)=10$, then we have $\cM(H_i)=14$ for $0\le i\le 7$, so that
  $p_2\ge 16$ -- contradiction. If $\cM(\pi)=3$, then we have $\cM(H_i)=7$
  for seven hyperplanes, so that $p_2\ge 14$ -- contradiction. If $\pi$
  contains a point $P_3$ of multiplicity $3$ and $\cM(\pi)<19$, then there
  exists a line $P_3\le L\le \pi$ with multiplicity $3$. Since each plane
  through $L$ has multiplicity at least $5$, each hyperplane $H\ge L$ has  multiplicity at least $3+2\cdot 8$, i.e., multiplicity at least $21$.
  This excludes the cases $\pi\in\{5,7,11\}$. Now let $L$ directly be some
  line with multiplicity $3$, which exists since $p_3\ge 1$ and
  $p_2+p_3< [4]_7$. Since each plane $\pi\ge L$ has multiplicity at least
  $12$, we have $\cM(H)\ge 3+9\cdot 8>35$ for each hyperplane $H\ge L$
  -- contradiction.

  In the following we have $\cM(Q)=3$, so that $p_2=0$. If $\cM(P)\in\{5,6\}$
  for some point $P$, then each hyperplane $H\ge P$ satisfies $\cM(H)\ge 14$, so
  that Lemma~\ref{lemma_subspace_mult_bound} gives a contradiction. Thus, we have $p_1=p_2=p_5=p_6=0$,
  $p_3\equiv 2\pmod 4$, and $p_4\equiv 0\pmod 3$. If $\cM(P)\equiv 0\pmod 3$,
  then $\tfrac{1}{3}\cdot \cM$ would be a $7$-divisible multiset of points in
  $\PG(r-1,7)$ with cardinality $14$ and maximum point multiplicity at most $2$
  -- contradiction. Thus, we have $p_4\in\{3,6,9\}$ and $p_3\in\{2,6,10\}$.
  From Lemma~\ref{lemma_existence_subspace} we conclude the existence of
  $(r-2)$-spaces $S_3$, $S_4$ with $\cM(S_3)=3$ and $\cM(S_4)=4$.
  Applying Equation~(\ref{eq_qp1_hyperplanes}) to $S_3$ yields that seven of
  the eight hyperplanes through $S_3$ have multiplicity $7$, so that $p_4\ge 7$.
  Similarly, applying Equation~(\ref{eq_qp1_hyperplanes}) to $S_4$ yields that
  six of the eight hyperplanes through $S_4$ have multiplicity $7$, so that $p_3\ge 6$ -- contradiction.
\end{proof}

\begin {corollary}
  If $\cM$ is a $7$-divisible multiset of points in $\PG(4,7)$ with
  cardinality $42$, then there exist points $P_1,\dots,P_6$ so that
  $\cM=\sum_{i=1}^6 7\cdot \chi_{P_i}$.
\end {corollary}

\begin{proposition}
  \label{prop_7_div_card_43}
  Each $7$-divisible multiset of points in $\PG(4,7)$ with cardinality $43$ has
  maximum point multiplicity at least $7$.
\end{proposition}
\begin{proof}
  Let $\cM$ be a spanning $7$-divisible multiset of points in $\PG(r-1,7)$
  with cardinality $43$ and maximum point multiplicity at most $6$. Due to
  the classification of $q$-divisible multisets of points in $\PG(0,q)$ and
  $\PG(1,q)$, we can assume $r\ge 3$. Since no $7$-divisible multiset of
  points with cardinality $41$ exists over $\F_7$, see e.g.\ \cite[Theorem 1]{kiermaier2020lengths}, we can use Lemma~\ref{lemma_projection_divisible}
  to conclude $\cM(P)\neq 2$ for all points $P$ in $\PG(r-1,q)$.\footnote{Alternatively, one can consider the eight hyperplanes through
  an $(r-2)$-space with multiplicity $2$.}

  Assume $\cM(P)\neq 1$ for all points $P$ for a moment. If $P_6$ is a point
  of multiplicity $6$, then there exists a $(r-2)$-space $S_6$ of multiplicity
  $6$, see Lemma~\ref{lemma_existence_subspace}. Since there are no points
  of multiplicity $1$ or $2$, all eight hyperplanes through $S_6$ have
  multiplicity at least $15$ -- contradiction. If $P_3$ is a point of
  multiplicity $3$, then there exists an $(r-2)$-space $S_3$ of multiplicity
  $3$, see Lemma~\ref{lemma_existence_subspace}. All eight hyperplanes through
  $S_3$ have multiplicity exactly $8$. Since there are no points
  of multiplicity $1$ or $2$, we have $p_5\ge 8$, so that $p_5=8$ and $p_3=1$.
  However, considering the eight hyperplanes through an $(r-2)$-space $S_5$
  of multiplicity $5$ yields a contradiction. Since $\#\cM$ is not divisible
  by $4$, there exists a point $P_5$ and an $(r-2)$-space $S_5$ of
  multiplicity  $5$. Since each point with positive multiplicity has
  multiplicity at least $4$, all eight hyperplanes through $S_5$ have
  multiplicity at least $5$ -- contradiction. Thus, there exists a point
  $P_1$ of multiplicity $1$. By $\cM'$ we denote the projection of $\cM$
  through $P_1$. From Lemma~\ref{lemma_projection_divisible} we conclude that
  $\cM'$ is a spanning $7$-divisible multiset of points in $\PG(r-2,7)$ with
  cardinality $42$. From our previous enumeration we know that
  $\cM'=\sum_{i=1}^6 7\cdot \chi_{Q_i}$ for some points $Q_1,\dots, Q_6$.
  If $S_3$ is an $(r-2)$-space with multiplicity $3$, then all eight hyperplanes
  through $S_3$ have multiplicity $8$. If $S_4$ is an $(r-2)$-space with multiplicity $4$, then seven hyperplanes through $S_4$ have multiplicity $8$
  and one has multiplicity $15$. If $S_5$ is an $(r-2)$-space with multiplicity $5$, then at least six hyperplanes through $S_5$ have multiplicity $8$. Either
  there are two with multiplicity $15$ or one with multiplicity $22$. Since
  there is no $7$-divisible multiset of points in $\PG(4,7)$ with cardinality
  $35$ and maximum point multiplicity at most six, $\cM$ cannot contain a full
  line in its support, i.e., every line through a point with multiplicity $1$
  has either multiplicity $1$ or contains at least one point with multiplicity
  at least $3$.

  Assume $r=3$ for a moment. If $L$ is a line with $\cM(L)>15$, then we have
  $\cM(L)=22$ and $L$ consists of four points of multiplicity $5$ and $2$
  points of multiplicity $1$. Considering the eight lines through a point of
  multiplicity $5$ on $L$ yields $p_4=0$, $p_5=4$, and $p_3\le 7$, so that
  $a_{22}=1$. Using Lemma~\ref{lemma_qp1_hyperplanes} we conclude
  $7a_{15}=3p_3-16$, which gives a contradiction. Using $a_i=0$ for $i>15$,
  Lemma~\ref{lemma_qp1_hyperplanes} gives
  $$
    a_{15}=\frac{3p_3+6p_4+10p_5}{7}-5.
  $$
  Assume that $L$ is a line consisting of three points of
  multiplicity $5$ and five points of multiplicity $0$. Considering the lines
  through a point $P_5$ with multiplicity $5$ on $L$ gives $p_4+p_5=5$. The two
  points with multiplicity at least $4$ outside of $L$ span a line $L'$ that
  meets $L$ in a unique point, which gives a contradiction if we consider
  a different point of multiplicity $5$ on $L$. Thus, no such line exists and
  we have $p_5\le 3$. Now assume that $L$ is a line consisting of two
  points $P_5$, $P_5'$ of multiplicity $5$, one point $P_4$ of multiplicity
  $4$, one point $P_1$   of multiplicity $1$, and four points of multiplicity
  zero. Considering the lines through $P_5$ yields $p_4+p_5=5$ again and the
  two points with multiplicity at least $4$ outside of $L$ span a line meeting
  $L$ in $P_5$. Applying the same argument to $P_5'$ gives a contradiction,
  so that $p_5\in\{0,1\}$. If $P_5$ is the unique point of multiplicity $5$,
  then the two lines of multiplicity $15$ through $P_5$ both contain exactly
  two points of multiplicity $4$ and none of multiplicity $3$, so that
  $p_4=4$, $p_3\le 6$, $a_{15}=2$, and Lemma~\ref{lemma_qp1_hyperplanes} yields
  a contradiction. If $p_4>0$ and $p_5=0$, then the unique line of multiplicity
  $15$  through a point of multiplicity $4$ consists of three points of
  multiplicity $4$, three points of multiplicity $1$, and two points of
  multiplicity zero, so that $p_4\equiv 0\pmod 3$ and $p_3\le 7$. With this
  we conclude $a_4\le 9$ and $a_{15}\le 3$, so that Lemma~\ref{lemma_qp1_hyperplanes} yields a contradiction. If $p_4+p_5=0$, then
  we have $a_{15}=0$ and Lemma~\ref{lemma_qp1_hyperplanes} yields a
  contradiction again. So, we can assume $r\in\{4,5\}$ in the remaining part.

  Let $P_5$ be a point of multiplicity $5$ and $\cM''$ be the projection of
  $\cM$ through $P_5$. From Lemma~\ref{lemma_projection_divisible} we
  conclude that $\cM''$ is $7$-divisible with cardinality $38$. From our
  previous computer enumeration we conclude the existence of three lines
  $L'$, $L''$, $L'''$ and two points $P'$, $P''$ such that $\cM''=\chi_{L'}+
  \chi_{L''}+\chi_{L'''}+7\cdot\chi_{P'}+7\cdot \chi_{P''}$. Either we
  have $L'=L''=L'''$ or there exists a point $\widetilde{P}$ with
  $\cM''(\widetilde{P})\in\{1,2\}$. In the latter case let $L\ge P_5$ be
  the corresponding line with $\cM(L)\in \{6,7\}$, so that $L$ contains
  at least one point with multiplicity $1$. Since each line through a
  point with multiplicity $1$ has a multiplicity that is congruent to $1$
  modulo $7$, we conclude $L'=L''=L'''$. So, for any line $L\ge P_5$ we
  have $\cM(L)\in\{8,12,15,19,22\}$. Since lines through a point of
  multiplicity $1$ have a multiplicity that is congruent to $1$ modulo $7$ and
  hyperplanes that contain a point of multiplicity $3$ have multiplicity $8$,
  we conclude $\cM(L)\in\{8,15,19,22\}$, each line $L_{19}\ge P_5$ with $\cM(L_{19})=19$ consists of exactly three points of multiplicity $5$, and
  each line $L_{22}\ge P_5$ with $\cM(L_{22})=22$ consists of exactly two
  points of multiplicity $5$. In both cases the line has to contain a
  point of multiplicity $4$ which is contained in a hyperplane with multiplicity at least $22$ -- contradiction. Thus, we have
  $\cM(L)\in\{8,15\}$. However, then we conclude $P'\le L'$ and $P''\le L'$,
  which contradicts $r\ge 4$, so that $P_5$ cannot exist and we have $p_5=0$.

 It remains to consider the cases where $r\in\{4,5\}$ and the maximum
 point multiplicity is $3$ or $4$. Let $S$ be an arbitrary $(r-2)$-space.
 For every hyperplane $H\ge S$ we have $\cM(H)=8$ if $H$ contains a point
 of multiplicity $3$ and $\cM(H)\in\{8,15\}$ if $H$ contains a point of
 multiplicity $4$. Since every line that contains a point of multiplicity
 $1$ has a multiplicity that is congruent to $1$ modulo $7$,
 Lemma~\ref{lemma_qp1_hyperplanes} gives $\cM(S)\in\{0,1,3,4,6,7,8,12,15,22,29\}$. If $\cM(S)=12$, then seven of the eight hyperplanes through $S$ have
 multiplicity $15$ and one has multiplicity $22$, so that $S$ cannot contain
 a point of multiplicity larger than $1$. This is impossible, so that we have
 $\cM(S)\neq 12$. If $\cM(S)\in\{15,22,29\}$, then $S$ cannot contain a
 point with multiplicity larger than $1$, which is only possible if $S$
 contains a full line in its support -- contradiction. If $\cM(S)=7$, then
 $S$ consists of a point $P_3$ of multiplicity $3$ and a point $P_4$ of
 multiplicity $4$. However, not all hyperplanes through $S$ can have
 multiplicity $8$, so that $\cM(S)\neq 7$. If $S$ contains both a point
 of multiplicity $3$ and $4$, then $\cM(S)=8$, which is impossible. Thus,
 we have $p_3\cdot p_4=0$. If $P_4$ and $P_4'$ are two different points with
 multiplicity $4$, then the line $L$ spanned by $P_4$ and $P_4'$ has multiplicity
 $8$, so that $r=4$. Let $\pi\ge L$ be a plane with multiplicity $15$ and
 $P_1\le \pi$ be a point with multiplicity $1$. Then the two
 lines $\left\langle P_4,P_1\right\rangle$ and $\left\langle P_4',P_1
 \right\rangle$ have multiplicity $8$ and $\pi$ contains $7$ points
 of multiplicity $1$. Considering another line through two points
 of multiplicity $1$, not being equal to $P_1$, in $\pi$ yields a contradiction.
 If $P_4$ is the unique point with multiplicity $4$, then considering the lines
 through a point of multiplicity $1$ yields a full line in the support of
 $\cM$ -- contradiction. Thus, we have $p_4=0$ and the maximum point
 multiplicity of $\cM$ equals $3$. Considering the lines through a point
 of multiplicity $1$ we conclude $p_3\ge 6$. Let $\pi$ be the subspace
 spanned by three points of multiplicity $3$, so that $\cM(\pi)\ge 9$ and
 $\dim(\pi)\le 3$. Thus, we have $r=4$ and $\pi$ is a hyperplane, which
 is also impossible.
\end{proof}

We remark that it is indeed possible to show that each $7$-divisible multiset
of points $\cM$ in $\PG(r-1,q)$ with $\#\cM\le 38$ or $\#\cM\in\{41,42,43\}$
can be written as the sum of characteristic functions of seven-fold points and
lines, without using any computer enumerations. Also the full characterization
for cardinalities $\#\cM\in\{39,40\}$ is doable by hand. The question is whether
such results can be obtained in a more elegant (and general) way than indicated
in the proof of Proposition~\ref{prop_7_div_card_42}.

\pagebreak

\section{Additive two-weight codes}
\label{sec_two_weight}
As mentioned at the end of Subsection~\ref{subsec_coding_notation}, (projective) linear two-weight codes have received a lot of attention in the literature. So, here we collect a few basic observations on additive two-weight codes in the geometric framework. To this end, a faithful projective $h-(n,r,\{s_1,\dots,s_a\},\mu)_q$ system $\cS$ is a faithful projective $h-(n,r,\min \{s_1,\dots,s_a\},\mu)_q$ system where the number of elements of $\cS$ that are contained
in a hyperplane $H$ of $\PG(r-1,q)$ is contained in $\{s_1,\dots,s_a\}$ for each
hyperplane $H$. So, for $a=2$, each hyperplane contains either $s_1$ or $s_2$ elements of $\cS$. We are interested in the possible parameters of such combinatorial objects. For $r=h$ each hyperplane contains none of the elements of $\cS$. The other cases
where $s_2$ is not attained by at least one hyperplane, i.e.\ additive one-weight codes, are characterized in Theorem~\ref{thm_partition}.\footnote{Here we either assume that the projective systems are faithful or that the size is maximal which implies faithfulness. In general, unfaithful additive one-weight codes are interesting and not fully understood yet, see e.g.\ \cite{krotov2023multispreads}. From a geometric point of view they are called multispreads generalizing spreads.} I.e.\ for $r>h$ there exists a positive integer $l$ such that $n=l\cdot\tfrac{[r]_q}{[\gcd(r,h)]_q}$, $s=l\cdot \tfrac{[r-h]_q}{[\gcd(r,h)]_q}$, and $\mu=l\cdot \tfrac{[h]_q}{[\gcd(r,h)]_q}$. For $r=h+1$ we can choose $0< l<[h+1]_q$ $h$-spaces with multiplicity $s_1$ and the other $\left([h+1]_q-l\right)$ $h$-spaces with multiplicity $s_2>s_1$ for all $s_1,s_2\in\N$, so that $n=s_2[h+1]_q-(s_2-s_1)l$.
In general we have:
\begin{lemma}
  \label{lemma_union_t_weigth}
  Let $\cS$ be a faithful projective $h-(n,r,\{s_1,\dots,s_a\})_q$ system and
  $\cS'$ be a faithful projective $h-(n',r,\{s_1',\dots,s_b'\})_q$ system, then
  $\cS\cup\cS'$ is a faithful projective
  $$
    h-\left(n+n',r,\left\{s_i+s_j'\,:\, 1\le i\le a,1\le j\le b\right\}\right)_q
  $$
  system.
\end{lemma}
So, choosing $a=1$ and $b=2$ e.g.\ gives constructions for two-weight codes using Theorem~\ref{thm_partition}. Another variant is to embed an example from
Theorem~\ref{thm_partition} for $\PG(r'-1,q)$ in $\PG(r-1,q)$, where $r>r'$,
and choose it as $\cS'$. See also Corollary~\ref{cor_few_weight_application}  for additive $t$-weight codes obtained from the Solomon--Stiffler construction.

For a given faithful projective $h-(n,r,\{s_1,\dots,s_a\})_q$ system $\cS$ we can also consider the $l$-fold copy of $\cS$ for $l\in\N$, which is a faithful projective $h-(ln,\{ls_1,\dots,ls_a\})_q$ system. So, one might assume that $\cS$ is a set of $h$-spaces and no subset of $\cS$ partitions the points of the
$\tfrac{[h]_q}{[\gcd(r,h)]_q}$-fold copy of the ambient space $\PG(r-1,q)$ as in Theorem~\ref{thm_partition}.

For $r=2h$ we can consider a partial spread $\cP$ of $h$-spaces in
$\PG(2h-1,q)$, i.e.\ a set of $h$-spaces with pairwise trivial intersection. Due to the existence of a spread of $h$-spaces we have constructions for all $1\le \#\cP\le q^h+1$. I.e.\ faithful projective $h-(n,2h,\{0,1\})_q$ systems
exist for all $1\le n\le q^h+1$. Using copies of a spread of $h$-spaces and Lemma~\ref{lemma_union_t_weigth} we can conclude that faithful projective
$h-(n,2h,\{s_1,s_2\})_q$ systems exist for all $s_1,s_2\in\N$.

For $r=2h+1$ and $h\ge 2$ we have that $[2h+1]-[h+1]$ and $[h-1]_q\cdot [2h+1]+q^{h-1}\cdot[1]$ are  $h$-partitionable over $\F_q$, see Lemma~\ref{lemma_vsp_type} and Lemma~\ref{lemma_construction_x_consequence}, so that projective faithful  $h-\left(q^{h+1},2h+1,\{0,q\}\right)_q$ and $h-\left([2h]_q-q^h,2h+1,\left\{[h]_q-1,[h]_q\right\}\right)_q$ systems do exist for all $h\ge 2$.

\begin{table}[htp]
  \begin{center}
    \begin{tabular}{rrr|rrr|rrr|rrr|rrr}
      \hline
      $n$ & $s_1$ & $s_2$ & $n$ & $s_1$ & $s_2$ & $n$ & $s_1$ & $s_2$ & $n$ & $s_1$ & $s_2$ & $n$ & $s_1$ & $s_2$\\
      \hline
       1 & 0 & 1 & 11 & 2 & 3 &  5 & 1 & 5 & 21 & 3 & 5 & 22 & 2 & 6 \\
       6 & 0 & 2 &  8 & 0 & 4 & 13 & 1 & 5 & 20 & 4 & 5 & 26 & 2 & 6 \\
       8 & 0 & 2 & 16 & 0 & 4 & 17 & 1 & 5 & 21 & 4 & 5 & 20 & 4 & 6 \\
       7 & 1 & 3 & 10 & 2 & 4 & 21 & 1 & 5 & 24 & 0 & 6 & 22 & 4 & 6 \\
       9 & 1 & 3 & 12 & 2 & 4 & 15 & 3 & 5 & 10 & 2 & 6 & 24 & 4 & 6 \\
      11 & 1 & 3 & 14 & 2 & 4 & 17 & 3 & 5 & 14 & 2 & 6 & \\
      10 & 2 & 3 & 16 & 2 & 4 & 19 & 3 & 5 & 18 & 2 & 6 & \\
      \hline
    \end{tabular}
    \caption{Parameters of additive two-weight codes in $\PG(4,2)$ for $s_1<s_2<7$.}
    \label{table_parameters_two_weight_codes_4_2}
  \end{center}
\end{table}

In Table~\ref{table_parameters_two_weight_codes_4_2} we have collected all feasible parameters of sets of $n$ lines in $\PG(4,2)$ such that each hyperplane
contains either $s_1$ or $s_2$ lines, where $s_1<s_2<7$. Note that Theorem~\ref{thm_partition} yields a set of $31$ lines in $\PG(4,2)$ such that
each hyperplane contains exactly $7$ lines. Since the span of two lines is at most
$4$-dimensional, the only possibility for $\left\{s_1,s_2\right\}=\{0,1\}$ is a single line. The previously stated general constructions explain the
cases $\left(n,s_1,s_2\right)\in\left\{(8,0,2),(11,2,3),(16,0,4),(5,1,5)\right\}$. So, there is much to explore.

\begin{remark}
  In \cite{kurz2024computer} it was shown that no projective $[66,5,\{48,56\}]_4$
  code exists. The question whether such a two-weight code exists if we only
  assume additivity instead of linearity over $\F_4$ remains open. A projective
  $[198,10,\{96,112\}]_2$ was constructed in \cite{kohnert2007constructing}.\footnote{A generator matrix can be obtained from \url{http://www.tec.hkr.se/~chen/research/2-weight-codes/search.php}.} So,
  can the corresponding set of $198$ points in $\PG(9,2)$ be partitioned into
  $66$ lines such that each hyperplane contains either $10$ or $18$ lines?
\end{remark}

By double-counting we can infer some necessary existence conditions via linear equation systems. To this end let $x_{s_i}$ denote the number of hyperplanes of $\PG(r-1,q)$ that contain exactly $s_i$ elements of a putative faithful projective  $h-\left(n,r,\left\{s_1,\dots,s_t\right\}\right)_q$ system $\cS$. Counting the number of hyperplanes gives
\begin{equation}
  \sum_{i=1}^t x_{s_i} = [r]_q.
\end{equation}
Double-counting the incidences between elements of $\cS$ and hyperplanes yields
\begin{equation}
  \sum_{i=1}^t s_i \cdot x_{s_i} =n\cdot[r-h]_q.
\end{equation}
In order to double-count the number of incidences between pairs of elements of $\cS$ and hyperplanes we let $y_i$ denote the number of pairs of elements of $\cS$ whose span has dimension $i$. With this, we compute
\begin{equation}
  \sum_{i=1}^t { {s_i} \choose 2} \cdot x_{s_i} = \sum_{i=h}^{2h} y_i\cdot [r-i]_q,
\end{equation}
where we additionally have
\begin{equation}
  \sum_{i=h}^{2h} y_i={n \choose 2}.
\end{equation}
Of course, all occurring variables have to be non-negative and integral.
If $\cS$ is a set of lines, i.e.\ $h=2$, the linear equation system simplifies
to
\begin{eqnarray*}
  x_{s_1}+x_{s_2}=[r]_q,\\
  s_1x_{s_1}+s_2x_{s_2}=n[r-2]_q,\\
  {{s_1}\choose 2}\cdot x_{s_1}+{{s_2}\choose 2}\cdot x_{s_2}
  ={n\choose 2}\cdot [r-4]_q +q^{r-4}\cdot y_3.
\end{eqnarray*}
In Table~\ref{table_parameters_non_existence_two_weight_codes_4_2} we collect
all parameters where the equation system for $x_{s_1}$, $s_{s_2}$, and $y_3$
has non-negative integer solutions but no corresponding additive two-weight code in $\PG(4,2)$ exists, where we restrict the parameters to $s_1<s_2<7$ and sets of lines. If not stated otherwise those nonexistence results have been obtained by
integer linear programming (ILP) computations directly modelling faithful projective systems (or the multiset of covered points) with given parameters.

\begin{table}[htp]
  \begin{center}
    \begin{tabular}{rrr|rrr|rrr|rrr|rrr}
      \hline
      $n$ & $s_1$ & $s_2$ & $n$ & $s_1$ & $s_2$ & $n$ & $s_1$ & $s_2$ & $n$ & $s_1$ & $s_2$ & $n$ & $s_1$ & $s_2$\\
      \hline
       2 & 0 & 2 & 12 & 0 & 4 & 14 & 2 & 5 & 18 & 1 & 6 & 26 & 5 & 6 \\
       4 & 0 & 2 & 13 & 1 & 4 & 17 & 2 & 5 & 15 & 3 & 6 \\
       3 & 0 & 3 & 16 & 1 & 4 &  6 & 0 & 6 & 18 & 3 & 6 \\
      12 & 0 & 3 &  5 & 0 & 5 & 12 & 0 & 6 & 18 & 4 & 6 \\
      13 & 1 & 3 & 20 & 0 & 5 & 18 & 0 & 6 & 26 & 4 & 6 \\
       4 & 0 & 4 &  9 & 1 & 5 & 13 & 1 & 6 & 25 & 5 & 6 \\
      \hline
    \end{tabular}
    \caption{Parameters of nonexistent additive two-weight codes in $\PG(4,2)$ for $s_1<s_2<7$.}
    \label{table_parameters_non_existence_two_weight_codes_4_2}
  \end{center}
\end{table}

\medskip

The example attaining $n_3(6,2;3)=21$ from \cite{clerck2001perp}, mentioned
in the introduction and improving $\overline{n}_3(6,2;3)=17$, is quite exceptional and corresponds to an additive two-weight code over $\F_9$ that
is not linear. As mentioned in Remark~\ref{remark_perp_system} it is a
special case of a so-called \emph{perp-system}, see \cite{clerck2001perp}
for details. An example of a perp-system is a (multi-)set $\cS_q$ of
$2$-spaces in $\PG(5,q)$ with cardinality $q(q^2-q+1)$ such that each
$4$-space contains either $0$ or $q$ elements from $\cS$. They do indeed
exist for even $q$ or $q=3$, see Remark~\ref{remark_perp_system}.
From each perp-system we get a two-weight code
\cite[Theorem 2.2]{clerck2001perp}. For $\cS_2$ the corresponding
projective $[18,6,\{8,12\}]_2$ code is unique, see e.g.\ \cite{projective_divisible_binary_codes}, and can be obtained from the
unique projective $[6,3,\{4,6\}]_4$ code, i.e.\ a hyperoval in geometric
terms. With respect to $\cS_3$ we remark that there are at least six
nonisomorphic projective $[84,6,\{54,63\}]_3$ codes, see e.g.\
\cite{bierbrauer1997family,dissett2000combinatorial} for constructions.
Not all of these sets of points can be partitioned into lines. The dual of
$\cS_q$ is a set of $q(q^2-q+1)$ $4$-spaces that each point is contained
in $0$ or $q$ of these $4$-spaces. Another example would be a set $\cS'$
of $22$ $6$-spaces in $\PG(9,2)$ such that each point is contained in either $0$
or $2$ elements. The uncovered points correspond to a projective $[330,10,\{160,176\}]_2$ code. There exist more than $1700$ nonisomorphic such two-weight
codes, see e.g.\ \cite{dissett2000combinatorial,kohnert2007constructing,momihara2014certain}. So far it is not known whether some of these point sets can be partitioned into $4$-spaces. There also exist projective $[110,5,\{80,88\}]_4$ codes, see e.g.\
\cite{dissett2000combinatorial}, and the question arises whether some of the
corresponding sets of points can be partitioned into lines. Framed differently, the existence of an $[22,2.5,\{20,22\}]_4^2$ additive code, corresponding to strongly regular graph with parameters $(1024, 330, 98, 110)$ \cite{polhill2008generalizations}, is an open problem. We remark that an $[20,2.5,18]_4^2$ code is known \cite{krotovkurz2025}, but currently no $[21,2.5,19]_4^2$ code is known.

\pagebreak

\section{Computer searches}
\label{sec_searches}

In this section we list the examples that we have found by computer searches.
Faithful projective systems can be easily modeled as ILPs. To reduce the search
space we prescribe subgroups of the automorphism group.
Alternatively we can try to partition suitable multisets of points.
Those multisets of points can again be modeled as ILPs and we may prescribe
subgroups of the automorphism group.
Alternatively we use the database of \emph{best known linear codes} (BKLC)
in \texttt{Magma} or enumerate suitable linear codes using \texttt{LinCode}
\cite{bouyukliev2021computer}. For each case we give a list of generator
matrices for the subspaces.


\bigskip

\medskip

\noindent
$n_2(7,2;6)\ge 22$:
$\left(
\right)$.

\bigskip
\bigskip

\noindent
Setting $v=\omega^2$ and $1=\omega^3$ we have the following $\F_4$ examples.

\def\subspace#1#2{\left(\!\begin{smallmatrix}#1\end{smallmatrix}\!\right)^{\!\raisebox{0mm}{\makebox[0mm][l]{$\scriptscriptstyle #2$}}}}
\def\2{\omega} \def\3{\upsilon}

\noindent
$n_4(5,2;2)\ge 20$:
$\subspace{01010\\00101}{}$,
$\subspace{100{\2}{\2}\\001{\2}{\2}}{}$,
$\subspace{110{\3}{\2}\\001{\3}1}{}$,
$\subspace{1{\2}10{\3}\\00011}{}$,
$\subspace{1{\3}01{\3}\\0011{\2}}{}$,
$\subspace{1001{\2}\\01001}{}$,
$\subspace{10001\\01100}{}$,
$\subspace{10111\\0111{\2}}{}$,
$\subspace{10{\2}1{\2}\\011{\3}{\2}}{}$,
$\subspace{10{\3}10\\011{\2}1}{}$,
$\subspace{010{\3}{\3}\\00100}{}$,
$\subspace{1001{\3}\\001{\2}0}{}$,
$\subspace{110{\2}{\2}\\001{\3}0}{}$,
$\subspace{1{\2}{\3}0{\3}\\00010}{}$,
$\subspace{1{\3}0{\3}{\2}\\00110}{}$,
$\subspace{100{\3}1\\01000}{}$,
$\subspace{10000\\01{\3}01}{}$,
$\subspace{101{\3}{\2}\\01{\3}{\2}1}{}$,
$\subspace{10{\2}{\3}1\\01{\3}1{\2}}{}$,
$\subspace{10{\3}{\3}{\2}\\01{\3}{\3}{\2}}{}$.

\medskip

\noindent
$n_4(5,2;3)\ge 39$:
$\subspace{00100\\00010}{}$,
$\subspace{01000\\00001}{}$,
$\subspace{01100\\00011}{}$,
$\subspace{110{\2}{\3}\\001{\2}1}{}$,
$\subspace{110{\3}{\2}\\001{\3}1}{}$,
$\subspace{1{\2}0{\2}1\\0011{\2}}{}$,
$\subspace{1{\2}01{\2}\\001{\3}{\2}}{}$,
$\subspace{1{\3}0{\3}1\\0011{\3}}{}$,
$\subspace{1{\3}01{\3}\\001{\2}{\3}}{}$,
$\subspace{101{\2}{\2}\\0101{\3}}{}$,
$\subspace{101{\3}{\3}\\0101{\2}}{}$,
$\subspace{10{\2}11\\010{\2}1}{}$,
$\subspace{10{\2}{\2}{\2}\\010{\2}{\3}}{}$,
$\subspace{10{\3}11\\010{\3}1}{}$,
$\subspace{10{\3}{\3}{\3}\\010{\3}{\2}}{}$,
$\subspace{1011{\3}\\011{\2}{\3}}{}$,
$\subspace{1011{\2}\\011{\3}{\2}}{}$,
$\subspace{10{\2}{\3}1\\0111{\3}}{}$,
$\subspace{10{\2}{\3}{\2}\\011{\3}1}{}$,
$\subspace{10{\3}{\2}1\\0111{\2}}{}$,
$\subspace{10{\3}{\2}{\3}\\011{\2}1}{}$,
$\subspace{101{\3}0\\01{\2}01}{}$,
$\subspace{10{\2}{\2}0\\01{\2}0{\2}}{}$,
$\subspace{10{\3}10\\01{\2}0{\3}}{}$,
$\subspace{101{\2}0\\01{\3}01}{}$,
$\subspace{10{\2}10\\01{\3}0{\2}}{}$,
$\subspace{10{\3}{\3}0\\01{\3}0{\3}}{}$,
$\subspace{1010{\2}\\01{\2}10}{}$,
$\subspace{10{\2}01\\01{\2}{\2}0}{}$,
$\subspace{10{\3}0{\3}\\01{\2}{\3}0}{}$,
$\subspace{1010{\3}\\01{\3}10}{}$,
$\subspace{10{\2}0{\2}\\01{\3}{\2}0}{}$,
$\subspace{10{\3}01\\01{\3}{\3}0}{}$,
$\subspace{101{\3}{\2}\\01{\2}11}{}$,
$\subspace{10{\2}{\2}1\\01{\2}{\2}{\2}}{}$,
$\subspace{10{\3}1{\3}\\01{\2}{\3}{\3}}{}$,
$\subspace{101{\2}{\3}\\01{\3}11}{}$,
$\subspace{10{\2}1{\2}\\01{\3}{\2}{\2}}{}$,
$\subspace{10{\3}{\3}1\\01{\3}{\3}{\3}}{}$.

\medskip

\noindent
$n_4(5,2;5)\ge 75$:
$\subspace{01{\3}0{\3}\\0001{\3}}{}$,
$\subspace{100{\3}{\2}\\0010{\3}}{}$,
$\subspace{110{\3}0\\00111}{}$,
$\subspace{1{\2}0{\3}{\3}\\001{\3}{\3}}{}$,
$\subspace{1{\3}0{\3}{\3}\\001{\2}1}{}$,
$\subspace{010{\2}1\\001{\2}{\2}}{}$,
$\subspace{100{\2}{\2}\\0011{\2}}{}$,
$\subspace{11{\3}00\\00011}{}$,
$\subspace{1{\2}01{\3}\\00101}{}$,
$\subspace{1{\3}0{\3}0\\001{\3}1}{}$,
$\subspace{010{\3}{\2}\\001{\2}{\2}}{}$,
$\subspace{100{\3}{\3}\\0011{\2}}{}$,
$\subspace{11100\\00011}{}$,
$\subspace{1{\2}0{\2}1\\00101}{}$,
$\subspace{1{\3}010\\001{\3}1}{}$,
$\subspace{01000\\001{\3}0}{}$,
$\subspace{10000\\00010}{}$,
$\subspace{11000\\001{\2}0}{}$,
$\subspace{1{\2}000\\00110}{}$,
$\subspace{1{\3}000\\00100}{}$,
$\subspace{010{\3}{\2}\\001{\3}{\2}}{}$,
$\subspace{10{\3}0{\2}\\0001{\2}}{}$,
$\subspace{110{\3}{\2}\\001{\2}{\3}}{}$,
$\subspace{1{\2}010\\0011{\3}}{}$,
$\subspace{1{\3}0{\2}{\3}\\0010{\2}}{}$,
$\subspace{1001{\3}\\0100{\3}}{}$,
$\subspace{1000{\3}\\0110{\2}}{}$,
$\subspace{10110\\0111{\2}}{}$,
$\subspace{10{\2}1{\2}\\011{\3}{\3}}{}$,
$\subspace{10{\3}1{\2}\\011{\2}0}{}$,
$\subspace{10001\\010{\3}0}{}$,
$\subspace{10{\3}{\2}{\2}\\01001}{}$,
$\subspace{10{\3}00\\01101}{}$,
$\subspace{10{\3}{\3}{\2}\\01{\2}{\2}{\2}}{}$,
$\subspace{10{\3}11\\01{\3}1{\2}}{}$,
$\subspace{1010{\2}\\010{\2}1}{}$,
$\subspace{10{\2}01\\01010}{}$,
$\subspace{10{\2}{\3}{\2}\\0101{\2}}{}$,
$\subspace{100{\3}1\\01{\3}0{\2}}{}$,
$\subspace{10100\\01{\3}{\2}1}{}$,
$\subspace{101{\2}0\\010{\2}{\3}}{}$,
$\subspace{100{\2}1\\01111}{}$,
$\subspace{10{\2}11\\011{\2}{\3}}{}$,
$\subspace{1011{\3}\\01{\2}00}{}$,
$\subspace{10{\2}0{\3}\\01{\2}11}{}$,
$\subspace{10101\\010{\3}{\3}}{}$,
$\subspace{101{\2}{\3}\\010{\3}1}{}$,
$\subspace{10{\3}01\\01011}{}$,
$\subspace{100{\2}0\\01{\2}0{\3}}{}$,
$\subspace{10{\3}0{\3}\\01{\2}10}{}$,
$\subspace{10{\2}{\2}0\\01011}{}$,
$\subspace{10101\\011{\2}{\2}}{}$,
$\subspace{1011{\2}\\011{\3}{\2}}{}$,
$\subspace{1001{\2}\\01{\2}{\3}1}{}$,
$\subspace{10{\2}{\2}1\\01{\2}{\3}0}{}$,
$\subspace{10{\2}00\\010{\2}{\2}}{}$,
$\subspace{10{\2}{\2}{\2}\\010{\2}0}{}$,
$\subspace{10{\3}{\2}{\3}\\0101{\3}}{}$,
$\subspace{1010{\3}\\01{\2}{\3}{\2}}{}$,
$\subspace{10{\2}0{\2}\\01{\2}{\2}{\3}}{}$,
$\subspace{101{\2}0\\0111{\3}}{}$,
$\subspace{10{\2}{\2}{\3}\\01{\2}00}{}$,
$\subspace{10{\2}1{\3}\\01{\2}{\3}1}{}$,
$\subspace{10{\3}{\3}1\\01{\2}{\3}{\3}}{}$,
$\subspace{10{\2}{\2}1\\01{\3}11}{}$,
$\subspace{10{\2}{\3}0\\011{\3}1}{}$,
$\subspace{10011\\01{\2}0{\2}}{}$,
$\subspace{101{\3}{\2}\\01{\3}01}{}$,
$\subspace{101{\2}{\2}\\01{\3}{\3}{\2}}{}$,
$\subspace{10{\2}{\3}1\\01{\3}00}{}$,
$\subspace{10{\3}{\2}1\\011{\2}1}{}$,
$\subspace{100{\2}{\3}\\01{\2}{\2}0}{}$,
$\subspace{10111\\01{\2}1{\3}}{}$,
$\subspace{101{\3}{\3}\\01{\2}{\2}1}{}$,
$\subspace{10{\3}{\2}0\\01{\3}{\3}{\3}}{}$.

\medskip

\noindent
$n_4(5,2;6)\ge 90$:
$\subspace{0100{\2}\\00011}{}$,
$\subspace{01001\\0001{\2}}{}$,
$\subspace{0100{\3}\\0001{\3}}{}$,
$\subspace{1000{\2}\\00101}{}$,
$\subspace{10001\\0010{\2}}{}$,
$\subspace{1000{\3}\\0010{\3}}{}$,
$\subspace{01011\\001{\2}1}{}$,
$\subspace{010{\2}{\3}\\001{\2}{\2}}{}$,
$\subspace{010{\3}{\2}\\001{\2}{\3}}{}$,
$\subspace{10011\\001{\3}1}{}$,
$\subspace{100{\2}{\3}\\001{\3}{\2}}{}$,
$\subspace{100{\3}{\2}\\001{\3}{\3}}{}$,
$\subspace{1100{\3}\\001{\2}1}{}$,
$\subspace{1100{\2}\\001{\2}{\2}}{}$,
$\subspace{11001\\001{\2}{\3}}{}$,
$\subspace{1100{\2}\\001{\3}1}{}$,
$\subspace{11001\\001{\3}{\2}}{}$,
$\subspace{1100{\3}\\001{\3}{\3}}{}$,
$\subspace{1{\2}011\\00100}{}$,
$\subspace{1{\2}0{\2}{\3}\\00100}{}$,
$\subspace{1{\2}0{\3}{\2}\\00100}{}$,
$\subspace{1{\3}10{\2}\\00010}{}$,
$\subspace{1{\3}{\2}01\\00010}{}$,
$\subspace{1{\3}{\3}0{\3}\\00010}{}$,
$\subspace{1{\2}01{\2}\\0010{\3}}{}$,
$\subspace{1{\2}0{\2}1\\00101}{}$,
$\subspace{1{\2}0{\3}{\3}\\0010{\2}}{}$,
$\subspace{1{\3}10{\3}\\00011}{}$,
$\subspace{1{\3}{\2}0{\2}\\0001{\2}}{}$,
$\subspace{1{\3}{\3}01\\0001{\3}}{}$,
$\subspace{100{\2}1\\01010}{}$,
$\subspace{100{\3}{\3}\\010{\2}0}{}$,
$\subspace{1001{\2}\\010{\3}0}{}$,
$\subspace{10{\3}00\\0110{\2}}{}$,
$\subspace{10100\\01{\2}01}{}$,
$\subspace{10{\2}00\\01{\3}0{\3}}{}$,
$\subspace{101{\3}1\\010{\3}{\3}}{}$,
$\subspace{10{\2}1{\3}\\0101{\2}}{}$,
$\subspace{10{\3}{\2}{\2}\\010{\2}1}{}$,
$\subspace{1010{\2}\\011{\2}{\3}}{}$,
$\subspace{10{\2}01\\01{\2}{\3}{\2}}{}$,
$\subspace{10{\3}0{\3}\\01{\3}11}{}$,
$\subspace{100{\3}1\\0111{\2}}{}$,
$\subspace{10{\2}{\2}1\\0110{\3}}{}$,
$\subspace{1001{\3}\\01{\2}{\2}1}{}$,
$\subspace{10{\3}{\3}{\3}\\01{\2}0{\2}}{}$,
$\subspace{100{\2}{\2}\\01{\3}{\3}{\3}}{}$,
$\subspace{1011{\2}\\01{\3}01}{}$,
$\subspace{100{\2}{\2}\\011{\2}1}{}$,
$\subspace{101{\3}{\2}\\0110{\3}}{}$,
$\subspace{100{\3}1\\01{\2}{\3}{\3}}{}$,
$\subspace{10{\2}11\\01{\2}0{\2}}{}$,
$\subspace{1001{\3}\\01{\3}1{\2}}{}$,
$\subspace{10{\3}{\2}{\3}\\01{\3}01}{}$,
$\subspace{100{\3}0\\011{\2}{\2}}{}$,
$\subspace{10{\3}{\2}1\\01100}{}$,
$\subspace{10010\\01{\2}{\3}1}{}$,
$\subspace{101{\3}{\3}\\01{\2}00}{}$,
$\subspace{100{\2}0\\01{\3}1{\3}}{}$,
$\subspace{10{\2}1{\2}\\01{\3}00}{}$,
$\subspace{10010\\011{\3}{\3}}{}$,
$\subspace{10{\3}1{\3}\\01100}{}$,
$\subspace{100{\2}0\\01{\2}1{\2}}{}$,
$\subspace{101{\2}{\2}\\01{\2}00}{}$,
$\subspace{100{\3}0\\01{\3}{\2}1}{}$,
$\subspace{10{\2}{\3}1\\01{\3}00}{}$,
$\subspace{10111\\011{\2}{\3}}{}$,
$\subspace{10{\2}1{\3}\\01111}{}$,
$\subspace{10{\2}{\2}{\3}\\01{\2}{\3}{\2}}{}$,
$\subspace{10{\3}{\2}{\2}\\01{\2}{\2}{\3}}{}$,
$\subspace{101{\3}1\\01{\3}{\3}{\2}}{}$,
$\subspace{10{\3}{\3}{\2}\\01{\3}11}{}$,
$\subspace{101{\3}0\\011{\2}0}{}$,
$\subspace{10{\3}{\2}0\\011{\2}0}{}$,
$\subspace{101{\3}0\\01{\2}{\3}0}{}$,
$\subspace{10{\2}10\\01{\2}{\3}0}{}$,
$\subspace{10{\2}10\\01{\3}10}{}$,
$\subspace{10{\3}{\2}0\\01{\3}10}{}$,
$\subspace{10{\2}{\2}0\\0111{\2}}{}$,
$\subspace{10{\3}{\3}{\3}\\01110}{}$,
$\subspace{1011{\2}\\01{\2}{\2}0}{}$,
$\subspace{10{\3}{\3}0\\01{\2}{\2}1}{}$,
$\subspace{10110\\01{\3}{\3}{\3}}{}$,
$\subspace{10{\2}{\2}1\\01{\3}{\3}0}{}$,
$\subspace{10{\2}1{\2}\\011{\3}1}{}$,
$\subspace{10{\3}11\\011{\2}{\2}}{}$,
$\subspace{101{\2}{\3}\\01{\2}{\3}1}{}$,
$\subspace{10{\3}{\2}1\\01{\2}1{\3}}{}$,
$\subspace{101{\3}{\3}\\01{\3}{\2}{\2}}{}$,
$\subspace{10{\2}{\3}{\2}\\01{\3}1{\3}}{}$.

\medskip

\noindent
$n_4(5,2;7)\ge 107$:
$\subspace{00010\\00001}{}$,
$\subspace{00010\\00001}{}$,
$\subspace{0100{\2}\\00110}{}$,
$\subspace{10010\\0011{\2}}{}$,
$\subspace{1101{\2}\\0010{\2}}{}$,
$\subspace{1001{\2}\\01010}{}$,
$\subspace{10110\\0101{\2}}{}$,
$\subspace{1000{\2}\\01110}{}$,
$\subspace{1011{\2}\\0110{\2}}{}$,
$\subspace{010{\2}0\\00110}{}$,
$\subspace{10010\\001{\3}0}{}$,
$\subspace{110{\3}0\\001{\2}0}{}$,
$\subspace{100{\3}0\\01010}{}$,
$\subspace{10110\\010{\3}0}{}$,
$\subspace{100{\2}0\\01110}{}$,
$\subspace{101{\3}0\\011{\2}0}{}$,
$\subspace{0101{\3}\\001{\2}{\3}}{}$,
$\subspace{100{\2}{\3}\\001{\3}0}{}$,
$\subspace{110{\3}0\\0011{\3}}{}$,
$\subspace{100{\3}0\\010{\2}{\3}}{}$,
$\subspace{101{\2}{\3}\\010{\3}0}{}$,
$\subspace{1001{\3}\\011{\2}{\3}}{}$,
$\subspace{101{\3}0\\0111{\3}}{}$,
$\subspace{01000\\001{\3}{\3}}{}$,
$\subspace{100{\3}{\3}\\001{\3}{\3}}{}$,
$\subspace{110{\3}{\3}\\00100}{}$,
$\subspace{100{\3}{\3}\\010{\3}{\3}}{}$,
$\subspace{101{\3}{\3}\\010{\3}{\3}}{}$,
$\subspace{10000\\011{\3}{\3}}{}$,
$\subspace{101{\3}{\3}\\01100}{}$,
$\subspace{010{\3}1\\001{\3}1}{}$,
$\subspace{100{\3}1\\00100}{}$,
$\subspace{11000\\001{\3}1}{}$,
$\subspace{10000\\010{\3}1}{}$,
$\subspace{101{\3}1\\01000}{}$,
$\subspace{100{\3}1\\011{\3}1}{}$,
$\subspace{10100\\011{\3}1}{}$,
$\subspace{1{\2}001\\0011{\3}}{}$,
$\subspace{10{\2}{\3}0\\0101{\3}}{}$,
$\subspace{10{\3}0{\3}\\0111{\3}}{}$,
$\subspace{1001{\3}\\01{\2}{\2}0}{}$,
$\subspace{1011{\3}\\01{\2}11}{}$,
$\subspace{10{\2}{\2}{\2}\\01{\2}{\3}1}{}$,
$\subspace{10{\2}1{\2}\\01{\3}01}{}$,
$\subspace{1{\2}0{\2}0\\00111}{}$,
$\subspace{10{\2}0{\3}\\01011}{}$,
$\subspace{10{\3}10\\01111}{}$,
$\subspace{10011\\01{\2}0{\2}}{}$,
$\subspace{10111\\01{\2}{\2}1}{}$,
$\subspace{10{\2}{\3}{\2}\\01{\2}{\2}{\3}}{}$,
$\subspace{10{\2}{\3}1\\01{\3}{\2}0}{}$,
$\subspace{1{\2}0{\3}{\2}\\00111}{}$,
$\subspace{10{\2}{\2}0\\01011}{}$,
$\subspace{10{\3}{\2}1\\01111}{}$,
$\subspace{10011\\01{\2}10}{}$,
$\subspace{10111\\01{\2}0{\2}}{}$,
$\subspace{10{\2}0{\3}\\01{\2}1{\2}}{}$,
$\subspace{10{\2}{\2}{\3}\\01{\3}{\3}{\2}}{}$,
$\subspace{1{\2}0{\2}{\3}\\001{\3}{\2}}{}$,
$\subspace{10{\2}10\\010{\3}{\2}}{}$,
$\subspace{10{\3}1{\2}\\011{\3}{\2}}{}$,
$\subspace{100{\3}{\2}\\01{\2}{\3}0}{}$,
$\subspace{101{\3}{\2}\\01{\2}0{\3}}{}$,
$\subspace{10{\2}01\\01{\2}{\3}{\3}}{}$,
$\subspace{10{\2}11\\01{\3}{\2}{\3}}{}$,
$\subspace{1{\2}0{\2}{\3}\\001{\3}{\2}}{}$,
$\subspace{10{\2}10\\010{\3}{\2}}{}$,
$\subspace{10{\3}1{\2}\\011{\3}{\2}}{}$,
$\subspace{100{\3}{\2}\\01{\2}{\3}0}{}$,
$\subspace{101{\3}{\2}\\01{\2}0{\3}}{}$,
$\subspace{10{\2}01\\01{\2}{\3}{\3}}{}$,
$\subspace{10{\2}11\\01{\3}{\2}{\3}}{}$,
$\subspace{1{\3}001\\00101}{}$,
$\subspace{10{\3}01\\01001}{}$,
$\subspace{10{\2}0{\2}\\01101}{}$,
$\subspace{10{\3}00\\01{\2}01}{}$,
$\subspace{10001\\01{\3}0{\2}}{}$,
$\subspace{10101\\01{\3}0{\2}}{}$,
$\subspace{10{\3}01\\01{\3}00}{}$,
$\subspace{1{\3}01{\3}\\0010{\3}}{}$,
$\subspace{10{\3}{\3}{\3}\\0100{\3}}{}$,
$\subspace{10{\2}{\2}1\\0110{\3}}{}$,
$\subspace{10{\3}10\\01{\2}1{\3}}{}$,
$\subspace{1000{\3}\\01{\3}11}{}$,
$\subspace{1010{\3}\\01{\3}{\3}1}{}$,
$\subspace{10{\3}{\2}{\3}\\01{\3}{\2}0}{}$,
$\subspace{1{\3}0{\3}{\2}\\0011{\2}}{}$,
$\subspace{10{\3}0{\2}\\0101{\2}}{}$,
$\subspace{10{\2}1{\3}\\0111{\2}}{}$,
$\subspace{10{\3}{\2}0\\01{\2}{\3}{\2}}{}$,
$\subspace{1001{\2}\\01{\3}0{\3}}{}$,
$\subspace{1011{\2}\\01{\3}{\3}{\3}}{}$,
$\subspace{10{\3}{\2}{\2}\\01{\3}{\3}0}{}$,
$\subspace{1{\3}001\\001{\2}{\2}}{}$,
$\subspace{10{\3}{\3}0\\010{\2}{\2}}{}$,
$\subspace{10{\2}0{\2}\\011{\2}{\2}}{}$,
$\subspace{10{\3}{\2}{\3}\\01{\2}01}{}$,
$\subspace{100{\2}{\2}\\01{\3}10}{}$,
$\subspace{101{\2}{\2}\\01{\3}{\2}1}{}$,
$\subspace{10{\3}1{\3}\\01{\3}{\3}1}{}$,
$\subspace{1{\3}01{\3}\\001{\2}1}{}$,
$\subspace{10{\3}00\\010{\2}1}{}$,
$\subspace{10{\2}{\2}1\\011{\2}1}{}$,
$\subspace{10{\3}{\3}{\2}\\01{\2}1{\3}}{}$,
$\subspace{100{\2}1\\01{\3}00}{}$,
$\subspace{101{\2}1\\01{\3}1{\3}}{}$,
$\subspace{10{\3}{\3}{\2}\\01{\3}1{\3}}{}$.

\medskip

\noindent
$n_4(5,2;9)\ge 141$:
$\subspace{00010\\00001}{}$,
$\subspace{010{\2}{\3}\\0010{\2}}{}$,
$\subspace{1000{\2}\\001{\2}1}{}$,
$\subspace{110{\2}1\\001{\2}{\3}}{}$,
$\subspace{100{\2}1\\0100{\2}}{}$,
$\subspace{1010{\2}\\010{\2}1}{}$,
$\subspace{100{\2}{\3}\\0110{\2}}{}$,
$\subspace{101{\2}1\\011{\2}{\3}}{}$,
$\subspace{010{\3}1\\00101}{}$,
$\subspace{10001\\001{\3}0}{}$,
$\subspace{110{\3}0\\001{\3}1}{}$,
$\subspace{100{\3}0\\01001}{}$,
$\subspace{10101\\010{\3}0}{}$,
$\subspace{100{\3}1\\01101}{}$,
$\subspace{101{\3}0\\011{\3}1}{}$,
$\subspace{010{\3}1\\0010{\3}}{}$,
$\subspace{1000{\3}\\001{\3}{\2}}{}$,
$\subspace{110{\3}{\2}\\001{\3}1}{}$,
$\subspace{100{\3}{\2}\\0100{\3}}{}$,
$\subspace{1010{\3}\\010{\3}{\2}}{}$,
$\subspace{100{\3}1\\0110{\3}}{}$,
$\subspace{101{\3}{\2}\\011{\3}1}{}$,
$\subspace{01000\\0011{\3}}{}$,
$\subspace{1001{\3}\\0011{\3}}{}$,
$\subspace{1101{\3}\\00100}{}$,
$\subspace{1001{\3}\\0101{\3}}{}$,
$\subspace{1011{\3}\\0101{\3}}{}$,
$\subspace{10000\\0111{\3}}{}$,
$\subspace{1011{\3}\\01100}{}$,
$\subspace{010{\2}{\3}\\00111}{}$,
$\subspace{10011\\001{\3}{\2}}{}$,
$\subspace{110{\3}{\2}\\001{\2}{\3}}{}$,
$\subspace{100{\3}{\2}\\01011}{}$,
$\subspace{10111\\010{\3}{\2}}{}$,
$\subspace{100{\2}{\3}\\01111}{}$,
$\subspace{101{\3}{\2}\\011{\2}{\3}}{}$,
$\subspace{010{\2}{\2}\\001{\2}1}{}$,
$\subspace{100{\2}1\\0010{\3}}{}$,
$\subspace{1100{\3}\\001{\2}{\2}}{}$,
$\subspace{1000{\3}\\010{\2}1}{}$,
$\subspace{101{\2}1\\0100{\3}}{}$,
$\subspace{100{\2}{\2}\\011{\2}1}{}$,
$\subspace{1010{\3}\\011{\2}{\2}}{}$,
$\subspace{1{\2}00{\3}\\0010{\2}}{}$,
$\subspace{10{\2}00\\0100{\2}}{}$,
$\subspace{10{\3}0{\2}\\0110{\2}}{}$,
$\subspace{1000{\2}\\01{\2}00}{}$,
$\subspace{1010{\2}\\01{\2}0{\3}}{}$,
$\subspace{10{\2}01\\01{\2}0{\3}}{}$,
$\subspace{10{\2}01\\01{\3}0{\3}}{}$,
$\subspace{1{\2}0{\3}0\\0011{\2}}{}$,
$\subspace{10{\2}{\2}1\\0101{\2}}{}$,
$\subspace{10{\3}{\2}0\\0111{\2}}{}$,
$\subspace{1001{\2}\\01{\2}1{\3}}{}$,
$\subspace{1011{\2}\\01{\2}0{\2}}{}$,
$\subspace{10{\2}0{\3}\\01{\2}11}{}$,
$\subspace{10{\2}{\2}{\2}\\01{\3}{\3}0}{}$,
$\subspace{1{\2}00{\2}\\001{\2}0}{}$,
$\subspace{10{\2}1{\3}\\010{\2}0}{}$,
$\subspace{10{\3}01\\011{\2}0}{}$,
$\subspace{100{\2}0\\01{\2}{\3}{\2}}{}$,
$\subspace{101{\2}0\\01{\2}{\2}{\3}}{}$,
$\subspace{10{\2}{\3}1\\01{\2}11}{}$,
$\subspace{10{\2}{\2}{\2}\\01{\3}0{\2}}{}$,
$\subspace{1{\2}0{\2}0\\001{\2}{\2}}{}$,
$\subspace{10{\2}{\2}1\\010{\2}{\2}}{}$,
$\subspace{10{\3}10\\011{\2}{\2}}{}$,
$\subspace{100{\2}{\2}\\01{\2}1{\3}}{}$,
$\subspace{101{\2}{\2}\\01{\2}1{\2}}{}$,
$\subspace{10{\2}{\2}{\3}\\01{\2}01}{}$,
$\subspace{10{\2}0{\2}\\01{\3}{\2}0}{}$,
$\subspace{1{\2}00{\2}\\001{\3}0}{}$,
$\subspace{10{\2}{\2}{\3}\\010{\3}0}{}$,
$\subspace{10{\3}01\\011{\3}0}{}$,
$\subspace{100{\3}0\\01{\2}1{\2}}{}$,
$\subspace{101{\3}0\\01{\2}{\3}{\3}}{}$,
$\subspace{10{\2}11\\01{\2}{\2}1}{}$,
$\subspace{10{\2}{\3}{\2}\\01{\3}0{\2}}{}$,
$\subspace{1{\2}0{\2}{\3}\\001{\3}{\3}}{}$,
$\subspace{10{\2}1{\3}\\010{\3}{\3}}{}$,
$\subspace{10{\3}1{\2}\\011{\3}{\3}}{}$,
$\subspace{100{\3}{\3}\\01{\2}{\3}{\2}}{}$,
$\subspace{101{\3}{\3}\\01{\2}0{\2}}{}$,
$\subspace{10{\2}0{\3}\\01{\2}{\3}0}{}$,
$\subspace{10{\2}10\\01{\3}{\2}{\3}}{}$,
$\subspace{1{\3}010\\00100}{}$,
$\subspace{10{\3}{\3}0\\01000}{}$,
$\subspace{10{\2}{\2}0\\01100}{}$,
$\subspace{10{\3}10\\01{\2}10}{}$,
$\subspace{10000\\01{\3}10}{}$,
$\subspace{10100\\01{\3}{\3}0}{}$,
$\subspace{10{\3}{\2}0\\01{\3}{\2}0}{}$,
$\subspace{1{\3}0{\2}0\\00101}{}$,
$\subspace{10{\3}1{\2}\\01001}{}$,
$\subspace{10{\2}{\3}0\\01101}{}$,
$\subspace{10{\3}{\2}1\\01{\2}{\2}0}{}$,
$\subspace{10001\\01{\3}{\2}{\3}}{}$,
$\subspace{10101\\01{\3}11}{}$,
$\subspace{10{\3}{\3}{\3}\\01{\3}{\3}{\2}}{}$,
$\subspace{1{\3}000\\00111}{}$,
$\subspace{10{\3}{\2}{\2}\\01011}{}$,
$\subspace{10{\2}00\\01111}{}$,
$\subspace{10{\3}11\\01{\2}00}{}$,
$\subspace{10011\\01{\3}{\3}{\3}}{}$,
$\subspace{10111\\01{\3}11}{}$,
$\subspace{10{\3}{\3}{\3}\\01{\3}{\2}{\2}}{}$,
$\subspace{1{\3}0{\2}{\2}\\00110}{}$,
$\subspace{10{\3}{\3}1\\01010}{}$,
$\subspace{10{\2}{\3}{\3}\\01110}{}$,
$\subspace{10{\3}{\3}{\2}\\01{\2}{\2}{\2}}{}$,
$\subspace{10010\\01{\3}1{\2}}{}$,
$\subspace{10110\\01{\3}01}{}$,
$\subspace{10{\3}0{\3}\\01{\3}1{\3}}{}$,
$\subspace{1{\3}0{\2}{\2}\\00110}{}$,
$\subspace{10{\3}{\3}1\\01010}{}$,
$\subspace{10{\2}{\3}{\3}\\01110}{}$,
$\subspace{10{\3}{\3}{\2}\\01{\2}{\2}{\2}}{}$,
$\subspace{10010\\01{\3}1{\2}}{}$,
$\subspace{10110\\01{\3}01}{}$,
$\subspace{10{\3}0{\3}\\01{\3}1{\3}}{}$,
$\subspace{1{\3}0{\3}1\\0011{\2}}{}$,
$\subspace{10{\3}00\\0101{\2}}{}$,
$\subspace{10{\2}1{\2}\\0111{\2}}{}$,
$\subspace{10{\3}{\2}{\3}\\01{\2}{\3}1}{}$,
$\subspace{1001{\2}\\01{\3}00}{}$,
$\subspace{1011{\2}\\01{\3}{\3}1}{}$,
$\subspace{10{\3}{\2}{\3}\\01{\3}{\3}1}{}$,
$\subspace{1{\3}0{\3}1\\001{\2}0}{}$,
$\subspace{10{\3}1{\3}\\010{\2}0}{}$,
$\subspace{10{\2}1{\2}\\011{\2}0}{}$,
$\subspace{10{\3}11\\01{\2}{\3}1}{}$,
$\subspace{100{\2}0\\01{\3}{\2}1}{}$,
$\subspace{101{\2}0\\01{\3}0{\3}}{}$,
$\subspace{10{\3}0{\2}\\01{\3}{\2}{\2}}{}$,
$\subspace{1{\3}0{\2}{\3}\\001{\3}{\3}}{}$,
$\subspace{10{\3}0{\3}\\010{\3}{\3}}{}$,
$\subspace{10{\2}{\3}1\\011{\3}{\3}}{}$,
$\subspace{10{\3}10\\01{\2}{\2}{\3}}{}$,
$\subspace{100{\3}{\3}\\01{\3}01}{}$,
$\subspace{101{\3}{\3}\\01{\3}{\2}1}{}$,
$\subspace{10{\3}1{\3}\\01{\3}{\2}0}{}$.

\medskip

\noindent
$n_4(5,2;10)\ge 156$:
$\subspace{001{\2}0\\00001}{}$,
$\subspace{001{\3}0\\00001}{}$,
$\subspace{01000\\00010}{}$,
$\subspace{10000\\00100}{}$,
$\subspace{01000\\00010}{}$,
$\subspace{10000\\00100}{}$,
$\subspace{0101{\2}\\00101}{}$,
$\subspace{0101{\3}\\00101}{}$,
$\subspace{010{\2}1\\0010{\2}}{}$,
$\subspace{010{\2}{\2}\\0010{\2}}{}$,
$\subspace{010{\3}1\\0010{\3}}{}$,
$\subspace{010{\3}{\3}\\0010{\3}}{}$,
$\subspace{1010{\2}\\00011}{}$,
$\subspace{1010{\3}\\00011}{}$,
$\subspace{10{\2}01\\0001{\2}}{}$,
$\subspace{10{\2}0{\2}\\0001{\2}}{}$,
$\subspace{10{\3}01\\0001{\3}}{}$,
$\subspace{10{\3}0{\3}\\0001{\3}}{}$,
$\subspace{1{\2}01{\3}\\0011{\2}}{}$,
$\subspace{1{\2}0{\2}{\2}\\0011{\3}}{}$,
$\subspace{1{\2}0{\3}1\\00111}{}$,
$\subspace{1{\3}01{\2}\\0011{\3}}{}$,
$\subspace{1{\3}0{\2}1\\00111}{}$,
$\subspace{1{\3}0{\3}{\3}\\0011{\2}}{}$,
$\subspace{1{\2}01{\2}\\001{\2}{\2}}{}$,
$\subspace{1{\2}0{\2}1\\001{\2}{\3}}{}$,
$\subspace{1{\2}0{\3}{\3}\\001{\2}1}{}$,
$\subspace{1{\3}01{\3}\\001{\3}{\3}}{}$,
$\subspace{1{\3}0{\2}{\2}\\001{\3}1}{}$,
$\subspace{1{\3}0{\3}1\\001{\3}{\2}}{}$,
$\subspace{100{\2}0\\0101{\2}}{}$,
$\subspace{100{\3}0\\0101{\3}}{}$,
$\subspace{10010\\010{\2}{\2}}{}$,
$\subspace{100{\3}0\\010{\2}1}{}$,
$\subspace{10010\\010{\3}{\3}}{}$,
$\subspace{100{\2}0\\010{\3}1}{}$,
$\subspace{10{\2}0{\2}\\01100}{}$,
$\subspace{10{\3}0{\3}\\01100}{}$,
$\subspace{1010{\2}\\01{\2}00}{}$,
$\subspace{10{\3}01\\01{\2}00}{}$,
$\subspace{1010{\3}\\01{\3}00}{}$,
$\subspace{10{\2}01\\01{\3}00}{}$,
$\subspace{101{\2}0\\01001}{}$,
$\subspace{101{\3}0\\01001}{}$,
$\subspace{10{\2}10\\0100{\3}}{}$,
$\subspace{10{\2}{\3}0\\0100{\3}}{}$,
$\subspace{10{\3}10\\0100{\2}}{}$,
$\subspace{10{\3}{\2}0\\0100{\2}}{}$,
$\subspace{1000{\3}\\011{\2}0}{}$,
$\subspace{1000{\2}\\011{\3}0}{}$,
$\subspace{10001\\01{\2}10}{}$,
$\subspace{1000{\2}\\01{\2}{\3}0}{}$,
$\subspace{10001\\01{\3}10}{}$,
$\subspace{1000{\3}\\01{\3}{\2}0}{}$,
$\subspace{10111\\01011}{}$,
$\subspace{10{\2}{\2}{\3}\\010{\2}{\3}}{}$,
$\subspace{10{\3}{\3}{\2}\\010{\3}{\2}}{}$,
$\subspace{10101\\01111}{}$,
$\subspace{10{\2}0{\3}\\01{\2}{\2}{\3}}{}$,
$\subspace{10{\3}0{\2}\\01{\3}{\3}{\2}}{}$,
$\subspace{10111\\01011}{}$,
$\subspace{10{\2}{\2}{\3}\\010{\2}{\3}}{}$,
$\subspace{10{\3}{\3}{\2}\\010{\3}{\2}}{}$,
$\subspace{10101\\01111}{}$,
$\subspace{10{\2}0{\3}\\01{\2}{\2}{\3}}{}$,
$\subspace{10{\3}0{\2}\\01{\3}{\3}{\2}}{}$,
$\subspace{101{\2}0\\010{\2}{\2}}{}$,
$\subspace{101{\3}0\\010{\3}{\3}}{}$,
$\subspace{10{\2}10\\0101{\2}}{}$,
$\subspace{10{\2}{\3}0\\010{\3}1}{}$,
$\subspace{10{\3}10\\0101{\3}}{}$,
$\subspace{10{\3}{\2}0\\010{\2}1}{}$,
$\subspace{1010{\2}\\011{\2}0}{}$,
$\subspace{1010{\3}\\011{\3}0}{}$,
$\subspace{10{\2}0{\2}\\01{\2}10}{}$,
$\subspace{10{\2}01\\01{\2}{\3}0}{}$,
$\subspace{10{\3}0{\3}\\01{\3}10}{}$,
$\subspace{10{\3}01\\01{\3}{\2}0}{}$,
$\subspace{1001{\3}\\0110{\2}}{}$,
$\subspace{1001{\2}\\0110{\3}}{}$,
$\subspace{100{\2}{\2}\\01{\2}01}{}$,
$\subspace{100{\2}1\\01{\2}0{\2}}{}$,
$\subspace{100{\3}{\3}\\01{\3}01}{}$,
$\subspace{100{\3}1\\01{\3}0{\3}}{}$,
$\subspace{100{\2}{\2}\\0110{\2}}{}$,
$\subspace{100{\3}{\3}\\0110{\3}}{}$,
$\subspace{1001{\2}\\01{\2}0{\2}}{}$,
$\subspace{100{\3}1\\01{\2}01}{}$,
$\subspace{1001{\3}\\01{\3}0{\3}}{}$,
$\subspace{100{\2}1\\01{\3}01}{}$,
$\subspace{1001{\3}\\011{\2}{\2}}{}$,
$\subspace{1001{\2}\\011{\3}{\3}}{}$,
$\subspace{10{\2}1{\2}\\0110{\3}}{}$,
$\subspace{10{\3}1{\3}\\0110{\2}}{}$,
$\subspace{100{\2}1\\01{\2}1{\2}}{}$,
$\subspace{100{\2}{\2}\\01{\2}{\3}1}{}$,
$\subspace{101{\2}{\2}\\01{\2}01}{}$,
$\subspace{10{\3}{\2}1\\01{\2}0{\2}}{}$,
$\subspace{100{\3}1\\01{\3}1{\3}}{}$,
$\subspace{100{\3}{\3}\\01{\3}{\2}1}{}$,
$\subspace{101{\3}{\3}\\01{\3}01}{}$,
$\subspace{10{\2}{\3}1\\01{\3}0{\3}}{}$,
$\subspace{100{\3}{\2}\\011{\2}1}{}$,
$\subspace{100{\2}{\3}\\011{\3}1}{}$,
$\subspace{10{\2}{\3}{\2}\\01101}{}$,
$\subspace{10{\3}{\2}{\3}\\01101}{}$,
$\subspace{100{\3}{\2}\\01{\2}1{\3}}{}$,
$\subspace{10011\\01{\2}{\3}{\3}}{}$,
$\subspace{101{\3}{\2}\\01{\2}0{\3}}{}$,
$\subspace{10{\3}11\\01{\2}0{\3}}{}$,
$\subspace{100{\2}{\3}\\01{\3}1{\2}}{}$,
$\subspace{10011\\01{\3}{\2}{\2}}{}$,
$\subspace{101{\2}{\3}\\01{\3}0{\2}}{}$,
$\subspace{10{\2}11\\01{\3}0{\2}}{}$,
$\subspace{101{\2}{\3}\\011{\2}1}{}$,
$\subspace{101{\3}{\2}\\011{\3}1}{}$,
$\subspace{10{\2}11\\01{\2}1{\3}}{}$,
$\subspace{10{\2}{\3}{\2}\\01{\2}{\3}{\3}}{}$,
$\subspace{10{\3}11\\01{\3}1{\2}}{}$,
$\subspace{10{\3}{\2}{\3}\\01{\3}{\2}{\2}}{}$,
$\subspace{101{\3}1\\011{\2}{\3}}{}$,
$\subspace{101{\2}1\\011{\3}{\2}}{}$,
$\subspace{10{\2}{\3}{\3}\\011{\3}{\2}}{}$,
$\subspace{10{\3}{\2}{\2}\\011{\2}{\3}}{}$,
$\subspace{101{\3}1\\01{\2}{\3}{\2}}{}$,
$\subspace{10{\2}{\3}{\3}\\01{\2}11}{}$,
$\subspace{10{\2}1{\3}\\01{\2}{\3}{\2}}{}$,
$\subspace{10{\3}1{\2}\\01{\2}11}{}$,
$\subspace{101{\2}1\\01{\3}{\2}{\3}}{}$,
$\subspace{10{\2}1{\3}\\01{\3}11}{}$,
$\subspace{10{\3}{\2}{\2}\\01{\3}11}{}$,
$\subspace{10{\3}1{\2}\\01{\3}{\2}{\3}}{}$,
$\subspace{10{\2}{\3}1\\01110}{}$,
$\subspace{10{\2}{\2}0\\011{\3}{\3}}{}$,
$\subspace{10{\3}{\2}1\\01110}{}$,
$\subspace{10{\3}{\3}0\\011{\2}{\2}}{}$,
$\subspace{101{\3}{\3}\\01{\2}{\2}0}{}$,
$\subspace{10110\\01{\2}{\3}1}{}$,
$\subspace{10{\3}{\3}0\\01{\2}1{\2}}{}$,
$\subspace{10{\3}1{\3}\\01{\2}{\2}0}{}$,
$\subspace{10110\\01{\3}{\2}1}{}$,
$\subspace{101{\2}{\2}\\01{\3}{\3}0}{}$,
$\subspace{10{\2}{\2}0\\01{\3}1{\3}}{}$,
$\subspace{10{\2}1{\2}\\01{\3}{\3}0}{}$,
$\subspace{10{\2}{\3}{\2}\\0111{\2}}{}$,
$\subspace{10{\2}{\2}1\\011{\3}1}{}$,
$\subspace{10{\3}{\2}{\3}\\0111{\3}}{}$,
$\subspace{10{\3}{\3}1\\011{\2}1}{}$,
$\subspace{101{\3}{\2}\\01{\2}{\2}{\2}}{}$,
$\subspace{1011{\3}\\01{\2}{\3}{\3}}{}$,
$\subspace{10{\3}{\3}{\3}\\01{\2}1{\3}}{}$,
$\subspace{10{\3}11\\01{\2}{\2}1}{}$,
$\subspace{1011{\2}\\01{\3}{\2}{\2}}{}$,
$\subspace{101{\2}{\3}\\01{\3}{\3}{\3}}{}$,
$\subspace{10{\2}{\2}{\2}\\01{\3}1{\2}}{}$,
$\subspace{10{\2}11\\01{\3}{\3}1}{}$.

\medskip

\noindent
$n_4(5,2;11)\ge 175$:
$\subspace{0101{\2}\\00101}{}$,
$\subspace{10001\\0011{\3}}{}$,
$\subspace{1101{\3}\\0011{\2}}{}$,
$\subspace{1001{\3}\\01001}{}$,
$\subspace{10101\\0101{\3}}{}$,
$\subspace{1001{\2}\\01101}{}$,
$\subspace{1011{\3}\\0111{\2}}{}$,
$\subspace{010{\2}0\\0010{\2}}{}$,
$\subspace{1000{\2}\\001{\2}{\2}}{}$,
$\subspace{110{\2}{\2}\\001{\2}0}{}$,
$\subspace{100{\2}{\2}\\0100{\2}}{}$,
$\subspace{1010{\2}\\010{\2}{\2}}{}$,
$\subspace{100{\2}0\\0110{\2}}{}$,
$\subspace{101{\2}{\2}\\011{\2}0}{}$,
$\subspace{010{\2}{\3}\\0010{\3}}{}$,
$\subspace{1000{\3}\\001{\2}0}{}$,
$\subspace{110{\2}0\\001{\2}{\3}}{}$,
$\subspace{100{\2}0\\0100{\3}}{}$,
$\subspace{1010{\3}\\010{\2}0}{}$,
$\subspace{100{\2}{\3}\\0110{\3}}{}$,
$\subspace{101{\2}0\\011{\2}{\3}}{}$,
$\subspace{010{\3}0\\0010{\2}}{}$,
$\subspace{1000{\2}\\001{\3}{\2}}{}$,
$\subspace{110{\3}{\2}\\001{\3}0}{}$,
$\subspace{100{\3}{\2}\\0100{\2}}{}$,
$\subspace{1010{\2}\\010{\3}{\2}}{}$,
$\subspace{100{\3}0\\0110{\2}}{}$,
$\subspace{101{\3}{\2}\\011{\3}0}{}$,
$\subspace{01001\\00110}{}$,
$\subspace{10010\\00111}{}$,
$\subspace{11011\\00101}{}$,
$\subspace{10011\\01010}{}$,
$\subspace{10110\\01011}{}$,
$\subspace{10001\\01110}{}$,
$\subspace{10111\\01101}{}$,
$\subspace{01010\\00111}{}$,
$\subspace{10011\\00101}{}$,
$\subspace{11001\\00110}{}$,
$\subspace{10001\\01011}{}$,
$\subspace{10111\\01001}{}$,
$\subspace{10010\\01111}{}$,
$\subspace{10101\\01110}{}$,
$\subspace{010{\2}{\2}\\0011{\2}}{}$,
$\subspace{1001{\2}\\001{\3}0}{}$,
$\subspace{110{\3}0\\001{\2}{\2}}{}$,
$\subspace{100{\3}0\\0101{\2}}{}$,
$\subspace{1011{\2}\\010{\3}0}{}$,
$\subspace{100{\2}{\2}\\0111{\2}}{}$,
$\subspace{101{\3}0\\011{\2}{\2}}{}$,
$\subspace{010{\3}1\\0011{\2}}{}$,
$\subspace{1001{\2}\\001{\2}{\3}}{}$,
$\subspace{110{\2}{\3}\\001{\3}1}{}$,
$\subspace{100{\2}{\3}\\0101{\2}}{}$,
$\subspace{1011{\2}\\010{\2}{\3}}{}$,
$\subspace{100{\3}1\\0111{\2}}{}$,
$\subspace{101{\2}{\3}\\011{\3}1}{}$,
$\subspace{0100{\3}\\001{\2}{\3}}{}$,
$\subspace{100{\2}{\3}\\001{\2}0}{}$,
$\subspace{110{\2}0\\0010{\3}}{}$,
$\subspace{100{\2}0\\010{\2}{\3}}{}$,
$\subspace{101{\2}{\3}\\010{\2}0}{}$,
$\subspace{1000{\3}\\011{\2}{\3}}{}$,
$\subspace{101{\2}0\\0110{\3}}{}$,
$\subspace{0100{\3}\\001{\3}1}{}$,
$\subspace{100{\3}1\\001{\3}{\2}}{}$,
$\subspace{110{\3}{\2}\\0010{\3}}{}$,
$\subspace{100{\3}{\2}\\010{\3}1}{}$,
$\subspace{101{\3}1\\010{\3}{\2}}{}$,
$\subspace{1000{\3}\\011{\3}1}{}$,
$\subspace{101{\3}{\2}\\0110{\3}}{}$,
$\subspace{0101{\3}\\001{\3}{\2}}{}$,
$\subspace{100{\3}{\2}\\001{\2}1}{}$,
$\subspace{110{\2}1\\0011{\3}}{}$,
$\subspace{100{\2}1\\010{\3}{\2}}{}$,
$\subspace{101{\3}{\2}\\010{\2}1}{}$,
$\subspace{1001{\3}\\011{\3}{\2}}{}$,
$\subspace{101{\2}1\\0111{\3}}{}$,
$\subspace{1{\2}00{\2}\\00100}{}$,
$\subspace{10{\2}0{\3}\\01000}{}$,
$\subspace{10{\3}01\\01100}{}$,
$\subspace{10000\\01{\2}0{\2}}{}$,
$\subspace{10100\\01{\2}0{\3}}{}$,
$\subspace{10{\2}01\\01{\2}01}{}$,
$\subspace{10{\2}0{\2}\\01{\3}0{\2}}{}$,
$\subspace{1{\2}0{\3}{\3}\\00100}{}$,
$\subspace{10{\2}11\\01000}{}$,
$\subspace{10{\3}{\2}{\2}\\01100}{}$,
$\subspace{10000\\01{\2}{\3}{\3}}{}$,
$\subspace{10100\\01{\2}11}{}$,
$\subspace{10{\2}{\2}{\2}\\01{\2}{\2}{\2}}{}$,
$\subspace{10{\2}{\3}{\3}\\01{\3}{\3}{\3}}{}$,
$\subspace{1{\2}01{\3}\\00110}{}$,
$\subspace{10{\2}11\\01010}{}$,
$\subspace{10{\3}{\3}{\2}\\01110}{}$,
$\subspace{10010\\01{\2}{\3}{\3}}{}$,
$\subspace{10110\\01{\2}{\3}1}{}$,
$\subspace{10{\2}1{\2}\\01{\2}0{\2}}{}$,
$\subspace{10{\2}0{\3}\\01{\3}1{\3}}{}$,
$\subspace{1{\2}0{\2}{\2}\\0011{\3}}{}$,
$\subspace{10{\2}01\\0101{\3}}{}$,
$\subspace{10{\3}11\\0111{\3}}{}$,
$\subspace{1001{\3}\\01{\2}0{\3}}{}$,
$\subspace{1011{\3}\\01{\2}{\2}0}{}$,
$\subspace{10{\2}{\3}0\\01{\2}{\2}{\3}}{}$,
$\subspace{10{\2}{\3}1\\01{\3}{\2}{\2}}{}$,
$\subspace{1{\2}0{\3}0\\00111}{}$,
$\subspace{10{\2}{\2}{\3}\\01011}{}$,
$\subspace{10{\3}{\2}0\\01111}{}$,
$\subspace{10011\\01{\2}1{\2}}{}$,
$\subspace{10111\\01{\2}01}{}$,
$\subspace{10{\2}0{\2}\\01{\2}1{\3}}{}$,
$\subspace{10{\2}{\2}1\\01{\3}{\3}0}{}$,
$\subspace{1{\2}00{\3}\\001{\2}{\2}}{}$,
$\subspace{10{\2}10\\010{\2}{\2}}{}$,
$\subspace{10{\3}0{\2}\\011{\2}{\2}}{}$,
$\subspace{100{\2}{\2}\\01{\2}{\3}0}{}$,
$\subspace{101{\2}{\2}\\01{\2}{\2}{\3}}{}$,
$\subspace{10{\2}{\3}1\\01{\2}1{\3}}{}$,
$\subspace{10{\2}{\2}1\\01{\3}0{\3}}{}$,
$\subspace{1{\2}001\\001{\3}{\3}}{}$,
$\subspace{10{\2}{\2}0\\010{\3}{\3}}{}$,
$\subspace{10{\3}0{\3}\\011{\3}{\3}}{}$,
$\subspace{100{\3}{\3}\\01{\2}10}{}$,
$\subspace{101{\3}{\3}\\01{\2}{\3}1}{}$,
$\subspace{10{\2}1{\2}\\01{\2}{\2}1}{}$,
$\subspace{10{\2}{\3}{\2}\\01{\3}01}{}$,
$\subspace{1{\2}001\\001{\3}{\3}}{}$,
$\subspace{10{\2}{\2}0\\010{\3}{\3}}{}$,
$\subspace{10{\3}0{\3}\\011{\3}{\3}}{}$,
$\subspace{100{\3}{\3}\\01{\2}10}{}$,
$\subspace{101{\3}{\3}\\01{\2}{\3}1}{}$,
$\subspace{10{\2}1{\2}\\01{\2}{\2}1}{}$,
$\subspace{10{\2}{\3}{\2}\\01{\3}01}{}$,
$\subspace{1{\2}00{\3}\\001{\3}{\3}}{}$,
$\subspace{10{\2}{\2}{\3}\\010{\3}{\3}}{}$,
$\subspace{10{\3}0{\2}\\011{\3}{\3}}{}$,
$\subspace{100{\3}{\3}\\01{\2}1{\2}}{}$,
$\subspace{101{\3}{\3}\\01{\2}{\3}{\2}}{}$,
$\subspace{10{\2}1{\3}\\01{\2}{\2}0}{}$,
$\subspace{10{\2}{\3}0\\01{\3}0{\3}}{}$,
$\subspace{1{\2}0{\2}0\\001{\3}0}{}$,
$\subspace{10{\2}10\\010{\3}0}{}$,
$\subspace{10{\3}10\\011{\3}0}{}$,
$\subspace{100{\3}0\\01{\2}{\3}0}{}$,
$\subspace{101{\3}0\\01{\2}00}{}$,
$\subspace{10{\2}00\\01{\2}{\3}0}{}$,
$\subspace{10{\2}10\\01{\3}{\2}0}{}$,
$\subspace{1{\2}0{\3}{\2}\\001{\3}1}{}$,
$\subspace{10{\2}{\3}0\\010{\3}1}{}$,
$\subspace{10{\3}{\2}1\\011{\3}1}{}$,
$\subspace{100{\3}1\\01{\2}{\2}0}{}$,
$\subspace{101{\3}1\\01{\2}{\2}{\2}}{}$,
$\subspace{10{\2}{\3}{\3}\\01{\2}0{\2}}{}$,
$\subspace{10{\2}0{\3}\\01{\3}{\3}{\2}}{}$,
$\subspace{1{\3}01{\3}\\0010{\2}}{}$,
$\subspace{10{\3}{\3}1\\0100{\2}}{}$,
$\subspace{10{\2}{\2}1\\0110{\2}}{}$,
$\subspace{10{\3}11\\01{\2}1{\3}}{}$,
$\subspace{1000{\2}\\01{\3}1{\2}}{}$,
$\subspace{1010{\2}\\01{\3}{\3}0}{}$,
$\subspace{10{\3}{\2}0\\01{\3}{\2}{\2}}{}$,
$\subspace{1{\3}00{\3}\\001{\2}1}{}$,
$\subspace{10{\3}{\3}0\\010{\2}1}{}$,
$\subspace{10{\2}01\\011{\2}1}{}$,
$\subspace{10{\3}{\2}{\2}\\01{\2}0{\3}}{}$,
$\subspace{100{\2}1\\01{\3}10}{}$,
$\subspace{101{\2}1\\01{\3}{\2}{\3}}{}$,
$\subspace{10{\3}1{\2}\\01{\3}{\3}{\3}}{}$,
$\subspace{1{\3}011\\001{\2}1}{}$,
$\subspace{10{\3}01\\010{\2}1}{}$,
$\subspace{10{\2}{\2}{\2}\\011{\2}1}{}$,
$\subspace{10{\3}{\3}0\\01{\2}11}{}$,
$\subspace{100{\2}1\\01{\3}0{\2}}{}$,
$\subspace{101{\2}1\\01{\3}1{\2}}{}$,
$\subspace{10{\3}{\3}1\\01{\3}10}{}$.

\medskip

\noindent
$n_4(5,2;14)\ge 222$:
$\subspace{00010\\00001}{}$,
$\subspace{00100\\00001}{}$,
$\subspace{01000\\00001}{}$,
$\subspace{10000\\00001}{}$,
$\subspace{01000\\00010}{}$,
$\subspace{10000\\00100}{}$,
$\subspace{01001\\00011}{}$,
$\subspace{0100{\3}\\0001{\2}}{}$,
$\subspace{0100{\2}\\0001{\3}}{}$,
$\subspace{10001\\00101}{}$,
$\subspace{1000{\3}\\0010{\2}}{}$,
$\subspace{1000{\2}\\0010{\3}}{}$,
$\subspace{01001\\00111}{}$,
$\subspace{0100{\3}\\0011{\2}}{}$,
$\subspace{0100{\2}\\0011{\3}}{}$,
$\subspace{10001\\00111}{}$,
$\subspace{1000{\3}\\0011{\2}}{}$,
$\subspace{1000{\2}\\0011{\3}}{}$,
$\subspace{11001\\00011}{}$,
$\subspace{1100{\3}\\0001{\2}}{}$,
$\subspace{1100{\2}\\0001{\3}}{}$,
$\subspace{11001\\00101}{}$,
$\subspace{1100{\3}\\0010{\2}}{}$,
$\subspace{1100{\2}\\0010{\3}}{}$,
$\subspace{11001\\00110}{}$,
$\subspace{1100{\2}\\00110}{}$,
$\subspace{1100{\3}\\00110}{}$,
$\subspace{11000\\00111}{}$,
$\subspace{11000\\0011{\2}}{}$,
$\subspace{11000\\0011{\3}}{}$,
$\subspace{1{\2}000\\001{\2}0}{}$,
$\subspace{1{\3}000\\001{\3}0}{}$,
$\subspace{1{\2}000\\001{\2}0}{}$,
$\subspace{1{\3}000\\001{\3}0}{}$,
$\subspace{1{\2}000\\001{\2}0}{}$,
$\subspace{1{\3}000\\001{\3}0}{}$,
$\subspace{1{\2}00{\2}\\001{\2}1}{}$,
$\subspace{1{\2}001\\001{\2}{\2}}{}$,
$\subspace{1{\2}00{\3}\\001{\2}{\3}}{}$,
$\subspace{1{\3}00{\3}\\001{\3}1}{}$,
$\subspace{1{\3}00{\2}\\001{\3}{\2}}{}$,
$\subspace{1{\3}001\\001{\3}{\3}}{}$,
$\subspace{1{\2}010\\001{\2}{\2}}{}$,
$\subspace{1{\2}011\\001{\2}{\2}}{}$,
$\subspace{1{\2}0{\2}0\\001{\2}{\3}}{}$,
$\subspace{1{\2}0{\2}{\3}\\001{\2}{\3}}{}$,
$\subspace{1{\2}0{\3}0\\001{\2}1}{}$,
$\subspace{1{\2}0{\3}{\2}\\001{\2}1}{}$,
$\subspace{1{\3}010\\001{\3}{\3}}{}$,
$\subspace{1{\3}011\\001{\3}{\3}}{}$,
$\subspace{1{\3}0{\2}0\\001{\3}1}{}$,
$\subspace{1{\3}0{\2}{\3}\\001{\3}1}{}$,
$\subspace{1{\3}0{\3}0\\001{\3}{\2}}{}$,
$\subspace{1{\3}0{\3}{\2}\\001{\3}{\2}}{}$,
$\subspace{10{\2}{\3}{\2}\\011{\2}0}{}$,
$\subspace{10{\2}{\3}0\\011{\2}1}{}$,
$\subspace{10{\3}{\2}{\3}\\011{\3}0}{}$,
$\subspace{10{\3}{\2}0\\011{\3}1}{}$,
$\subspace{101{\3}{\2}\\01{\2}10}{}$,
$\subspace{101{\3}0\\01{\2}1{\3}}{}$,
$\subspace{10{\3}11\\01{\2}{\3}0}{}$,
$\subspace{10{\3}10\\01{\2}{\3}{\3}}{}$,
$\subspace{101{\2}{\3}\\01{\3}10}{}$,
$\subspace{101{\2}0\\01{\3}1{\2}}{}$,
$\subspace{10{\2}11\\01{\3}{\2}0}{}$,
$\subspace{10{\2}10\\01{\3}{\2}{\2}}{}$,
$\subspace{1010{\2}\\01010}{}$,
$\subspace{1010{\3}\\01010}{}$,
$\subspace{10100\\0101{\2}}{}$,
$\subspace{10100\\0101{\3}}{}$,
$\subspace{10{\2}01\\010{\2}0}{}$,
$\subspace{10{\2}0{\2}\\010{\2}0}{}$,
$\subspace{10{\2}00\\010{\2}1}{}$,
$\subspace{10{\2}00\\010{\2}{\2}}{}$,
$\subspace{10{\3}01\\010{\3}0}{}$,
$\subspace{10{\3}0{\3}\\010{\3}0}{}$,
$\subspace{10{\3}00\\010{\3}1}{}$,
$\subspace{10{\3}00\\010{\3}{\3}}{}$,
$\subspace{1010{\3}\\010{\2}1}{}$,
$\subspace{1010{\2}\\010{\2}{\2}}{}$,
$\subspace{1010{\2}\\010{\3}1}{}$,
$\subspace{1010{\3}\\010{\3}{\3}}{}$,
$\subspace{10{\2}0{\2}\\0101{\2}}{}$,
$\subspace{10{\2}01\\0101{\3}}{}$,
$\subspace{10{\2}01\\010{\3}1}{}$,
$\subspace{10{\2}0{\2}\\010{\3}{\3}}{}$,
$\subspace{10{\3}01\\0101{\2}}{}$,
$\subspace{10{\3}0{\3}\\0101{\3}}{}$,
$\subspace{10{\3}01\\010{\2}1}{}$,
$\subspace{10{\3}0{\3}\\010{\2}{\2}}{}$,
$\subspace{10010\\01100}{}$,
$\subspace{100{\2}0\\01{\2}00}{}$,
$\subspace{100{\3}0\\01{\3}00}{}$,
$\subspace{10010\\01100}{}$,
$\subspace{100{\2}0\\01{\2}00}{}$,
$\subspace{100{\3}0\\01{\3}00}{}$,
$\subspace{1001{\3}\\0110{\2}}{}$,
$\subspace{1001{\2}\\0110{\3}}{}$,
$\subspace{100{\2}{\2}\\01{\2}01}{}$,
$\subspace{100{\2}1\\01{\2}0{\2}}{}$,
$\subspace{100{\3}{\3}\\01{\3}01}{}$,
$\subspace{100{\3}1\\01{\3}0{\3}}{}$,
$\subspace{10011\\011{\2}0}{}$,
$\subspace{10011\\011{\2}{\2}}{}$,
$\subspace{10011\\011{\3}0}{}$,
$\subspace{10011\\011{\3}{\3}}{}$,
$\subspace{10{\2}10\\01101}{}$,
$\subspace{10{\2}1{\2}\\01101}{}$,
$\subspace{10{\3}10\\01101}{}$,
$\subspace{10{\3}1{\3}\\01101}{}$,
$\subspace{100{\2}{\3}\\01{\2}10}{}$,
$\subspace{100{\2}{\3}\\01{\2}1{\2}}{}$,
$\subspace{100{\2}{\3}\\01{\2}{\3}0}{}$,
$\subspace{100{\2}{\3}\\01{\2}{\3}1}{}$,
$\subspace{101{\2}0\\01{\2}0{\3}}{}$,
$\subspace{101{\2}{\2}\\01{\2}0{\3}}{}$,
$\subspace{10{\3}{\2}0\\01{\2}0{\3}}{}$,
$\subspace{10{\3}{\2}1\\01{\2}0{\3}}{}$,
$\subspace{100{\3}{\2}\\01{\3}10}{}$,
$\subspace{100{\3}{\2}\\01{\3}1{\3}}{}$,
$\subspace{100{\3}{\2}\\01{\3}{\2}0}{}$,
$\subspace{100{\3}{\2}\\01{\3}{\2}1}{}$,
$\subspace{101{\3}0\\01{\3}0{\2}}{}$,
$\subspace{101{\3}{\3}\\01{\3}0{\2}}{}$,
$\subspace{10{\2}{\3}0\\01{\3}0{\2}}{}$,
$\subspace{10{\2}{\3}1\\01{\3}0{\2}}{}$,
$\subspace{1001{\3}\\011{\2}{\2}}{}$,
$\subspace{1001{\3}\\011{\2}{\3}}{}$,
$\subspace{1001{\2}\\011{\3}{\2}}{}$,
$\subspace{1001{\2}\\011{\3}{\3}}{}$,
$\subspace{10{\2}1{\2}\\0110{\3}}{}$,
$\subspace{10{\2}1{\3}\\0110{\3}}{}$,
$\subspace{10{\3}1{\2}\\0110{\2}}{}$,
$\subspace{10{\3}1{\3}\\0110{\2}}{}$,
$\subspace{100{\2}1\\01{\2}11}{}$,
$\subspace{100{\2}1\\01{\2}1{\2}}{}$,
$\subspace{100{\2}{\2}\\01{\2}{\3}1}{}$,
$\subspace{100{\2}{\2}\\01{\2}{\3}{\2}}{}$,
$\subspace{101{\2}1\\01{\2}01}{}$,
$\subspace{101{\2}{\2}\\01{\2}01}{}$,
$\subspace{10{\3}{\2}1\\01{\2}0{\2}}{}$,
$\subspace{10{\3}{\2}{\2}\\01{\2}0{\2}}{}$,
$\subspace{100{\3}1\\01{\3}11}{}$,
$\subspace{100{\3}1\\01{\3}1{\3}}{}$,
$\subspace{100{\3}{\3}\\01{\3}{\2}1}{}$,
$\subspace{100{\3}{\3}\\01{\3}{\2}{\3}}{}$,
$\subspace{101{\3}1\\01{\3}01}{}$,
$\subspace{101{\3}{\3}\\01{\3}01}{}$,
$\subspace{10{\2}{\3}1\\01{\3}0{\3}}{}$,
$\subspace{10{\2}{\3}{\3}\\01{\3}0{\3}}{}$,
$\subspace{101{\2}0\\01111}{}$,
$\subspace{101{\2}1\\0111{\3}}{}$,
$\subspace{101{\3}0\\01111}{}$,
$\subspace{101{\3}1\\0111{\2}}{}$,
$\subspace{10{\2}{\2}{\3}\\011{\2}0}{}$,
$\subspace{10{\2}{\2}1\\011{\2}{\3}}{}$,
$\subspace{10{\3}{\3}{\2}\\011{\3}0}{}$,
$\subspace{10{\3}{\3}1\\011{\3}{\2}}{}$,
$\subspace{10111\\01{\2}10}{}$,
$\subspace{1011{\3}\\01{\2}11}{}$,
$\subspace{10{\2}1{\3}\\01{\2}{\2}1}{}$,
$\subspace{10{\2}10\\01{\2}{\2}{\3}}{}$,
$\subspace{10{\2}{\3}{\3}\\01{\2}{\2}{\2}}{}$,
$\subspace{10{\2}{\3}0\\01{\2}{\2}{\3}}{}$,
$\subspace{10{\3}{\3}{\2}\\01{\2}{\3}0}{}$,
$\subspace{10{\3}{\3}{\3}\\01{\2}{\3}{\2}}{}$,
$\subspace{10111\\01{\3}10}{}$,
$\subspace{1011{\2}\\01{\3}11}{}$,
$\subspace{10{\2}{\2}{\3}\\01{\3}{\2}0}{}$,
$\subspace{10{\2}{\2}{\2}\\01{\3}{\2}{\3}}{}$,
$\subspace{10{\3}1{\2}\\01{\3}{\3}1}{}$,
$\subspace{10{\3}10\\01{\3}{\3}{\2}}{}$,
$\subspace{10{\3}{\2}0\\01{\3}{\3}{\2}}{}$,
$\subspace{10{\3}{\2}{\2}\\01{\3}{\3}{\3}}{}$,
$\subspace{1011{\2}\\011{\2}{\2}}{}$,
$\subspace{10110\\011{\2}{\3}}{}$,
$\subspace{10110\\011{\3}{\2}}{}$,
$\subspace{1011{\3}\\011{\3}{\3}}{}$,
$\subspace{10{\2}1{\3}\\01110}{}$,
$\subspace{10{\2}1{\2}\\0111{\2}}{}$,
$\subspace{10{\3}1{\2}\\01110}{}$,
$\subspace{10{\3}1{\3}\\0111{\3}}{}$,
$\subspace{101{\2}1\\01{\2}{\2}0}{}$,
$\subspace{101{\2}{\2}\\01{\2}{\2}{\2}}{}$,
$\subspace{10{\2}{\2}0\\01{\2}11}{}$,
$\subspace{10{\2}{\2}{\2}\\01{\2}1{\2}}{}$,
$\subspace{10{\2}{\2}1\\01{\2}{\3}1}{}$,
$\subspace{10{\2}{\2}0\\01{\2}{\3}{\2}}{}$,
$\subspace{10{\3}{\2}{\2}\\01{\2}{\2}0}{}$,
$\subspace{10{\3}{\2}1\\01{\2}{\2}1}{}$,
$\subspace{101{\3}1\\01{\3}{\3}0}{}$,
$\subspace{101{\3}{\3}\\01{\3}{\3}{\3}}{}$,
$\subspace{10{\2}{\3}{\3}\\01{\3}{\3}0}{}$,
$\subspace{10{\2}{\3}1\\01{\3}{\3}1}{}$,
$\subspace{10{\3}{\3}0\\01{\3}11}{}$,
$\subspace{10{\3}{\3}{\3}\\01{\3}1{\3}}{}$,
$\subspace{10{\3}{\3}1\\01{\3}{\2}1}{}$,
$\subspace{10{\3}{\3}0\\01{\3}{\2}{\3}}{}$,
$\subspace{101{\3}{\2}\\011{\2}1}{}$,
$\subspace{101{\2}{\3}\\011{\3}1}{}$,
$\subspace{10{\2}{\3}0\\0111{\3}}{}$,
$\subspace{10{\2}{\2}{\2}\\011{\3}0}{}$,
$\subspace{10{\2}{\3}{\2}\\011{\3}1}{}$,
$\subspace{10{\3}{\2}0\\0111{\2}}{}$,
$\subspace{10{\3}{\2}{\3}\\011{\2}1}{}$,
$\subspace{10{\3}{\3}{\3}\\011{\2}0}{}$,
$\subspace{101{\3}0\\01{\2}{\2}1}{}$,
$\subspace{1011{\2}\\01{\2}{\3}0}{}$,
$\subspace{101{\3}{\2}\\01{\2}{\3}{\3}}{}$,
$\subspace{10{\2}{\3}{\2}\\01{\2}1{\3}}{}$,
$\subspace{10{\2}11\\01{\2}{\3}{\3}}{}$,
$\subspace{10{\3}11\\01{\2}1{\3}}{}$,
$\subspace{10{\3}{\3}1\\01{\2}10}{}$,
$\subspace{10{\3}10\\01{\2}{\2}{\2}}{}$,
$\subspace{1011{\3}\\01{\3}{\2}0}{}$,
$\subspace{101{\2}{\3}\\01{\3}{\2}{\2}}{}$,
$\subspace{101{\2}0\\01{\3}{\3}1}{}$,
$\subspace{10{\2}11\\01{\3}1{\2}}{}$,
$\subspace{10{\2}{\2}1\\01{\3}10}{}$,
$\subspace{10{\2}10\\01{\3}{\3}{\3}}{}$,
$\subspace{10{\3}{\2}{\3}\\01{\3}1{\2}}{}$,
$\subspace{10{\3}11\\01{\3}{\2}{\2}}{}$.

\medskip

\noindent
$n_4(5,2;18)\ge 290$:
$\subspace{001{\2}0\\00001}{}$,
$\subspace{001{\3}0\\00001}{}$,
$\subspace{1{\2}000\\00001}{}$,
$\subspace{1{\3}000\\00001}{}$,
$\subspace{00101\\00010}{}$,
$\subspace{0010{\2}\\00010}{}$,
$\subspace{0010{\3}\\00010}{}$,
$\subspace{00100\\00011}{}$,
$\subspace{00100\\0001{\2}}{}$,
$\subspace{00100\\0001{\3}}{}$,
$\subspace{10001\\01000}{}$,
$\subspace{1000{\2}\\01000}{}$,
$\subspace{1000{\3}\\01000}{}$,
$\subspace{10000\\01001}{}$,
$\subspace{10000\\0100{\2}}{}$,
$\subspace{10000\\0100{\3}}{}$,
$\subspace{0100{\3}\\00101}{}$,
$\subspace{0100{\2}\\0010{\2}}{}$,
$\subspace{01001\\0010{\3}}{}$,
$\subspace{1000{\3}\\00011}{}$,
$\subspace{1000{\2}\\0001{\2}}{}$,
$\subspace{10001\\0001{\3}}{}$,
$\subspace{0101{\2}\\0010{\2}}{}$,
$\subspace{010{\2}1\\0010{\3}}{}$,
$\subspace{010{\3}{\3}\\00101}{}$,
$\subspace{1010{\2}\\0001{\2}}{}$,
$\subspace{10{\2}01\\0001{\3}}{}$,
$\subspace{10{\3}0{\3}\\00011}{}$,
$\subspace{1000{\2}\\0101{\2}}{}$,
$\subspace{10001\\010{\2}1}{}$,
$\subspace{1000{\3}\\010{\3}{\3}}{}$,
$\subspace{1010{\2}\\0100{\2}}{}$,
$\subspace{10{\2}01\\01001}{}$,
$\subspace{10{\3}0{\3}\\0100{\3}}{}$,
$\subspace{0101{\2}\\0010{\2}}{}$,
$\subspace{010{\2}1\\0010{\3}}{}$,
$\subspace{010{\3}{\3}\\00101}{}$,
$\subspace{1010{\2}\\0001{\2}}{}$,
$\subspace{10{\2}01\\0001{\3}}{}$,
$\subspace{10{\3}0{\3}\\00011}{}$,
$\subspace{1000{\2}\\0101{\2}}{}$,
$\subspace{10001\\010{\2}1}{}$,
$\subspace{1000{\3}\\010{\3}{\3}}{}$,
$\subspace{1010{\2}\\0100{\2}}{}$,
$\subspace{10{\2}01\\01001}{}$,
$\subspace{10{\3}0{\3}\\0100{\3}}{}$,
$\subspace{01000\\00110}{}$,
$\subspace{10000\\00110}{}$,
$\subspace{11000\\00010}{}$,
$\subspace{11000\\00100}{}$,
$\subspace{1{\2}011\\0011{\3}}{}$,
$\subspace{1{\2}0{\2}{\3}\\00111}{}$,
$\subspace{1{\2}0{\3}{\2}\\0011{\2}}{}$,
$\subspace{1{\3}01{\3}\\00111}{}$,
$\subspace{1{\3}0{\2}{\2}\\0011{\2}}{}$,
$\subspace{1{\3}0{\3}1\\0011{\3}}{}$,
$\subspace{101{\2}{\2}\\011{\2}1}{}$,
$\subspace{101{\3}{\2}\\011{\3}{\3}}{}$,
$\subspace{10{\2}11\\01{\2}1{\2}}{}$,
$\subspace{10{\2}{\3}1\\01{\2}{\3}{\3}}{}$,
$\subspace{10{\3}1{\3}\\01{\3}1{\2}}{}$,
$\subspace{10{\3}{\2}{\3}\\01{\3}{\2}1}{}$,
$\subspace{1{\2}01{\2}\\001{\2}{\3}}{}$,
$\subspace{1{\2}0{\2}1\\001{\2}1}{}$,
$\subspace{1{\2}0{\3}{\3}\\001{\2}{\2}}{}$,
$\subspace{1{\3}010\\001{\3}1}{}$,
$\subspace{1{\3}0{\2}0\\001{\3}{\2}}{}$,
$\subspace{1{\3}0{\3}0\\001{\3}{\3}}{}$,
$\subspace{10{\2}{\3}0\\011{\2}{\2}}{}$,
$\subspace{10{\3}{\2}1\\011{\3}0}{}$,
$\subspace{101{\3}{\3}\\01{\2}10}{}$,
$\subspace{10{\3}10\\01{\2}{\3}1}{}$,
$\subspace{101{\2}0\\01{\3}1{\3}}{}$,
$\subspace{10{\2}1{\2}\\01{\3}{\2}0}{}$,
$\subspace{1{\2}000\\001{\3}0}{}$,
$\subspace{1{\3}000\\001{\2}0}{}$,
$\subspace{1{\2}000\\001{\3}0}{}$,
$\subspace{1{\3}000\\001{\2}0}{}$,
$\subspace{1{\2}000\\001{\3}0}{}$,
$\subspace{1{\3}000\\001{\2}0}{}$,
$\subspace{1010{\3}\\01011}{}$,
$\subspace{10101\\0101{\3}}{}$,
$\subspace{10{\2}0{\3}\\010{\2}{\2}}{}$,
$\subspace{10{\2}0{\2}\\010{\2}{\3}}{}$,
$\subspace{10{\3}0{\2}\\010{\3}1}{}$,
$\subspace{10{\3}01\\010{\3}{\2}}{}$,
$\subspace{1010{\3}\\0101{\2}}{}$,
$\subspace{1010{\2}\\0101{\3}}{}$,
$\subspace{10{\2}0{\2}\\010{\2}1}{}$,
$\subspace{10{\2}01\\010{\2}{\2}}{}$,
$\subspace{10{\3}0{\3}\\010{\3}1}{}$,
$\subspace{10{\3}01\\010{\3}{\3}}{}$,
$\subspace{101{\2}{\2}\\0101{\3}}{}$,
$\subspace{101{\2}{\3}\\0101{\3}}{}$,
$\subspace{10{\2}{\3}1\\010{\2}{\2}}{}$,
$\subspace{10{\2}{\3}{\2}\\010{\2}{\2}}{}$,
$\subspace{10{\3}11\\010{\3}1}{}$,
$\subspace{10{\3}1{\3}\\010{\3}1}{}$,
$\subspace{10{\3}01\\011{\3}1}{}$,
$\subspace{10{\3}01\\011{\3}{\3}}{}$,
$\subspace{1010{\3}\\01{\2}1{\2}}{}$,
$\subspace{1010{\3}\\01{\2}1{\3}}{}$,
$\subspace{10{\2}0{\2}\\01{\3}{\2}1}{}$,
$\subspace{10{\2}0{\2}\\01{\3}{\2}{\2}}{}$,
$\subspace{101{\3}1\\01011}{}$,
$\subspace{101{\3}{\2}\\01011}{}$,
$\subspace{10{\2}11\\010{\2}{\3}}{}$,
$\subspace{10{\2}1{\3}\\010{\2}{\3}}{}$,
$\subspace{10{\3}{\2}{\2}\\010{\3}{\2}}{}$,
$\subspace{10{\3}{\2}{\3}\\010{\3}{\2}}{}$,
$\subspace{10{\2}0{\3}\\011{\2}1}{}$,
$\subspace{10{\2}0{\3}\\011{\2}{\3}}{}$,
$\subspace{10{\3}0{\2}\\01{\2}{\3}{\2}}{}$,
$\subspace{10{\3}0{\2}\\01{\2}{\3}{\3}}{}$,
$\subspace{10101\\01{\3}11}{}$,
$\subspace{10101\\01{\3}1{\2}}{}$,
$\subspace{1011{\3}\\010{\2}0}{}$,
$\subspace{101{\3}{\3}\\010{\3}0}{}$,
$\subspace{10{\2}1{\2}\\01010}{}$,
$\subspace{10{\2}{\2}{\2}\\010{\3}0}{}$,
$\subspace{10{\3}{\3}1\\01010}{}$,
$\subspace{10{\3}{\2}1\\010{\2}0}{}$,
$\subspace{10100\\011{\2}{\2}}{}$,
$\subspace{10{\2}00\\0111{\3}}{}$,
$\subspace{10{\2}00\\01{\2}{\3}1}{}$,
$\subspace{10{\3}00\\01{\2}{\2}{\2}}{}$,
$\subspace{10100\\01{\3}{\3}1}{}$,
$\subspace{10{\3}00\\01{\3}1{\3}}{}$,
$\subspace{1011{\3}\\010{\2}{\3}}{}$,
$\subspace{101{\3}1\\010{\3}{\2}}{}$,
$\subspace{10{\2}1{\3}\\01011}{}$,
$\subspace{10{\2}{\2}{\2}\\010{\3}{\2}}{}$,
$\subspace{10{\3}{\3}1\\01011}{}$,
$\subspace{10{\3}{\2}{\2}\\010{\2}{\3}}{}$,
$\subspace{10101\\011{\2}{\3}}{}$,
$\subspace{10{\2}0{\3}\\0111{\3}}{}$,
$\subspace{10{\2}0{\3}\\01{\2}{\3}{\2}}{}$,
$\subspace{10{\3}0{\2}\\01{\2}{\2}{\2}}{}$,
$\subspace{10101\\01{\3}{\3}1}{}$,
$\subspace{10{\3}0{\2}\\01{\3}11}{}$,
$\subspace{101{\2}1\\010{\2}0}{}$,
$\subspace{10111\\010{\3}0}{}$,
$\subspace{10{\2}{\2}{\3}\\01010}{}$,
$\subspace{10{\2}{\3}{\3}\\010{\3}0}{}$,
$\subspace{10{\3}1{\2}\\01010}{}$,
$\subspace{10{\3}{\3}{\2}\\010{\2}0}{}$,
$\subspace{10100\\011{\3}{\2}}{}$,
$\subspace{10{\3}00\\01111}{}$,
$\subspace{10100\\01{\2}{\2}{\3}}{}$,
$\subspace{10{\2}00\\01{\2}11}{}$,
$\subspace{10{\2}00\\01{\3}{\3}{\2}}{}$,
$\subspace{10{\3}00\\01{\3}{\2}{\3}}{}$,
$\subspace{10011\\01100}{}$,
$\subspace{10010\\01101}{}$,
$\subspace{100{\2}{\3}\\01{\2}00}{}$,
$\subspace{100{\2}0\\01{\2}0{\3}}{}$,
$\subspace{100{\3}{\2}\\01{\3}00}{}$,
$\subspace{100{\3}0\\01{\3}0{\2}}{}$,
$\subspace{1001{\2}\\01100}{}$,
$\subspace{10010\\0110{\2}}{}$,
$\subspace{100{\2}1\\01{\2}00}{}$,
$\subspace{100{\2}0\\01{\2}01}{}$,
$\subspace{100{\3}{\3}\\01{\3}00}{}$,
$\subspace{100{\3}0\\01{\3}0{\3}}{}$,
$\subspace{1001{\3}\\0110{\3}}{}$,
$\subspace{100{\2}{\2}\\01{\2}0{\2}}{}$,
$\subspace{100{\3}1\\01{\3}01}{}$,
$\subspace{1001{\3}\\0110{\3}}{}$,
$\subspace{100{\2}{\2}\\01{\2}0{\2}}{}$,
$\subspace{100{\3}1\\01{\3}01}{}$,
$\subspace{100{\3}1\\01111}{}$,
$\subspace{100{\2}{\2}\\011{\2}{\3}}{}$,
$\subspace{101{\3}1\\0110{\3}}{}$,
$\subspace{10{\2}{\2}{\3}\\0110{\3}}{}$,
$\subspace{1001{\3}\\01{\2}{\2}{\3}}{}$,
$\subspace{100{\3}1\\01{\2}{\3}{\2}}{}$,
$\subspace{10{\2}1{\3}\\01{\2}0{\2}}{}$,
$\subspace{10{\3}{\3}{\2}\\01{\2}0{\2}}{}$,
$\subspace{1001{\3}\\01{\3}11}{}$,
$\subspace{100{\2}{\2}\\01{\3}{\3}{\2}}{}$,
$\subspace{10111\\01{\3}01}{}$,
$\subspace{10{\3}{\2}{\2}\\01{\3}01}{}$,
$\subspace{10010\\011{\2}1}{}$,
$\subspace{10{\2}11\\01100}{}$,
$\subspace{100{\2}0\\01{\2}{\3}{\3}}{}$,
$\subspace{10{\3}{\2}{\3}\\01{\2}00}{}$,
$\subspace{100{\3}0\\01{\3}1{\2}}{}$,
$\subspace{101{\3}{\2}\\01{\3}00}{}$,
$\subspace{100{\3}{\3}\\011{\2}0}{}$,
$\subspace{100{\2}1\\011{\3}{\2}}{}$,
$\subspace{10{\2}{\3}{\3}\\0110{\2}}{}$,
$\subspace{10{\3}{\2}0\\0110{\2}}{}$,
$\subspace{100{\3}{\3}\\01{\2}11}{}$,
$\subspace{1001{\2}\\01{\2}{\3}0}{}$,
$\subspace{101{\3}0\\01{\2}01}{}$,
$\subspace{10{\3}1{\2}\\01{\2}01}{}$,
$\subspace{100{\2}1\\01{\3}10}{}$,
$\subspace{1001{\2}\\01{\3}{\2}{\3}}{}$,
$\subspace{101{\2}1\\01{\3}0{\3}}{}$,
$\subspace{10{\2}10\\01{\3}0{\3}}{}$,
$\subspace{10011\\011{\3}0}{}$,
$\subspace{10011\\011{\3}{\3}}{}$,
$\subspace{10{\3}10\\01101}{}$,
$\subspace{10{\3}1{\3}\\01101}{}$,
$\subspace{100{\2}{\3}\\01{\2}10}{}$,
$\subspace{100{\2}{\3}\\01{\2}1{\2}}{}$,
$\subspace{101{\2}0\\01{\2}0{\3}}{}$,
$\subspace{101{\2}{\2}\\01{\2}0{\3}}{}$,
$\subspace{100{\3}{\2}\\01{\3}{\2}0}{}$,
$\subspace{100{\3}{\2}\\01{\3}{\2}1}{}$,
$\subspace{10{\2}{\3}0\\01{\3}0{\2}}{}$,
$\subspace{10{\2}{\3}1\\01{\3}0{\2}}{}$,
$\subspace{10011\\011{\3}0}{}$,
$\subspace{10011\\011{\3}{\3}}{}$,
$\subspace{10{\3}10\\01101}{}$,
$\subspace{10{\3}1{\3}\\01101}{}$,
$\subspace{100{\2}{\3}\\01{\2}10}{}$,
$\subspace{100{\2}{\3}\\01{\2}1{\2}}{}$,
$\subspace{101{\2}0\\01{\2}0{\3}}{}$,
$\subspace{101{\2}{\2}\\01{\2}0{\3}}{}$,
$\subspace{100{\3}{\2}\\01{\3}{\2}0}{}$,
$\subspace{100{\3}{\2}\\01{\3}{\2}1}{}$,
$\subspace{10{\2}{\3}0\\01{\3}0{\2}}{}$,
$\subspace{10{\2}{\3}1\\01{\3}0{\2}}{}$,
$\subspace{101{\2}{\3}\\0111{\2}}{}$,
$\subspace{101{\2}0\\0111{\3}}{}$,
$\subspace{10{\3}{\3}1\\011{\3}0}{}$,
$\subspace{10{\3}{\3}{\3}\\011{\3}1}{}$,
$\subspace{1011{\3}\\01{\2}10}{}$,
$\subspace{1011{\2}\\01{\2}1{\3}}{}$,
$\subspace{10{\2}{\3}{\2}\\01{\2}{\2}1}{}$,
$\subspace{10{\2}{\3}0\\01{\2}{\2}{\2}}{}$,
$\subspace{10{\2}{\2}{\2}\\01{\3}{\2}0}{}$,
$\subspace{10{\2}{\2}1\\01{\3}{\2}{\2}}{}$,
$\subspace{10{\3}10\\01{\3}{\3}1}{}$,
$\subspace{10{\3}11\\01{\3}{\3}{\3}}{}$,
$\subspace{101{\3}0\\01111}{}$,
$\subspace{101{\3}1\\0111{\2}}{}$,
$\subspace{10{\2}{\2}{\3}\\011{\2}0}{}$,
$\subspace{10{\2}{\2}1\\011{\2}{\3}}{}$,
$\subspace{10{\2}1{\3}\\01{\2}{\2}1}{}$,
$\subspace{10{\2}10\\01{\2}{\2}{\3}}{}$,
$\subspace{10{\3}{\3}{\2}\\01{\2}{\3}0}{}$,
$\subspace{10{\3}{\3}{\3}\\01{\2}{\3}{\2}}{}$,
$\subspace{10111\\01{\3}10}{}$,
$\subspace{1011{\2}\\01{\3}11}{}$,
$\subspace{10{\3}{\2}0\\01{\3}{\3}{\2}}{}$,
$\subspace{10{\3}{\2}{\2}\\01{\3}{\3}{\3}}{}$,
$\subspace{10111\\011{\2}1}{}$,
$\subspace{10110\\011{\2}{\2}}{}$,
$\subspace{10{\2}1{\2}\\01110}{}$,
$\subspace{10{\2}11\\01111}{}$,
$\subspace{10{\2}{\2}0\\01{\2}{\3}1}{}$,
$\subspace{10{\2}{\2}{\3}\\01{\2}{\3}{\3}}{}$,
$\subspace{10{\3}{\2}1\\01{\2}{\2}0}{}$,
$\subspace{10{\3}{\2}{\3}\\01{\2}{\2}{\3}}{}$,
$\subspace{101{\3}{\3}\\01{\3}{\3}0}{}$,
$\subspace{101{\3}{\2}\\01{\3}{\3}{\2}}{}$,
$\subspace{10{\3}{\3}{\2}\\01{\3}1{\2}}{}$,
$\subspace{10{\3}{\3}0\\01{\3}1{\3}}{}$,
$\subspace{101{\3}0\\011{\2}{\3}}{}$,
$\subspace{10{\2}{\3}{\2}\\0111{\3}}{}$,
$\subspace{10{\2}{\2}{\2}\\011{\3}1}{}$,
$\subspace{10{\3}{\2}{\2}\\011{\2}0}{}$,
$\subspace{101{\3}1\\01{\2}{\3}0}{}$,
$\subspace{10{\2}10\\01{\2}{\3}{\2}}{}$,
$\subspace{10{\3}{\3}1\\01{\2}1{\3}}{}$,
$\subspace{10{\3}11\\01{\2}{\2}{\2}}{}$,
$\subspace{1011{\3}\\01{\3}{\2}{\2}}{}$,
$\subspace{101{\2}{\3}\\01{\3}{\3}1}{}$,
$\subspace{10{\2}1{\3}\\01{\3}10}{}$,
$\subspace{10{\3}{\2}0\\01{\3}11}{}$,
$\subspace{1011{\2}\\011{\3}{\3}}{}$,
$\subspace{10{\3}1{\3}\\0111{\2}}{}$,
$\subspace{101{\2}{\2}\\01{\2}{\2}1}{}$,
$\subspace{10{\2}{\2}1\\01{\2}1{\2}}{}$,
$\subspace{10{\2}{\3}1\\01{\3}{\3}{\3}}{}$,
$\subspace{10{\3}{\3}{\3}\\01{\3}{\2}1}{}$,
$\subspace{101{\2}{\3}\\011{\3}1}{}$,
$\subspace{10{\2}{\3}{\2}\\011{\3}1}{}$,
$\subspace{10{\3}{\2}0\\0111{\2}}{}$,
$\subspace{10{\3}{\3}{\3}\\011{\2}0}{}$,
$\subspace{101{\3}0\\01{\2}{\2}1}{}$,
$\subspace{1011{\2}\\01{\2}{\3}0}{}$,
$\subspace{10{\2}{\3}{\2}\\01{\2}1{\3}}{}$,
$\subspace{10{\3}11\\01{\2}1{\3}}{}$,
$\subspace{101{\2}{\3}\\01{\3}{\2}{\2}}{}$,
$\subspace{10{\2}{\2}1\\01{\3}10}{}$,
$\subspace{10{\2}10\\01{\3}{\3}{\3}}{}$,
$\subspace{10{\3}11\\01{\3}{\2}{\2}}{}$.

\medskip

\noindent
$n_4(5,2;19)\ge 307$:
$\subspace{00100\\00010}{}$,
$\subspace{00100\\00010}{}$,
$\subspace{00100\\00010}{}$,
$\subspace{01000\\00001}{}$,
$\subspace{10000\\00001}{}$,
$\subspace{01000\\00011}{}$,
$\subspace{01000\\0001{\2}}{}$,
$\subspace{01000\\0001{\3}}{}$,
$\subspace{10000\\00101}{}$,
$\subspace{10000\\0010{\2}}{}$,
$\subspace{10000\\0010{\3}}{}$,
$\subspace{0100{\2}\\00101}{}$,
$\subspace{0100{\3}\\00101}{}$,
$\subspace{01001\\0010{\2}}{}$,
$\subspace{0100{\2}\\0010{\2}}{}$,
$\subspace{01001\\0010{\3}}{}$,
$\subspace{0100{\3}\\0010{\3}}{}$,
$\subspace{1000{\2}\\00011}{}$,
$\subspace{1000{\3}\\00011}{}$,
$\subspace{10001\\0001{\2}}{}$,
$\subspace{1000{\2}\\0001{\2}}{}$,
$\subspace{10001\\0001{\3}}{}$,
$\subspace{1000{\3}\\0001{\3}}{}$,
$\subspace{0101{\3}\\001{\2}1}{}$,
$\subspace{010{\2}{\2}\\001{\2}{\2}}{}$,
$\subspace{010{\3}1\\001{\2}{\3}}{}$,
$\subspace{0101{\2}\\001{\3}1}{}$,
$\subspace{010{\2}1\\001{\3}{\2}}{}$,
$\subspace{010{\3}{\3}\\001{\3}{\3}}{}$,
$\subspace{10010\\001{\2}1}{}$,
$\subspace{100{\2}0\\001{\2}{\2}}{}$,
$\subspace{100{\3}0\\001{\2}{\3}}{}$,
$\subspace{10010\\001{\3}1}{}$,
$\subspace{100{\2}0\\001{\3}{\2}}{}$,
$\subspace{100{\3}0\\001{\3}{\3}}{}$,
$\subspace{11000\\00001}{}$,
$\subspace{11000\\00111}{}$,
$\subspace{11000\\0011{\2}}{}$,
$\subspace{11000\\0011{\3}}{}$,
$\subspace{11010\\00111}{}$,
$\subspace{11011\\00111}{}$,
$\subspace{110{\2}0\\0011{\2}}{}$,
$\subspace{110{\2}{\3}\\0011{\2}}{}$,
$\subspace{110{\3}0\\0011{\3}}{}$,
$\subspace{110{\3}{\2}\\0011{\3}}{}$,
$\subspace{1100{\3}\\001{\2}1}{}$,
$\subspace{1100{\2}\\001{\2}{\2}}{}$,
$\subspace{11001\\001{\2}{\3}}{}$,
$\subspace{1100{\2}\\001{\3}1}{}$,
$\subspace{11001\\001{\3}{\2}}{}$,
$\subspace{1100{\3}\\001{\3}{\3}}{}$,
$\subspace{10000\\01000}{}$,
$\subspace{10011\\01001}{}$,
$\subspace{100{\2}{\3}\\0100{\3}}{}$,
$\subspace{100{\3}{\2}\\0100{\2}}{}$,
$\subspace{10001\\01101}{}$,
$\subspace{1000{\3}\\01{\2}0{\3}}{}$,
$\subspace{1000{\2}\\01{\3}0{\2}}{}$,
$\subspace{10011\\01001}{}$,
$\subspace{100{\2}{\3}\\0100{\3}}{}$,
$\subspace{100{\3}{\2}\\0100{\2}}{}$,
$\subspace{10001\\01101}{}$,
$\subspace{1000{\3}\\01{\2}0{\3}}{}$,
$\subspace{1000{\2}\\01{\3}0{\2}}{}$,
$\subspace{10100\\01010}{}$,
$\subspace{10{\2}00\\010{\2}0}{}$,
$\subspace{10{\3}00\\010{\3}0}{}$,
$\subspace{10100\\01010}{}$,
$\subspace{10{\2}00\\010{\2}0}{}$,
$\subspace{10{\3}00\\010{\3}0}{}$,
$\subspace{10100\\01010}{}$,
$\subspace{10{\2}00\\010{\2}0}{}$,
$\subspace{10{\3}00\\010{\3}0}{}$,
$\subspace{101{\2}1\\01011}{}$,
$\subspace{101{\3}1\\01011}{}$,
$\subspace{10{\2}1{\3}\\010{\2}{\3}}{}$,
$\subspace{10{\2}{\3}{\3}\\010{\2}{\3}}{}$,
$\subspace{10{\3}1{\2}\\010{\3}{\2}}{}$,
$\subspace{10{\3}{\2}{\2}\\010{\3}{\2}}{}$,
$\subspace{10{\2}0{\3}\\011{\2}{\3}}{}$,
$\subspace{10{\3}0{\2}\\011{\3}{\2}}{}$,
$\subspace{10101\\01{\2}11}{}$,
$\subspace{10{\3}0{\2}\\01{\2}{\3}{\2}}{}$,
$\subspace{10101\\01{\3}11}{}$,
$\subspace{10{\2}0{\3}\\01{\3}{\2}{\3}}{}$,
$\subspace{101{\2}{\3}\\0101{\2}}{}$,
$\subspace{101{\3}{\2}\\0101{\3}}{}$,
$\subspace{10{\2}11\\010{\2}{\2}}{}$,
$\subspace{10{\2}{\3}{\2}\\010{\2}1}{}$,
$\subspace{10{\3}11\\010{\3}{\3}}{}$,
$\subspace{10{\3}{\2}{\3}\\010{\3}1}{}$,
$\subspace{10{\2}0{\2}\\011{\2}1}{}$,
$\subspace{10{\3}0{\3}\\011{\3}1}{}$,
$\subspace{1010{\2}\\01{\2}1{\3}}{}$,
$\subspace{10{\3}01\\01{\2}{\3}{\3}}{}$,
$\subspace{1010{\3}\\01{\3}1{\2}}{}$,
$\subspace{10{\2}01\\01{\3}{\2}{\2}}{}$,
$\subspace{1010{\2}\\010{\2}{\2}}{}$,
$\subspace{1010{\3}\\010{\3}{\3}}{}$,
$\subspace{10{\2}0{\2}\\0101{\2}}{}$,
$\subspace{10{\2}01\\010{\3}1}{}$,
$\subspace{10{\3}0{\3}\\0101{\3}}{}$,
$\subspace{10{\3}01\\010{\2}1}{}$,
$\subspace{1011{\3}\\010{\2}1}{}$,
$\subspace{1011{\2}\\010{\3}1}{}$,
$\subspace{10{\2}{\2}1\\0101{\3}}{}$,
$\subspace{10{\2}{\2}{\2}\\010{\3}{\3}}{}$,
$\subspace{10{\3}{\3}1\\0101{\2}}{}$,
$\subspace{10{\3}{\3}{\3}\\010{\2}{\2}}{}$,
$\subspace{10{\2}01\\0111{\3}}{}$,
$\subspace{10{\3}01\\0111{\2}}{}$,
$\subspace{1010{\3}\\01{\2}{\2}1}{}$,
$\subspace{10{\3}0{\3}\\01{\2}{\2}{\2}}{}$,
$\subspace{1010{\2}\\01{\3}{\3}1}{}$,
$\subspace{10{\2}0{\2}\\01{\3}{\3}{\3}}{}$,
$\subspace{101{\2}{\3}\\010{\2}0}{}$,
$\subspace{101{\3}{\2}\\010{\3}0}{}$,
$\subspace{10{\2}11\\01010}{}$,
$\subspace{10{\2}{\3}{\2}\\010{\3}0}{}$,
$\subspace{10{\3}11\\01010}{}$,
$\subspace{10{\3}{\2}{\3}\\010{\2}0}{}$,
$\subspace{10100\\011{\2}1}{}$,
$\subspace{10100\\011{\3}1}{}$,
$\subspace{10{\2}00\\01{\2}1{\3}}{}$,
$\subspace{10{\2}00\\01{\2}{\3}{\3}}{}$,
$\subspace{10{\3}00\\01{\3}1{\2}}{}$,
$\subspace{10{\3}00\\01{\3}{\2}{\2}}{}$,
$\subspace{101{\3}1\\010{\2}{\3}}{}$,
$\subspace{101{\2}1\\010{\3}{\2}}{}$,
$\subspace{10{\2}{\3}{\3}\\01011}{}$,
$\subspace{10{\2}1{\3}\\010{\3}{\2}}{}$,
$\subspace{10{\3}{\2}{\2}\\01011}{}$,
$\subspace{10{\3}1{\2}\\010{\2}{\3}}{}$,
$\subspace{10{\2}0{\3}\\011{\3}{\2}}{}$,
$\subspace{10{\3}0{\2}\\011{\2}{\3}}{}$,
$\subspace{10101\\01{\2}{\3}{\2}}{}$,
$\subspace{10{\3}0{\2}\\01{\2}11}{}$,
$\subspace{10101\\01{\3}{\2}{\3}}{}$,
$\subspace{10{\2}0{\3}\\01{\3}11}{}$,
$\subspace{10011\\01100}{}$,
$\subspace{10010\\01101}{}$,
$\subspace{100{\2}{\3}\\01{\2}00}{}$,
$\subspace{100{\2}0\\01{\2}0{\3}}{}$,
$\subspace{100{\3}{\2}\\01{\3}00}{}$,
$\subspace{100{\3}0\\01{\3}0{\2}}{}$,
$\subspace{1001{\3}\\0110{\2}}{}$,
$\subspace{1001{\2}\\0110{\3}}{}$,
$\subspace{100{\2}{\2}\\01{\2}01}{}$,
$\subspace{100{\2}1\\01{\2}0{\2}}{}$,
$\subspace{100{\3}{\3}\\01{\3}01}{}$,
$\subspace{100{\3}1\\01{\3}0{\3}}{}$,
$\subspace{100{\2}{\2}\\01101}{}$,
$\subspace{100{\2}{\3}\\0110{\2}}{}$,
$\subspace{100{\3}{\3}\\01101}{}$,
$\subspace{100{\3}{\2}\\0110{\3}}{}$,
$\subspace{10011\\01{\2}0{\2}}{}$,
$\subspace{1001{\2}\\01{\2}0{\3}}{}$,
$\subspace{100{\3}{\2}\\01{\2}01}{}$,
$\subspace{100{\3}1\\01{\2}0{\3}}{}$,
$\subspace{1001{\3}\\01{\3}0{\2}}{}$,
$\subspace{10011\\01{\3}0{\3}}{}$,
$\subspace{100{\2}{\3}\\01{\3}01}{}$,
$\subspace{100{\2}1\\01{\3}0{\2}}{}$,
$\subspace{100{\2}1\\011{\2}0}{}$,
$\subspace{100{\3}1\\011{\3}0}{}$,
$\subspace{101{\2}0\\0110{\3}}{}$,
$\subspace{101{\3}0\\0110{\2}}{}$,
$\subspace{1001{\3}\\01{\2}10}{}$,
$\subspace{100{\3}{\3}\\01{\2}{\3}0}{}$,
$\subspace{10{\2}10\\01{\2}01}{}$,
$\subspace{10{\2}{\3}0\\01{\2}0{\2}}{}$,
$\subspace{1001{\2}\\01{\3}10}{}$,
$\subspace{100{\2}{\2}\\01{\3}{\2}0}{}$,
$\subspace{10{\3}10\\01{\3}01}{}$,
$\subspace{10{\3}{\2}0\\01{\3}0{\3}}{}$,
$\subspace{10111\\01110}{}$,
$\subspace{10110\\01111}{}$,
$\subspace{10{\2}{\2}{\3}\\01{\2}{\2}0}{}$,
$\subspace{10{\2}{\2}0\\01{\2}{\2}{\3}}{}$,
$\subspace{10{\3}{\3}{\2}\\01{\3}{\3}0}{}$,
$\subspace{10{\3}{\3}0\\01{\3}{\3}{\2}}{}$,
$\subspace{101{\2}{\3}\\01110}{}$,
$\subspace{101{\3}{\2}\\01110}{}$,
$\subspace{10{\2}{\2}0\\011{\2}1}{}$,
$\subspace{10{\3}{\3}0\\011{\3}1}{}$,
$\subspace{10110\\01{\2}1{\3}}{}$,
$\subspace{10{\2}11\\01{\2}{\2}0}{}$,
$\subspace{10{\2}{\3}{\2}\\01{\2}{\2}0}{}$,
$\subspace{10{\3}{\3}0\\01{\2}{\3}{\3}}{}$,
$\subspace{10110\\01{\3}1{\2}}{}$,
$\subspace{10{\2}{\2}0\\01{\3}{\2}{\2}}{}$,
$\subspace{10{\3}11\\01{\3}{\3}0}{}$,
$\subspace{10{\3}{\2}{\3}\\01{\3}{\3}0}{}$,
$\subspace{101{\2}{\3}\\0111{\3}}{}$,
$\subspace{101{\3}{\2}\\0111{\2}}{}$,
$\subspace{10{\2}{\2}1\\011{\2}1}{}$,
$\subspace{10{\3}{\3}1\\011{\3}1}{}$,
$\subspace{1011{\3}\\01{\2}1{\3}}{}$,
$\subspace{10{\2}11\\01{\2}{\2}1}{}$,
$\subspace{10{\2}{\3}{\2}\\01{\2}{\2}{\2}}{}$,
$\subspace{10{\3}{\3}{\3}\\01{\2}{\3}{\3}}{}$,
$\subspace{1011{\2}\\01{\3}1{\2}}{}$,
$\subspace{10{\2}{\2}{\2}\\01{\3}{\2}{\2}}{}$,
$\subspace{10{\3}11\\01{\3}{\3}1}{}$,
$\subspace{10{\3}{\2}{\3}\\01{\3}{\3}{\3}}{}$,
$\subspace{1011{\2}\\011{\2}1}{}$,
$\subspace{1011{\3}\\011{\3}1}{}$,
$\subspace{10{\2}11\\0111{\2}}{}$,
$\subspace{10{\3}11\\0111{\3}}{}$,
$\subspace{101{\2}{\3}\\01{\2}{\2}{\2}}{}$,
$\subspace{10{\2}{\2}{\2}\\01{\2}1{\3}}{}$,
$\subspace{10{\2}{\2}1\\01{\2}{\3}{\3}}{}$,
$\subspace{10{\3}{\2}{\3}\\01{\2}{\2}1}{}$,
$\subspace{101{\3}{\2}\\01{\3}{\3}{\3}}{}$,
$\subspace{10{\2}{\3}{\2}\\01{\3}{\3}1}{}$,
$\subspace{10{\3}{\3}{\3}\\01{\3}1{\2}}{}$,
$\subspace{10{\3}{\3}1\\01{\3}{\2}{\2}}{}$,
$\subspace{10110\\011{\2}{\2}}{}$,
$\subspace{10110\\011{\3}{\3}}{}$,
$\subspace{10{\2}1{\2}\\01110}{}$,
$\subspace{10{\3}1{\3}\\01110}{}$,
$\subspace{101{\2}{\2}\\01{\2}{\2}0}{}$,
$\subspace{10{\2}{\2}0\\01{\2}1{\2}}{}$,
$\subspace{10{\2}{\2}0\\01{\2}{\3}1}{}$,
$\subspace{10{\3}{\2}1\\01{\2}{\2}0}{}$,
$\subspace{101{\3}{\3}\\01{\3}{\3}0}{}$,
$\subspace{10{\2}{\3}1\\01{\3}{\3}0}{}$,
$\subspace{10{\3}{\3}0\\01{\3}1{\3}}{}$,
$\subspace{10{\3}{\3}0\\01{\3}{\2}1}{}$,
$\subspace{10111\\011{\2}{\2}}{}$,
$\subspace{10111\\011{\3}{\3}}{}$,
$\subspace{10{\2}1{\2}\\01111}{}$,
$\subspace{10{\3}1{\3}\\01111}{}$,
$\subspace{101{\2}{\2}\\01{\2}{\2}{\3}}{}$,
$\subspace{10{\2}{\2}{\3}\\01{\2}1{\2}}{}$,
$\subspace{10{\2}{\2}{\3}\\01{\2}{\3}1}{}$,
$\subspace{10{\3}{\2}1\\01{\2}{\2}{\3}}{}$,
$\subspace{101{\3}{\3}\\01{\3}{\3}{\2}}{}$,
$\subspace{10{\2}{\3}1\\01{\3}{\3}{\2}}{}$,
$\subspace{10{\3}{\3}{\2}\\01{\3}1{\3}}{}$,
$\subspace{10{\3}{\3}{\2}\\01{\3}{\2}1}{}$,
$\subspace{101{\3}1\\011{\2}{\2}}{}$,
$\subspace{101{\2}1\\011{\3}{\3}}{}$,
$\subspace{10{\2}{\3}1\\011{\3}{\2}}{}$,
$\subspace{10{\3}{\2}1\\011{\2}{\3}}{}$,
$\subspace{101{\3}{\3}\\01{\2}{\3}{\2}}{}$,
$\subspace{10{\2}{\3}{\3}\\01{\2}1{\2}}{}$,
$\subspace{10{\2}1{\3}\\01{\2}{\3}1}{}$,
$\subspace{10{\3}1{\3}\\01{\2}11}{}$,
$\subspace{101{\2}{\2}\\01{\3}{\2}{\3}}{}$,
$\subspace{10{\2}1{\2}\\01{\3}11}{}$,
$\subspace{10{\3}{\2}{\2}\\01{\3}1{\3}}{}$,
$\subspace{10{\3}1{\2}\\01{\3}{\2}1}{}$,
$\subspace{10{\2}{\2}{\2}\\0111{\2}}{}$,
$\subspace{10{\3}{\3}{\3}\\0111{\3}}{}$,
$\subspace{1011{\2}\\01{\2}{\2}{\2}}{}$,
$\subspace{10{\3}{\3}1\\01{\2}{\2}1}{}$,
$\subspace{1011{\3}\\01{\3}{\3}{\3}}{}$,
$\subspace{10{\2}{\2}1\\01{\3}{\3}1}{}$,
$\subspace{10{\2}{\3}0\\01111}{}$,
$\subspace{10{\2}{\2}{\3}\\011{\3}0}{}$,
$\subspace{10{\3}{\2}0\\01111}{}$,
$\subspace{10{\3}{\3}{\2}\\011{\2}0}{}$,
$\subspace{101{\3}0\\01{\2}{\2}{\3}}{}$,
$\subspace{10111\\01{\2}{\3}0}{}$,
$\subspace{10{\3}{\3}{\2}\\01{\2}10}{}$,
$\subspace{10{\3}10\\01{\2}{\2}{\3}}{}$,
$\subspace{10111\\01{\3}{\2}0}{}$,
$\subspace{101{\2}0\\01{\3}{\3}{\2}}{}$,
$\subspace{10{\2}{\2}{\3}\\01{\3}10}{}$,
$\subspace{10{\2}10\\01{\3}{\3}{\2}}{}$,
$\subspace{10{\2}1{\3}\\011{\2}0}{}$,
$\subspace{10{\2}10\\011{\2}{\3}}{}$,
$\subspace{10{\3}1{\2}\\011{\3}0}{}$,
$\subspace{10{\3}10\\011{\3}{\2}}{}$,
$\subspace{101{\2}1\\01{\2}10}{}$,
$\subspace{101{\2}0\\01{\2}11}{}$,
$\subspace{10{\3}{\2}{\2}\\01{\2}{\3}0}{}$,
$\subspace{10{\3}{\2}0\\01{\2}{\3}{\2}}{}$,
$\subspace{101{\3}1\\01{\3}10}{}$,
$\subspace{101{\3}0\\01{\3}11}{}$,
$\subspace{10{\2}{\3}{\3}\\01{\3}{\2}0}{}$,
$\subspace{10{\2}{\3}0\\01{\3}{\2}{\3}}{}$,
$\subspace{10{\2}{\3}0\\011{\2}0}{}$,
$\subspace{10{\3}{\2}0\\011{\3}0}{}$,
$\subspace{101{\3}0\\01{\2}10}{}$,
$\subspace{10{\3}10\\01{\2}{\3}0}{}$,
$\subspace{101{\2}0\\01{\3}10}{}$,
$\subspace{10{\2}10\\01{\3}{\2}0}{}$,
$\subspace{10{\2}1{\2}\\011{\3}0}{}$,
$\subspace{10{\2}10\\011{\3}{\3}}{}$,
$\subspace{10{\3}1{\3}\\011{\2}0}{}$,
$\subspace{10{\3}10\\011{\2}{\2}}{}$,
$\subspace{101{\2}{\2}\\01{\2}{\3}0}{}$,
$\subspace{101{\2}0\\01{\2}{\3}1}{}$,
$\subspace{10{\3}{\2}1\\01{\2}10}{}$,
$\subspace{10{\3}{\2}0\\01{\2}1{\2}}{}$,
$\subspace{101{\3}{\3}\\01{\3}{\2}0}{}$,
$\subspace{101{\3}0\\01{\3}{\2}1}{}$,
$\subspace{10{\2}{\3}1\\01{\3}10}{}$,
$\subspace{10{\2}{\3}0\\01{\3}1{\3}}{}$,
$\subspace{10{\2}1{\2}\\011{\3}{\3}}{}$,
$\subspace{10{\3}1{\3}\\011{\2}{\2}}{}$,
$\subspace{101{\2}{\2}\\01{\2}{\3}1}{}$,
$\subspace{10{\3}{\2}1\\01{\2}1{\2}}{}$,
$\subspace{101{\3}{\3}\\01{\3}{\2}1}{}$,
$\subspace{10{\2}{\3}1\\01{\3}1{\3}}{}$.

\bigskip
\bigskip

\noindent
$n_5(5,2;3)\ge 50$:
$\left(
\right)$.


\pagebreak

\section{Generalized Hamming weights and geometric counterparts}
\label{sec_generalized_weights}

The distance $d$ we defined in the introduction, i.e.\ the number of positions where two codewords differ, is called the \emph{Hamming distance}. For a single codeword $c$ in a code $C$ with a finite field $\F_q$ as alphabet we can define its \emph{Hamming weight} $\wt(c)$ as the number of non-zero positions. With this, we have $d\!\left(c_1,c_2\right)=\wt\!\left(c_1-c_2\right)$ if $c_1,c_2,c_1-c_2\in C$. So, especially for linear or additive codes the occurring distances are given by the occurring weights. The Hamming weight turns $\F_q^n$ into a normed vector space.\footnote{The notion of Hamming weight also exists for general matroids, see e.g.\ \cite{johnsen2013hamming}.} For $c=\left(c_1,\dots,c_n\right)\in\F_q^n$ we call
\begin{equation}
  \supp(c):=\left\{1\le i\le n\,:\, c_i\neq 0\right\}
\end{equation}
the \emph{support} of $c$, so that $\wt(c)=|\supp(c)|$. For some linear subspace $C$ in $\F_q^n$ let
\begin{equation}
  \supp(C):=\left\{1\le i\le n\,:\, \exists c=\left(c_1,\dots,c_n\right)\in C, c_i\neq 0\right\}
\end{equation}
be the \emph{support} of $C$ and $\dim(C)$ its $\F_q$-dimension. For two
$\F_q$ vector spaces $C$, $C'$ in $\F_q^n$ we write $C\le C'$ is $C$ is contained in $C'$. With this, the \emph{$f$th generalized Hamming weight} of
a linear code $C$ \cite{helleseth1977weight,klove1978weight}, denoted as $d_f(C)$, is the size of the
smallest support of an $f$-dimensional subcode of $C$ , i.e.\
\begin{equation}
   d_f(C):= \min\!\left\{|\supp(C')|\,:\, C'\le C, \dim(C')=f\right\}.
\end{equation}
I.e.\ $d_1(C)$ is the minimum Hamming distance of a linear code $C$. The sequence $\left(d_1(C),\dots,d_k(C)\right)$ is called the \emph{weight hierarchy} of a linear $[n,k]_q$ code $C$. Clearly, we have $1\le d_1(C)\le \dots\le d_k(C)\le n$. The generalized Hamming weights can be used to describe the cryptography performance of a linear code over the wire-tap channel of type~II \cite{wei1991generalized} and to determine the trellis complexity of the code \cite{chen2001trellis,forney1994density,forney1994dimension,kasami1993optimum}. The weight hierarchy of a linear code can be obtained from a quadratic form over a finite field \cite{li2020weight,li2022weight,liu2023generalized}. Also the geometric reformulation of the generalized Hamming weights in terms of multisets of points is well known \cite{helleseth1992generalized,tsfasman1995geometric}. Let $\cM$ be a multiset of points in $\PG(k-1,q)$ and $C$ its corresponding $[n,k]_q$ code. Then, we have
\begin{equation}
  \label{eq_gen_ham_geometry}
  m-d_f(C) =\max\left\{\cM(U)\,:\, U\text{ subspace of codimension }f \right\}
\end{equation}
for all $1\le f\le k$. In order to keep the paper self-contained we state a brief argument and generalize to additive codes,  c.f.\ \cite{ball2025additive}.\footnote{Generalized Hamming weights have been also defined for so-called almost affine codes \cite{simonis1998almost}. Many properties for linear codes carry over, see e.g.\ \cite{johnsen2017generalized}. For generalized weights for codes over rings we refer e.g.\ to \cite{dougherty2002generalized}.}
Given a linear $[n,k,d]_q$ code $C$, a codeword (a $1$-dimensional subcode) of $C$ is obtained by left multiplication of a generator matrix $G$ by a vector $v\in\F_q^k$. Considering $v$ as a point in $\PG(k-1, q)$, the hyperplane $v^\perp$
contains the point $x$ iff $\langle v,x\rangle=0$. Viewing the set of columns of $G$ as a multiset of points in $\PG(k-1,q)$, thus have that the codeword $vG$ has weight $w$ iff $n-w$ points are contained in the hyperplane $v^\perp$. More generally, for a $j$-dimensional subspace $V$ of $\F_q^n$ the codimension $j$ subspace $V^\perp$ contains $n-w$ points if
the subspace $\left\{vG\,:\, v\in V\right\}$ has support $w$, which proves equation~(\ref{eq_gen_ham_geometry}). We can apply the same argument to the subfield generator matrix $\widetilde{G}$ of an additive $[n,r/h,d]_q^h$ code $C$ to conclude
\begin{equation}
  \label{eq_gen_ham_geometry_add}
  n-d_f(C) =\max\left\{\left|\left\{S\in\cX_G(C)\,:\, S\le U\right\}\right|\,:\, U\text{ subspace of codimension }f \right\}
\end{equation} for all $1\le f\le k$. Given Equation~(\ref{eq_gen_ham_geometry_add}), we generalize the notion of a projective system from Definition~\ref{def_system}:
\begin{definition}
  \label{def_system_gen}
  A projective $(h,f)-(n,r,s)_q$ system, where $h+f\le r$, is a multiset
  $\cS$ of $n$ subspaces of $\PG(r-1, q)$ of dimension at most $h$ such
  that each subspace of codimension $f$ contains at most $s$ elements
  of $\cS$ and some subspace of codimension $f$ contains exactly $s$
  elements of $\cS$. We say that $\cS$ is faithful if all elements have
  dimension $h$. A projective $(h,f)-(n,r,s)_q$ system $\cS$ is a projective
  $(h,f)-(n,r,s,\mu)_q$ system if each $f$-space is contained
  in at most $\mu$ elements from $\cS$ and there is some $f$-space that is
  contained in exactly  $\mu$ elements from $\cS$.
\end{definition}
So, a projective $(h,1)-(n,r,s)_q$ system is just a projective $h-(n,r,s)_q$ system and a general projective $(h,f)-(n,r,s)_q$ system corresponds to an additive $[n,r/h,d_f]_q^h$ code $C$ with $s=n-d_f$, where $d_f$ denotes the minimum  $f$th generalized Hamming weight of $C$. By $\qbin{r}{s}{q}$ we denote the number of $s$-spaces in $\PG(r-1,q)$, i.e.\
\begin{equation}
  \qbin{r}{s}{q}=\prod_{i=0}^{s-1} \frac{q^{r-i}-1}{q^{s-i}-1}
  = \prod_{i=0}^{s-1} \frac{[r-i]_q}{[s-i]_q},
\end{equation}
and by $n_q(r,h,f;s)$ we denote the maximum integer $n$ such that a projective $(h,f)-(n,r,s)_q$ system exists. Clearly we can convert any given projective $(h,f)-(n,r,s)_q$ system into a faithful projective $(h,f)-(n,r,\le s)_q$ system by replacing each element $U$ by an arbitrary $h$-space containing $U$.
\begin{lemma}
   \label{lemma_sum_gen_ham}
   $n_q\!\left(r,h,f;s_1+s_2\right)\le n_q\!\left(r,h,f;s_1\right)+n_q\!\left(r,h,f;s_2\right)$
\end{lemma}
\begin{proof}
  Consider the union of a projective $(h,f)-\left(n_q(r,h,f;s_1),r,s_1\right)_q$ and a projective $(h,f)-\left(n_q(r,h,f;s_2),r,s_2\right)_q$ system.
\end{proof}
The upper bound for $n_q(n,h;s)$ from Lemma~\ref{lemma_one_weight_bound} can be easily generalized:
\begin{lemma}
\label{lemma_one_weight_bound_gen_hamming_weight}
We have
\begin{equation}
  n_q(r,h,f;s)\le \frac{\qbin{r}{f}{q}\cdot s}{\qbin{r-h}{f}{q}}
  =
  \prod_{i=0}^{f-1} \frac{[r-i]_q}{[r-h-i]_q}\cdot s.
\end{equation}
\end{lemma}
\begin{proof}
   Let $\cS$ be a faithful projective $(h,f)-(n,r,s)_q$ system with
  $n=n_q(r,h,f;s)$. Since each element $S\in\cS$ is contained in
  $\qbin{r-h}{f}{q}$ subspaces of codimension $f$ and there are $\qbin{r}{f}{q}$ subspaces of codimension $f$ in total,
  we conclude $n\le \tfrac{\qbin{r}{f}{q}\cdot s}{\qbin{r-h}{f}{q}}$.
\end{proof}

Considering the set of all $n=\qbin{r}{h}{q}$ $h$-spaces in $\PG(r-1,q)$ we see that the upper bound in Lemma~\ref{lemma_one_weight_bound_gen_hamming_weight} is tight for $s=\qbin{r-f}{h}{q}$. Using $\lambda$ copies of this construction yields
\begin{equation}
  \lim_{s\to\infty} n_q(r,h,f;s) \cdot \frac{\qbin{r-h}{f}{q}}{\qbin{r}{f}{q}\cdot s}=1.
\end{equation}
The determination of $n_q(r,h,f;s)$ as a function of $s$ seems to be a hard problem, even for small parameters, see e.g.\ \cite{jozefien} for the currently known bounds for $n_2(5,2,2;s)$.

For $h=1$, i.e.\ linear codes w.r.t.\ the $f$-generalized Hamming weight several bounds have been obtained in the literature, see e.g.\ \cite{tsfasman1995geometric}. Here we single out a Griesmer-type bound to which we devote the entire subsequent subsection. As for the special case $f=1$ it can indeed always be attained if the minimum distance $d$ is sufficiently large. In Subsection~\ref{subsec_gen_ham} we determine a few exact values of $n_q(r,1,f;s)$ for small parameters. Many of these results are from \cite{GenHamMitIvanAssia}.

\subsection{The Griesmer bound}
\label{subsec_griesmer_bound}

It is well known that the Griesmer bound
\begin{equation}
  \label{eq_griesmer_bound2}
  n\ge \sum_{i=0}^{k-1} \left\lceil\frac{d}{q^i}\right\rceil=:g_q(k,d).
\end{equation}
for the minimum length of an $[n,k]_q$ codes with given minimum Hamming distance $d$ is attained if $d$ is sufficiently large. Here we will show that the same is true for the $f$th generalized Hamming distance.\footnote{This should be well known, but we did not find a corresponding reference.} First we study the geometric counterpart. To this end, let an $(n,s)$-arc in $\PG(v-1,q)$ denote a multiset of points with cardinality $n$ such that each hyperplane contains at most $s$ points and at least one hyperplane contains exactly $s$ points. In other words, an $(n,s)$-arc in $\PG(v-1,q)$ corresponds to an $[n,v,n-s]_q$ code.

\begin{lemma}
  \label{lemma_residual}
  Let $\cK$ be an $(n,s)$-arc in $\PG(k-1,q)$, where $k\ge 3$. If $H$ is a hyperplane with $\cK(H)=m$, then we have $\cK(S)\le \left\lfloor (sq+m-n)/q\right\rfloor=s-\left\lceil(n-m)/q\right\rceil\le s-\left\lceil(n-s)/q\right\rceil$ for every subspace $S$ of codimension two.
\end{lemma}
\begin{proof}
  Let $S$ be an arbitrary hyperline and $H_0,\dots, H_q$ the $q+1$ hyperplanes through $S$. With this and $\cK(H_0)=m$ we have
  $$
    n=\sum_{i=0}^q \cK(H_i)-q\cdot \cK(S)\le m+q\cdot s-q\cdot \cK(S),
  $$
  so that
  $$
    \cK(S)\le \frac{qs+m-n}{q}.
  $$
  Note that $\cK(S)$ is a non-negative integer and $m\le s$.
\end{proof}

\begin{lemma}
  \label{lemma_upper_bound_multiplicity_j_space}
  Let $\cK$ be an $\left(\widehat{\gamma}_k,\widehat{\gamma}_{k-1}\right)$-arc in $\PG(k-1,q)$, where $k\ge 3$. For $1\le j\le k$ the maximum
  multiplicity of a $j$-space is at most $\widehat{\gamma}_j$, where
  \begin{equation}
    \label{ie_upper_bound_multiplicity_j_space}
    \widehat{\gamma}_j := \left\lfloor\frac{[k-j]_q\widehat{\gamma}_{j+1}-\widehat{\gamma}_k}{[k-j]_q-1}\right\rfloor =
    \widehat{\gamma}_{j+1}-\left\lceil\frac{\widehat{\gamma}_k-\widehat{\gamma}_{j+1}}{[k-j]_q-1}\right\rceil.
  \end{equation}
  for $j=k-2,\dots,1$.
\end{lemma}
\begin{proof}
  We prove by induction for $j=k-2,\dots,1$. So, let $X$ be an arbitrary $j$-space and
  $Y_1,\dots,Y_l$ the $l:=[k-j]_q$ subspaces of dimension $j+1$ that contain $X$. Thus, we conclude
  $$
    \#\cK =\sum_{i=1}^l \cK\!\left(Y_i\right) \,-\, (l-1)\cK(X)\le l\cdot \widehat{\gamma}_{j+1}-(l-1)\cK(X)
  $$
  from the induction hypothesis, so that $\cK(X)\in\N_0$ and $\#\cK=\widehat{\gamma}_k$ imply
  $$
    \cK(X)\le \left\lfloor \frac{l\widehat{\gamma}_{j+1}-\widehat{\gamma}_k}{l-1}\right\rfloor.
  $$
\end{proof}
For e.g.\ a $(1010,204)$-arc in $\PG(4,5)$ we have $\widehat{\gamma}_5=1010$, $\widehat{\gamma}_4=204$, $\widehat{\gamma}_3=42$, $\widehat{\gamma}_2=9$, and $\widehat{\gamma}_1=2$. Setting $\widehat{\gamma}_k=n$ and $\widehat{\gamma}_{k-1}=s$ we can rewrite the recursive definition of $\widehat{\gamma}_j$ in
Equation~(\ref{ie_upper_bound_multiplicity_j_space}) to
\begin{equation}
  \label{ie_upper_bound_multiplicity_j_space_explicit}
  \widehat{\gamma}_j =\left\lfloor\frac{\left\lfloor
                      \begin{array}{cc}
                      \left\lfloor\frac{\left\lfloor \frac{s\cdot[2]_q-n}{[2]_q-1}\right\rfloor \cdot [3]_q-n}{[3]_q-1}\right\rfloor & \\ 
                                                                                                                                     & \ddots
                      \end{array}
                      \right\rfloor \cdot [k-j]_q -n}{[k-j]_q-1}\right\rfloor.
\end{equation}
For the special case $m=s$ Lemma~\ref{lemma_upper_bound_multiplicity_j_space} is a generalization of Lemma~\ref{lemma_residual}. We can also rewrite
Lemma~\ref{lemma_residual} for $[n,k,\ge d]_q$-codes instead of $(n,\le s)$-arcs:
\begin{lemma}
  \label{lemma_residual_code_bound}
  If an $[n,k,\ge d]_q$-code $C$ contains a codeword $c$ of weight $w$ with $w<\frac{dq}{q-1}$, then the residual code of
  $C$ with respect to $c$ is an $\left[n-w,k-1,\ge d-w+\left\lceil \tfrac{w}{q}\right\rceil\right]_q$-code.
\end{lemma}
The repeated application of Lemma~\ref{lemma_residual_code_bound} with $w$ chosen as the minimum Hamming distance $d$ implies the Griesmer bound in Inequality~(\ref{eq_griesmer_bound2}).
An $[n,k,d]_q$-code whose parameters satisfy Inequality~(\ref{eq_griesmer_bound2}) with equality is called a \emph{Griesmer code}. The corresponding $(n,n-d)$-arcs in $\PG(k-1,q)$ are called a \emph{Griesmer arcs}. Next we restate Lemma~\ref{lemma_parameters_griesmer_code} and provide a brief proof.

\begin{lemma}
  \label{lemma_parameters_griesmer_code2}
  Let $k\ge 1$ and $d$ be positive integers. Write $d$ as
  \begin{equation}
    \label{eq_griesmer_representation_min_dist2}
    d=\sigma q^{k-1}-\sum_{i=0}^{k-2}\varepsilon_iq^i,
  \end{equation}
  where $\sigma\in\N_0$ and the $0\le\varepsilon_i<q$ are integers for all $0\le i\le k-2$. Then, Inequality~(\ref{eq_griesmer_bound2})
  is satisfied with equality iff
  \begin{equation}
    \label{eq_griesmer_representation_length2}
    n=\sigma[k]_q-\sum_{i=0}^{k-2}\varepsilon_i[i+1]_q,
  \end{equation}
  which is equivalent to
  \begin{equation}
    \label{eq_griesmer_representation_species2}
    n-d=\sigma[k-1]_q-\sum_{i=1}^{k-2}\varepsilon_i[i]_q.
  \end{equation}
  Moreover, for each integer $0\le j\le k-1$ we have
  \begin{equation}
    \label{eq_griesmer_representation_32}
    \sum_{i=j}^{k-1} \left\lceil\frac{d}{q^i}\right\rceil=\sigma [k-j]_q-\sum_{i=j}^{k-2}\varepsilon_i[i-j+1]_q.
  \end{equation}
\end{lemma}
\begin{proof}
  From Equation~(\ref{eq_griesmer_representation_min_dist2}) and
  $$
    0\le \sum_{i=0}^{h-1} \varepsilon_i q^i\le (q-1)\cdot\sum_{i=0}^{h-1}q^i=q^h-1<q^h
  $$
  we conclude
  \begin{equation}
    \left\lceil\frac{d}{q^h}\right \rceil =\sigma q^{k-1-h}-\sum_{i=h}^{k-2} \varepsilon_i q^{i-h}
  \end{equation}
  for all $0\le h\le k-1$. With this we have
  \begin{eqnarray*}
    \sum_{h=j}^{k-1} \left\lceil\frac{d}{q^h}\right\rceil &=& \sigma  \sum_{h=j}^{k-1} q^{k-1-h}- \sum_{h=j}^{k-1}\sum_{i=h}^{k-2} \varepsilon_i q^{i-h} \\
    &=& \sigma\sum_{h=0}^{k-1-j} q^h -\sum_{i=j}^{k-2}\varepsilon_i\sum_{h=j}^{i} q^{i-h} \\
    &=& \sigma[k-j]_q -\sum_{i=j}^{k-2}\varepsilon_i\sum_{h=0}^{i-j} q^{h} \\
    &=& \sigma[k-j]_q -\sum_{i=j}^{k-2}\varepsilon_i [i-j+1]_q
  \end{eqnarray*}
  for $j=0,\dots,k-1$. The case $j=0$ implies that the value for $n$, such that Inequality~(\ref{eq_griesmer_bound2})
  is satisfied with equality, is given by Equation~(\ref{eq_griesmer_representation_length2}). The case $j=1$ gives the equivalent condition for $n-d$.
\end{proof}

In other words a Griesmer code has parameters $n$ and $d$ satisfying Equation~(\ref{eq_griesmer_representation_length2}) and
Equation~(\ref{eq_griesmer_representation_min_dist2}), respectively. Note that each positive integer $d$ admits a unique representation as in
Equation~(\ref{eq_griesmer_representation_min_dist2}). The parameters $n$ and $s=n-d$ of a Griesmer arc
are specified by Equation~(\ref{eq_griesmer_representation_length2}) and Equation~(\ref{eq_griesmer_representation_species2}), respectively. Note that
not every positive integer $s$ admits a representation as in Equation~(\ref{eq_griesmer_representation_species2}), satisfying the conditions for $\sigma$ and
the $\varepsilon_i$. The expressions in (\ref{eq_griesmer_representation_32}) equal
$\widehat{\gamma}_{k-j}$ and are indeed attained for Griesmer arcs.

\begin{lemma}
  \label{lemma_parameters_griesmer_arc}
  For an integer $k\ge 3$ let $\sigma$ and $0\le \varepsilon_i<q$, where $0\le i\le k-2$, be non-negative integers. If
  \begin{equation}
    \label{eq_griesmer_representation_min_dist2arc}
    n=\sigma[k]_q-\sum_{i=0}^{k-2}\varepsilon_i[i+1]_q
  \end{equation}
  and
  \begin{equation}
    \label{eq_griesmer_representation_species2arc}
    s=\sigma [k-1]_q-\sum_{i=1}^{k-2}\varepsilon_i[i]_q,
  \end{equation}
  then we have
  \begin{equation}
    \label{eq_griesmer_representation_gammaarc}
    \widehat{\gamma}_{j}=\sigma [j]_q-\sum_{i=k-j}^{k-2}\varepsilon_i[i-k+j+1]_q
  \end{equation}
  for each $(n,s)$-arc $\cK$ in $\PG(k-1,q)$ and for all $1\le j\le k$. Moreover, each $j$-space $X$ with $\cK(X)=\widehat{\gamma}_j$ contains a $(j-1)$-space $Y$
  with $\cK(Y)=\widehat{\gamma}_{j-1}$, where $2\le j\le k$.
\end{lemma}
\begin{proof}
  First we proof Equation~(\ref{eq_griesmer_representation_gammaarc}) by induction for $j=k,k-1,\dots,1$. The case $j=k$ is given by Equation~(\ref{eq_griesmer_representation_min_dist2arc})
  and the case $j=k-1$ is given by Equation~(\ref{eq_griesmer_representation_species2arc}). Note that we have $[a]_q-[b]_q=q^b[a-b]_q$ and $[b]_q-1=q[b-1]_q$ for integers $1\le b\le a$.
  With this we compute
  \begin{eqnarray*}
    \frac{\widehat{\gamma}_k-\widehat{\gamma}_{j+1}}{[k-j]_q-1} &=& \frac{\sigma\left([k]_q-[j+1]_q\right)-\sum\limits_{i=k-j-1}^{k-2} \varepsilon_i\left([i+1]_q-[i+1-k+j+1]_q\right)
    -\sum\limits_{i=0}^{k-j-2} \varepsilon_i[i+1]_q}{q\cdot[k-j-1]_q}\\
    &=& \sigma q^j-\sum\limits_{i=k-j-1}^{k-2} \varepsilon_i q^{i-k+j+1}-\frac{\sum\limits_{i=0}^{k-j-2} \varepsilon_i[i+1]_q }{[k-j]_q-1},
  \end{eqnarray*}
  where $1\le j\le k-2$. Since
  $$
    0\le \sum\limits_{i=0}^{k-j-2} \varepsilon_i[i+1]_q \le \sum_{i=0}^{k-j-2} \left(q^{i+1}-1\right)=\left([k-j]_q-1\right)-(k-j-1)< [k-j]_q-1
  $$
  we conclude
  $$
    \left\lceil \frac{\widehat{\gamma}_k-\widehat{\gamma}_{j+1}}{[k-j]_q-1} \right\rceil = \sigma q^j-\sum\limits_{i=k-j-1}^{k-2} \varepsilon_i q^{i-k+j+1}.
  $$
  Plugging into Equation~(\ref{ie_upper_bound_multiplicity_j_space}) gives
  \begin{eqnarray*}
    \widehat{\gamma}_j &=& \widehat{\gamma}_{j+1}-\left\lceil \frac{\widehat{\gamma}_k-\widehat{\gamma}_{j+1}}{[k-j]_q-1} \right\rceil =
    \sigma \left([j+1]_q-q^j\right)-\sum\limits_{i=k-j-1}^{k-2} \varepsilon_i \left([i+1-k+j+1]_q-q^{i-k+j+1}\right) \\
    &=& \sigma [j]_q-\sum\limits_{i=k-j-1}^{k-2} \varepsilon_i [i+1-k+j]_q =  \sigma [j]_q-\sum\limits_{i=k-j}^{k-2} \varepsilon_i [i+1-k+j]_q,
  \end{eqnarray*}
  so that Equation~(\ref{eq_griesmer_representation_gammaarc}) is satisfied.

  For the final statement note that $\cK|_X$ is a $\left(\widehat{\gamma}_j,\le \widehat{\gamma}_{j-1}\right)$-arc in $\PG(j-1,q)$. Observe that the arc is
  spanning, due to $\widehat{\gamma}_j>\widehat{\gamma}_{j-1}$, and the non-existence of a hyperplane $Y$ in $X$ with $\cK(Y)=\widehat{\gamma}_{j-1}$
  contradicts the Griesmer bound, i.e., Inequality~(\ref{eq_griesmer_bound2}).
\end{proof}

For the $r$th generalized Hamming distance we have:
\begin{theorem}(Griesmer-type bound) \cite[Theorem 4]{helleseth1995bounds}, \cite[Theorem 5]{helleseth1992generalized}\\\label{thm_griesmer_gen_ham}
   For each $[n,k]_q$ code and each $1\le r\le k$ we have
   \begin{equation}
     \label{ie_griesmer_gen_ham_weight}
     n \ge d_r +\sum_{j=1}^{k-r} \left\lceil \frac{d_r}{[r]_q\cdot q^j}\right\rceil=:g_q^r\!\left(k,d_r\right).
   \end{equation}
\end{theorem}
For other bounds we refer to e.g.~\cite{tsfasman1995geometric}.
A few auxiliary results are given as follows

\begin{lemma}
  \label{lemma_q_bin_difference}
  For $a\ge b$ we have
  \begin{equation}
    [a]_q-[a-b]_q=q^{a-b}\cdot [b]_q.
  \end{equation}
\end{lemma}
\begin{proof}
  $$
    [a]_q-[a-b]_q =\frac{\left(q^a-1\right)-\left(q^b-1\right)}{q-1}
    =q^{a-b}\cdot \frac{q^b-1}{q-1}=q^{a-b}\cdot [b]_q.
  $$
\end{proof}

\begin{lemma}
  \label{lemma_g_function_properties}
  We have
  \begin{itemize}
   \item[(a)] $g_q^1(k,d)=g_q(k,d)$;
   \item[(b)] $g_q^r(k,a\cdot q[r]_q+b)=b-q[r]_q+g_q^r(k,(a+1)\cdot q[r]_q)$ for all $a\in\N$ and all $b\in\{1,\dots,q[r]_q\}$;
   \item[(c)] $g_q^r(k,\sigma q^{k-r}[r]_q+\lambda)=\sigma[k]_q+g_q^r(k,\lambda)$ for all $\sigma,\lambda\in\N$.
  \end{itemize}
 \end{lemma}
\begin{proof}
  We compute
  $$
    g_q^1(k,d)=d +\sum_{j=1}^{k-1} \left\lceil \frac{d}{[1]_q\cdot q^j}\right\rceil=\sum_{j=0}^{k-1} \left\lceil \frac{d}{q^j}\right\rceil
    =g_q(k,d),
  $$
  \begin{eqnarray*}
    g_q^r(k,aq[r]_q+b)-aq[r]_q-b &=&
    \sum_{j=1}^{k-r} \left\lceil \frac{aq[r]_q+b}{[r]_q\cdot q^j}\right\rceil =\sum_{j=1}^{k-r} \left\lceil \frac{(a+1)q[r]_q}{[r]_q\cdot q^j}\right\rceil\\&=&g_q^r(k,(a+1)q[r]_q)-(a+1)q[r]_q,
  \end{eqnarray*}
  and
  \begin{eqnarray*}
   g_q^r(k,\sigma q^{k-r}[r]_q+\lambda)&=& \left(\sigma q^{k-r}[r]_q+\lambda\right) +\sum_{j=1}^{k-r} \left\lceil\frac{\sigma q^{k-r}[r]_q+\lambda}{[r]_q\cdot q^j}\right\rceil
   \\&=&\sigma[k]_q+\left(\lambda+\sum_{j=1}^{k-r} \left\lceil\frac{\lambda}{[r]_q\cdot q^j}\right\rceil\right)
  \end{eqnarray*}
  using Lemma~\ref{lemma_q_bin_difference}.
\end{proof}
So, for $r=1$ Inequality~(\ref{ie_griesmer_gen_ham_weight}) coincides with Inequality~(\ref{eq_griesmer_bound}) and we have $g_q^r(k,d)=g_q^r(k,d-1)+1$ if
$d$ is not divisible by $q[r]_q$.

\smallskip

A geometric construction that attains Inequality~(\ref{eq_griesmer_bound}) for sufficiently large $d$ was given by Solomon and Stiffler \cite{solomon1965algebraically}, cf.\ Section~\ref{sec_solomon_stiffler}. To describe the construction let $\chi_S$ denote the \emph{characteristic function} of a subspace $S$, i.e.\ $\chi_S(P)$ is point $P$ is contained in $S$ and $\chi_S(P)$ otherwise. We say that a multiset of points $\cM$ in $\PG(v-1,q)$ is of \emph{type}
  $\sigma[v]-\sum_{i=1}^{v-1} \varepsilon_i[i]$,
  if there exist subspaces $S_1,\dots,S_l$ such
  that $\cM=\sigma\cdot \chi_A-\sum_{i=1}^l \chi_{S_i}$
  and their are exactly $\varepsilon_i$ subspaces $S_j$ of dimension $i$, where $\chi_A$ denotes the
  characteristic function of the ambient space $\PG(v-1,q)$ and $\sigma$, $\varepsilon_i\in \N$ for all $1\le i\le v-1$.

\begin{lemma}
  \label{lemma_S_f_parameters}
  Let $\cM$ be a multiset of points in $\PG(k-1,q)$ with  type $\sigma[k]-\sum_{i=1}^{k-1}\varepsilon_i[i]$, where $\sigma\in\N$ and $\varepsilon_i\in\N$ for all $1\le i\le k-1$. Then, we have $\#\cM=n$
  and $\cM(U)\le s_r$ for each subspace of codimension $r$, where
  \begin{equation}
    n=\sigma[k]_q-\sum_{i=1}^{k-1}\varepsilon_i[i]_q
  \end{equation}
  and
  \begin{equation}
    s_r=\sigma[k-r]_q-\sum_{i=r}^{k-1}\varepsilon_i[i-r]_q
  \end{equation}
  for all $1\le r\le k$. Moreover, we have
  \begin{equation}
    \label{eq_d_f}
    d_r:=n-s_r= [r]_q\cdot\left(
    \sigma\cdot q^{k-r}-\sum_{i=r+1}^{k-1} \varepsilon_i \cdot q^{i-r}
    \right)
    -\sum_{i=1}^{r} \varepsilon_i[i]_q.
  \end{equation}
\end{lemma}
\begin{proof}
  Counting points gives the equation for $n$. Note that each subspace of codimension $r$ contains exactly $[k-r]_q$ points while the intersection of an $i$-space and a subspace of codimension $r$ contains at least $[i-r]_q$ points if $i\ge r$. Thus, we have $\cM(U)\le s_r$ for each subspace of codimension $r$, where $1\le r\le k$. Using Lemma~\ref{lemma_q_bin_difference} we compute
  \begin{eqnarray*}
    d_r &=& n-s_r =\sigma\cdot q^{k-r}\cdot [r]_q-\sum_{i=r+1}^{k-1}  \varepsilon_i\cdot q^{i-r}\cdot [r]_q-\sum_{i=1}^r \varepsilon_i[i]_q.
  \end{eqnarray*}
\end{proof}
Clearly, we can always use the construction $\cM=\sigma\cdot \chi_A-\sum_{i=1}^l \chi_{S_i}$ if $\sigma$ is sufficiently large, which may depend on a clever choice of the subspaces $S_j$. We remark that there is no need to assume that the $\varepsilon_i$ are non-negative if assume that subspaces $S_j$ of the same dimension are equal and arranged in a chain $S_1\le \dots\le S_{r-1}$, where $\dim(S_i)=i$ for all $1\le i\le r-1$. If the $\varepsilon_i$ are non-negative, then choosing $\sigma\ge\sum_{i=1}^{v-1}\varepsilon_i$ is sufficient.

\begin{lemma}
  \label{lemma_sol_stif_aux1}
  Using the notation from Lemma~\ref{lemma_S_f_parameters} let $\varepsilon_i\le q-1$ for $1\le i< r+j$, where $1\le j\le k-r$. Then, we
  have
  \begin{equation}
    \left\lceil\frac{d_r}{[r]_q\cdot q^j}\right\rceil
    =\sigma\cdot q^{k-r-j}-\sum_{i=r+j}^{k-1} \varepsilon_i\cdot q^{i-r-j}.
  \end{equation}
\end{lemma}
\begin{proof}
   Since $0\le \varepsilon_i\le q-1$ for $1\le i< r+j$ we have
   \begin{equation}
     \label{ie_rounding_gen_ham}
     0\le\sum_{i=1}^r \varepsilon_i [i]_q +\sum_{i=r+1}^{r+j-1} \varepsilon_i\cdot q^{i-r}\cdot [r]_q
     < q[r]_q+(q-1)[r]_q\cdot\sum_{i=1}^{j-1} q^i=[r]_q\cdot q^j,
   \end{equation}
   so that Equation~(\ref{eq_d_f}) implies the statement.
\end{proof}

\begin{lemma}
  \label{lemma_sol_stif_aux2}
  Using the notation from Lemma~\ref{lemma_S_f_parameters} let $\varepsilon_i\le q-1$ for $1\le i\le k-1$. Then Inequality~(\ref{ie_griesmer_gen_ham_weight}) is attained
  for each $1\le r\le k$.
\end{lemma}
\begin{proof}
  Using Lemma~\ref{lemma_sol_stif_aux1} and Lemma~\ref{lemma_S_f_parameters}
  we compute
  $$
    \sum_{j=1}^{k-r} \left\lceil \frac{d_r}{[r]_q\cdot q^j}\right\rceil
    = \sum_{j=1}^{k-r} \left(
    \sigma\cdot q^{k-r-j}-\sum_{i=r+j}^{k-1} \varepsilon_i\cdot q^{i-r-j}
    \right)
    = \sigma\cdot [k-r]_q -\sum_{i=r+1}^{k-1} \varepsilon_i\cdot [i-r]_q,
  $$
  so that Lemma~\ref{lemma_q_bin_difference} and Equation~(\ref{eq_d_f})
  yield
  $$
    d_r +\sum_{j=1}^{k-r} \left\lceil \frac{d_r}{[r]_q\cdot q^j}\right\rceil
    =\sigma\cdot [k]_q -\sum_{i=r+1}^{k-1} \varepsilon_i\cdot [i]_q
    -\sum_{i=1}^r \varepsilon_i[i]_q=\sigma\cdot [k]_q -\sum_{i=1}^{k-1} \varepsilon_i\cdot [i]_q=n.
  $$
\end{proof}
In other words, the Solomon--Stiffler construction and indeed every Griesmer arc is optimal for the $r$th generalized Hamming weight for all $1\le r\le k$. For the special case when $\varepsilon_i\in\{0,1\}$ for all $1\le i\le k-1$ this was e.g.\ also observed in \cite{pan2025optimal}, while it is certainly known for a longer time. From Inequality~(\ref{ie_rounding_gen_ham}) we can also conclude the existence of Solomon--Stiffler constructions with $\varepsilon\ge q$ for some small $i$ which are optimal for the $f$th generalized Hamming weight for all $f\ge j$ but not optimal for the $(j-1)$th generalized Hamming weight, where $j$ is arbitrary.
\begin{theorem}
  For $1\le r\le k$ let $C$ be an $[n,k,d]_q$ code with
  $$
    n=\sigma[k]-\sum_{i=r+1}^{k-1} \varepsilon_i[i]_q
  $$
  and
  $$
    d=\sigma\cdot q^{k-1}-\sum_{i=r+1}^{k-1}\varepsilon_i\cdot q^{i-1},
  $$
  where $\sigma\in\N$ and $\varepsilon_i\in\{0,\dots,q-1\}$. Then, we have
  $$
    d_r(C)=[r]_q\cdot\left(\sigma\cdot q^{k-r}-\sum_{i=r+1}^k \varepsilon_i q^{i-r}\right)
  $$
  and $C$ attains the Griesmer bound for the $r$th generalized Hamming weight, i.e.\ $g_q^r(k,d_r)=n$.
\end{theorem}
\begin{proof}
  Let $\cM$ be the multiset of points in $\PG(k-1,q)$ corresponding to $C$. Due to the minimum distance $d$ each hyperplane contains at most
  $$
    s=\sigma[k-1]_q-\sum_{i=r+1}^{k-1} \varepsilon_i [i-1]_q
  $$
  points and there exists a hyperplane containing exactly $s$ points. Lemma~\ref{lemma_upper_bound_multiplicity_j_space} yields that each $(k-r)$-space contains at most
  $$
    s_{k-r}=\sigma[k-r]_q-\sum_{i=r+1}^{k-1} \varepsilon_i[i-r]_q
  $$
  points and that there exist at least one $(k-r)$-space attaining this bound. Thus, we have
  \begin{equation}
    \label{eq_dr_C}
    d_r(C)=n-s_{k-r}=[r]_q\cdot\left(\sigma\cdot q^{k-r}-\sum_{i=r+1}^k \varepsilon_i q^{i-r}\right)
  \end{equation}
  and Lemma~\ref{lemma_sol_stif_aux2} yields $n=g_q^r(k,d_r(C))$.
\end{proof}
As mentioned before, we have $g_q^r(k,d)=g_q^r(k,d-1)+1$ if
$d$ is not divisible by $q[r]_q$. So, in order to determine $n_q(k,r;s)$ we can look at all $d_r$ that are multiples of $q[r]_q$. Equation~(\ref{eq_dr_C}) uniquely determines the auxiliary parameters $\sigma$, $\varepsilon_{r+1},\varepsilon_{r+2},\dots,\varepsilon_{k-1}$, which then determine the length $n$ and the minimum Hamming distance $d$ of the desired Griesmer code $C$. If $d_r$ is large, so is $d$, so that we only need to check the existence of the Griesmer codes for some small multiples of $q[r]_q$, since Griesmer codes always exist if $d$ is sufficiently large. For e.g.\ $q=3$, $k=5$, $r=2$, and $d_r=4\cdot 12=48$ we look for an $[55,5,36]_3$ Griesmer code, which indeed exists, so that $n_3(5,2;7)=55$. For the same parameters and $d_r=3\cdot 12=36$ we look for an $[41,5,25]_3$, which indeed exists. Since it is not a Griesmer code for $d_r=4\cdot 12=48$ we end up with $s=7$, the cases $s=6$ needs to be considered separately.

From Inequality~(\ref{ie_rounding_gen_ham}) we can also conclude the existence of Solomon--Stiffler constructions with $\varepsilon\ge q$ for some small $i$ which are optimal for the $r$th generalized Hamming weight for all $r\ge j$ but not optimal for the $(j-1)$th generalized Hamming weight, where $j$ is arbitrary. See also \cite[Theorem 7.1]{chen2025t} for a closely related application of the Solomon--Stiffler construction.

\medskip

In \cite[Theorem 1]{ball2025additive} a linear $[n,k,d]_q$ code $C$ is used to construct an additive $[n,k/h,\ge d_h(C)]_q^h$ code. We remark that the Solomon--Stiffler construction and the Griesmer bound gives $n_q(5,2;\sigma\cdot [3]_q-\varepsilon_4\cdot (q+1)-\varepsilon_3)=\sigma\cdot[5]_q-\varepsilon_4\cdot[4]_q-\varepsilon_3\cdot[3]_q$ for $\varepsilon_4,\varepsilon_3\in\{0,\dots,q-1\}$ and sufficiently large $\sigma\in\N$. Applying the mentioned construction we obtain optimal additive codes iff $\varepsilon_3=0$. For bounds and exact values for $2.5$-dimensional additive codes over $\F_q$, where $q\in\{2,3,4,5\}$, we refer to \cite{bierbrauer2021optimal,krotovkurz2025}, cf.~Subsection~\ref{subsec_dimension_2_5}. However, for many other parameters \cite[Theorem 1]{ball2025additive} indeed attains the Griesmer upper bound.

\medskip

Currently, we do not know a Griesmer-type bound for the $f$th generalized Hamming weight of an additive code, or for $n_q(r,h,f;s)$ in geometric terms. For linear codes in the so-called $b$-symbol metric a Griesmer-type bound is known and it can indeed always be attained for sufficiently large minimum distances \cite{b_symbol}. For Griesmer-type bounds for certain subclasses of non-linear codes we refer e.g.\ to \cite{bellini2015griesmer,pan2025griesmer,shiromoto2001griesmer}.

\subsection{Exact values of $n_q(r,1,f;s)$ for small parameters}
\label{subsec_gen_ham}

While it is in general hard to determine the values $n_q(r,1,f;s)$ parametric $q$ there is one easy case:

\begin{proposition}
  We have $n_q(f+1,1,f;s)=s\cdot [f+1]_q$ for each $f\ge 1$ and each $s\ge 1$.
\end{proposition}
\begin{proof}
  The lower bound is given by Lemma~\ref{lemma_one_weight_bound_gen_hamming_weight} and a construction is given by a multiset of points $\cM$ in $\PG(f,q)$ with $\cM(P)=s$ for every point $P$.
\end{proof}

The \emph{Griesmer upper bound} for $n_q(r,1,f;s)$ is the largest integer $n$ such that $n\ge g_q^f(r,n-s)$. The \emph{coding upper bound} for $n_q(r,1,f;s)$ is given by $s_k$, where $s_{r-f}=s$ and $s_i=n_q\!\left(i,1,1;s_{i-1}\right)$ for $r-f<i\le r$.
In other words, the coding upper bound uses the parameters of optimal linear codes to recursively upper bound the number of points in subspaces of codimension $f-1,\dots,1,0$.

\begin{example}
  For $n_2(7,1,2;21)$ the Griesmer upper bound is $81$ and the coding upper bound is $75$, where $s_6=39$.
\end{example}

\begin{corollary}
  $n_q(f+1,1,f;s)$ is given by the Griesmer upper bound for all $f\ge 1$ and all $s\ge 1$.
\end{corollary}

From Lemma~\ref{lemma_g_function_properties}.(c) we directly conclude:
\begin{lemma}
  If $n_q(r,1,f;s)$ attains the Griesmer upper bound, then also $n_q(r,1,f;s+t\cdot[r]_q)$ attains the Griesmer upper bound for all $t\in\N$.
\end{lemma}

In the following we determine the exact values of $n_2(r,1,f;s)$ for all $r\le 7$. We remark that $n_2(8,1,1;s)$ is completely known while there only partial results for $n_2(9,1,1;s)$. However, the determination of $n_2(8,1,1;s)$ is scattered among several papers, so that we do not attempt to determine $n_2(8,1,f;s)$ here. However, we determine the exact values of $n_3(r,1,f;s)$ for all $r\le 5$.

\begin{proposition}
  \label{prop_n_2_4_2_s}
  We have $n_2(4,1,2;3t+2)=15t+8$, $n_2(4,1,2;3t+3)=15t+15$, and
  $n_2(4,1,2;3t+4)=15t+16$ for all $t\in N$.
\end{proposition}
\begin{proof}
  The upper bounds are given by the Griesmer upper bound. Constructions for $n_2(4,1,2;2)\ge 8$,
  $n_2(4,1,2;3)\ge 15$, and $n_2(4,1,2;4)\ge 16$ are given by an affine solid (type $[4]-[3]$), a solid (type $[4]$), and a solid plus a point (type $[4]+[1]$), respectively. Combining those examples with $t$ copies of a solid yields the remaining upper bounds by Lemma~\ref{lemma_sum_gen_ham}.
\end{proof}

\begin{table}[htp]
  \begin{center}
    \begin{tabular}{rrcc}
      \hline
      $s$ & $n_2(4,1,2;s)$ & construction & upper bound \\
      \hline
      2 &  8 & $[4]-[3]$ & Griesmer upper bound \\
      3 & 15 & $[4]$ & Griesmer upper bound \\
      4 & 16 & plus point & Griesmer upper bound \\
      \hline
    \end{tabular}
    \caption{Exact values for $n_2(4,1,2;s)$.}
    \label{table_n_2_4_2_s}
  \end{center}
\end{table}

\begin{corollary}
  $n_2(4,1,2;s)$ is given by the Griesmer upper bound for all $s\ge 2$.
\end{corollary}

Proposition~\ref{prop_n_2_4_2_s} can be generalized:
\begin{proposition}
  For each $f\ge 2$ and each $t\in \N$ we have  $n_2(f+2,1,f;3t+2)=t\cdot [f+2]_2+2^{f+1}$, $n_2(f+2,1,f;3t+3)=t\cdot [f+2]_2+[f+2]_2$, and $n_2(f+2,1,f;3t+4)=(t+1)\cdot [f+2]_2+1$.
\end{proposition}
\begin{proof}
  Constructions for $n_2(f+2,1,f;2)\ge 2^{f+1}$,
  $n_2(f+2,1,f;3)\ge [f+2]_2$, and $n_2(f+2,1,f;4)\ge [f+2]_2$ are given by an affine $(f+2)$-space  (type $[f+2]-[f+1]$), an $(f+2)$-space (type $[f+2]$), and an $(f+2)$-space plus a point (type $[f+2]+[1]$), respectively. Combining those examples with $t$ copies of the ambient space yields the remaining upper bounds by Lemma~\ref{lemma_sum_gen_ham}.

  The upper bounds are given by the Griesmer upper bound. More precisely, applying Theorem~\ref{thm_griesmer_gen_ham} with $d_f=(2t+1)\cdot 2[f]_2$ gives
  $n\ge (2t+1)\cdot 2[f]_2+(2t+1)+(t+1)=t\cdot[f+2]_2+2^{f+1}$ and applying Theorem~\ref{thm_griesmer_gen_ham} with $d_f=(2t+2)\cdot 2[f]_2$ gives
  $n\ge (2t+2)\cdot 2[f]_2+(2t+2)+(t+1)=t\cdot[f+2]_2+[f+2]_2$.
\end{proof}

\begin{corollary}
  For each $f\ge 2$ we have that $n_2(f+2,1,f;s)$ is given by the Griesmer upper bound for all $s\ge 2$.
\end{corollary}

\begin{lemma}
  \label{lemma_projective_base}
  For $r\ge 5$ we have $n_2(r,1,2;r-2)=r+1$.
\end{lemma}
\begin{proof}
  A projective base (or frame) shows $n_2(r,1,2;r-2)\ge r+1$ (even for $r\ge 3)$. From the Griesmer upper bound we conclude $n_2(r-1,1,1;r-2)\le r$, so that we can assume that a multiset $\cM$ of at least $r+1$ points contains the points spanned by the $r$ unit vectors. Any further point in $\cM$ spanned by $v\in\F_2^r$ then needs to have Hamming weight $r-1$ or $r$, since otherwise $r-1$ points would be contained in a subspace of codimension two. The sum of two different such vectors with Hamming weight $r$ or $r-1$ has Hamming weight strictly less than $r-1$, so that we can find $r-1$ points in a subspace of codimension two.
\end{proof}
Note that ovoids imply $n_q(4,1,2;2)\ge q^2+1$.

\begin{proposition}
  \label{prop_n_2_5_1_2}
  We have $n_2(5,1,2;7t+4)=31t+16$, $n_2(5,1,2;7t+5)=31t+17$, $n_2(5,1,2;7t+6)=31t+24$, $n_2(5,1,2;7t+7)=31t+31$, $n_2(5,1,2;7t+8)=31t+32$, $n_2(5,1,2;7t+9)=31t+33$, and $n_2(5,1,2;7t+10)=31t+40$ for all $t\in N$. Moreover, we have $n_2(5,1,2;3)=6$.
\end{proposition}
\begin{proof}
 Lemma~\ref{lemma_projective_base} yields $n_2(5,1,2;3)=6$. Constructions for $4\le s\le 10$ are given by multisets of points with types $[5]-[4]$, $[5]-[4]+[1]$, $[5]-[3]$, $[5]$, $[5]+[1]$, $[5]+2[1]$, and $2[5]-[4]-[3]$, respectively. Combining those examples with $t$ copies of a $5$-space yields the remaining constructions by Lemma~\ref{lemma_sum_gen_ham}.
  The upper bounds for $n_2(5,1,2;s)$ are given by the Griesmer upper bound for all $s\ge 4$.
\end{proof}

\begin{table}[htp]
  \begin{center}
    \begin{tabular}{rrcc}
      \hline
      $s$ & $n_2(5,^,2;s)$ & construction & upper bound \\
      \hline
      3 &  6 & projective base & Lemma~\ref{lemma_projective_base} \\
      4 & 16 & $[5]-[4]$ & Griesmer upper bound \\
      5 & 17 & plus point & Griesmer upper bound \\
      6 & 24 & $[5]-[3]$ & Griesmer upper bound \\
      7 & 31 & $[5]$ & Griesmer upper bound \\
      8 & 32 & plus point & Griesmer upper bound \\
      9 & 33 & plus point & Griesmer upper bound \\
      10 & 40 & $2[5]-[4]-[3]$ & Griesmer upper bound \\
      \hline
    \end{tabular}
    \caption{Exact values for $n_2(5,1,2;s)$.}
    \label{table_n_2_5_2_s}
  \end{center}
\end{table}

\begin{corollary}
  $n_2(5,1,2;s)$ is given by the Griesmer upper bound for all $s\ge 4$.
\end{corollary}

We can generalize Proposition~\ref{prop_n_2_5_1_2} as follows:
\begin{proposition}
  For each $f\ge 1$ and all $s\ge 4$ we have that $n_2(f+3,1,f;s)$ is given by the Griesmer upper bound.
\end{proposition}
\begin{proof}
  For Solomon--Stiffler constructions for $[f+3]-[f+2]$, $[f+3]-[f+2]$, $[f+3]$, and $2[f+3]-[f+2]-[f+1]$ give $n_2(f+3,1,r;4)\ge 2^{f+2}$, $n_2(f+3,1,f;6)\ge 3\cdot 2^{f+1}$, $n_2(f+3,1,f;7)\ge [f+3]_2$, and $n_2(f+3,1,f;10)\ge 5\cdot 2^{f+1}$. Adding points gives $n_2(f+3,1,f;5)\ge n_2(f+3,1,f;4)+1$, $n_2(f+3,1,f;8)\ge n_2(f+3,1,f;7)+1$, and  $n_2(f+3,1,f;9)\ge n_2(f+3,1,f;7)+2$. For $s>10$ the lower bounds are given by $n_2(f+3,1,f;s)\ge n_2(f+3,1,f;s-7)+n_2(f+3,1,3;7)$.

  The upper bounds are given by the Griesmer upper bound. More precisely, applying Theorem~\ref{thm_griesmer_gen_ham} with $d_f=(4t+2)\cdot 2[f]_2$ gives
  $n\ge (4t+2)\cdot 2[f]_2+(2t+1)+(t+2)+1=t\cdot[f+3]_2+2^{f+2}$, applying Theorem~\ref{thm_griesmer_gen_ham} with $d_f=(4t+3)\cdot 2[f]_2$ gives
  $n\ge (4t+3)\cdot 2[f]_2+(2t+2)+(t+1)+1=t\cdot[f+3]_2+3\cdot 2^{f+1}$, applying Theorem~\ref{thm_griesmer_gen_ham} with $d_f=(4t+4)\cdot 2[f]_2$ gives
  $n\ge (4t+4)\cdot 2[f]_2+(2t+2)+(t+1)+1=t\cdot[f+3]_2+[f+3]_2$, and applying Theorem~\ref{thm_griesmer_gen_ham} with $d_f=(4t+5)\cdot 2[f]_2$ gives
  $n\ge (4t+5)\cdot 2[f]_2+(2t+3)+(t+2)+2=t\cdot[f+3]_2+5\cdot 2^{f+1}$.
\end{proof}
So, the only non-trivial value that is not determined yet is $n_2(f+3,1,f;3)$. The first values are given by  $n_2(4,1,1;3)=5$ and $n_2(5,1,2;3)=6$.

\begin{table}[htp]
  \begin{center}
    \begin{tabular}{rrcc}
      \hline
      $s$ & $n_2(6,1,3;s)$ & construction & upper bound \\
      \hline
      3 &  8 & Lemma~\ref{lemma_n_2_6_3_3}  & Lemma~\ref{lemma_n_2_6_3_3} \\
      4 & 32 & $[6]-[5]$ & Griesmer upper bound \\
      5 & 33 & plus point & Griesmer upper bound \\
      6 & 48 & $[6]-[4]$ & Griesmer upper bound \\
      7 & 63 & $[6]$ & Griesmer upper bound \\
      8 & 64 & plus point & Griesmer upper bound \\
      9 & 65 & plus point & Griesmer upper bound \\
      10 & 80 & $2[6]-[5]-[4]$ & Griesmer upper bound \\
      \hline
    \end{tabular}
    \caption{Exact values for $n_2(6,1,3;s)$.}
    \label{table_n_2_6_3_s}
  \end{center}
\end{table}

\begin{lemma}
  \label{lemma_n_2_6_3_3}
  We have $n_2(6,1,3;3)=8$.
\end{lemma}
\begin{proof}
  A feasible example is given by
  $$
    \begin{pmatrix}
      1 & 0 & 0 & 0 & 0 & 0 & 1 & 1 \\[-1mm]
      0 & 1 & 0 & 0 & 0 & 0 & 1 & 1 \\[-1mm]
      0 & 0 & 1 & 0 & 0 & 0 & 1 & 0 \\[-1mm]
      0 & 0 & 0 & 1 & 0 & 0 & 1 & 0 \\[-1mm]
      0 & 0 & 0 & 0 & 1 & 0 & 0 & 1 \\[-1mm]
      0 & 0 & 0 & 0 & 0 & 1 & 0 & 1 \\
    \end{pmatrix}
  $$
  Let $\cM$ be a multiset of $n\ge 9$ points in $\PG(5,2)$ such that each plane contains at most three points. Since $n_2(5,1,2;3)=6$ we assume w.l.o.g.\ that $\cM$ contains the six points spanned by the six unit vectors. Clearly, the maximum point multiplicity is one and every additional point is spanned by a vector $x$ with Hamming weight at least $4$. Since $n_2(6,1,2;4)=7$ we assume w.l.o.g.\ that $\cM$ also contains the point spanned by $x=(1,1,1,1,0,0)^\top$. Let two further points be spanned by $y,z\in\F_2^6$ . Since every plane contains at most three points we have $\wt(y),\wt(z)\ge 4$ and $d_H(x,z),d_H(x,y),d_H(y,z)\ge 4$. If $\wt(y)=4$, then no such vector $z$ exists, so that $\wt(y),\wt(z)\ge 5$, which contradicts $d_H(y,z)\ge 4$.
\end{proof}

\begin{lemma}
  \label{lemma_n_2_7_4_3}
  We have $n_2(7,1,4;3)=11$ and $n_2(8,1,5;3)=17$.
\end{lemma}
\begin{proof}
  An example showing $n_2(7,1,4;3)\ge 11$ is given by
  $$
    \begin{pmatrix}
      1000000&11&10\\[-1mm]
      0100000&11&01\\[-1mm]
      0010000&10&10\\[-1mm]
      0001000&10&01\\[-1mm]
      0000100&01&10\\[-1mm]
      0000010&01&01\\[-1mm]
      0000001&00&11
    \end{pmatrix}.
  $$
  Let $\cM$ be a multiset of $n\ge 12$ points in $\PG(6,2)$ such that each plane contains at most three points. Since $n_2(6,1,3;3)=8$ we assume w.l.o.g.\ that $\cM$ contains the seven points spanned by the seven unit vectors. Clearly, the maximum point multiplicity is one and every additional point is spanned by a vector $x$ with Hamming weight at least $4$. Via a small ILP computation we excluded $n\ge 12$. 
  An example showing $n_2(8,1,5;3)\ge 17$ is given by
  $$
    \begin{pmatrix}
      10000000&1110&01111\\[-1mm]
      01000000&1101&11101\\[-1mm]
      00100000&1010&10011\\[-1mm]
      00010000&1001&01110\\[-1mm]
      00001000&0110&10101\\[-1mm]
      00000100&0101&01011\\[-1mm]
      00000010&0011&00111\\[-1mm]
      00000001&0000&11111
    \end{pmatrix}.
  $$
  Again we can prescribe the eight points spanned by the unit vectors, so that any further point is spanned by a vector $x\in\F_2^8$ with Hamming weight at least $4$. If there is no point spanned by a vector with Hamming weight $4$, than a small ILP computations shows that the maximum cardinality is $16$. So we can additionally prescribe an arbitrary point spanned by a vector with Hamming weight $4$. Another ILP computation shows that the maximum cardinality is $17$.
\end{proof}

\begin{table}[htp]
  \begin{center}
    \begin{tabular}{rrcc}
      \hline
      $s$ & $n_2(7,1,4;s)$ & construction & upper bound \\
      \hline
      3 &  11 & Lemma~\ref{lemma_n_2_7_4_3}  & Lemma~\ref{lemma_n_2_7_4_3} \\
      4 & 64 & $[7]-[6]$ & Griesmer upper bound \\
      5 & 65 & plus point & Griesmer upper bound \\
      6 & 96 & $[7]-[5]$ & Griesmer upper bound \\
      7 & 127 & $[7]$ & Griesmer upper bound \\
      8 & 128 & plus point & Griesmer upper bound \\
      9 & 129 & plus point & Griesmer upper bound \\
      10 & 160 & $2[7]-[6]-[5]$ & Griesmer upper bound \\
      \hline
    \end{tabular}
    \caption{Exact values for $n_2(7,1,4;s)$.}
    \label{table_n_2_7_4_s}
  \end{center}
\end{table}

\begin{lemma}
  \label{lemma_n_2_6_2_11}
  We have $n_2(6,1,2;11)=38$.
\end{lemma}
\begin{proof}
  The lower bound is given by the unique $[38,6,18]_2$ code \cite{bouyukliev2001optimal}, which is a Griesmer code. A generator matrix is e.g.\ given by
  $$
    \left(\begin{smallmatrix}
    00000100000000000000001111111111111111\\
    00001000000001111111110000000111111111\\
    00010000011110000111110001111000001111\\
    00100001101110011000110110011001110011\\
    01000010110111101011011010101010010101\\
    10000011111010110101011101001100110110
    \end{smallmatrix}\right).
  $$
  Let $\cM$ be a multiset of points in $\PG(5,2)$ with at most $11$ points in each solid.
  Considering projections through a point and $n_2(5,1,2;9)=33$ implies $\cM(P)\le 1$ for every point $P$.
  Since $n_2(5,1,1;11)=21$ we have
  $\cM(H)\le 21$ for every hyperplane $H$. Since $n_2(6,1,1;20)=38$ we can assume the existence of a hyperplane $H^\star$ with $\cM(H^\star)=21$. There are exactly two $[21,5,10]_2$ codes, see e.g.\ \cite[Theorem 6]{aggarwal2012characterisation}. Prescribing these two possibilities for $H^\star$ a small ILP computation quickly shows $\#\cM\le 38$.
\end{proof}
We remark that the Griesmer upper bound gives $n_2(6,1,2;11)\le 41$, while the coding upper bound yields $n_2(6,1,2;11)\le 39$ via
$n_2(6,1,1;21)=39$. Note that the complement of a hypothetical set of $39$ points in $\PG(5,2)$ points with at most $11$ points per solid is a multiset of points of cardinality $24$ that blocks every solid at least four times. The union of a solid and a plane (plus two arbitrary points) gives such an example as a multiset but not as a set of points.

\begin{proposition}
  If $s\ge 12$ or $s\in\{6,8,9,10\}$, then $n_2(6,1,2;s)$ is given by the Griesmer upper bound. Moreover, we have $n_2(6,1,2;4)=7$, $n_2(6,1,2;5)=11$,
  $n_2(6,1,2;7)=19$, and $n_2(6,1,2;11)=38$.
\end{proposition}
\begin{proof}
  Using Lemma~\ref{lemma_projective_base} and Lemma~\ref{lemma_n_2_6_2_11} we state $n_2(6,1,2;4)=7$ and $n_2(6,1,2;11)=38$, respectively.
  We consider the following constructions
  \begin{itemize}
     \item types $t[6]$, $t[6]+[1]$, $t[6]+2[1]$, $t[6]+3[1]$, $t[6]-[5]$, $t[6]-[5]+[1]$, $t[6]-[5]+2[1]$, $t[6]-[4]$, $t[6]-[4]+[1]$, and $t[6]-[3]$ for $t\ge 1$;
     \item types $t[6]-[5]-[4]$, $t[6]-[5]-[3]$, and $t[6]-[4]-[3]$ for $t\ge 2$; and
     \item type $t[6]-[5]-[4]-[3]$ for $t\ge 3$.
  \end{itemize}
  So, it remains to provide constructions for $s\in\{5,\dots,7,19\}$. For $s=7$ we add an arbitrary point to an example for $s=6$. For $s\in\{5,6,19\}$ we provide explicit examples:
  $$
    \left(\begin{smallmatrix}
    00000100111\\
    00001011111\\
    00010001011\\
    00100011100\\
    01000011001\\
    10000011010
    \end{smallmatrix}\right),
    \left(\begin{smallmatrix}
    000000100001111111\\
    000011000110001111\\
    000100001110110011\\
    001000011010011101\\
    010000011101000111\\
    100000010111101001
    \end{smallmatrix}\right)
    ,\text{ and}
  $$
  $$
    \left(\begin{smallmatrix}
    0001000000000000111111111111111000000000001111111111000000000001111111111\\
    0010000001111111000001111111111000000111110000001111000000111110000001111\\
    0000001110000111001110000011111010011001110001110011010011001110001110011\\
    0000010110111000110010001100111101100001110001111100101100001110001111100\\
    0100111011001001010110010101011000101010110110010101000101010110110010101\\
    1000110011010011100100100110001001011100011010111001001011100011010111001
    \end{smallmatrix}\right).
  $$
  The Griesmer upper bound is not attained for $s\in \{4,5,7,11\}$ while it is for all other cases $s>3$ For $s\in\{5,7\}$ the coding upper bound is attained.
\end{proof}
The stored generator matrices of $[n_2(6,1,2;s),6]_2$ codes in the database of \emph{best known linear codes} (BKLC) in \texttt{Magma} give optimal examples for $s\in\{5,6,11,19\}$.
Note that for $s=6$ we can take any $[18,6,8]_2$ Griesmer code and for $s=19$ we can take any $[73,6,36]_2$ Griesmer code.

\begin{table}[htp]
  \begin{center}
    \begin{tabular}{rrcc}
      \hline
      $s$ & $n_2(6,1,2;s)$ & construction & upper bound \\
      \hline
      4 &  7 & projective base & Lemma~\ref{lemma_projective_base} \\
      5 & 11 & BKLC/ILP & Coding upper bound \\
      6 & 18 & BKLC/ILP & Griesmer upper bound \\
      7 & 19 & plus point & Coding upper bound \\
      8 & 32 & $[6]-[5]$ & Griesmer upper bound \\
      9 & 33 & plus point & Griesmer upper bound \\
      10 & 34 & plus point & Griesmer upper bound \\
      11 & 38 & BKLC/ILP & Lemma~\ref{lemma_n_2_6_2_11} \\
      12  & 48 & $[6]-[4]$ & Griesmer upper bound \\
      13  & 49 & plus point & Griesmer upper bound \\
      14  & 56 & $[6]-[3]$& Griesmer upper bound \\
      15 & 63 & [6] & Griesmer upper bound \\
      16 & 64 & plus point & Griesmer upper bound \\
      17 & 65 & plus point & Griesmer upper bound \\
      18 & 66 & plus point  & Griesmer upper bound \\
      19 & 73 & BKLC/ILP & Griesmer upper bound \\
      20 & 80 & $2[6]-[5]-[4]$ & Griesmer upper bound \\
      21 & 81 & plus point & Griesmer upper bound \\
      22 & 88 & $2[6]-[5]-[3]$ & Griesmer upper bound \\
      23 & 95 & sum construction & Griesmer upper bound \\
      24 & 96 & plus point & Griesmer upper bound \\
      25 & 97 & plus point & Griesmer upper bound \\
      26 & 104 & $2[6]-[4]-[3]$ & Griesmer upper bound \\
      27 & 111 & sum construction & Griesmer upper bound \\
      28 & 112 & plus point & Griesmer upper bound \\
      29 & 119 & sum construction & Griesmer upper bound \\
      30 & 126 & sum construction & Griesmer upper bound \\
      31 & 127 & plus point & Griesmer upper bound \\
      32 & 128 & plus point & Griesmer upper bound \\
      33 & 129 & plus point & Griesmer upper bound \\
      34 & 136 & $3[6]-[5]-[4]-[3]$ & Griesmer upper bound \\
      \hline
    \end{tabular}
    \caption{Exact values for $n_2(6,1,2;s)$.}
    \label{table_n_2_6_2_s}
  \end{center}
\end{table}

\begin{lemma}
  \label{lemma_n_2_7_s_9_and_10}
  We have $n_2(7,1,2;9)=27$ and  $n_2(7,1,2;10)=28$.
\end{lemma}
\begin{proof}
  An example showing $n_2(7,1,2;9)\ge 27$ is given by $$
    \left(\begin{smallmatrix}
    000000100000000011111111111\\
    000001000001111100000111111\\
    000010001110111100011000011\\
    000100010111000101111000101\\
    001000011001001110101011100\\
    010000001011011010110101010\\
    100000010100110111010110001
    \end{smallmatrix}\right),
  $$
  so that adding an arbitrary point gives $n_2(7,1,2;10)\ge 28$. For $s=9$ the coding upper bound is attained. Next we consider a multiset $\cM$ of points in $\PG(6,2)$ such that every subspace of codimension two contains at most ten points. Starting from $n_2(6,1,1;10)=18$ we have used \texttt{LinCode} \cite{bouyukliev2021computer} to enumerate the two non-isomorphic $[18,6,8]_2$ codes.
Prescribing the two possible configurations and an arbitrary point making the span $7$-dimensional, we have used an ILP computation to conclude $\#\cM\le 28$. In the following we assume $\cM(H)\le 17$ for each hyperplane. Since $n_2(6,1,2;7)=19$ we also assume $\cM(P)\le 2$ for every point $P$. If there exists a solid $S$ with $\cM(S)\ge 7$, then we have $\#\cM \le 7\cdot 3+7=28$, so that we assume $\cM(S)\le 6$ for every solid $S$. We have used \texttt{LinCode} \cite{bouyukliev2021computer} to enumerate the four non-isomorphic $[10,5,4]_2$ codes.
  Prescribing the four possible configurations and two points that make the span $7$-dimensional, we have used ILP computations to conclude $\#\cM\le 28$.
\end{proof}

\begin{lemma}
  \label{lemma_n_2_7_s_20}
  We have $n_2(7,1,2;20)=71$.
\end{lemma}
\begin{proof}
  An example showing $n_2(7,1,2;20)\ge 71$ is given by
  $$
    \left(\begin{smallmatrix}
    00000010000000000000000000000000000000111111111111111111111111111111111\\
    00000100000000000000011111111111111111000000000000000011111111111111111\\
    00001000000011111111100000000011111111000000001111111100000000011111111\\
    00010000111100001111100000111100001111000011110000111100000111100001111\\
    00100001001101110001100011001100110011001100110011001100111001100110011\\
    01000001010110110010101101010101010101110101010101010101011010101010101\\
    10000001111011010100110110100110010110011010011001011010001011001101001
    \end{smallmatrix}\right).
  $$
  Next we consider a multiset $\cM$ of points in $\PG(6,2)$ such that every subspace of codimension two contains at most $20$ points. Starting from $n_2(6,1,1;20)=38$ we have used \texttt{LinCode} \cite{bouyukliev2021computer} to construct the unique  $[38,6,18]_2$ code \cite{bouyukliev2001optimal}.
  Prescribing the corresponding unqiue configuration and a further point that makes the span $7$-dimensional, we have used an ILP computation to conclude $\#\cM\le 71$. Observing $n_2(7,1,1;37)=71$ finishes the proof.
\end{proof}

\begin{lemma}
  \label{lemma_n_2_7_s_22_and_23}
  We have $n_2(7,1,2;22)=82$ and  $n_2(7,1,2;23)=83$.
\end{lemma}
\begin{proof}
  An example showing $n_2(7,1,2;22)\ge 82$ is given by each of the eleven $[82,7,40]_2$ Griesmer codes \cite{bouyukliev2001optimal}. One generator matrix is e.g.\ given by
  $$ 
    \left(\begin{smallmatrix}
    0000001000000000000000000000000000000000000111111111111111111111111111111111111111\\
    0000010000000000000000001111111111111111111000000000000000000011111111111111111111\\
    0000100000011111111111110000000111111111111000000011111111111100000000111111111100\\
    0001000011100000011111110001111000001111111000111100000001111100001111000001111100\\
    0010000101100011100011110110011000110000111011001100011110011100110011001110001100\\
    0100000110101100101100111010101001010011001101010101100110101101010101010110010100\\
    1000000111010101010101011101001010010101010110100110101010110110010110011010100100
    \end{smallmatrix}\right),
  $$
  so that adding an arbitrary point gives $n_2(7,1,2;23)\ge 83$. For $s=22$ the Griesmer upper bound is attained. Next we consider a multiset $\cM$ of points in $\PG(6,2)$ such that every subspace of codimension two contains at most $23$ points. Starting from $n_2(6,1;23)=34$ we have used \texttt{LinCode} \cite{bouyukliev2021computer} to construct the unique $[45,6,22]_2$ code \cite{bouyukliev2001optimal}.
  Prescribing the corresponding configuration and an arbitrary point making the span $7$-dimensional, we have used an ILP computation to conclude $\#\cM\le 83$. In the following we assume $\cM(H)\le 44$ for each hyperplane. Since $n_2(6,1,2;20)=80$ we also assume $\cM(P)\le 1$ for every point $P$. We have used \texttt{LinCode} \cite{bouyukliev2021computer} to enumerate the unique $[44,6,21]_2$ code with maximum column multiplicity one.
  Prescribing the corresponding configuration and an arbitrary point making the span $7$-dimensional, we have used an ILP computation to conclude $\#\cM\le 83$. In the following we assume $\cM(H)\le 43$ for each hyperplane.
  If there exists a solid $S$ with $\cM(S)\ge 13$, then we have $\#\cM \le 7\cdot 10+13=83$, so that we assume $\cM(S)\le 12$ for every solid $S$. We have used \texttt{LinCode} \cite{bouyukliev2021computer} to construct the unique $[23,11,5]_2$ code \cite{simonis200023}.
  Prescribing the corresponding configuration and two points that make the span $7$-dimensional, we have used an ILP computation to conclude $\#\cM\le 83$.
\end{proof}

\begin{proposition}
  If $s\ge 24$ or $s\in\{7,11,14,16,\dots,19,22\}$, then $n_2(7,1,2;s)$ is given by the Griesmer upper bound. Moreover, we have $n_2(7,1,2;5)=8$, $n_2(7,1,2;6)=12$, $n_2(7,1,2;8)=20$, $n_2(7,1,2;9)=27$, $n_2(7,1,2;10)=28$, $n_2(7,1,2;12)=36$, $n_2(7,1,2;13)=43$, $n_2(7,1,2;15)=51$, $n_2(7,1,2;20)=71$, $n_2(7,1,2;21)=75$, and $n_2(7,1,2;23)=83$.
\end{proposition}
\begin{proof}
  For $n_2(7,1,2;9)=27$ and $n_2(7,1,2;10)=28$ we refer to Lemma~\ref{lemma_n_2_7_s_9_and_10}. For $n_2(7,1,2;20)=71$ we refer to Lemma~\ref{lemma_n_2_7_s_20}.
  For $n_2(7,1,2;22)=82$ and $n_2(7,1,2;23)=83$ we refer to Lemma~\ref{lemma_n_2_7_s_22_and_23}.
  We consider the following constructions
  \begin{itemize}
     \item types $t[7]$, $t[7]-[6]$, $t[7]-[5]$, $t[7]-[4]$, and $t[7]-[3]$
           for $t\ge 1$;
     \item types $t[7]-[6]-[5]$, $t[7]-[6]-[4]$, $t[7]-[6]-[3]$,
          $t[7]-[5]-[4]$, $t[7]-[5]-[3]$, and $t[7]-[4]-[3]$ for $t\ge 2$;
     \item types $t[7]-[6]-[5]-[4]$, $t[7]-[6]-[5]-[3]$, $t[7]-[6]-[4]-[3]$,
           and $t[7]-[5]-[4]-[3]$ for $t\ge 3$; and
     \item type $t[7]-[6]-[5]-[4]-[3]$ for $t\ge 4$,
  \end{itemize}
  as well as adding up to four additional points to those constructions. By selecting the removed subspaces more carefully than by a chain, we can also have constructions for $[7]-[4]-[3]$, $2[7]-[6]-[4]-[3]$, $2[7]-[5]-[4]-[3]$, and $3[7]-[6]-[5]-[4]-[3]$, i.e.\ for $s\in\{27,43,51,67\}$, which meet the Griesmer upper bound.  The case $s=5$ is treated in Lemma~\ref{lemma_projective_base}. For $s\in\{8,12,15,38\}$ the best known construction can be obtained by adding an arbitrary point to a construction for $s-1$. For $s\in\{6,7,11,13,14,21\}$ 
  we give explicit examples found by ILP searches.
  $$
    \left(\begin{smallmatrix}
    000000100111\\
    000001001011\\
    000010010011\\
    000100011100\\
    001000011111\\
    010000001101\\
    100000011001
    \end{smallmatrix}\right),
    \left(\begin{smallmatrix}
    0000001000001111111\\
    0000010001110001111\\
    0000100110010010111\\
    0001000011100110011\\
    0010000010111100101\\
    0100000101101101001\\
    1000000011011011001
    \end{smallmatrix}\right),
    \left(\begin{smallmatrix}
    00000010000000000000111111111111111\\
    00000100000111111111000001111111111\\
    00001000011000011111001110000011111\\
    00010001111001100111010110001100001\\
    00100000101110100011011010010101110\\
    01000001011011111001100100101000110\\
    10000000110011101010100111000111010
    \end{smallmatrix}\right),
  $$
  $$
    \left(\begin{smallmatrix}
    0000001000000000000011111111111111111111111\\
    0000010000011111111100000000011111111111111\\
    0000100001100001111100001111100000001111111\\
    0001000110101110011100110111100011110001111\\
    0010000011010110001111001001100100110010111\\
    0100000111011000100101010010101001010111011\\
    1000000100101011100110111000110100010011101
    \end{smallmatrix}\right),
    \left(\begin{smallmatrix}
    00000010000000000000000001111111111111111111111100\\
    00000100000000111111111110000000000011111111111100\\
    00001001111111000000011110000000111100001111111100\\
    00010000000111000011101110011111000100010001111111\\
    00100000011001001100110110100011011101110010001111\\
    01000000101010011111000011101101001010110100010111\\
    10000001010010100101011010110100101111011000100111
    \end{smallmatrix}\right),
  $$
  $$ 
    \left(\begin{smallmatrix}
    000000100000000000000000000000000000001111111111111111111111111111111111111\\
    000001000000000000000111111111111111110000000000000000000111111111111111111\\
    000010000000001111111000000001111111110000000001111111111000000000111111111\\
    000100000111110001111000111110000111110000111110000111111000001111000011111\\
    001000001001110110011011000110011000110011000110011000111001110011001100111\\
    010000011110011010101101011010101011011101011010101001011010110101010101001\\
    100000011010101101001110101011001001100110101011001011101100010110011010011
   \end{smallmatrix}\right),
  $$
  For $s\in\{36,37,39\}$ we can take $[138,7,68]_2$, $[145,7,72]_2$, and $[153,7,76]_2$ Griesmer codes \cite{van1981smallest}.
  For $s\in\{6,8,12,13,15,21\}$ the coding upper bound is attained.
  All other upper bounds are given by the Griesmer upper bound.
\end{proof}

\begin{table}[htp]
  \begin{center}
    \begin{tabular}{rrcc}
      \hline
      $s$ & $n_2(7,1,2;s)$ & construction & upper bound \\
      \hline
      5 &  8 & projective base & Lemma~\ref{lemma_projective_base} \\
      6 & 12 & BKLC/ILP & Coding upper bound \\
      7 & 19 & BKLC/ILP & Griesmer upper bound \\
      8 & 20 & plus point & Coding upper bound \\
      9 & 27 & BKLC/ILP & Coding upper bound \\
      10 & 28 & plus point & Lemma~\ref{lemma_n_2_7_s_9_and_10} \\
      11 & 35 & BKLC/ILP & Griesmer upper bound \\
      12 & 36 & plus point & Coding upper bound \\
      13 & 43 & BKLC/ILP & Coding upper bound \\
      14 & 50 & BKLC/ILP & Griesmer upper bound \\
      15 & 51 & plus point & Coding upper bound \\
      16 & 64 & $[7]-[6]$ & Griesmer upper bound \\
      17 & 65 & plus point& Griesmer upper bound \\
      18 & 66 & plus point & Griesmer upper bound \\
      19 & 67 & plus point & Griesmer upper bound \\
      20 & 71 & BKLC/ILP & Lemma~\ref{lemma_n_2_7_s_20} \\
      21 & 75 & BKLC/ILP & Coding upper bound \\
      22 & 82 & BKLC/ILP & Griesmer upper bound \\
      23 & 83 & plus point & Lemma~\ref{lemma_n_2_7_s_22_and_23}\\
      24 & 96 & $[7]-[5]$ & Griesmer upper bound \\
      25 & 97 & plus point & Griesmer upper bound \\
      26 & 98 & plus point & Griesmer upper bound \\
      27 & 105 & $[7]-[4]-[3]$ & Griesmer upper bound \\
      28 & 112 & $[7]-[4]$ & Griesmer upper bound \\
      29 & 113 & plus point & Griesmer upper bound \\
      30 & 120 & $[7]-[3]$ & Griesmer upper bound \\
      31 & 127 & $[7]$ & Griesmer upper bound \\
      32 & 128 & plus point & Griesmer upper bound \\
      33 & 129 & plus point & Griesmer upper bound \\
      34 & 130 & plus point & Griesmer upper bound \\
      35 & 131 & plus point & Griesmer upper bound \\
      36 & 138 & BKLC/ILP & Griesmer upper bound \\
      37 & 145 & BKLC/ILP & Griesmer upper bound \\
      38 & 146 & plus point & Griesmer upper bound \\
      39 & 153 & BKLC/ILP & Griesmer upper bound \\
      40 & 160 & $2[7]-[6]-[5]$& Griesmer upper bound \\
      41 & 161 & plus point & Griesmer upper bound \\
      42 & 162 & plus point & Griesmer upper bound \\
      43 & 169 & $2[7]-[6]-[4]-[3]$ & Griesmer upper bound \\
      44 & 176 & $2[7]-[6]-[4]$ & Griesmer upper bound \\
      45 & 177 & plus point & Griesmer upper bound \\
      \hline
    \end{tabular}
    \caption{Exact values for $n_2(7,1,2;s)$ -- part 1.}
    \label{table_n_2_7_2_s_part1}
  \end{center}
\end{table}

\begin{table}[htp]
  \begin{center}
    \begin{tabular}{rrcc}
      \hline
      $s$ & $n_2(7,1,2;s)$ & construction & upper bound \\
      \hline
      46 & 184 & $2[7]-[6]-[3]$ & Griesmer upper bound \\
      47 & 191 & sum construction & Griesmer upper bound \\
      48 & 192 & plus point & Griesmer upper bound \\
      49 & 193 & plus point & Griesmer upper bound \\
      50 & 194 & plus point & Griesmer upper bound \\
      51 & 201 & $2[7]-[5]-[4]-[3]$ & Griesmer upper bound \\
      52 & 208 & $2[7]-[5]-[4]$ & Griesmer upper bound \\
      53 & 209 & plus point & Griesmer upper bound \\
      54 & 216 & $2[7]-[5]-[3]$ & Griesmer upper bound \\
      55 & 223 & sum construction & Griesmer upper bound \\
      56 & 224 & plus point & Griesmer upper bound \\
      57 & 225 & plus point & Griesmer upper bound \\
      58 & 232 & sum construction & Griesmer upper bound \\
      59 & 239 & sum construction & Griesmer upper bound \\
      60 & 240 & plus point & Griesmer upper bound \\
      61 & 247 & sum construction & Griesmer upper bound \\
      62 & 254 & sum construction & Griesmer upper bound \\
      63 & 255 & plus point & Griesmer upper bound \\
      64 & 256 & plus point & Griesmer upper bound \\
      65 & 257 & plus point & Griesmer upper bound \\
      66 & 258 & plus point & Griesmer upper bound \\
      67 & 265 & $3[7]-[6]-[5]-[4]-[3]$ & Griesmer upper bound \\
      68 & 272 & $3[7]-[6]-[5]-[4]$ & Griesmer upper bound \\
      69 & 273 & plus point & Griesmer upper bound \\
      70 & 280 & $3[7]-[6]-[5]-[3]$ & \\
      \hline
    \end{tabular}
    \caption{Exact values for $n_2(7,1,2;s)$ -- part 2.}
    \label{table_n_2_7_2_s_part2}
  \end{center}
\end{table}

The stored generator matrices of $[n_2(7,1,2;s),7]_2$ codes in the database of \emph{best known linear codes} (BKLC) in \texttt{Magma} give optimal examples for $s\in\{6,7,9,11,13,14,20,21,22,36,37,39\}$.

\begin{lemma}
  \label{lemma_n_2_7_3_s_4_and_5}
  We have $n_2(7,1,3;4)=9$ and $n_2(7,1,3;5)=19$.
\end{lemma}
\begin{proof}
  Examples showing $n_2(7,1,3;4)\ge 9$ and $n_2(7,1,3;5)\ge 19$
  are given by
  $$
    \begin{pmatrix}
    0&0&0&0&0&1&1&1&1\\[-1mm]
    0&0&0&1&1&0&0&1&1\\[-1mm]
    0&0&0&0&0&0&1&0&1\\[-1mm]
    0&0&0&0&1&1&1&0&1\\[-1mm]
    1&1&1&0&0&0&1&1&0\\[-1mm]
    0&1&1&0&0&0&0&0&0\\[-1mm]
    1&1&0&0&0&0&0&0&0
    \end{pmatrix}
  $$
  and
  $$
    \begin{pmatrix}
    1000000000111101101\\[-1mm]
    0100000101010111001\\[-1mm]
    0010000111100010011\\[-1mm]
    0001000110111000110\\[-1mm]
    0000100011011100011\\[-1mm]
    0000010100100111110\\[-1mm]
    0000001010010011111
    \end{pmatrix},
  $$
  respectively. For upper bounds are obtained by ILP computations prescribing the seven points generated by the unit vectors in $\F_2^7$.
\end{proof}

\begin{lemma}
  \label{lemma_n_2_7_3_s_6}
  We have $n_2(7,1,3;6)=28$.
\end{lemma}
\begin{proof}
  An example showing $n_2(7,1,3;6)\ge 28$ is given by $$
    \left(\begin{smallmatrix}
    0 0 0 0 0 0 0 0 0 0 0 0 1 1 1 1 1 1 0 0 0 0 1 1 1 1 1 1\\
    0 0 0 0 0 0 0 0 0 1 1 1 0 0 0 1 1 1 0 1 1 1 0 0 0 1 1 1\\
    0 0 0 1 1 1 1 1 1 0 0 0 0 0 0 1 1 1 0 0 0 1 0 0 1 0 1 1\\
    1 1 1 0 0 0 1 1 1 0 0 0 0 0 0 0 0 0 0 0 1 1 0 1 1 1 0 1\\
    0 0 0 1 1 1 0 0 0 0 0 0 1 1 1 1 1 1 1 1 1 1 0 0 0 0 0 0\\
    0 1 1 0 1 1 0 1 1 0 1 1 0 1 1 0 1 1 0 0 0 0 0 0 0 0 0 0\\
    1 1 0 1 1 0 1 1 0 1 1 0 1 1 0 1 1 0 0 0 0 0 0 0 0 0 0 0
\end{smallmatrix}\right).
  $$
  Let $\cM$ be a multiset of $n$ points in $\PG(6,2)$ such that every solid contains at most six points. We have used \texttt{LinCode} \cite{bouyukliev2021computer} to enumerate the two non-isomorphic $[18,6,8]_2$ codes.
  Prescribing the two possible configurations and an arbitrary point making the span $7$-dimensional, we have used an ILP computation to conclude $\#\cM\le 25$. In the following we assume $\cM(H)\le 17$ for each hyperplane.
  We have used \texttt{LinCode} \cite{bouyukliev2021computer} to enumerate the three non-isomorphic $[17,6,7]_2$ codes.
  Prescribing the three possible configurations and an arbitrary point making the span $7$-dimensional, we have used an ILP computation to conclude $\#\cM\le 28$. In the following we assume $\cM(H)\le 16$ for each hyperplane.
  We have used \texttt{LinCode} \cite{bouyukliev2021computer} to enumerate the four non-isomorphic $[10,5,4]_2$ codes.
  Prescribing the four possible configurations and two points that make the span $7$-dimensional, we have used ILP computations to conclude $\#\cM\le 28$.
\end{proof}

\begin{lemma}
  \label{lemma_n_2_7_3_s_11}
  We have $n_2(7,1,3;11)=72$.
\end{lemma}
\begin{proof}
  An example showing $n_2(7,1,3;11)\ge 72$ is given by $$
    \left(\begin{smallmatrix}
    100000011011100011101001110000110011101011010001010111001001011111111100\\
010000010110010010011101001000101010011110111001111100101101110010000010\\
001000001111011001101110001100001101111010010100011101010010111001000010\\
000100000010001101000111101110010110111101001010100111001101111000100010\\
000010000000110111110111000111000011001111001101011010100010111100010010\\
000001000001001010001111110011101001110110100110001110110101011100001010\\
000000111110100111011010010001010111110100100010111010010011110000000110\\
\end{smallmatrix}\right).
  $$
  Let $\cM$ be a multiset of $n$ points in $\PG(6,2)$ such that every solid contains at most eleven points. We have $\cM(H)\le n_2(6,1,2;11)=38$ for every hyperplane $H$, so that $\#\cM\le n_2(7,1,1;38)\le 72$.
\end{proof}

\begin{proposition}
  If $s\ge 12$ or $s\in\{8,9,10\}$, then $n_2(7,1,3;s)$ is given by the Griesmer upper bound. Moreover, we have $n_2(7,1,3;4)=9$, $n_2(7,1,3;5)=19$, $n_2(7,1,3;6)=28$, $n_2(7,1,3;7)=35$, and  $n_2(7,1,3;11)=72$.
\end{proposition}
\begin{proof}
  For $n_2(7,1,3;4)=9$ and $n_2(7,1,3;5)=19$ we refer to Lemma~\ref{lemma_n_2_7_3_s_4_and_5}. For $n_2(7,1,3;6)=28$ we refer to Lemma~\ref{lemma_n_2_7_3_s_6}. For $n_2(7,1,3;11)=72$ we refer to Lemma~\ref{lemma_n_2_7_3_s_11}.
  We consider the following constructions
  \begin{itemize}
     \item types $t[7]$, $t[7]-[6]$, $t[7]-[5]$, and $t[7]-[4]$
           for $t\ge 1$;
     \item types $t[7]-[6]-[5]$, $t[7]-[6]-[4]$, $t[7]-[6]$,
          $t[7]-[5]-[4]$, and  $t[7]-[5]$ for $t\ge 2$,
  \end{itemize}
  as well as adding up to four additional points to those constructions. Examples showing $n_2(7,1,3;7)\ge 35$ 
  given by
  $$
    \left(\begin{smallmatrix}
    1 0 0 1 0 1 0 1 1 0 1 0 1 0 1 1 0 1 0 1 1 1 0 1 0 1 1 0 0 0 0 1 1 1 1\\
    0 1 0 1 0 0 0 1 0 1 1 0 1 1 0 0 0 1 1 1 0 0 0 0 1 0 1 0 1 1 1 0 1 0 0\\
    0 0 1 1 0 1 0 0 0 0 1 0 0 1 0 0 0 0 1 1 0 0 0 1 1 1 0 1 0 0 1 1 1 1 1\\
    0 0 0 0 1 1 0 0 1 1 1 0 0 0 1 0 0 0 0 1 0 0 1 0 1 1 1 0 1 1 0 1 1 0 1\\
    0 0 0 0 0 0 1 1 1 1 1 0 0 0 0 1 0 0 0 0 1 0 0 1 1 1 1 1 0 1 1 0 1 1 0\\
    0 0 0 0 0 0 0 0 0 0 0 1 1 1 1 1 0 0 0 0 0 1 1 1 1 1 1 1 1 0 0 0 1 1 1\\
    0 0 0 0 0 0 0 0 0 0 0 0 0 0 0 0 1 1 1 1 1 1 1 1 1 1 1 1 1 1 1 1 0 0 0
    \end{smallmatrix}\right).
  $$
  For $m_2(7,1,3;19)\ge 145$ we can use a $[145,7,72]_2$ Griesmer code \cite{van1981smallest}.
  The upper bound $n_2(7,1,3;7)\le 35$ is given by the coding upper bound. All other upper bounds are obtained from the Griesmer upper bound.
\end{proof}

The stored generator matrices of $[n_2(7,1,3;s),7]_2$ codes in the database of \emph{best known linear codes} (BKLC) in \texttt{Magma} give optimal examples for $s\in\{5,7,19\}$.

\begin{table}[htp]
  \begin{center}
    \begin{tabular}{rrcc}
      \hline
      $s$ & $n_2(7,1,3;s)$ & construction & upper bound \\
      \hline
       4 & 9 & ILP & Lemma~\ref{lemma_n_2_7_3_s_4_and_5} \\
       5 & 19 & BKLC/ILP & Lemma~\ref{lemma_n_2_7_3_s_4_and_5} \\
       6 & 28 & ILP & Lemma~\ref{lemma_n_2_7_3_s_6} \\
       7 & 35 & BKLC/ILP & Coding upper bound \\
       8 & 64 & $[7]-[6]$ & Griesmer upper bound \\
       9 & 65 & plus point & Griesmer upper bound \\
      10 & 66 & plus point & Griesmer upper bound \\
      11 & 72 & Lemma~\ref{lemma_n_2_7_3_s_11} & Lemma~\ref{lemma_n_2_7_3_s_11} \\
      12 & 96 & $[7]-[5]$ & Griesmer upper bound \\
      13 & 97 & plus point & Griesmer upper bound \\
      14 & 112 & $[7]-[4]$ & Griesmer upper bound \\
      15 & 127 & $[7]$ & Griesmer upper bound \\
      16 & 128 & plus point & Griesmer upper bound \\
      17 & 129 & plus point & Griesmer upper bound \\
      18 & 130 & plus point & Griesmer upper bound \\
      19 & 145 & BKLC/ILP & Griesmer upper bound \\
      20 & 160 & $2[7]-[6]-[5]$ & Griesmer upper bound \\
      21 & 161 & plus point & Griesmer upper bound \\
      22 & 176 & $2[7]-[6]-[4]$ & Griesmer upper bound \\
      23 & 191 & sum construction & Griesmer upper bound \\
      24 & 192  & sum construction & Griesmer upper bound \\
      25 & 193 & sum construction & Griesmer upper bound \\
      26 & 208 & $2[7]-[5]-[4]$ & Griesmer upper bound \\
      \hline
    \end{tabular}
    \caption{Exact values for $n_2(7,1,3;s)$.}
    \label{table_n_2_7_3_s}
  \end{center}
\end{table}

\begin{proposition}
  If $s\ge 3$ , then $n_3(4,1,2;s)$ is given by the Griesmer upper bound. Moreover, we have $n_3(4,1,2;2)=10$.
\end{proposition}
\begin{proof}
  The upper bound $n_3(4,1,2;2)\le 10$ follows from the coding upper bound and all other upper bounds follow from the Griesmer upper bound. The existence of an ovoid in $\PG(3,3)$ yields $n_3(4,1,2;2)\ge 10$. A $[27,4,18]_3$ Griesmer code yields $n_3(4,1,2;3)\ge 27$. Note that Griesmer $[n,4,d]_3$ codes exist for all $d\ge 16$.
\end{proof}

We remark that we have $n_q(4,1,2;2)=q^2+1$ for all $q>2$ \cite{bose1947mathematical,qvist1952some}. Moreover, we have $n_4(4,1,2;3)=31$ and $n_5(4,1,2;3)=44$ \cite{edel2010multiple}.

\begin{lemma}
  \label{lemma_3_5_2_4}
  We have $n_3(5,1,2;4)=20$.
\end{lemma}
\begin{proof}
  An example showing $n_3(5,1,2;4)\ge 20$ is given by
  $$
    \begin{pmatrix}
    10000220001102111221\\[-1mm]
    01000002221001121111\\[-1mm]
    00100212202022200211\\[-1mm]
    00010202111222011010\\[-1mm]
    00001111021120010111
    \end{pmatrix}.
  $$
  After prescribing the unique $[10,4;6]_3$ code a small ILP computation verifies $n_3(5,1,2;4)\le 20$.
\end{proof}

\begin{lemma}
  \label{lemma_3_5_2_6}
  We have $n_3(5,1,2;6)=38$.
\end{lemma}
\begin{proof}
  An example showing $n_3(5,1,2;6)\ge 38$ is given by
  $$
    \left(\begin{smallmatrix}
    1 0 0 0 0 0 1 2 0 1 0 1 1 1 2 2 1 2 0 2 1 1 0 0 0 2 0 0 1 2 1 2 2 0 1 2 1 1\\
    0 1 0 0 0 2 1 2 1 2 2 1 2 2 0 1 0 2 2 2 2 1 2 2 1 0 1 0 1 0 1 1 0 1 0 1 1 1\\
    0 0 1 0 0 2 0 1 2 0 0 1 2 1 1 2 0 1 0 2 2 0 1 1 2 1 1 2 0 2 0 0 1 0 1 2 2 1\\
    0 0 0 1 0 1 2 0 1 2 0 1 0 1 2 2 1 2 2 1 2 0 1 0 2 2 0 2 0 1 2 2 1 1 0 0 0 2\\
    0 0 0 0 1 2 2 2 2 2 2 2 2 2 2 2 2 2 2 2 0 1 0 1 1 1 2 2 2 1 0 2 2 2 1 2 0 2]\end{smallmatrix}\right).
  $$
  After prescribing the three non-equivalent  $[15,4;9]_3$ codes  small ILP computations verify $n_3(5,1,2;6)\le 38$.
\end{proof}

\begin{lemma}
  \label{lemma_3_5_2_11}
  We have $n_3(5,1,2;11)=91$.
\end{lemma}
\begin{proof}
  An example showing $n_3(5,1,2;6)\ge 38$ is given by the $[91,5,60]_3$ code in the database of \emph{best known linear codes} (BKLC) in \texttt{Magma}.
  After prescribing the unique non-equivalent  $[32,4;21]_3$ code an ILP computation verifies $n_3(5,1,2;11)\le 91$.
\end{proof}

\begin{lemma}
  \label{lemma_3_5_2_17}
  We have $n_3(5,1,2;17)=143$.
\end{lemma}
\begin{proof}
  An example showing $n_3(5,1,2;17)\ge 143$ is given by $$
  \left(\begin{smallmatrix}
  1111111111111111110000000001000000000000000000000011111111111111111111111111\\1111111111111111111111000000111111111110011111110000001111111111100\\
0000011111122222221111111100100000000011111111111100000000000000111111111111\\1122222222222220001221001111000111222211000012210011110001112222110\\
0011200112200011220011122210010011111100001111122200001111222222000001111222\\2200001111222221220122010122001012002110012201220101220010120021100\\
1202012120101212021201201210001100112201120012201201120112001122001220122011\\2200121122011221120012022112020110120000111200120221120201101200001\\
0000000000000000000000000000000212120121221211222221211122121212121122112112\\1212111212212120000000121211122222221201100000001212111222222212011,
  \end{smallmatrix}\right)
  $$
  where each pair or rows has to be read as a single row of a generator matrix.
  After prescribing the two non-equivalent  $[32,4;21]_3$ codes  ILP computations verify $n_3(5,1,2;17)\le 143$.
\end{proof}

\begin{proposition}
  If $s\ge 18$ or $s\in\{7,9,10,12,13,14,15,16\}$, then $n_3(5,1,2;s)$ is given by the Griesmer upper bound. Moreover, we have $n_3(5,1,2;3)=11$, $n_3(5,1,2;4)=20$, $n_3(5,1,2;5)=29$, $n_3(5,1,2;6)=38$, $n_3(5,1,2;8)=56$, $n_3(5,1,2;11)=91$, and $n_3(5,1,2;17)=143$.
\end{proposition}
\begin{proof}
  For $n_3(5,1,2;4)=20$ we refer to Lemma~\ref{lemma_3_5_2_4} and for $n_3(5,1,2;6)=38$ we refer to Lemma~\ref{lemma_3_5_2_6}. For $n_3(5,1,2;11)=91$ we refer to Lemma~\ref{lemma_3_5_2_11} and for $n_3(5,1,2;17)=143$ we refer to Lemma~\ref{lemma_3_5_2_17}.
  For $s\in\{7,9,12,13,16\}$ the existence of $[55,5,36]_3$, $[81,5,54]_3$, $[108,5,72]_3$, $[121,5,81]_3$, and $[136,5,90]_3$ Griesmer codes yields the lower bounds.
  For $s\in \{10,14,15\}$ the lower bound is attained by adding arbitrary points.  Also $n_3(5,1,2;8)\ge 56$ is given by adding a point.
  Since Griesmer $[n,5,d]_3$ codes do exist for all $d\ge 100$, $n_3(5,2;s)$ is given by the Griesmer upper bound for all $s\ge 18$.
  For $s\in\{3,4,5,6,11\}$ the stored generator matrices of $[n,5]_3$ codes in the database of \emph{best known linear codes} (BKLC) in \texttt{Magma}
  yield $n_3(5,1,2;3)\ge 11$,  $n_3(5,1,2;4)\ge 20$, $n_3(5,1,2;5)\ge 29$, $n_3(5,1,2;6)\ge 38$, and $n_3(5,1,2;11)\ge 91$, respectively.

  The coding upper bound gives $n_3(5,1,2;3)\le 11$, 
  $n_3(5,1,2;5)\le 29$, and  
  $n_3(5,1,2;8)\le 56$. 
  All other upper bounds are given by the Griesmer upper bound.
\end{proof}

\begin{proposition}
  If $s\ge 3$, then $n_3(5,1,3;s)$ is given by the Griesmer upper bound. Moreover, we have $n_3(5,1,3;2)=20$.
\end{proposition}
\begin{proof}
  For $n_3(5,1,3;2)=20$ we refer e.g.\ to \cite{hill1983pellegrino}.\footnote{Up to projective equivalence there are exactly nine different examples.}
  The $[81,5,54]_3$ and the $[121,5,81]_3$ Griesmer codes give examples for $n_3(5,1,3;3)\ge 81$ and $n_3(5,1,3;4)\ge 121$, respectively. $n_3(5,3;5)\ge 122$ is obtained by adding an arbitrary point. The other lower bounds follow from the fact that Griesmer $[n,5,d]_3$ codes exist for all $d\ge 100$. The upper bounds are given by the Griesmer upper bound except for $s=2$.
 \end{proof}

We have $n_4(5,1,3;2)=41$ \cite{edel199941}\footnote{Up to symmetry there exist two $41$-caps in $\PG(4,4)$.} and $n_3(6,1,4;2)=56$ \cite{hirschfeld1985finite}.\footnote{Up to symmetry there exists a unique $56$-cap in $\PG(5,3)$.}

\pagebreak

\section{Almost affine codes}
\label{sec_almost_affine_codes}

Let $C$ be an $[n,k]_q$ code. It is well known that for each subset $S\subseteq\{1,\dots,n\}$ the projection $C_S$ of $C$ into any coordinate space $\F_q^S:=\prod_{i\in S} \F_q$ is again a linear code. Especially, we have that $\# C_S$ is a power of $q$. \emph{Affine codes}, i.e.\ the cosets or translates of a linear subspace of $\F_q^n$ have the same property that all projections have a size which is a power of the alphabet size. In \cite{simonis1998almost} it was shown that for $q\in\{2,3\}$ there are no other possibilities. However, this changes if we increase the alphabet size.
\begin{definition}
 A block code $C\subseteq \cA^n$ is called \emph{almost affine} if it satisfies the condition
  \begin{equation}
   r(S):=\log_{\#\cA} \# C_S\,\,\in\,\,\N
 \end{equation}
  for all $S\subseteq \{1,\dots,n\}$.
\end{definition}
So, for $\cA=\F_q$ affine codes are almost affine codes. Note that $r(S)$ can be considered as a \emph{rank function} turning $C$ into a \emph{matroid} $M_C$ \cite{simonis1998almost}. In general, for each matroid $M$ a Hamming weight and a Hamming distance can be defined \cite{johnsen2013hamming}. Generalized Hamming weights for almost affine codes are considered in \cite{johnsen2017generalized}.

A block code $C$ is called \emph{quasi-uniform} if every alphabet symbol occurs for each coordinate in the same number of codewords. We remarks that almost affine codes are quasi-uniform \cite[Proposition 2]{chan2013quasi}. A quasi-uniform code is also distance invariant, see \cite{chan2013quasi} for a definition and proof.

Since an additive code does not have to have a cardinality with is a power of the alphabet size, they are not almost affine codes in general. However, they satisfy a relaxed version:
\begin{definition}
 A block code $C\subseteq \cA^n$ is called \emph{almost subalphabet affine} if there exists an integer $q>1$ such that $\#\cA$ is a power of $q$ and we have
  \begin{equation}
   r(S):=\log_{q} \# C_S\,\,\in\,\,\N
 \end{equation}
  for all $S\subseteq \{1,\dots,n\}$.
\end{definition}
Clearly, $r(S)$ can be considered as a rank function any many properties of affine codes should transfer.

\end{document}